\documentclass[prb,aps,twocolumn,10pt,floatfix,nofootinbib,superscriptaddress]{revtex4-2}
\usepackage{mystyle}

\title{Delicate topology protected by rotation symmetry: Crystalline Hopf insulators and beyond}

\begin{document}

\title{Delicate topology protected by rotation symmetry: Crystalline Hopf insulators and beyond}

\author{Aleksandra Nelson\,\orcidlink{0000-0002-2168-3846}}
\email{anelson@physik.uzh.ch}
\affiliation{Department of Physics, University of Zurich, Winterthurerstrasse 190, 8057 Zurich, Switzerland}
\author{Titus Neupert\,\orcidlink{0000-0003-0604-041X}}
\email{titus.neupert@uzh.ch}
\affiliation{Department of Physics, University of Zurich, Winterthurerstrasse 190, 8057 Zurich, Switzerland}
\author{A. Alexandradinata\,\orcidlink{0000-0003-2235-497X}}    \email{aalexan6@ucsc.edu}
\affiliation{Physics Department, University of California Santa Cruz, Santa Cruz, CA 95064, USA}
\author{Tom\'{a}\v{s} Bzdu\v{s}ek\,\orcidlink{0000-0001-6904-5264}}
\email{tomas.bzdusek@psi.ch}
\affiliation{Condensed Matter Theory Group, Paul Scherrer Institute, 5232 Villigen PSI, Switzerland}
\affiliation{Department of Physics, University of Zurich, Winterthurerstrasse 190, 8057 Zurich, Switzerland}

\begin{abstract}
Pontrjagin's seminal topological classification of two-band Hamiltonians in three momentum dimensions is hereby enriched with the inclusion of a crystallographic rotational symmetry. 
The enrichment is attributed to a new topological invariant which quantifies a $2\pi$-quantized change in the Berry-Zak phase between a pair of rotation-invariant lines in the bulk, three-dimensional Brillouin zone; 
because this change is reversed on the complementary section of the Brillouin zone, we refer to this new invariant as a returning Thouless pump (RTP). 
We find that the RTP is associated to anomalous values for the angular momentum of surface states, which guarantees metallic in-gap states for open boundary condition with sharply terminated hoppings;
more generally for arbitrarily terminated hoppings, surface states are characterized by Berry-Zak phases that are quantized to a rational multiple of $2\pi$.  
The RTP adds to the family of topological invariants (the Hopf and Chern numbers) that are known to classify two-band Hamiltonians in Wigner-Dyson symmetry class~A. 
Of these, the RTP and Hopf invariants are \textit{delicate}, meaning that they can be trivialized by adding  a particular trivial band to either the valence or the conduction subspace. 
Not all trivial band additions will nullify the RTP invariant, which allows its generalization beyond two-band Hamiltonians to arbitrarily many bands; such generalization is a hallmark of \textit{symmetry-protected} delicate topology.
\end{abstract}
\date{\today}

\maketitle

\setcounter{tocdepth}{2}

\renewcommand{\baselinestretch}{1.25}
\makeatletter
{\small \tableofcontents }
\makeatother
\renewcommand{\baselinestretch}{1.0}

\section{Introduction}\label{sec:intro}

One of the central problems in topology  -- to determine equivalence classes of maps between two manifolds~\cite{Hatcher:2003,Hopf:1931,pontrjagin_classification,Whitehead:1947} -- has found diverse applications in physics~\cite{wilczekzee_linkingnumbers,stone_mathematicsforphysics,Radu:2008,Ackerman:2017,Mermin:1979}. 
We focus here on an application in condensed matter physics, where the momentum ($\bk$)-dependent Bloch Hamiltonian can be viewed as a map from the Brillouin zone (BZ) to a classifying space of Hamiltonians~\cite{kitaev_periodictable,Ryu:2010,read_compactwannier}. 
One of the seminal results is Pontrjagin's topological classification~\cite{pontrjagin_classification} of two-band, insulating Hamiltonians~\cite{Hopfinsulator_Moore, hopfsurfacestates_deng,kennedy_hopfchern} in Wigner-Dyson symmetry class~A~\cite{altland1997} and in three momentum dimensions. 
[By `two-band, insulating', we mean the Hamiltonian has a spectral gap throughout the Brillouin zone, with a valence subspace (also a conduction subspace) given by one linearly-independent Bloch function.]
The classification is given by four quantized invariants~\cite{Auckly:2010,kennedy_thesis}, three of which are Chern invariants defined on two-dimensional sub-tori of the 3-torus, while the fourth can be identified as the integer ($\mathbb{Z}$) valued Hopf invariant if the Chern class is trivial. 
However, Pontrjagin's classification has never been generalized to include crystallographic point-group symmetries, i.e., crystalline symmetries beyond discrete translations. 

Crystallographic point-group symmetries have had a transformative role in the topological classification of band structures, rendering some topological invariants to vanish~\cite{Chen_bulktopologicalinvariants}, as well as introducing new topological invariants~\cite{Hsieh_SnTe,Chen_bulktopologicalinvariants,AA_wilsonloopinversion,Shiozaki2014,Hourglass,benalcazar_quantizedmultipole,Schindler_higherorderTI}. 
One approach to characterize band topology is based on the analysis of crystalline symmetry representations of one-particle Bloch states~\cite{TQC, bandcombinatorics_kruthoff,Po_symmetryindicators,po_fragile}.
This approach led to a successful prediction of numerous new topological materials~\cite{Fang_completecatalog,Maia_completecatalog,Xiangang_completecatalog}; however, by design, these symmetry-representation-based filtering schemes do not identify topologically nontrivial band structures that are \textit{band representations}~\cite{zak1980,Zak_bandrepresentations}, i.e., band structures admitting a description by a set of exponentially-localized, locally-symmetric crystalline (Wannier) orbitals~\cite{nogo_AAJH,bouhon_wilsonloopapproach,crystalsplit_AAJHWCLL}, of which the Hopf insulator~\cite{Hopfinsulator_Moore} is an example. 

As was recently revealed~\cite{Nelson:2021,AA_teleportation}, the presence of crystalline symmetry also extends the characterization of two-band, insulating Hamiltonians with trivial first Chern class beyond the $\mathbb{Z}$-valued Hopf invariant. 
In particular, if the valence and the conduction bands transform differently under rotation, additional $\mathbb{Z}$-valued topological invariants can be defined: 
each of them is given by a $2\pi$-quantized difference in Berry-Zak phase  between two rotation-invariant lines (in the Brillouin zone) separated by half a reciprocal-lattice vector, with the Zak phase reverting to its original value over a full reciprocal vector. 
By viewing two of three independent momentum coordinates as adiabatic parameters, this Zak-phase reversion can be interpreted as an adiabatic evolution of the electric polarization defined for a single-momentum Hamiltonian: a unit charge is pumped across one (or more) real-space lattice periods within half an adiabatic cycle, and is then pumped back in the next half of the cycle. 
In analogy with the quantized charge pumping first investigated by Thouless \cite{thouless_pump}, we call the Zak-phase reversion a returning Thouless pump (RTP). We refer to the corresponding rotation-symmetry-enriched topological two-band Hamiltonians as 
\emph{crystalline Hopf insulators}. 

In this work, we classify   two-band, insulating Hamiltonians  (with trivial Chern class) by the Hopf and the rotation-protected RTP invariants, for all crystallographically allowed orders of rotation: $n=2,3,4,6$.
We discover that these invariants are not independent: the RTP invariants determine the Hopf invariant modulo $n$, resembling in spirit a known relation between rotation-symmetry representations and the Chern number~\cite{Chen_bulktopologicalinvariants}. 
At the same time, the RTP invariants characterize qualitatively distinct phases beyond the Pontrjagin classification; as proof of principle we will demonstrate a model that carries vanishing Chern and Hopf invariants but nonvanishing RTP invariants. 

A RTP leads to potentially measurable surface signatures which are qualitatively unique among the known topological insulators. 
Fundamentally, the RTP leads to \textit{anomalous} values for the angular momentum of \textit{all} surface states that are localized to a single, rotation-invariant crystalline \textit{facet}. 
By `facet', we mean a crystalline surface retaining two of three independent translational symmetries of the bulk space group; by `all surface states', we mean all surface-localized states irrespective of whether they are occupied or unoccupied by fermions; by `anomalous', we mean that these angular-momentum values are non-identical to the angular momentum of bulk valence states, and also non-identical to the angular momentum of bulk conduction states.  
Such an \textit{angular-momentum anomaly} is physically manifested in two~ways.\\

\noindent (\emph{i}) For open boundaries with sharply-terminated hoppings, the anomaly implies the existence of facet-localized states along any line in the  reduced/surface Brillouin zone connecting a certain pair of rotation-invariant reduced momenta. These states energetically connect the bulk valence and bulk conduction bands in a manner that  is robust against continuous, gap-preserving deformations of the \textit{bulk} Hamiltonian; however, the energetic connection crucially relies on the \textit{sharpness} of the boundary condition, namely that Hamiltonian matrix elements are everywhere identical as in the bulk, except for those hopping matrix elements crossing a facet; these exceptions are set to zero. 
Such \textit{conditionally robust} surface states have no analog in stable~\cite{TKNN,kane2005B,Hourglass} and fragile~\cite{po_fragile,nogo_AAJH,bouhon_wilsonloopapproach,crystalsplit_AAJHWCLL,Bouhon:2020} topological insulators, and are unique to delicate topological insulators (a notion introduced in Ref.~\cite{Nelson:2021} and elaborated below).
\\

\noindent (\emph{ii}) For a general choice of the Hamiltonian termination, facet-localized states are characterized by Berry-Zak phases that are quantized to a rational multiple of $2\pi$, which is distinct from the Berry-Zak phases of bulk states in both the bulk conduction and the bulk valence bands. 
How surface states can be topologically distinct from bulk states may be rationalized in a slab geometry with two facets: the Zak phase of top-facet-localized states are `conjugate' to the Zak phase of the bottom-facet-localized states, such that the totality of states localized to both facets are topologically equivalent to the bulk states. 
Colloquially, one may say that topologically `interesting' bands have arisen from being \textit{physically} separated across the bulk of the slab; such a hallmark of delicate topological insulators stands in stark contrast to the conventional wisdom of \textit{energetically} separating topologically-interesting bands across a bulk energy gap.\\

Unlike stable and fragile topology, both Hopf and RTP invariants are delicate~\cite{Nelson:2021}, meaning that they are not stable under addition of a particular trivial band to either the bulk conduction or the bulk valence subspace.
Addition of any trivial band makes the Hopf invariant ill-defined, however the RTP invariants remain well-defined and quantized upon addition of trivial bands with certain angular momenta. 
Such symmetry-constrained generalization of the RTP to Hamiltonians with an arbitrary number of bands is a distinctive feature of \textit{symmetry-protected delicate topology}, whose existence improves the odds of discovering delicate topology in realistic materials.

\section{Organization of the manuscript}

{\centering
\textbf{Berry dipole as origin of crystalline \nopagebreak\\ Hopf insulators: Sec.~\ref{sec:berry-dipole}--\ref{sec:lattice-RTP}.}
\par}\nopagebreak\smallskip\nopagebreak

The crystalline Hopf insulator originates from a unique semi-metallic phase -- characterized by a band-touching point that acts as a rotation-symmetric, dipolar source of Berry curvature. By elucidating the geometric-topological properties of this \textit{Berry dipole} in Sec.~\ref{sec:berry-dipole}, we can explain (in Sec.~\ref{sec:lattice-RTP}) why a Berry dipole mediates a topological phase transition from a trivial insulator to a crystalline Hopf insulator. 

It is also in  Sec.~\ref{sec:lattice-RTP} that we formalize the notion of the returning Thouless pump (RTP) and explain why it is quantized. Together with the Hopf invariant (reviewed therein), the RTP invariant classifies two-band, insulating, tight-binding Hamiltonians with trivial Chern class. A substantial amount of basic band-theoretic vocabulary (including the notion of band representations) is also reviewed in this section, to streamline subsequent reading.\\

{\centering
\textbf{Models of crystalline Hopf insulators: Sec.~\ref{sec:models}.}
\par}\nopagebreak\smallskip\nopagebreak

Beyond its explanatory power, a Berry dipole is also a theoretically useful strategy for inventing tight-binding models of crystalline Hopf insulators -- a program carried out in Sec.~\ref{sec:models}. As proof of principle, we invent representative models of crystalline Hopf insulators for all $\mathrm{P}n$ space groups\footnote{A $\mathrm{P}n$ space group,  with $n=2,3,4,6$,  is a semidirect product of a three-dimensional translational group and an order-$n$ rotational group.}
(as summarized   in Table \ref{tab:model-zoo}), illustrating how crystallographic symmetry enriches the topological classification of two-band Hamiltonians; in particular, we show explicitly a model with a trivial Hopf invariant but a nontrivial RTP invariant. \\

{\centering
\textbf{Complete classification of crystalline \nopagebreak\\ Hopf insulators: Sec.~\ref{sec:RTP-Hopf}.}
\par}\nopagebreak\smallskip\nopagebreak

The Hopf and RTP invariants are not independent but are related modulo $n$, as shown in Sec.~\ref{sec:RTP-Hopf} for all $\mathrm{P}n$ space groups, with the main result summarized in \s{sec:mainresult}. Without yet a rigorous algebraic-topological proof, we conjecture that this Hopf-RTP relation gives the complete classification of $\mathrm{P}n$-symmetric, two-band insulating Hamiltonians with trivial Chern class. \\

{\centering
\textbf{Bulk-boundary correspondence: Sec.~\ref{sec:BBC}.}
\par}\nopagebreak\smallskip\nopagebreak

We establish the boundary signatures of the bulk RTP invariant  in Sec.~\ref{sec:BBC}. The fundamental signature is the angular-momentum anomaly (described briefly above); the more physically accessible signatures are the conditionally-robust surface states as well as the anomalous values of the Zak phase of surface bands. Our result here, combined with the boundary signature of the bulk Hopf invariant established in \ocite{AA_teleportation}, completes the bulk-boundary correspondence of crystalline Hopf insulators. Because the angular-momentum anomaly offers one way to prove the mod-$n$ Hopf-RTP relation, we have chosen to present the bulk-boundary correspondence section (Sec.~\ref{sec:BBC}) before the above-mentioned Hopf-RTP section (Sec.~\ref{sec:RTP-Hopf}). However, there is no substantial loss in switching the order of reading, because we have found other means to independently prove the Hopf-RTP relation.  \\

{\centering
\textbf{Symmetry-protected delicate topology beyond \nopagebreak\\ two-band Hamiltonians: Sec.~\ref{sec:multi-band}.}
\par}\nopagebreak\smallskip\nopagebreak

We discuss the extent to which the Hopf and RTP invariants can be generalized beyond two-band Hamiltonians in Sec.~\ref{sec:multi-band}, which naturally leads to the related notions of delicate vs.~symmetry-protected delicate topology. We re-establish a known result that the Hopf invariant is delicate-topological and cannot be extended to (${>}{2}$)-band Hamiltonians; however, we highlight that the RTP is a symmetry-protected delicate topological invariant that can be extended to (${>}{2}$)-band Hamiltonians, with the caveat that the expanded Hilbert space satisfies certain symmetry constraints.\\

{\centering
\textbf{Multicellularity of crystalline Hopf insulators: Sec.~\ref{sec:CHI-multicel}.}
\par}\smallskip

A real-space implication of a nontrivial RTP or a nontrivial Hopf invariant is the impossibility to confine a representative Wannier orbital (of either the valence or the conduction band) to a single primitive unit cell. This obstruction exists despite all Wannier orbitals being exponentially-localized and locally-symmetric. 
Such `multicellularity' is investigated in Sec.~\ref{sec:CHI-multicel} for various models of crystalline Hopf insulators, and proven more generally to be a (symmetry-protected) delicate-topological property.
Going beyond the terminology introduced in our earlier work~\cite{Nelson:2021}, we here refine the notion of multicellularity through the introduction of \emph{grading}: a positive-integer-valued number that places a lower bound on the number of unit cells into which representative Wannier orbitals
can be contained. \\

{\centering
\textbf{Topological semimetals intermediate \nopagebreak\\ crystalline Hopf insulators: Sec.~\ref{sec:semi-metallic-transition-region}.}
\par}\nopagebreak\smallskip\nopagebreak

In spite of the conceptual simplicity of the Berry dipole, it is worth remarking that such a band touching is fine-tuned, i.e., occupying a region in the phase diagram which is `harder to reach' than generic band touchings for $\mathrm{P}n$-symmetric Hamiltonian. Aiming toward more naturalistic models of delicate topology,  Sec.~\ref{sec:semi-metallic-transition-region} investigates the more generic semimetallic phases (`generic' = requiring fewer tuning parameters) that intermediate a trivial insulator and a crystalline Hopf insulator. In particular, we find  that the \textit{rotation}-protected delicate topological insulators, studied in this work, are proximate to \textit{mirror}-protected delicate topological semimetals -- a notion that we introduce here. \\

{\centering
\textbf{Material applications and outlook \nopagebreak\\ of delicate topology: Sec.~\ref{sec:conclusion}.}
\par}\nopagebreak\smallskip\nopagebreak

We conclude the paper in Sec.~\ref{sec:conclusion}, with a partial summary of results that emphasizes the possible natural/(meta)material/cold-atomic applications. We also present  an outlook toward future developments in delicate topology.\\

{\centering
\textbf{Appendices.}
\par}\nopagebreak\smallskip\nopagebreak

The manuscript is supplemented by a series of appendices \ref{app:glossary}--\ref{app:incompatibility}, which mostly serve either a pedagogical purpose or add to the rigor of discussions in the main text. 
The sole exception are:
Appendix~\ref{app:glossary}, which is a glossary of all recurring technical symbols and abbreviations; 
and Appendix~\ref{app:index-of-notions} that contains a list of key terms and notions
accompanied with their brief definitions.

\section{Crystalline Hopf insulators from lifting a Berry-dipole band degeneracy}\label{sec:berry-dipole}

One convenient strategy to realize known topological insulators has been to perturb topological semi-metallic phases so as to develop an energy gap (cf.~Refs.~\cite{kane2005A, Murakami2007A} and chapter~8 of Ref.~\cite{bernevig_book}). 
We show here that crystalline Hopf insulators can be realized by perturbing a rotation-symmetric semi-metallic phase with a band touching (at an isolated $\bk$-point) that acts as a dipolar source of Berry curvature.  We present a $\bk\cdot\bp$ Hamiltonian for such a \textit{Berry dipole} in Sec.~\ref{sec:Berry-definition},  then argue in Sec.~\ref{sec:Berry-transitions} that a Berry-dipole band degeneracy mediates a topological phase transition. We then move from $\bk\cdot\bp$ Hamiltonians to tight-binding lattice models in \s{sec:lattice-RTP}, showing that a Berry dipole mediates a quantized change of the integer-valued Hopf and RTP invariants.

\subsection{Rotation-symmetric \texorpdfstring{$\bk{\cdot}\bp$}{k.p} model of a Berry dipole}\label{sec:Berry-definition}

In contrast to a Weyl point [Fig.~\ref{fig:Berry-curvature}(a)], the Berry curvature integrated over a sphere enclosing the Berry dipole necessarily vanishes.
However, the Berry flux through a hemisphere
oriented parallel to the dipole axis is quantized to integer multiplies of $2\pi$ [Fig.~\ref{fig:Berry-curvature}(b)],
\begin{equation}
    \Psi=\!\!\!\!\int\limits_\textrm{hemisphere}\!\!\!\!\bm{\mathcal{F}}\cdot d\bm{S}\;\quad\in{2\pi}\mathbb{Z},
    \label{eq:Berry-dipole-flux}
\end{equation}
where $\bm{\mathcal{F}}$ is the Berry curvature of the valence band and $d\bm{S}$ is an infinitesimal surface 
element normal to the hemisphere. We say that the Berry dipole is of the $\mathbbm{d}$-th order if the absolute value of the flux $\left|\Psi\right|$ through a hemisphere equals $2\pi \mathbbm{d}$. 

We further restrict our attention to \emph{rotation-symmetric} Berry-dipoles, for which the dipole axis coincides with a rotation axis of symmetry. We assume that a two-band Hamiltonian is symmetric under continuous rotations with respect to the $z$-axis:
\begin{equation}
    \label{eq:rot-sym-Berry-dipole-contin}
    \begin{split}
    &R_{\theta}\,h(\bk)\,R_{\theta}^{-1}=h(\theta\bk), \\
    &R_{\theta}=e^{-i\theta\Delta\ell/2\cdot\sigma_z},
    \end{split}
\end{equation}
with $\theta\bk$ being the momentum vector obtained from $\bk$ by a counterclockwise rotation of $(k_x,k_y)$ by an arbitrary angle $\theta$, $R_\theta$ being the matrix representation of a rotation by angle $\theta$, and $\Delta\ell$ (assumed $\,\neq\!\!0$) is the integer-valued \textit{Berry-dipole spin}, in units where $\hbar=1$.

To study the prospective gapped phases obtained by lifting the Berry-dipole band degeneracy, we construct a two-band $\bk\cdot\bp$ Hamiltonian that depends on a tuning parameter~$\phi$. 
At $\phi\neq0$ the two energy bands are separated by a finite gap, which closes for $\phi=0$ when the bands form a Berry dipole. 
A simple realization of a rotation-symmetric Berry-dipole Hamiltonian is encoded with a spinor
\begin{equation}
    \zeta(\bk)=\begin{pmatrix}\zeta_1(\bk)\\ \zeta_2(\bk)\end{pmatrix}=\begin{pmatrix}\left(k_{-\sign(\Delta\ell)}\right)^{|\Delta\ell|}\\ \upsilon\, k_z+i\phi\end{pmatrix} \in \mathbb{C}^2,
    \label{eq:kp-sym-spinor}
\end{equation}
with $k_\pm=k_x\pm ik_y$ and helicity $\upsilon\in\{-1,+1\}$ (cf.~footnote~\ref{foot:helicity} for a geometric interpretation of $\upsilon$). The spinor serves as an input for~the
\begin{align}
    \text{\textit{spinor-form Hamiltonian}:} \as h(\bk)=[\zeta^{\dagger}(\bk)\bm{\sigma}\zeta(\bk)]\cdot\bm{\sigma}, 
    \label{eq:zszs-hamiltonian}
\end{align}
with $\bm{\sigma}=(\sigma_x, \sigma_y, \sigma_z)$ being the vector of Pauli matrices.
We have obtained \q{eq:kp-sym-spinor} by first assuming that $\zeta$ has a multi-variable Taylor expansion in all momentum components and in the tuning parameter $\phi$, then deriving [in  Appendix~\ref{app:kp-spinor}] the lowest-order form for $\zeta$ such that the $\bk{\cdot}\bp$-Hamiltonian in Eq.~\eqref{eq:zszs-hamiltonian} fulfills the symmetry condition in Eq.~\eqref{eq:rot-sym-Berry-dipole-contin}. If one further allows for continuous deformations of the Hamiltonian that preserve its `topology' (in a sense that will be made increasingly precise), $\zeta$ can be simplified to the form shown in~\q{eq:kp-sym-spinor}.\footnote{We emphasize that a Hamiltonian with rotation symmetry given by Eq.~\eqref{eq:rot-sym-Berry-dipole-contin} can as well be constructed with the spinor 
\begin{equation}
    \zeta(\bk)=\begin{pmatrix}\upsilon\, k_z+i\phi \\ \left(k_{\sign(\Delta\ell)}\right)^{|\Delta\ell|}\end{pmatrix},
    \label{eq:kp-sym-spinor-2}
\end{equation}
as shown in Appendix~\ref{app:kp-spinor}. 
Spinor-form Hamiltonians [Eq.~\eqref{eq:zszs-hamiltonian}] constructed with spinors given by Eqs.~\eqref{eq:kp-sym-spinor} resp.~\eqref{eq:kp-sym-spinor-2} differ by symmetry eigenvalues of their eigenstates. 
Namely, define an angular momentum $\ell_v$ ($\ell_c$) of the valence (conduction) eigenstate as 
\begin{equation}
    R_\theta\ket{u_{v(c)}(0,0,k_z)}=e^{i\theta\ell_{v(c)}}\ket{u_{v(c)}(0,0,k_z)}.
    \label{eq:kp-angular-mom}
\end{equation}
Then for the Hamiltonian defined through spinor \eqref{eq:kp-sym-spinor} the difference $\ell_c-\ell_v$ coincides with the Berry-dipole spin $\Delta\ell$. 
This corresponds to a valence (conduction) state at rotation axis being induced from the first (second) basis orbital, where these orbitals define a basis in which the Hamiltonian and the rotation matrix are expressed. 
For the Hamiltonian defined though spinor \eqref{eq:kp-sym-spinor-2}, one instead finds the difference $\ell_c-\ell_v=-\Delta\ell$, meaning that the valence (conduction) state is induced by the second (first) basis orbital. The two classes of models are related by flipping the sign of their energy bands.
Throughout Sec.~\ref{sec:berry-dipole} we will present results for the Hamiltonian given by the spinor \eqref{eq:kp-sym-spinor}. \label{foot:spinor2}}

The energy-momentum dispersion of each band of the spinor-form Hamiltonian (indexed by $\pm$) is simply the squared norm of the spinor multiplied by $\pm 1$:\footnote{To reveal this, first recall the completeness relation for the Pauli matrices: $\frac{1}{2}\sum_{i\in\{0,x,y,z\}} \sigma^i_{\alpha\beta}\sigma^i_{\gamma\delta} = \delta_{\alpha\delta}\delta_{\beta\gamma}$~\cite{Pauli-completeness-relation}, where $\delta$ is the Kronecker symbol and $\sigma^0$ is the identity matrix. This is equivalent to
\begin{subequations}
\begin{equation}
\sum_{i\in\{x,y,z\}}\sigma^i_{\alpha\beta}\sigma^i_{\gamma\delta} = 2\delta_{\alpha\delta}\delta_{\beta\gamma} - \delta_{\alpha\beta}\delta_{\gamma\delta}.\label{eqn:pauli-completeness}
\end{equation}
Since the two-band Hamiltonian in Eq.~(\ref{eq:zszs-hamiltonian}) is traceless, the two eigenvalues have the same magnitude but opposite sign: $\pm E$. The determinant (being the product of eigenvalues) is then found as $\det h = -E^2$. Simultaneously, for a Pauli-matrix Hamiltonian $h = \sum_{i\in\{x,y,z\}} h_i \sigma^i$, the determinant is also expressed as $\det h = -\sum_{i\in\{x,y,z\}} h_i^2$. For the spinor-form Hamiltonian, $h_i = \zeta^*_\alpha \sigma^i_{\alpha \beta} \zeta_\beta$; therefore, matching the two expressions for $\det h$ results in
\begin{equation}
E^2 = \sum_{i\in\{x,y,z\}}\zeta^*_\alpha \sigma^i_{\alpha\beta} \zeta_\beta \zeta^*_\gamma \sigma^i_{\gamma\delta} \zeta_\delta \stackrel{\textrm{(\ref{eqn:pauli-completeness})}}{=} \zeta_\alpha^*\zeta_\gamma^*\left(2\delta_{\alpha\delta}\delta_{\beta\gamma} - \delta_{\alpha\beta}\delta_{\gamma\delta}\right)\zeta_\beta\zeta_\delta ,
\end{equation}
\end{subequations}
where summation over repeated indices $\alpha,\beta,\gamma,\delta$ is implied. From the last expression, $E^2 = (\zeta_\alpha^* \zeta_\alpha)^2$ follows easily, which is equivalent to Eq.~(\ref{eqn:spinor-Ham-energy}).}
\begin{equation}
E_{\pm}(\bk) = \pm \norm{\zeta(\bk)}^2.\label{eqn:spinor-Ham-energy}
\end{equation}
The particular form of the Hamiltonian defined by Eqs.~(\ref{eq:kp-sym-spinor}) and~(\ref{eq:zszs-hamiltonian}) implies that the energy gap grows quadratically with the $k_z$-component of momentum and with the parameter $\phi$. 
Furthermore, the Berry-dipole spin $\Delta\ell$ determines the exponent of $k_{\pm}$ to be $|\Delta\ell|$, and therefore the energy gap grows with $2|\Delta \ell|$-th power of the momentum distance from the rotation axis.
In certain parts of our discussion, it will be convenient to consider the \emph{spectrally flattened}\footnote{For traceless, two-band Hamiltonians that are \emph{not} of the spinor form, the spectral flattening is instead defined as 
\begin{equation}
h_\mathrm{flat}(\bk) = \frac{h(\bk)}{\sqrt{-\det[h(\bk)]}}.\label{eqn:flatten-nonspinor}
\end{equation}} 
spinor-form Hamiltonian, 
\begin{equation}
h_\mathrm{flat}(\bk) = \frac{1}{\norm{\zeta(\bk)}^2}h(\bk),\label{eq:zszs-hamiltonian-normalized}
\end{equation}
which is defined for all $(\bk,\phi)\neq(\boldsymbol{0},0)$, and for which the spectrum is simply $E_\mathrm{flat}(\bk) = \pm 1$.

\begin{figure}
    \centering
    \includegraphics{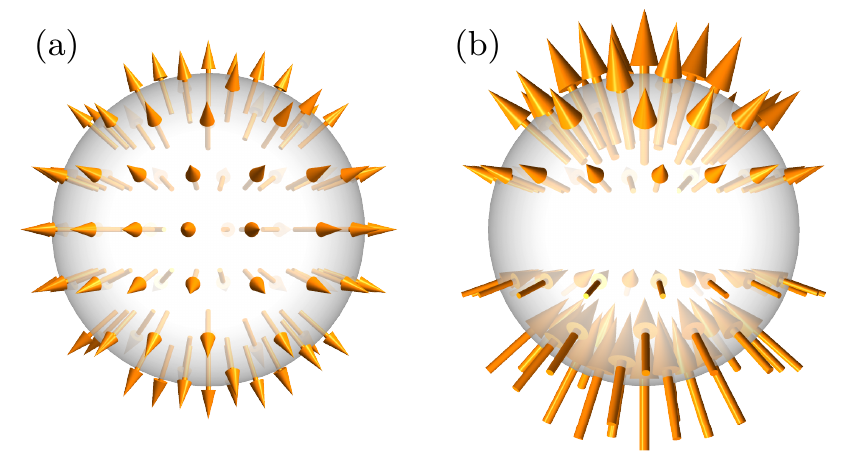}
    \caption{Berry-curvature vector fields of (a) an isotropic Weyl point with positive chirality, and of (b) a rotation-symmetric Berry dipole in Eq.~(\ref{eq:kp-sym-spinor}) with spin $\Delta\ell{=} -1$.
    Both vector fields are plotted over a sphere centered at
    the band-touching point.}
    \label{fig:Berry-curvature}
\end{figure}

We derive in Appendix~\ref{app:berry-flux} for the Hamiltonian defined by Eqs.~(\ref{eq:kp-sym-spinor}) and~(\ref{eq:zszs-hamiltonian}) that the Berry-dipole spin $\Delta\ell$ determines the Berry flux through the upper hemisphere as\footnote{
To derive the Berry flux for the Hamiltonian defined through the spinor \eqref{eq:kp-sym-spinor-2}, observe that this spinor can be brought to the form
\eqref{eq:kp-sym-spinor} by first flipping the order of the basis orbitals (i.e., rotating the Hilbert space with Pauli matrix $\sigma_x$) and by subsequently changing the sign of the Berry-dipole spin. The order of the basis orbitals (i.e., the rotation of the Hilbert space) does not affect the Berry flux, and therefore
\begin{equation}
   \Psi\left(\substack{\textrm{Berry dipole}\\\textrm{spinor \eqref{eq:kp-sym-spinor-2}}}\right)=2\pi\Delta\ell. 
\end{equation} 
Recalling the expression for the valence and conduction angular momenta given in footnote~\ref{foot:spinor2}, we conclude that the Berry flux through the upper hemisphere of both spinor-form Hamiltonians is given by
\begin{equation}
   \Psi=-2\pi(\ell_c-\ell_v). 
\end{equation}
\label{foot:Berry-flux2}} 
\begin{equation}
\Psi\left(\substack{\textrm{rotation-sym.}\\\textrm{Berry dipole}}\right)=-2\pi\Delta\ell; \label{eqn:Berry-dipole-flux}
\end{equation} 
implying that the absolute value of the Berry-dipole spin 
equals the order of the Berry dipole $|\Delta\ell|=\mathbbm{d}$.

\subsection{Berry dipole as a topological phase transition}\label{sec:Berry-transitions}

While the Hamiltonian given in Eqs.~(\ref{eq:kp-sym-spinor}) and ~(\ref{eq:zszs-hamiltonian}) has been defined for $\bk \in \mathbb{R}^3$, the fact that each Hamiltonian matrix element is a low-order polynomial in $\bk$ allows us to interpret the same Hamiltonian as a small-momentum, truncated Taylor expansion  of a tight-binding Hamiltonian defined over the Brillouin zone. Tuning $\phi$ through zero makes the energy gap close at an isolated point $\bk=\bze$, then subsequently reopens the gap; throughout this tuning, the rotation-symmetry representations at $\bk=\bze$ are \textit{not} inverted. We will find that the parameter $\phi=0$ marks a critical point between two gapped phases that are distinguished by the Hopf and Returning-Thouless-pump topological invariants, both of which are integers defined for the tight-binding Hamiltonian. 

\begin{figure}
    \centering
    \includegraphics[width=0.475\textwidth]{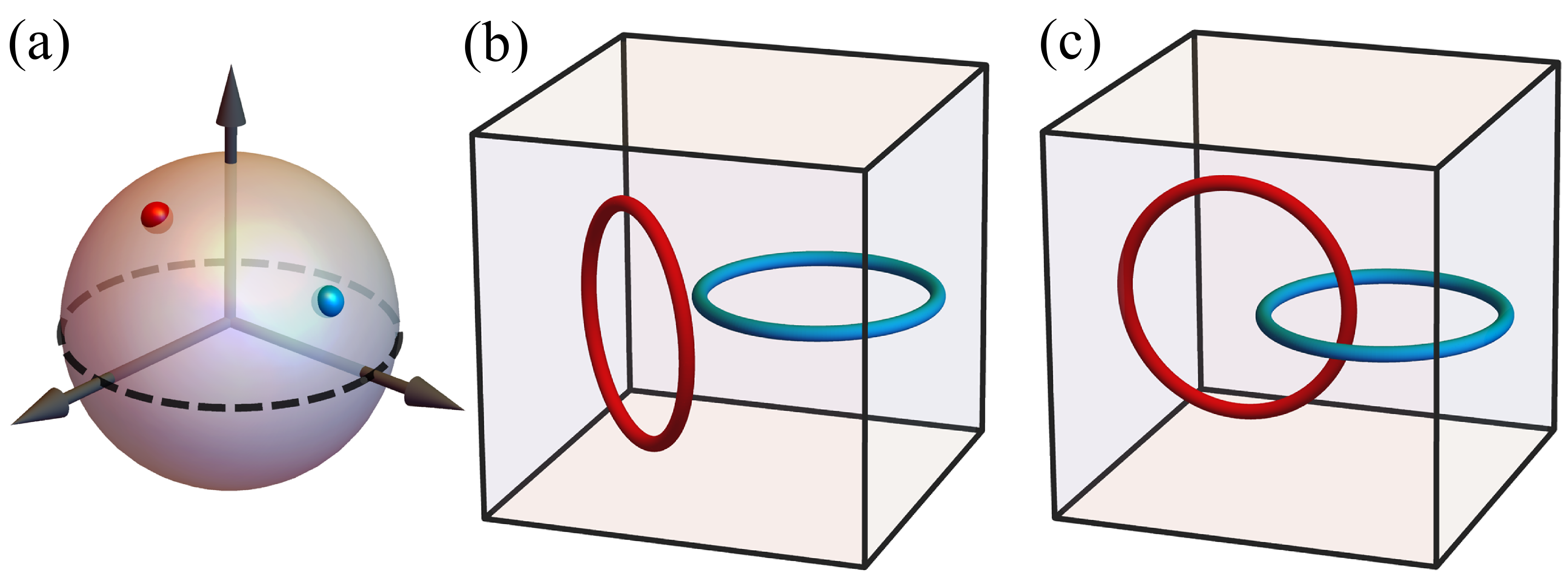}
    \caption{
    Relation between the Hopf invariant and the linking of preimages. (a) The Bloch sphere ($S^{\!2}$) represents the space of (spectrally {flattened}) gapped two-band Hamiltonians, with 
     the three axes indicating the decomposition of such Hamiltonians into Pauli matrices: $h_\mathrm{flat}=\boldsymbol{h}_\mathrm{flat}(\bk)\cdot \bsigma$. (b) For a Hamiltonian with trivial Hopf invariant, and for two arbitrarily chosen points on the Bloch sphere: $\boldsymbol{h}_{\textrm{red}},\boldsymbol{h}_{\textrm{blue}}\in S^{\!2}$ (indicated by red and blue dots), we plot their respective preimages inside the Brillouin zone (black frame), i.e., the collection of momenta $\bk$ where $\boldsymbol{h}_\mathrm{flat}(\bk)=\boldsymbol{h}_\textrm{red}$ and $\boldsymbol{h}_\mathrm{flat}(\bk)=\boldsymbol{h}_\textrm{blue}$ respectively. Panel (c) shows the linked preimages for a Hamiltonian with absolute value of the Hopf invariant $\abs{\chi} = 1$. 
    }
    \label{fig:Hopf-linking}
\end{figure}

To demonstrate why the Berry dipole intermediates a change in the Hopf invariant, we utilize the known,  linking-number interpretation of the Hopf invariant, that applies to two-band Hamiltonians with trivial Chern class.  (This triviality is assumed throughout the manuscript.) 
The Hamiltonian can be viewed as a continuous map from the BZ three-torus to the Bloch two-sphere; for a three-toroidal $\bk$-space, the Hopf invariant equals the linking number of the preimages corresponding to any two distinct points on the Bloch two-sphere \cite{Hopf:1931,wilczekzee_linkingnumbers,kennedy_hopfchern}, cf.~Fig.~\ref{fig:Hopf-linking}.
In the $\bk\cdot \bp$ model, this linking number is not well-defined since the preimage manifolds are noncompact, forming a family of infinite (and for $\phi\neq 0$ non-intersecting) straight lines\footnote{
According to the correspondence $S^{\!2} \approx\mathbb{C}P^1$ (the complex projective plane), points on the Bloch sphere are in one-to-one correspondence with rays in the Hilbert space $\mathbb{C}^2$~\cite{stone_mathematicsforphysics}. 
Therefore, for any two $\bk$-points $(\bk_1$ and $\bk_2$) in the preimage of a single Bloch-sphere point, the corresponding spinors $\zeta(\bk_1)$ and $\zeta(\bk_2)$, as given by the Hamiltonian in Eq.~(\ref{eq:zszs-hamiltonian}), are related to each other by rescaling with a non-zero complex number. 
Such rescaling keeps the ratio of the two spinor components invariant; specifying the preimage is equivalent to fixing a constant $c\in\mathbb{C}$ in $\zeta_1(\bk)=c \zeta_2(\bk)$. 
Plugging in the spinor components of a first-order Berry dipole from \q{eq:kp-sym-spinor} for the case $\abs{\Delta\ell}=1$, one obtains $k_x\pm i\,k_y=c(\upsilon k_z+i\phi)$, where the sign on the left-hand side equals $-\sign(\Delta\ell)$. This may be viewed as a straight $\bk$-line (parametrized by $k_z$) which intersects the $k_z=0$ plane at point $k_x \pm ik_y = ic \phi$. Note that the inclination 
of the straight line,  as captured by $d(k_x\pm i k_y)/dk_z = c\upsilon$, is independent of the tuning parameter $\phi$.\label{foot:rescale-spinor}} in $\mathbb{R}^3$, as illustrated representatively by two lines in \fig{fig:Hopf-winding}(a) for the case $\Delta\ell=1$ and $\upsilon=1$.\footnote{The geometric constellation of the preimages motivates the adopted name \emph{helicity} for the parameter $\upsilon\in\{-1,+1\}$; namely: if we extract the polar coordinate of any particular preimage line, $\theta(k_z)=\arg[k_x(k_z)+
i k_y(k_z)]$ 
(where polar coordinate $\theta$ and the radial coordinate $k$ are defined via $k_x = k\cos\theta$ and $k_y = k\sin\theta$), as a function of $k_z$ [the dependence being emphasized by the round brackets: `$(k_z)$'], then 
\begin{equation}
\sign\left [\frac{d \theta (k_z)}{ dk_z} \right]= \sign[\phi\upsilon].\label{eqn:chirality-def}
\end{equation} 
Pictorially, as one moves along the positive $k_z$ direction, we observe all preimage lines to wind in the positive direction around the $k_z$ axis when $\sign[\phi\upsilon]=+1$ (resp.~in the negative direction when $\sign[\phi\upsilon]=-1$), indicating a notion of handedness. [Note that these conclusions are derived for a Hamiltonian given by Eqs.~\eqref{eq:kp-sym-spinor} and \eqref{eq:zszs-hamiltonian} with the Berry-dipole spin fixed to $\Delta\ell=1$; the sign of Eq.~(\ref{eqn:chirality-def}) is flipped for $\Delta\ell = -1$.]
\label{foot:helicity}} [The red line is the preimage of $\bm{h}\propto(1,0,0)$, and the cyan line a preimage of $\bm{h}\propto(-1,0,0)$.] Given any two distinct points on the two-sphere,  a \emph{change} in their associated linking number (as attributed to the Berry-dipole band touching) remains well-defined. 
For instance, comparing  Figs.~\ref{fig:Hopf-winding}(a) and (c) (for $\phi<0$ and $\phi>0$, respectively) shows that the red and cyan lines `pass through' each other  as $\phi$ is tuned through zero, which can be interpreted as a unit change in the linking number. The sense of this interpretation is that no matter how we `close off' the cyan and red lines at $\bk$-infinity, their linking number must change by unity under the `passing through'.

\begin{figure*}
    \centering
    \includegraphics{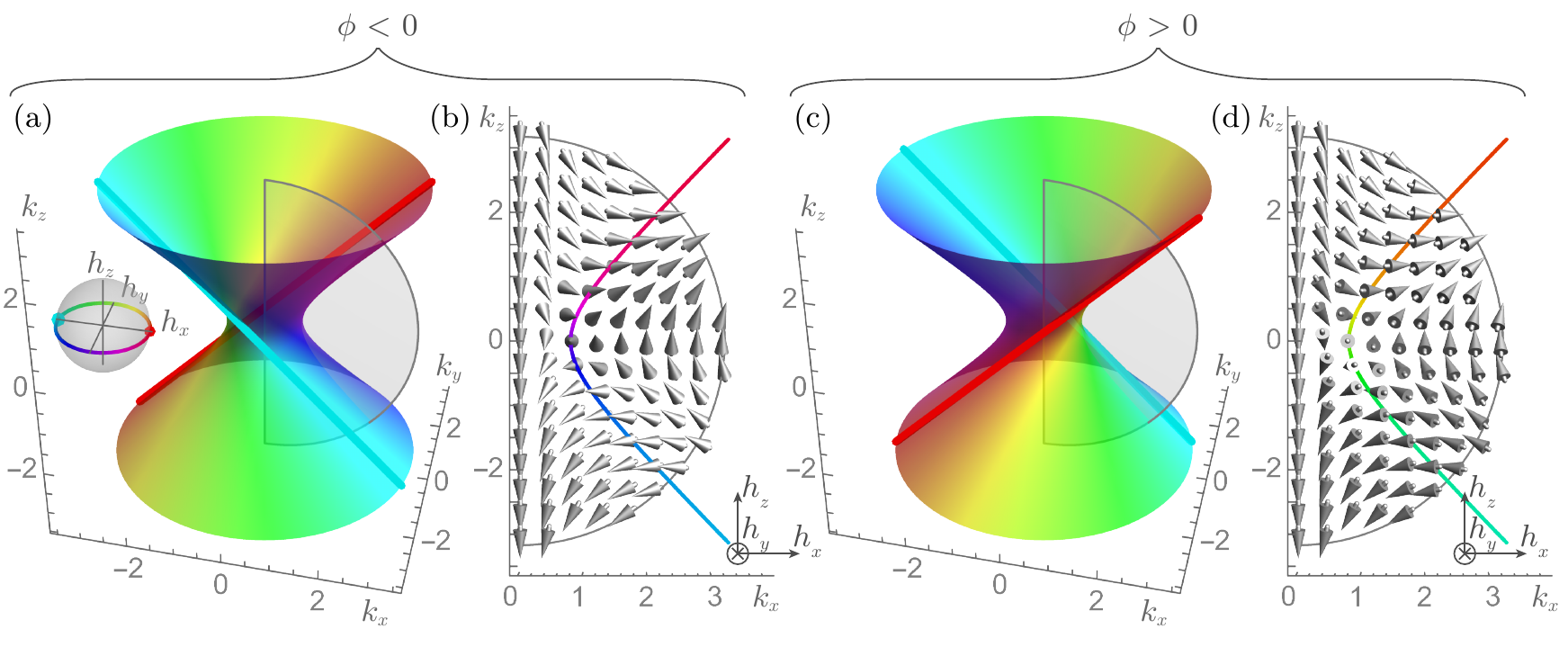}
    \caption{ Panels (a--b) characterize the Hamiltonian of 
    Eqs.~(\ref{eq:kp-sym-spinor}--\ref{eq:zszs-hamiltonian}) with 
    $\Delta\ell=\upsilon=+1$ and $\phi=-0.8$; for panels (c--d), we maintain $\Delta\ell=\upsilon=+1$ but change $\phi=+0.8$. For both Hamiltonians, the corresponding momentum surfaces where the $z$-component of $\bm{h}_{flat}(\bk)$ vanishes [$h_{flat,z}(\bk)=0$] form hyperboloids illustrated in panels (a) and (c) respectively.     The coloring of the surfaces indicates the values of $h_{flat,x}$ and $h_{flat,y}$ according to an inset in panel (a).
    The cyan and red lines indicate the momenta where the Hamiltonian $\bm{h}_{flat}\propto(1,0,0)$ and $\bm{h}_{flat}\propto(-1,0,0)$, respectively. 
    The gray semicircle represents a 2D radial cut given by $\varphi_\bk=0$, $\theta_\bk\in[0,\pi]$ and $|\bk|\in[0,R]$, with spherical coordinates $(|\bk|,\theta_\bk,\varphi_\bk)$ specified in the main text. The vector field of $\bm{h}_{flat}(\bk)$ is illustrated over the gray semicircle in panels (b) and (d), for $\phi=-0.8$ and $+0.8$ respectively.
    The colored line in panel (b) [and also (d)] denotes momenta where the radial cut is intersected by the hyperboloid; the color scheme follows again the legend in panel (a).
    }
    \label{fig:Hopf-winding}
\end{figure*}

To diagnose a change in the Hopf invariant of the tight-binding Hamiltonian, we look at changes in the \textit{continuum Hopf number} (of the  $\bk\cdot \bp$ Hamiltonian), defined as an integral of the Chern-Simons three-form:
\begin{equation}
    \chi^\textrm{cont.}=-\frac{1}{4\pi^2}\int\limits_{\mathbb{R}^3}d^3k\,\bm{\mathcal{F}}\cdot\boldsymbol{\mathcal{A}},\as \bm{\mathcal{F}}=\bm{\nabla}_{\bm{k}}\times \boldsymbol{\mathcal{A}}(\bm{k})
    \label{eq:hopfinvar-cont}
\end{equation}
with  $\boldsymbol{\mathcal{A}}(\bm k)=\bra{u_v(\bm k)}\ket{i\bm{\nabla}_{\bm{k}}u_v(\bm k)}$ the Berry gauge connection of the intra-cell wave function $\ket{u_{v}(\bk)}$ (the subscript `$v$' stands for `valence band'), and $\bm{\mathcal{F}}(\bm{k})$ the corresponding Berry curvature.   

The continuum analog of  the RTP invariant is given by the Berry-Zak phase computed along the rotation axis:
\begin{equation}
    \mathscr{Z}^{\mathrm{cont.}} = \int\limits_{-\infty}^\infty \mathcal{A}_z(0, 0, k_z)dk_z, \label{eq:pol-continuum}
\end{equation}
where $\mathcal{A}_z$ is the $z$ component of the Berry connection. We refer to \q{eq:pol-continuum} as the \textit{continuum Zak phase}. Let us clarify that the quantity is indeed a phase, meaning it is  well-defined modulo integer multiples of $2\pi$. This holds under the conditions that the wave function is differentiable everywhere on the $k_z$-line, and moreover coincides at $k_z\to\pm\infty$, allowing to  `compactify' the $k_z$-line by identifying the points at infinity. (Under these conditions, gauge transformations can at most change the value of $\mathscr{Z}^{\mathrm{cont.}}$ by an integer multiple of $2\pi$.\footnote{Under a gauge transformation $\ket{u(\bk)}\ri e^{-i\beta(\bk)}\ket{u(\bk)}$, the Berry connection transforms as $\mathcal{A}_z\ri \mathcal{A}_z+ \partial_{k_z}\beta$. If the resultant wave function is to remain differentiable with respect to $k_z$ and periodic over the compactified $k_z$ axis, then $\beta/2\pi$ must be a differentiable function of $k_z$ with  $\int (\partial_{k_z}\beta) \, dk_z/2\pi$ being an integer-valued winding number.  }) A stronger statement, namely that $\mathscr{Z}^{\mathrm{cont.}}/2\pi$ is integer-valued, follows from the Hamiltonian matrix being diagonal along the rotation-invariant $\bk$-line, according to the rotation constraint in \q{eq:rot-sym-Berry-dipole-contin}.\footnote{The diagonal constraint ensures the existence of a gauge where the wave function is constant along the rotation-invariant $\bk$-line, and hence the Berry connection vanishes at each point on said line.}

Changes in the continuum Hopf number $\chi^\textrm{cont.}$ and Zak phase $\mathscr{Z}^{\mathrm{cont.}}/2\pi$ across a gapless point are quantized to integer values, assuming that the $\bk \cdot \bp$ Hamiltonian is a small-momentum expansion of a tight-binding Hamiltonian. In the next section, we will identify $\delta\chi^\textrm{cont.}
\in \Z$ as the change in the Hopf invariant of the tight-binding model, and $\delta \mathscr{Z}^{\mathrm{cont.}}/2\pi \in \Z$ as the corresponding change in the RTP invariant. Assuming the gapless point is described by the Berry-dipole Hamiltonian [Eqs.~\eqref{eq:kp-sym-spinor},\eqref{eq:zszs-hamiltonian}],  $\chi^\textrm{cont.}$ and $\mathscr{Z}^\textrm{cont.}$ are modified by
\begin{align}
    &\delta\chi^\textrm{cont.}
    =-\upsilon\Delta\ell,
    \label{eq:Berry-dipole-hopf-change} \\
    &\delta \mathscr{Z}^\textrm{cont.}
    =-2\pi\upsilon, \label{eq:Berry-dipole-pol-change}
\end{align}
as $\phi$ is tuned from negative to positive values.\footnote{In analogy with footnote~\ref{foot:Berry-flux2}, we argue that the change in the continuum Hopf number over a Berry dipole given by the spinor \eqref{eq:kp-sym-spinor-2} has an opposite sign than for the spinor \eqref{eq:kp-sym-spinor}, $\delta\chi^{\mathrm{cont.}}=\upsilon\Delta\ell$. 
Therefore, more generally for both spinors \eqref{eq:kp-sym-spinor} and \eqref{eq:kp-sym-spinor-2} this change is determined by the difference in angular momenta between the conduction and the valence band, defined in footnote~\ref{foot:spinor2}:
\begin{equation}
    \delta\chi^\textrm{cont.}=-\upsilon(\ell_c-\ell_v).
\end{equation}
The change in the continuum Zak phase given by Eq.~\eqref{eq:Berry-dipole-pol-change} is valid for Hamiltonians defined through both spinors \eqref{eq:kp-sym-spinor} and \eqref{eq:kp-sym-spinor-2}. \label{foot:hopf-change-2}} 
Note that although the continuum Zak phase is defined only modulo $2\pi$, the change in the Zak phase across a Berry dipole \eqref{eq:Berry-dipole-pol-change} is well defined modulo nothing if we require a wave function to be smooth for all $(\bk,\phi)\neq(0,0)$.

Eq.~\eqref{eq:Berry-dipole-pol-change} can be directly derived after noticing that the valence intra-cell wave function of the Hamiltonian in Eq.~\eqref{eq:zszs-hamiltonian} can be analytically expressed as
\begin{equation}
   \ket{u_v(\bk)}=e^{i\beta(k_z)} i\sigma_y\zeta^*(\bk)/\|\zeta(\bk)\|,
   \label{eq:valence-band-spinor}
\end{equation}
with $\zeta$ given in \q{eq:kp-sym-spinor} and a phase $\beta(k_z)$ chosen to ensure a smooth wave function along the compactified $k_z$-axis.\footnote{
To guarantee this, we choose the phase $\beta(k_z)$, which depends smoothly on $k_z$, such that $\beta(-\infty)=0$ and $\beta(\infty)=\pi$.
} Plugging this wave function in the definition of the continuum Zak phase in Eq.~\eqref{eq:pol-continuum} gives the change in the Zak phase according to Eq.~\eqref{eq:Berry-dipole-pol-change}.\footnote{The Berry connection computed in the gauge specified in Eq.~\eqref{eq:valence-band-spinor} is $\mathcal{A}_z(0, 0, k_z)=-\upsilon\phi/(k_z^2+\phi^2)-\partial_{k_z}\beta(k_z)$ (cf.~supplementary Sec.~III of Ref.~\cite{Nelson:2021} for the explicit calculation), which plugged into the definition of the continuum Zak phase in Eq.~\eqref{eq:pol-continuum} gives $\mathscr{Z}^\textrm{cont.}
=-\pi\,\left[1+\upsilon\,\sign(\phi)\right]$. 
Hence, when $\phi$ changes sign from negative to positive, the continuum Zak phase changes according to Eq.~\eqref{eq:Berry-dipole-pol-change}. \label{foot:Zak-change}} 

The quantized jump of the Hopf number expressed in Eq.~\eqref{eq:Berry-dipole-hopf-change} is derived generally in Appendix~\ref{app:change-hopf}. 
For the particular case of $\Delta\ell=1$ and $\upsilon=1$, we present here a pictorial argument for $\chi^\textrm{cont.}=\pm1/2$, with the sign changing across $\phi=0$. We utilize Wilczek and Zee's ~\cite{wilczekzee_linkingnumbers, stone_mathematicsforphysics} correspondence of the unit Hopf invariant with a `self-eating, rotating' skyrmion texture for the \textit{pseudospin vector} $\bm{h}$ associated to the Hamiltonian $h(\bk)=\bm{h}(\bk)\cdot\bsigma$. By `self-eating' and `rotating', we mean that the skyrmion's `worldline' forms a loop in $\bk$-space, with the skyrmion pseudospin texture (in a plane orthogonal to the loop) rotating by $2\pi$ as $\bk$ is advanced along the loop, as illustrated in Fig.~1 of Ref.~\cite{wilczekzee_linkingnumbers}. 
Analogously, we now argue that $\chi^\textrm{cont.}=\pm1/2$ for the Berry-dipole transition is associated to `half' a self-eating, rotating skyrmion: consider a 2-dimensional radial cut given by $\varphi_\bk=0$, $\theta_\bk\in[0,\pi]$ and $|\bk|\in[0,R]$, where $(|\bk|,\theta_\bk, \varphi_\bk)$ are spherical coordinates in the $\bk$-space\footnote{We specify the spherical coordinates to be $k_x=|\bk|\sin\theta_\bk\cos\varphi_\bk$, $k_y=|\bk|\sin\theta_\bk\sin\varphi_\bk$ and $k_z=|\bk|\cos\theta_\bk$.~\label{foot:S2-coordinates}} and $R$ is a large radius. 
The pseudospin texture on this cut forms half of a skyrmion, as $\bm{h}(\bk)$ covers half of the Bloch two-sphere when $\bk$ is swept over the cut [cf.~Fig.~\ref{fig:Hopf-winding}(b)]. Indeed the vectors on the edge of the cut are mapped to the great circle of the two-sphere with $h_y=0$, while all vectors within the cut cover the hemisphere given by $h_y<0$ for $\phi<0$, and the hemisphere given by $h_y>0$ for $\phi>0$ [cf.~Fig.~\ref{fig:Hopf-winding}(d)]. 
Whether it is the hemisphere with $h_y<0$ or with $h_y>0$, determines the sign of $\chi^{cont.}=\pm1/2$. The sense in which this half-skyrmion is `rotating' is illustrated in \fig{fig:Hopf-winding}(a,c), where the hyperboloid is the preimage of the Bloch-sphere equator ($h_z=0)$: each point on the hyperboloid is colored according to the angle $\theta_{xy}$ between the $(h_x,h_y)$-vector and the $h_x$-axis ($h_x=\sqrt{h_x^2+h_y^2}\cos\theta_{xy}$), 
following a color-coding scheme in the inset of Fig.~\ref{fig:Hopf-winding}(a). For any point on the hyperboloid, we see that $\theta_{xy}$ changes by $2\pi$ as we encircle the rotation axis while circumnavigating the hyperboloid sheet at a fixed `elevation' $k_z$. 

The presented results demonstrate that the Berry-dipole transition alters the value of the topological numbers in Eqs.~(\ref{eq:hopfinvar-cont}) and~(\ref{eq:pol-continuum}).
It is worth remarking that, per codimension counting~\cite{vonNeumann:1929,Bzdusek:2017}, a generic band degeneracy
in a 4-dimensional momentum-parameter space for Wigner-Dyson class $\textrm{A}$ extends over
1-dimensional nodal lines. 
Therefore, the Berry dipole, which constitutes a single gapless \emph{point} in the $4$-dimensional space, requires certain amount of fine-tuned parameters. 
Nevertheless, as discussed in Appendix~\ref{app:phase-transition-general}, the Berry-dipole is a \emph{representative} model for the considered topological phase transitions, because any process that alters the value of the Hopf invariant while keeping the Chern class trivial can be, in a sense specified in the Appendix, deformed into the Berry-dipolar form in Eqs.~(\ref{eq:kp-sym-spinor}) and~(\ref{eq:zszs-hamiltonian}). If one is interested in realizing delicate topology for realistic Hamiltonians, it pays to explore the phase diagram in the vicinity of the fine-tuned Berry-dipolar form of the Hamiltonian. It turns out the vicinity contains a variety of toplogical (semi)metallic phases, some of which manifest a generalized notion of symmetry-protected delicate topology. Such a discussion is deferred to~Sec.~\ref{sec:semi-metallic-transition-region}.

In the next section, we show how the $\bk\cdot\bp$ model of the Berry dipole arises as the small-momentum expansion of a tight-binding Hamiltonian on a lattice. Such a lattice regularization makes precise the notion of `closing off' the infinite preimage lines.

\section{From Berry dipoles to RTP and Hopf invariants}\label{sec:lattice-RTP}

In this section we extend the notion of the Berry dipole from $\bk\cdot\bp$ models to $\mathrm{P}n$-symmetric tight-binding Hamiltonians ($n\in\{2,3,4,6\}$), with $\mathrm{P}n$ being the space group whose order-$n$ point group is generated by an $n$-fold rotation. These tight-binding Hamiltonians will be formulated in Sec.~\ref{sec:lattice-ham}, which also contains a pedagogical introduction to several basic concepts in band theory that underlie the subsequent exposition. In Sec.~\ref{sec:lattice-invariants} we define the lattice analogs of the continuum Hopf number and Zak phase. The lattice analog of a nontrivial continuum Zak phase will be interpreted as a `returning' Thouless pump (RTP), by analogy with charge pumping in one-momentum-parameter Hamiltonians.
Finally, in Sec.~\ref{sec:Hopf-RTP-correspondence}, we discuss the changes in the Hopf and in the RTP invariants across the possible
Berry-dipole band touching points. In particular, we show how the individual contributions from each touching point add up for a resulting change in a lattice invariant.

The introduced notions will be illustrated on concrete models in the subsequent Sec.~\ref{sec:models}.  
The present
discussion also constitutes the basis for deriving the bulk-boundary correspondence of these topological models (analyzed in Sec.~\ref{sec:BBC}), and will help us derive a constraint on the admissible values of the Hopf and the RTP invariants (discussed in Sec.~\ref{sec:RTP-Hopf}).

\subsection{\texorpdfstring{$\mathrm{P}n$}{Pn}-symmetric tight-binding Hamiltonians\label{sec:lattice-ham}}

We review here certain basic properties of $\mathrm{P}n$-symmetric tight-binding Hamiltonians which form the basis for a following higher-level discussions. 
Much of what is said in this subsection is either well-known or exists in the literature, but we hope that our pedagogical introduction of basic band-theoretic vocabulary streamlines the subsequent reading. It is not impossible for an experienced reader to selectively read parts of this subsection, or even skip this subsection on a first reading, returning only when clarification is needed for certain symbols or terminology.

The Hilbert space of our tight-binding Hamiltonians is first formulated in \s{sec:bandrep} as an infinite set of Wannier orbitals on which the space group $\mathrm{P}n$ has a certain action. Such a Hilbert space is our first example of a `band representation' -- a fundamental notion in the topological classification of crystalline insulators. Basic properties of the momentum-dependent tight-binding Hamiltonian are pedagogically reviewed in \s{sec:kdependentTBham}, with emphasis on the angular momentum of Bloch states -- the latter is an essential concept in formulating the returning Thouless pump. Our assumption that the tight-binding Hamiltonian has trivial first Chern class has implications (for the space-group action on Wannier orbitals) that are elaborated in \s{sec:trivialchern}.

\begin{figure}
    \centering
    \includegraphics{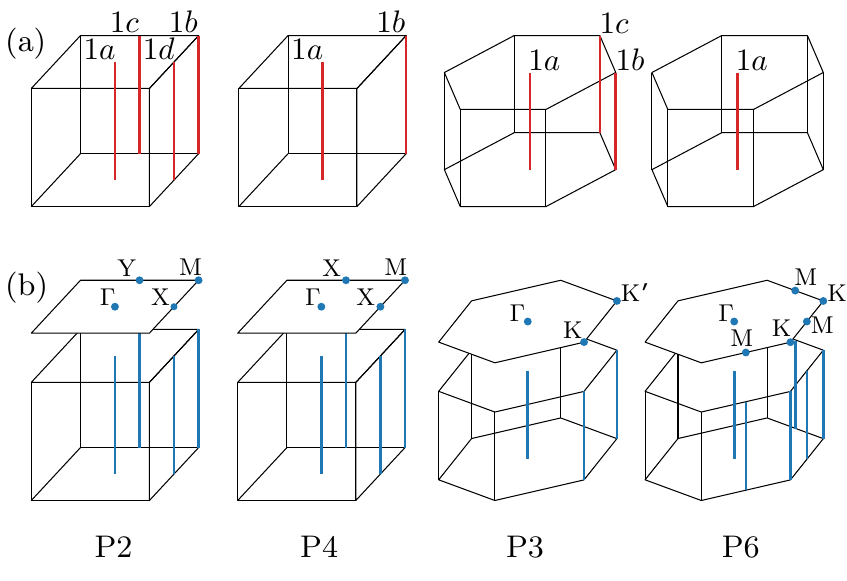}
    \caption{For each of the space groups P$n$ ($n\in\{2,4,3,6\}$), we show (a) all $C_n$-invariant Wyckoff positions in the primitive unit cell, and (b) all $C_m$-invariant momentum-lines in the Brillouin zone ($\textrm{BZ}$), with $m$ a nontrivial divisor of $n$, as well as all $C_m$-invariant momentum-points in the reduced Brillouin zone ($\textrm{rBZ}$); $\textrm{rBZ}$ is obtained from $\textrm{BZ}$ by projection in the direction of the rotation axis.
    }
    \label{fig:wp-rotinvmom}
\end{figure}

\subsubsection{\texorpdfstring{$\mathrm{P}n$}{Pn}-symmetric, tight-binding Hilbert space and band representations}\la{sec:bandrep}

We consider a class of $\mathrm{P}n$-symmetric tight-binding Hamiltonian whose Hilbert space is specified by two \textit{representative tight-binding-basis orbitals}: $\varphi_1$ and $\varphi_2$, with $\varphi_1$ unrelated to $\varphi_2$ by any space-group operation.\footnote{Our considerations can be generalized to any $\mathcal{G}$-symmetric tight-binding Hamiltonian, with $\mathcal{G}$ any space group (double or not, magnetic or not) that contains $\mathrm{P}n$ as a subgroup, such that  $\varphi_1$ is unrelated to $\varphi_2$ by elements~in~$\mathcal{G}$.}
Modulo Bravais-lattice translations, $\varphi_\alpha$ is defined to have  a positional center (or \textit{Wyckoff position}) $\br_\alpha=(\br_{\alpha,\perp},z_\alpha)$. 

For any point $\br$ in $\R^3$, its \textit{site-stabilizer} $\mathcal{G}_{\br}$ is defined as the subgroup of the space group $\mathcal{G}$ that preserves $\br$; in our example with $\mathcal{G}=\mathrm{P}n$, we assume $\mathcal{G}_{\br_\alpha}$ to be isomorphic to the order-$n$-rotational point group $C_n$; in short, we say that $\br_\alpha$ is a  \textit{$C_n$-invariant Wyckoff position}. All $C_n$-invariant Wyckoff positions for space
groups P$n$, $n\in\{2,3,4,6\}$ are presented in Fig.~\ref{fig:wp-rotinvmom}(a) as red lines, indicating that $\br_{\alpha,\perp}$ are symmetry-fixed but $z_\alpha$ are not. 

We further assume $\varphi_\alpha$ transforms under a one-dimensional representation
$e^{i2\pi\mathcal{L}_\alpha/n}$ of $C_n$, for both $\alpha=1,2$. We refer to $\mathcal{L}_\alpha\in \{0,1,\ldots,n-1\} $ as the \emph{on-site angular momenta}.\footnote{ 
Mod-$n$ integer-valued angular momenta most directly apply to bosonic particles, with the caveat that `insulating Hamiltonians' should be interpreted more generally as single-boson Hamiltonians with a spectral gap separating the two bands. The ultra-cold atomic realization of the Hopf invariant in \ocite{schuster_hopfdipolarspin} is a case in point.\medskip

\noindent Mod-$n$ integer-valued angular momenta also apply to three categories of single-\emph{fermion} Hamiltonians:\medskip

\noindent (\emph{i}) In the absence of spin-orbit-coupling, the wave function at each $\bk$ factorizes into a tensor product of orbital and spin wave functions.  Assuming the spin $SU(2)$ symmetry is not spontaneously broken by magnetism, the spin wave function plays a trivial `spectator role' and is removed from the Hilbert space under consideration. A `two-band Hamiltonian' is then implicitly the restriction of a spin-dependent Hamiltonian to its action on the orbital wave function; a conduction band (given by a single linearly-independent Bloch-type orbital wave function) is implicitly a shorthand for an energy-degenerate conduction band (given by two linearly-independent, spin-dependent Bloch functions), and likewise for the valence band.\medskip 

\noindent (\emph{ii}) If the spin $SU(2)$ symmetry is spontaneously broken down to a subgroup $H\subset SU(2)$, with $H$ containing  $U(1)$ as a subgroup (meaning one component of spin is conserved, say, $S_z$), then the Hilbert space decomposes into two $S_z$-labelled sectors, and our `two-band Hamiltonian' should be interpreted as a mean-field Hamiltonian applying  to either of these sectors. In each sector, one can remove the spectator spin quantum number, such that all angular momenta are integer-valued. For instance, a ferromagnetic, spin-$U(1)$-symmetric realization of the Hopf insulator was introduced by two of the present co-authors in \ocite{AA_teleportation}.\medskip 

\noindent (\emph{iii}) If the spin ${S}{U}(2)$ symmetry is spontaneously broken down to $H$ with no $U(1)$ subgroup, then
the mean-field Hamiltonian essentially has a interaction-induced spin-orbit coupling, which can be addressed in essentially the same manner as a single-particle, non-symmetry-breaking Hamiltonian whose spin-orbit coupling has a relativistic origin.   In both cases, the spin quantum number is not a spectator variable, and one should begin by considering \emph{half-integer-spin representations} of the site stabilizer. However, our restriction to one-dimensional representations of the site stabilizer allow to redefine the $n$-fold rotation operator by an overall multiplicative phase factor, such that all eigenvalues of the redefined rotation operator are shifted toward mod-$n$ integer-valued angular momenta. A `two-band Hamiltonian' is then unambiguously  `two-band', with the spin-orbit coupling needing to be sufficiently strong that the bands closest to the Fermi level are spin-split.}
(For the purpose of introducing notation used throughout the text, $\mathcal{L}_\alpha$ is an example of a \textit{mod-$p$ quantity}, $p\in\mathbb{Z}$:
\begin{subequations}
\begin{equation}   
\textrm{integer mod} \;p \in \{0,1,\ldots,p-1\},
\end{equation}
and we will denote a mod-$p$ equivalence between two integer-valued quantities by `$=_p$':
\begin{equation}
    a=_p b \as\textrm{meaning}\as a\;\mathrm{mod}\; p=b\;\mathrm{mod}\; p.)
    \label{eq:equal-mod-m}
\end{equation} 
\end{subequations}

By applying all Bravais-lattice translations to the representative basis orbital $\varphi_\alpha$, one generates an infinite set of Wannier orbitals that forms a representation space of $\mathrm{P}n$. 
(We use `representative' in the English sense of one representing many, and `representation' in the mathematical sense of group representations.)
In general, the process of going from a representation of subgroup $H$ to a representation of a group $\mathcal{G} \supset H$ is known mathematically as \textit{induction}. 
In the crystallographic context where $\mathcal{G}$ is a space group, an exponentially-localized Wannier orbital gives a representation of a site-stabilizer $H=\mathcal{G}_{\br}\subset \mathcal{G}$, and the resultant induced representation is known as a \textit{band representation}~\cite{Zak_bandrepresentations,nogo_AAJH,bouhon_wilsonloopapproach,crystalsplit_AAJHWCLL}. 
Not only are band representations relevant to the construction of tight-binding models, but also the categorization into  band representations vs.~non-band representations has been viewed in many works~\cite{shiozaki_review,TQC, bandcombinatorics_kruthoff,Po_symmetryindicators,po_fragile} as being synonymous to a categorization into topologically `trivial' vs.~`nontrivial' band insulators. 
Crystalline Hopf insulators present a counterpoint to these views by being `trivial' due to their band-representability,
yet also nontrivial in the more refined classification of delicate topological insulators. 

Unless specified otherwise in this paper, we will only concern ourselves with space groups $\mathcal{G}=\mathrm{P}n$, with Wyckoff positions $\br$ that are $C_n$-invariant, and with representations (of the site-stabilizer $\mathrm{P}n_{\br}$) that are one-dimensional. Under these assumptions, an equivalence class of band representation is fully specified by the data ($\mathcal{L},\br_{\perp}$) with $\mathcal{L}$ an on-site angular momentum, and with the understanding that the $z$-component of $\br$ is irrelevant to the representation of $\mathcal{G}_{\br}$ and hence also irrelevant to the representation of $\mathcal{G}$. Henceforth, we will often apply the figure-of-speech of a \textit{band representation induced from an ($\mathcal{L},\br_{\perp}$)-orbital}, and denote the corresponding equivalence class by BR$(\mathcal{L},\br_{\perp})$. Two band representations belonging in the same equivalence class are said to be \textit{symmetry-equivalent} [a formal definition is given in Eq.~\eqref{isoorbital} and footnote~\ref{foot:sym-equiv} of Sec.~\ref{sec:lattice-invariants}].

In the particular case that the representative Wannier orbital is a tight-binding-basis orbital $\varphi_\alpha$ with on-site angular momentum $\mathcal{L}_\alpha$ and Wyckoff position $\br_\alpha$, the induced representation will be referred to as a \textit{basis band representation} and denoted by BBR$[\varphi_\alpha]$. A basis band representation is a refined notion\footnote{The conventional notion of band representations carries an additional notion of equivalence which is absent in our definition of basis band representations. This equivalence is more specifically an isomorphism in space-group-equivariant vector bundle theory~\cite{crystalsplit_AAJHWCLL}, as elaborated in \s{sec:multicell}. Under this equivalence, $z_\alpha$ can be continuously deformed along the rotation axis, and the support of the representative Wannier orbital can be continuously deformed from one-site-localized to exponentially-localized.} of a band representation carrying two pieces of additional data: (a)  the $z_\alpha$ coordinate of the Wyckoff position, and (b) the fact that the representative Wannier orbital $\varphi_\alpha$ is not just exponentially-localized but also compactly supported on a single `site'. BBR$[\varphi_\alpha]$ is a member of the equivalence class BR$(\mathcal{L}_\alpha,\br_{\alpha,\perp})$. Data (a), despite being irrelevant to the symmetry representation, is of paramount relevance to the returning Thouless pump -- a topological invariant unique to delicate topological insulators. 

The sum of basis band representations induced from $\varphi_1$ and $\varphi_2$ defines the Hilbert space of our tight-binding Hamiltonian. 

\subsubsection{\texorpdfstring{$\mathrm{P}n$}{Pn}-symmetric, \texorpdfstring{$\bk$}{k}-dependent tight-binding Hamiltonian and itinerant angular momenta}\la{sec:kdependentTBham}

Linearly combining all Wannier orbitals of BBR$[\varphi_\alpha]$ with weight coefficient $e^{i\bk\cdot (\bR+\br_\alpha)}$, where $\bR$ is a Bravais-lattice vector, defines a momentum-dependent \textit{Bloch state}; the two Bloch states thus attained may be viewed as basis vectors for a two-by-two Hamiltonian matrix $h(\bk)$, as briefly reviewed in Appendix~\ref{app:tb-formalism}. We refer to $h(\bk)$ as the \textit{$\bk$-dependent tight-binding Hamiltonian}, to be distinguished from the \textit{real-space tight-binding Hamiltonian} which acts on the infinite set of Wannier orbitals.  

Because the Wyckoff positions of $\varphi_1$ and $\varphi_2$ are not necessarily equal (modulo lattice translations), the tight-binding Hamiltonian is not necessarily periodic in reciprocal-lattice translations, and instead satisfies:
\begin{equation}
h(\bk+\bG) = V(\bG)^{-1}h(\bk)V(\bG),
\label{eq:ham_non-periodic}
\end{equation}
with $\bG$ an arbitrary reciprocal-lattice vector, and $V$ a matrix that encodes the two Wyckoff positions:
$[V(\bG)]_{\alpha\beta}=\exp(i\bG\cdot\br_\alpha)\delta_{\alpha\beta}$ [note that $V(\bG)^{-1} = V(-\bG)$].
Consequently, the two-component (`{intra-cell}') eigenstates of $h(\bk)$ are two-component wave functions (with ket notation $\ket{u(\bk)}$) that are not necessarily periodic over the Brillouin zone; it is advantageous to fix the boundary condition as  $\ket{u(\bk+\bG)}=V(-\bG)\ket{u(\bk)}$, for subsequent computations of the Berry-Zak phase. 

The $\mathrm{P}n$ symmetry constrains the Hamiltonian as 
\begin{align}
R_{C_n}h(\bk)R_{C_n}^{-1}=h(C_n\bk),
\label{eq:hamilt-sym-condition}    
\end{align}
with $R_{C_n}$ a  \textit{rotation matrix} that is more generally defined as
\begin{equation}
R_{C_m}=\begin{pmatrix}
e^{i2\pi \mathcal{L}_1/m} & 0 \\
0 & e^{i2\pi \mathcal{L}_2/m}
\end{pmatrix}
\end{equation}
for  $1<m\leq n$ being any non-unit, positive divisor of $n$;  the diagonal elements of the rotation matrix specified by the on-site angular momenta of the two representative basis orbitals.

An important concept that underlies the returning Thouless pump (to be formally defined in the next subsection) is that individual Bloch states can carry an angular momentum. To elaborate, we focus on Bloch states with momenta on any of the vertical $C_m$-invariant $\bk$-lines in the BZ. Each such line is denoted $\gamma_\Pi$, where the \textit{reduced momentum} $\Pi$ is the projection of $\gamma_{\Pi}$ onto the 2D plane perpendicular to the rotation axis. Analogously, the projection of the 3D BZ will be referred to as the \textit{reduced BZ} (rBZ). 
For each of the space groups P$n$, $n\in\{2,3,4,6\}$, all rotation-invariant $\bk$-lines are presented in Fig.~\ref{fig:wp-rotinvmom}(b). Note that when $m<n$, the $C_m$-invariant $\bk$-lines form an $n/m$-fold multiplet ($n/m$\textit{-plet}, in short) of $C_n$-related $\bk$-lines. Therefore, we will sometimes speak of the \emph{set of rotation-inequivalent rotation-invariant $\bk$-lines}, which includes only one $\bk$-line per each $n/m$-plet.

For any $\bk \in \gamma_{\Pi}$, the Hamiltonian satisfies the commutation relation:  
\begin{equation}\label{eq:ham_commut}
\begin{split}
C_m\text{-invariant}\;\bk\text{-line} \;\gamma_{\Pi}:\as&[h(\Pi, k_z), V_{\bG_{\Pi}}R_{C_m}]=0,\\
&\bG_\Pi=C_m\Pi-\Pi.
\end{split}
\end{equation}
Hence, for a Hamiltonian restricted to a rotation-invariant line, we can define an \emph{itinerant rotation matrix}
\begin{equation}
\widetilde{R}_{C_m}(\Pi)=V_{\bG_\Pi}R_{C_m},
\label{eq:itinerant_rotation}
\end{equation}
whose two eigenvalues 
\begin{equation}
    \widetilde{\rho}_{m,\alpha}(\Pi)=\exp[i\left(\frac{2\pi\mathcal{L}_\alpha}{m}+\br_{\alpha,\perp}\cdot\bG_\Pi\right)\,]\coloneqq\exp(i\frac{2\pi\widetilde{\mathcal{L}}_\alpha(\Pi)}{m})
    \label{eq:ellv-via-rho}
\end{equation}
determine the symmetry of extended/itinerant {basis} Bloch states. The eigenvalues of $\widetilde{R}_{C_m}$ are a function of the on-site angular momenta and Wyckoff positions, because $V_{\bG_\Pi}$ and $R_{C_m}$ are both diagonal matrices with
eigenvalues $e^{i\br_{\alpha,\perp}\cdot\bG_\Pi}$ and $e^{i2\pi\mathcal{L}_\alpha/m},$ respectively. The logarithm of $\widetilde{\rho}_{m,\alpha}(\Pi)$ defines a mod-$m$ \emph{itinerant angular momentum} 
\begin{equation}
    \widetilde{\mathcal{L}}_\alpha(\Pi)=\mathcal{L}_\alpha+m\frac{\br_{\alpha,\perp}\cdot\bG_\Pi}{2\pi}\mod m,
    \label{eq:itin-am-from-basis}
\end{equation}
which sums two contributions: the first from rotating the representative basis orbital $\varphi_\alpha$ about its positional center, and the second from rotating the plane-wave phase factor $e^{i\bk\cdot(\bR+\br_\alpha)}$ in the definition of the Bloch state. The commutation relation \q{eq:ham_commut} guarantees that that the Hamiltonian matrix and the itinerant rotation matrix have common eigenstates at $C_m$-invariant $\bk$-lines. The eigenvalue of the itinerant rotation matrix corresponding to the valence state of the Hamiltonian defines the itinerant angular momentum of the valence band $\widetilde{\mathcal{L}}_v(\Pi)$, which at a fixed reduced momentum $\Pi$ coincides either with $\widetilde{\mathcal{L}}_1(\Pi)$ or with $\widetilde{\mathcal{L}}_2(\Pi)$, with a possibly different choice at different $\Pi$. Defined in a similar way is the itinerant angular momentum of the conduction band $\widetilde{\mathcal{L}}_c(\Pi)$, which coincides with $\widetilde{\mathcal{L}}_2(\Pi)$ or $\widetilde{\mathcal{L}}_1(\Pi)$, respectively.

\subsubsection{Triviality of Chern class and band-representability}\la{sec:trivialchern}

We restrict ourselves to  tight-binding insulating Hamiltonians with trivial Chern class. In three momentum dimensions, the triviality of the Chern class means the vanishing of all three independent first-Chern numbers, which are defined as integrals of the Berry curvature over three independent 2D sub-tori of the 3D Brillouin torus.   

A theorem proven in \ocite{nogo_AAJH}  ensures that the bulk valence band, being nondegenerate at each $\bk$ and with trivial Chern class, is a band representation; this also holds for the bulk conduction band. It follows that our models would be missed under classification schemes in which band representations are deemed trivial, such as the methods of band structure combinatorics~\cite{bandcombinatorics_kruthoff}, topological quantum chemistry~\cite{TQC}, and symmetry indicators~\cite{Po_symmetryindicators}.

To remind ourselves that the valence band is a band representation, we sometimes refer to the valence band as a \textit{valence band representation} ($\mathrm{VBR}$). Note however that the valence band is generically \textit{not} a basis band representation, even if the two are symmetry-equivalent. The inter-orbital matrix elements of the tight-binding Hamiltonian prevent the valence band from being \emph{continuously deformable} to a basis band representation. 

\subsection{Defining the RTP and Hopf invariants for tight-binding Hamiltonians\label{sec:lattice-invariants}}

We are ready to discuss the lattice analogues of the continuum Hopf number and Zak phase discussed in Sec.~\ref{sec:berry-dipole}. 
The Hopf invariant is given by an integral of the Chern-Simons form over the 3-dimensional BZ,
\begin{equation}
    \chi = -\frac{1}{4\pi^2}\int\limits_\textrm{BZ}d^3k\,\bm{\mathcal{F}}\cdot\boldsymbol{\mathcal{A}},
    \label{eq:hopfinvar}    
\end{equation}
which differs from Eq.~(\ref{eq:hopfinvar-cont}) in two aspects: (\emph{i}) the integration domain is changed from the 3-dimensional real space to the BZ, and (\emph{ii}) the intra-cell wave function $\ket{u_v(\bk)}$, used to compute the Berry connection and curvature, is now an eigenstate of the lattice tight-binding Hamiltonian. The Hopf invariant takes integer values for any two-band insulating tight-binding Hamiltonian with a trivial Chern class; this fact has been demonstrated for $\bk$-periodic Hamiltonians~\cite{Hopfinsulator_Moore}, and we further show in Appendix~\ref{app:Hopf-periodic-nonperiodic} that the same fact holds also for $\bk$-nonperiodic Hamiltonians for which the positional centers of the tight-binding-basis orbitals are distinct.

According to the geometric theory of polarization \cite{zak_berryphase, kingsmith_polarization}, the Zak phase
of the valence-band wave function [computed over a $\bk$-line with a fixed reduced momentum $\bk_\perp=(k_x,k_y)$] is equal (up to a factor $2\pi$) to the 1D electric polarization for a one-momentum-parameter Hamiltonian, obtained by restricting $h(\bk)$ to the same $\bk$-line:
\begin{equation}
    \mathscr{P}(\bk_\perp)=\frac{1}{2\pi}\int_{-\pi}^\pi \mathcal{A}_z(\bk_\perp, k_z)dk_z.
    \label{eq:pol-def}
\end{equation}
In referring to $\mathscr{P}$ as `polarization', we have taken the liberty of setting the electron charge $e$ and the lattice constant  to unity.  

Suppose that at a $C_m$-invariant reduced momentum $\Pi$, the itinerant angular momenta of the valence band $\widetilde{\mathcal{L}}_v(\Pi)$ and of the conduction band $\widetilde{\mathcal{L}}_c(\Pi)$ are distinct. In the sense that $\{\widetilde{\mathcal{L}}_{v}(\Pi)\}\cap\{\widetilde{\mathcal{L}}_{c}(\Pi)\}=\varnothing$, we call this a
\begin{equation}\label{eq:mut-disj}
\begin{split}
&\text{\textit{mutually-disjoint condition at}}\;\Pi: \as \Delta\widetilde{\mathcal{L}}(\Pi)\neq 0. \\
&\textrm{where}\;\;\Delta\widetilde{\mathcal{L}}(\Pi)=\widetilde{\mathcal{L}}_c(\Pi)-\widetilde{\mathcal{L}}_v(\Pi)\mod{m}.
\end{split}
\end{equation}
It follows from this condition that the tight-binding Hamiltonian $h(\bk)$ \textit{is diagonal along the line $\gamma_\Pi$ in a basis that simultaneously diagonalizes the rotation matrix} [cf.~\q{eq:hamilt-sym-condition}]. 

Let us not forget that $h(\bk)_{\alpha \beta}$ is a two-by-two matrix with row index $\alpha$ corresponding to a basis Bloch state that linearly combines the representative  tight-binding-basis orbital $\varphi_{\alpha}$ with all its Bravais-lattice translations [cf.~Eq.~\eqref{eq:app:Bloch-basis}]. 
Each of the two basis Bloch states may be said to be purely of one \textit{orbital character}; clearly, there are only two orbital characters in a two-band Hamiltonian. 
The above-positioned, italicized statement  may then be equivalently stated as: any Bloch energy eigenstate (with momentum $\bk \in \gamma_\Pi$) is (up to a trivial $U(1)$ phase) a basis Bloch state with momentum $\bk$, and is thus purely of one orbital character. 
Because the Hamiltonian is assumed to be insulating, the orbital character is constant along any rotation-invariant $\bk$-line  $\gamma_\Pi$. 
If the valence-band wave function along  $\gamma_\Pi$ is of the $\varphi_\alpha$-character, then the polarization at $\Pi$ is fixed, up to an integer, to the $z$-coordinate of $\varphi_\alpha$'s Wyckoff position:
\begin{equation}
\gamma_\Pi \;\text{of}\;\varphi_\alpha\text{-character}\;\;\imp\;\; \mathscr{P}(\Pi)=^{\mathrm{P}n}_1z_{\alpha}.
\label{eqn:pol-quant}
\end{equation} 
We have added a superscript to  $=^{\mathrm{P}n}_1$ to remind ourselves that the equivalence modulo integers is robust under continuous deformations of the Hamiltonian that preserves both the bulk energy gap and the $\mathrm{P}n$ symmetry.
For a derivation of the intuitive result in \q{eqn:pol-quant}, see Appendix~\ref{app:pol-position}.

Let us assume that the mutually-disjoint condition is satisfied at a pair of rotation-invariant reduced momenta $\Pi_1$ and $\Pi_2$. Two possibilities emerge: either (a) the valence-band wave functions along $\gamma_{\Pi_1}$ and $\gamma_{\Pi_2}$ are of the same orbital character,  or (b) they are not. 
Utilizing the Greek word `\emph{isos}' meaning `equal to', we say in case (a) that the \textit{iso-orbital condition} holds at $\Pi_1$ and $\Pi_2$.
An instructive restatement of the iso-orbital condition is that there exists a single basis band representation BBR$[\varphi_\alpha]$ induced from a tight-binding-basis orbital $\varphi_\alpha$ ($\alpha\in \{1,2\}$), such that the restriction of the valence band representation $\mathrm{VBR}$  to $\gamma_{\Pi_1}$ equals\footnote{To formalize being `equal', one may say that the projector to all Bloch states (on the BZ-line $\gamma_{\Pi}$ at fixed reduced momentum $\Pi$) are equal for $\mathrm{VBR}$ and for BBR$[\varphi_\alpha]$. Note that equality of the sum of projectors over a $\bk$-line $\gamma_{\Pi}$ implies equality of projectors at each point $\bk\in\gamma_{\Pi}$ individually.} the restriction of BBR$[\varphi_\alpha]$ to $\gamma_{\Pi_1}$, and likewise for $\gamma_{\Pi_1}\ri \gamma_{\Pi_2}$:\footnote{If this equality holds at \emph{all} rotation-invariant lines $\gamma_\Pi$, we say that the valence band representation VBR is \emph{symmetry-equivalent} to the basis band representation $\mathrm{BBR}[\varphi_\alpha]$. In other words, both band representations belong to the equivalent class $\mathrm{BR(\mathcal{L}_\alpha, \br_{\alpha,\perp})}$, cf.~Sec.~\ref{sec:bandrep}. \label{foot:sym-equiv}}
\begin{equation}\la{isoorbital}
\begin{split}
&\text{\textit{iso-orbital condition at}}\;\Pi_1\;\text{and}\;\Pi_2:\as \!\!\!\exists\; \alpha\in \{1,2\}\;\; \textrm{s.t.}\\
&  \;\; \mathrm{VBR}\bigg|_{\gamma_{\Pi_1}}\!\!\!\!\!\!=\mathrm{BBR}[\varphi_\alpha]\bigg|_{\gamma_{\Pi_1}}\;\;\!\!\text{and}\;\;\mathrm{VBR}\bigg|_{\gamma_{\Pi_2}}\!\!\!\!\!=\mathrm{BBR}[\varphi_\alpha]\bigg|_{\gamma_{\Pi_2}}\!\!\!\!\!\!.
\end{split}
\end{equation}
It follows from \q{eqn:pol-quant} that the polarizations at both $\Pi_1$ and $\Pi_2$ are symmetry-fixed to $z_\alpha$, modulo integers. This contrasts with case (b), where one of the two polarizations is symmetry-fixed to $z_1$, and the other to $z_2$. In the absence of fine-tuning, one expects that $z_1\neq_1 z_2$, because by assumption no space-group symmetry relates the two tight-binding-basis orbitals.\footnote{Depending on the context, there may exist chemical principles (beyond symmetry principles) that approximately enforce $z_1=_1 z_2$. For instance, this is true of crystals containing a single atom within each primitive unit cell, with both tight-binding-basis orbitals chosen to be atomic orbitals.} 

Since our interest is to formulate an invariant difference in polarizations, let us  assume that both mutually-disjoint and iso-orbital conditions hold:  
\begin{align}
\!\!&\text{\textit{RTP invariant}}:\as    \Delta\mathscr{P}_{\Pi_1\Pi_2}=\mathscr{P}(\Pi_2)-\mathscr{P}(\Pi_1), \label{eq:RTP-def}\\
\!\!&\Delta\mathscr{P}_{\Pi_1\Pi_2}\!\in\!\mathbb{Z} \Leftarrow \!\!
\begin{cases} \text{\textit{mutually-disjoint condition}: \q{eq:mut-disj} \!\!\!\!} \\ \text{\textit{iso-orbital condition}: \q{isoorbital}}
\end{cases} \label{twoconds}
\end{align}
It is worth emphasizing that the integer-valued polarization difference is uniquely defined, because the triviality of the first Chern class allows to find a global smooth gauge for the intra-cell wave function $\ket{u_v(\bk)}$ (cf.~Refs.~\cite{Panati_trivialityblochbundle,DeNittis:2018,Peterson:1959} and footnote~12 in Ref.~\cite{nogo_AAJH}), so that the polarization $\mathscr{P}(\bk_\perp)$ is a continuous and periodic function over the rBZ.

A nonzero value of this polarization difference can be interpreted as a \emph{returning Thouless pump} (RTP) in a 2D sub-torus of the BZ containing the lines $\gamma_{\Pi_1}$ and $\gamma_{\Pi_2}$~\cite{Nelson:2021}. For this interpretation, we view the Hamiltonian as having only one momentum component $k_z$ and view the perpendicular momentum component $\bkp$ as an external adiabatic parameter. $\Delta\mathscr{P}_{\Pi_1\Pi_2}\neq0$ then means that a charge is pumped across $\Delta\mathscr{P}_{\Pi_1\Pi_2}$-number of lattice periods in the $z$-direction, when $\bkp$ is advanced from $\Pi_1$ to $\Pi_2$. The assumption of a trivial first Chern class ensures that the pumped charge is reverted back in the second half of the adiabatic cycle. 
Henceforth, we refer to the quantized polarization difference as the \emph{RTP invariant}. 

The RTP invariant is well-defined for any pair of `mutually-disjoint' reduced momenta that also satisfies the iso-orbital condition \eqref{isoorbital}. 
It turns out that a combination of symmetry ($\mathrm{P}n$), topology (triviality of Chern class) and Hilbert-space constraint (two-band Hamiltonian) imposes that the iso-orbital condition holds for any `mutually-disjoint' pair in space groups $\textrm{P}n$ with $n=3,4,6$, as is proven in Appendix~\ref{app:no-band-inv-from-Chern}. This has implications -- for the mod-$n$ relation between the Hopf and RTP invariants -- that are  discussed in \s{sec:relaxdipole}, wherein we also discuss the exceptional case of $\mathrm{P}2$.

For any $\mathrm{P}n$-symmetric two-band tight-binding Hilbert space induced by any pair of basis orbitals on $C_n$-symmetric Wyckoff positions,  we have exhaustively identified all possible pairs of reduced momenta which satisfy the mutually-disjoint condition [the first line in \q{twoconds}]. 
We summarize in Table~\ref{tab:itinerant_ell} the itinerant angular momentum difference $\widetilde{\mathcal{L}}_2(\Pi)-\widetilde{\mathcal{L}}_1(\Pi)\pmod{m}$ between two basis Bloch states for (\emph{i}) all rotation-invariant lines at (\emph{ii}) all space groups P$n$, $n\in\{2,3,4,6\}$, given (\emph{iii}) any difference of the on-site angular momenta, $\Delta\mathcal{L} = \mathcal{L}_2 - \mathcal{L}_1$, and (\emph{iv}) any positional difference: $\Delta\br_\perp=\br_{2,\perp}-\br_{1,\perp}$ [projected to the $(x,y)$-plane] of the two tight-binding-basis orbitals.
All these
options [with the sole exception of $\Delta\mathcal{L}=0, \Delta \bm{r}_\perp=(0,0)$, which is not listed in the table]
exhibit at least two rotation-invariant reduced momenta with non-zero values of $\widetilde{\mathcal{L}}_2(\Pi)-\widetilde{\mathcal{L}}_1(\Pi)\pmod{m}$. 
These momenta thus have non-zero itinerant angular momentum differences between conduction and valence states   $\Delta\widetilde{\mathcal{L}}(\Pi)=\widetilde{\mathcal{L}}_c(\Pi)-\widetilde{\mathcal{L}}_v(\Pi)\pmod{m}$ and therefore satisfy the mutually-disjoint condition.

To recapitulate, Table~\ref{tab:itinerant_ell} tells us where (in the BZ) one might in principle find a nontrivial RTP, given a certain tight-binding Hilbert space. To derive the tight-binding Hamiltonian that realizes an RTP, our strategy (elucidated in Sec.~\ref{sec:models}) is to design Berry-dipole band-touching points along any rotation-invariant $\bk$-line specified in Table~\ref{tab:itinerant_ell}.

\begin{table}
    \caption{Itinerant angular momenta differences $\widetilde{\mathcal{L}}_2(\Pi)-\widetilde{\mathcal{L}}_1(\Pi) \mod m$
    at all $C_m$-invariant reduced momenta $\Pi$ (value of $m$ given in brackets) [cf.~Fig.~\ref{fig:wp-rotinvmom}(b)] for space groups P$n$, $n\in\{2,3,4,6\}$, and for basis orbitals with all possible differences in angular momentum $\Delta\mathcal{L}=\mathcal{L}_2-\mathcal{L}_1$  and with relative displacement
    $\Delta\br_\perp=\br_{2,\perp}-\br_{1,\perp}=\sum_i\Delta r_{\perp, i} \bm{a}_i$ where $\bm{a}_i$, $i=1,2$ are primitive lattice translation vectors. 
    For $\mathrm{P}2$ and $\mathrm{P}4$ symmetries we choose $\bm{a}_1=(1,0)$ and $\bm{a}_2=(0,1)$, while for P3 and P6 symmetries $\bm{a}_1=\frac{\sqrt{3}}{2}(\sqrt{3},1)$ and $\bm{a}_2=\frac{\sqrt{3}}{2}(\sqrt{3},-1)$. We omit the cases $\Delta \mathcal{L} = 0$ with $\Delta \br_\perp =(0,0)$, when the itinerant angular momenta of the two bands coincide at all~$\Pi$. 
    }    
    \label{tab:itinerant_ell}
    \begin{ruledtabular}
    \begin{tabular}{lll@{\hspace{4mm}}llll}
    \rule{0pt}{3ex}
        Group & $\Delta\mathcal{L}$ & $\Delta\br_\perp =\! (\Delta r_{\perp,1},\Delta r_{\perp,2})$ & \multicolumn{4}{c}{$\widetilde{\mathcal{L}}_2(\Pi)-\widetilde{\mathcal{L}}_1(\Pi) \mod m$} \\ 
        & & & \multicolumn{4}{c}{for $C_m$-invariant $\Pi$} \vspace{0.7ex}\\
        \hline
        \rule{0pt}{3ex}
         & & & $\Gamma$($C_2$) & X($C_2$) & Y($C_2$) & M($C_2$) \\
         \cline{4-7}
         P2 & 0 & $(0,1/2)$ & 0 & 0 & 1 & 1 \\
         & & $(1/2, 0)$ & 0 & 1 & 0 & 1 \\
         & & $(1/2,1/2)$ & 0 & 1 & 1 & 0 
         \vspace{2mm}\\
         & 1 & $(0,0)$ & 1 & 1 & 1 & 1 \\
         & & $(0,1/2)$ & 1 & 1 & 0 & 0 \\
         & & $(1/2, 0)$ & 1 & 0 & 1 & 0 \\
         & & $(1/2,1/2)$ & 1 & 0 & 0 & 1 
         \vspace{2mm}\\
         & & & $\Gamma$($C_3$) & K($C_3$) & $\tkpr$($C_3$) & \\
         \cline{4-6}
         P3 & 0 & $(1/3,1/3)$ & 0 & 1 & 2 &  \\
         & & $(2/3,-1/3)$ & 0 & 2 & 1 &  
         \vspace{2mm}\\
         & 1 & $(0,0)$ & 1 & 1 & 1 &  \\
         & & $(1/3,1/3)$ & 1 & 2 & 0 &  \\
         & & $(2/3,-1/3)$ & 1 & 0 & 2 & 
         \vspace{2mm}\\
         & -1 & $(0,0)$ & 2 & 2 & 2 &  \\
         & & $(1/3,1/3)$ & 2 & 0 & 1 &  \\
         & & $(2/3,-1/3)$ & 2 & 1 & 0 & 
         \vspace{2mm}\\
         & & & $\Gamma$($C_4$) & X($C_2$) & $\tm$($C_4$) & \\
         \cline{4-6}
         P4 & 0 & $(1/2,1/2)$ & 0 & 1 & 2 &  
         \vspace{2mm}\\
         & 1 & $(0,0)$ & 1 & 1 & 1 &  \\
         & & $(1/2,1/2)$ & 1 & 0 & 3 &  
         \vspace{2mm}\\
         & 2 & $(0,0)$ & 2 & 0 & 2 &  \\
         & & $(1/2,1/2)$ & 2 & 1 & 0 &  
         \vspace{2mm}\\
         & -1 & $(0,0)$ & 3 & 1 & 3 &  \\
         & & $(1/2,1/2)$ & 3 & 0 & 1 &  
         \vspace{2mm}\\
         & & & $\Gamma$($C_6$) & K($C_3$) & M($C_2$) & \\
         \cline{4-6}
         P6 & 1 & $(0,0)$ & 1 & 1 & 1 &  \\
          & 2 & $(0,0)$ & 2 & 2 & 0 &  \\
          & 3 & $(0,0)$ & 3 & 0 & 1 &  \\
          & -2 & $(0,0)$ & 4 & 1 & 0 &  \\
          & -1 & $(0,0)$ & 5 & 2 & 1 &  
    \end{tabular}
    \end{ruledtabular}
\end{table}

\subsection{Berry dipole mediates a change in the RTP and Hopf invariants}\label{sec:Hopf-RTP-correspondence}

Both the Hopf and the RTP invariants of a lattice tight-binding Hamiltonian can be altered across critical phase-transition points where the conduction and valence bands touch (as a Berry dipole)  along one (or more) rotation-invariant $\bk$-line(s). 

For any Berry-dipole band touching at $C_m$-invariant momentum $\bk^\textrm{t}$, we assume that the tight-binding Hamiltonian $h(\bk)$ has a convergent, multi-variable Taylor expansion in $(\bk-\bk^\textrm{t})$, such that the truncated Taylor expansion is a $\bk\cdot \bp$ Hamiltonian of the spinor-form [cf.\ \q{eq:zszs-hamiltonian}], with spinor given by \q{eq:kp-sym-spinor} for a certain Berry-dipole spin $[\Delta \ell(\bk^\textrm{t})\in \Z]$ 
and a certain Berry-dipole helicity $[\upsilon(\bk^\textrm{t})=\pm 1]$ 
that both depend on the particular band touching $\bk^\textrm{t}$. 
For this to be possible, we additionally assume that the valence band is of a $\varphi_1$ character at $\bk^t$ -- this guarantees that the valence band in the $\bk\cdot\bp$ expansion is also of the first orbital character, as required for the Hamiltonian defined through the spinor \eqref{eq:kp-sym-spinor} [cf.~footnote~\ref{foot:spinor2}]. Comparison of  the symmetry constraints \q{eq:rot-sym-Berry-dipole-contin} and Eqs.~(\ref{eq:ham_commut}--\ref{eq:ellv-via-rho}) allows to equate (modulo $m$) the Berry-dipole spin with the difference in itinerant angular momenta of the basis Bloch states:
\begin{equation}
\Delta{\ell}(\bk^\textrm{t}) =_m \widetilde{\mathcal{L}}_2(\Pi^\textrm{t})- \widetilde{\mathcal{L}}_1(\Pi^\textrm{t})=_m\Delta\widetilde{\mathcal{L}}(\Pi^\textrm{t}),\label{eqn:l-L-mod-m}
\end{equation}
with $\Pi^\textrm{t}$ the projection of $\bk^\textrm{t}$ onto the $(k_x,k_y)$-plane;\footnote{As we show in Appendix~\ref{app:spinor-lattice-sym}, to the lowest-in-$\bk$ order, $\Delta\ell(\bk^\textrm{t})$ of corresponding continuum expansion is chosen to be the smallest (in absolute value) integer that obeys this constraint.} the Berry-dipole helicity depends on details of the Hamiltonian hopping matrix elements. 
The second equality follows from the assumption that the valence band is of $\varphi_1$ character, and hence conduction band is of $\varphi_2$ character, at $\bk^\textrm{t}$.

Recall from Sec.~\ref{sec:Berry-transitions} that the Berry dipole results in an integer-valued change of the continuum Hopf number $\delta \chi^{\mathrm{cont.}}=-\Delta \ell\, \upsilon$ according to \q{eq:Berry-dipole-hopf-change}. Recall that $\chi^{\mathrm{cont.}}$ can be viewed as a momentum-local contribution to the Hopf invariant -- a BZ-integral of the Chern-Simons three-form. Suppose our tight-binding Hamiltonian depends on a non-momentum variable $\Phi\in \R$, and that for an isolated value of $\Phi=\Phi_c$, the conduction and valence bands touch as Berry dipoles at $T$ momenta, $\{\bk^\textrm{t}\}_{\textrm{t}=1,\ldots, T}$. The net change in the Hopf invariant is then obtained by summing over all $\bk^\textrm{t}$:\footnote{If the valence band is instead of a $\varphi_2$ character, the $\bk\cdot\bp$ expansion in the spinor-form Hamiltonian will be defined via the spinor given in Eq.~\eqref{eq:kp-sym-spinor-2}. In this case, the local contributions to the Hopf invariant will instead be given~by  \begin{equation}\delta\chi^\textrm{cont.}(\bk^\textrm{t})=\upsilon(\bk^\textrm{t})\Delta \ell(\bk^\textrm{t})\label{eq:chi-total-change-inverted}
\end{equation} 
(cf.~footnote~\ref{foot:hopf-change-2}). Note that both Eqs.~(\ref{eq:chi-total-change}) and~(\ref{eq:chi-total-change-inverted}) can be formulated with the \emph{same} sign on the right-hand side if we employ the difference $\ell_c-\ell_v$ (instead of $\Delta \ell = \ell_2 - \ell_1$) at $\bk^\textrm{t}$.}
\begin{equation}
    \delta\chi=\sum_{\textrm{t}=1}^T\delta\chi^{\mathrm{cont.}}(\bk^\textrm{t}), \as \!\! \delta\chi^{\mathrm{cont.}}(\bk^\textrm{t})=-\upsilon(\bk^\textrm{t})\Delta \ell(\bk^\textrm{t}) \in \Z.
    \label{eq:chi-total-change}
\end{equation}
If $\bk^\textrm{t}$ is a $C_m$-invariant momentum in a $\mathrm{P}n$-symmetric BZ with $m<n$, then $\bk^\textrm{t}$ belongs to 
an $n/m$-plet of dipole-band-touching momenta which are mutually related by $C_n$ symmetry, and $\delta \chi^{\mathrm{cont.}}$ is identical for each member of this $n/m$-plet.\footnote{$\delta \chi^{\mathrm{cont.}}$ is identical for each member of a $n/m$-plet because the Berry-connection vector $\boldsymbol{\mathcal{A}}$ transforms as a vector under spatial operations, Berry-curvature vector $\bm{\mathcal{F}}=\curl \boldsymbol{\mathcal{A}}$ transforms as a pseudo-vector, and hence the
Chern-Simons three-form $\bm{\mathcal{F}}\cdot\boldsymbol{\mathcal{A}}$ transforms as a pseudo-scalar.} 
  
Recall that a Berry-dipole band touching at $\bk^\textrm{t}$  changes also the continuum Zak phase by $\delta \mathscr{Z}^{\mathrm{cont.}}=-2\pi \upsilon$ according to \q{eq:Berry-dipole-pol-change}. $\delta \mathscr{Z}^{\mathrm{cont.}}$ can be viewed as a momentum-local contribution to the RTP invariant $\Delta\mathscr{P}_{\Pi_1\Pi_2}$ if $\bk^\textrm{t}$ is a point on either of the rotation-invariant lines: $\gamma_{\Pi_1}$ and $\gamma_{\Pi_2}$. The net change in the RTP invariant across the critical parameter value $\Phi_c$ is contributed by all dipole-band-touching points within either line:
\begin{equation}
    \delta \Delta\mathscr{P}_{\Pi_1\Pi_2}=-\sum_{\bk^\textrm{t}\in \gamma_{\Pi_1}} \upsilon(\bk^\textrm{t})+\sum_{\bk^\textrm{t}\in \gamma_{\Pi_2}} \upsilon(\bk^\textrm{t}).
    \label{eq:rtp-change}
\end{equation}
If $\bk^\textrm{t}$ belongs to a $n/m$-plet, then $\upsilon$ is identical for each member of this multiplet.\footnote{$\upsilon$ is identical for each member of a $n/m$-plet, because $\boldsymbol{\mathcal{A}}$ transforms as a vector under spatial operations, and in particular  $\mathcal{A}_z$ transforms as a scalar under rotations with axes parallel to $z$.} 

Let us begin from a  canonical trivial phase with only on-site energies and zero `hopping'  matrix elements of the real-space tight-binding Hamiltonian. In such a phase, both valence and conduction bands are basis band representations, and the Hopf and RTP invariants vanish. By tuning the Hamiltonian, we then transit to a different gapped phase via a Berry-dipole critical point, with both Hopf and RTP invariants simultaneously modified through Eqs.~(\ref{eq:chi-total-change}) and~(\ref{eq:rtp-change}). One can eliminate the Berry-dipole helicity $\upsilon(\bk^\textrm{t})$ from both equations, thus expressing the Hopf invariant in terms of the RTP invariants -- this is done in
Sec.~\ref{sec:RTP-Hopf}. A phase transition involving multiple Berry dipoles can in principle leave the Hopf invariant unchanged, while altering the RTP invariant. This happens when the individual band-touching contributions in Eq.~\eqref{eq:chi-total-change} cancel each other.\\

We name insulators that are classified by both Hopf and RTP invariants as \emph{crystalline Hopf insulators}. 
In the next section we illustrate the richness of their classification by means of explicit tight-binding models. The RTP invariant, in contrast to the Hopf invariant, admits generalization to more than two bands~\cite{Nelson:2021} as will be discussed in details in Sec.~\ref{sec:rtpbeyond2}. Therefore, we adopt the name of \emph{RTP insulator} to define a two-or-more band insulator characterized by at least one non-zero RTP invariant.

\section{Models of crystalline Hopf insulators}
\label{sec:models}
In this section we illustrate the concept of crystalline Hopf insulators on several explicit tight-binding models. In Secs.~\ref{sec:zero-hopf}, we present an example model for which the Hopf invariant remains zero although the RTP invariants are 
nontrivial. In Sec.~\ref{sec:same-hopf-diff-rtp}, we introduce a model that exhibits two phases with
the same non-vanishing value of the Hopf invariant but which are topologically distinguished by the RTP invariant. 
Based on these examples, we claim that rotational symmetry enriches the topological classification of Hopf insulators through the introduction of the RTP invariants. 
We also note a relation between the Hopf and the RTP invariants that holds in both models. 
This relation is formally generalized and proven in Sec.~\ref{sec:RTP-Hopf}. 
Finally, in Sec.~\ref{sec:repr-models} we present several representative tight-binding models, that possess quantized values of the Hopf and RTP invariants, in all space groups P$n$. 
All the presented models are summarized in Table~\ref{tab:model-zoo} and their phase diagrams are presented in Table~\ref{tab:models_summary}.

\subsection{Hopf-less RTP insulator\label{sec:zero-hopf}}

We present a model that has a trivial Hopf invariant but non-trivial RTP invariants. 
The existence of such a model demonstrates that in the presence of crystalline symmetries the topological classification of two-band Hamiltonians (with trivial first Chern class) is not exhausted by the Hopf invariant. 

One way to realize a tight-binding model with the above-mentioned property is to start with a trivial insulator and induce a phase transition via a pair of Berry dipoles at two rotation-invariant points $\bk^1\in\gamma_{\Pi_1}$, $\bk^2\in\gamma_{\Pi_2}$, that lie on two rotation-invariant lines in the BZ.
For simplicity, we assume that both $\gamma_{\Pi_j}$ ($j=1,2$) are invariant (modulo lattice translations) under $n$-fold rotation, with $n$ the order of the maximal point group of $\mathrm{P}n$. 
The changes in the Hopf invariant attributed to the two Berry dipoles are chosen to cancel each other: $\delta\chi(\bk^1)=-\delta\chi(\bk^2)$.
The model can nonetheless acquire
a nonzero RTP invariant $\Delta\mathscr{P}_{\Pi_1\Pi_2}$, 
if the corresponding changes in polarization at these points
have opposite signs, $\delta\mathscr{P}(\Pi_1) = -\delta\mathscr{P}(\Pi_2)$.

With the RTP-protecting symmetry chosen to be a 4-fold rotation $C_{4z}$, represented by $R_{C_{4z}}=\diag(1,i)$, a model with the required characteristics
is achieved by the Hamiltonian
\begin{equation}\label{eq:Hopf-less}
\begin{split}
h(\bm{k}) \eq [\zeta^{\dagger}(\bm{k})\bm{\sigma}\zeta(\bm{k})]\cdot\bm{\sigma},\;\;\;\;\zeta(\bk)=(\zeta_1,\zeta_2)^\top \\
\zeta_1(\bm{k}) \eq \sin k_x - i\sin k_y, \\
\zeta_2(\bm{k}) \eq \sin k_z + i\,\Phi\left[\Phi + \cos k_z (\cos k_x + \cos k_y)\right],
\end{split}
\end{equation}
with $\bm{\sigma}=(\sigma_x, \sigma_y, \sigma_z)$ a vector of the Pauli matrices and $\Phi\in\mathbb{R}$ a tuning parameter. 
The model is defined on a tetragonal lattice in $\bm{r} = (x,y,z)$ space, and we rescale the dimensionless lattice constants/periods in both the vertical and the in-plane direction to $1$, such that the first Brillouin zone is given by $\abs{k_{x,y,z}}\leq \pi$.
The presented model has a pair of basis orbitals $\{\varphi_j\}_{j=1,2}$ placed along a rotational axis passing through the center of the unit cell, which
transform under rotations like a spinless $s$ and $p_+=p_x+ip_y$ orbital (i.e., they carry angular momenta $\mathcal{L}_1=0$ and~$\mathcal{L}_2=1$, respectively).

Owing to the trivial value of the Chern number on all subplanes of the BZ, the eigenstates of the Hamiltonian \eqref{eq:Hopf-less} can be expressed as periodic and analytic functions of $\bk$ over the Brillouin torus. 
We adopt the gauge in which the valence (conduction) eigenstates are
given by $\ket{u_v}=i\sigma_y  \zeta^*/\|\zeta\|$
($\ket{u_c}=\zeta/\|\zeta\|$), with $\zeta(\bk)$ an analytic and periodic function specified in \q{eq:Hopf-less}.
One easily verifies that the valence and conduction wave functions along 
each rotation-invariant line 
transform under rotation with itinerant angular momenta $\widetilde{\mathcal{L}}_v=0$ resp.~$\widetilde{\mathcal{L}}_c=1$, implying that 
the mutually-disjoint condition is fulfilled. 
The RTP invariants are the polarization differences between pairs of $C_{4z}$-invariant reduced momenta $\Gamma=(0,0)$ and $\tm=(\pi,\pi)$ and $C_{2z}$-invariant reduced momentum $\tx=(\pi,0),(0,\pi)$.  
All possible pairwise differences in polarization  can be expressed as a linear combination of two independent ones: $\Delta\mathscr{P}_{\Gamma\tx}$ and~$\Delta\mathscr{P}_{\Gamma\tm}$. 

\begin{figure}
    \centering
    \includegraphics{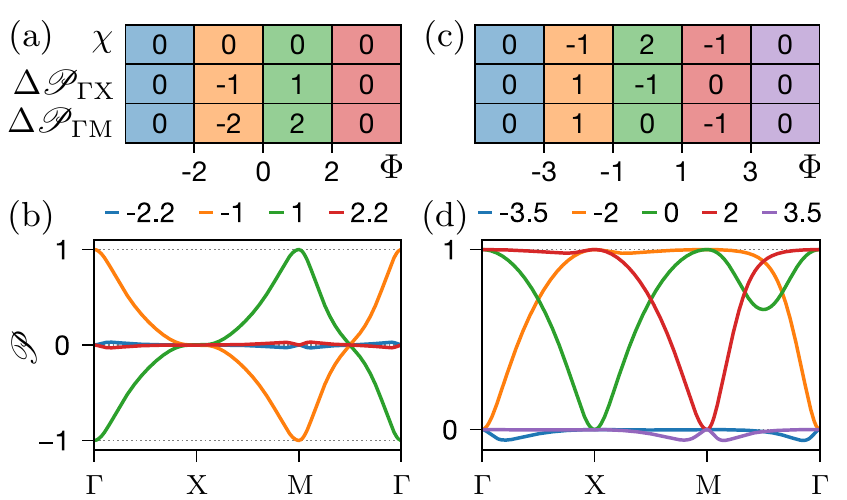}
    \caption{Phase diagram (a)~of the model with zero Hopf invariant [defined in Eq.~(\ref{eq:Hopf-less})], and (c)~of the Moore-Ran-Wen (MRW) model [described by Eq.~\eqref{eq:hamMRW}]. Both models belong to space group $\mathrm{P}4$, with $\mathcal{L}_1=0$ and $\mathcal{L}_2=1$ orbitals set along the same $C_4$ rotation axis. 
    The three rows corresponds to the values of the Hopf invariant ($\chi$) and of two quantized polarization differences between points $\Gamma$ and X ($\Delta\mathscr{P}_{\Gamma\textrm{X}}$) resp.~between $\Gamma$ and M ($\Delta\mathscr{P}_{\Gamma\textrm{M}}$) in the reduced Brillouin zone. 
    (b,d) For each phase in the phase diagram (a) resp.~(c), the $z$-component of the polarization as a function of momentum along path $\Gamma$XM$\Gamma$ is shown for a representative value of the parameter $\Phi$.}
    \label{fig:Hopfless-MRW-RTP}
\end{figure}

To verify that the Hamiltonian
indeed enters the phase with zero Hopf invariant but non-trivial RTP, we analyze in detail the phase transition at $\Phi_c=-2$. 
For large negative parameter values $\Phi\to -\infty$, the valence band is $\bk$-independent and the Hopf invariant and the differences in polarization are all zero. Increasing $\Phi$ to $-2$, the gap closes simultaneously at two BZ points $(\Gamma, 0)$ and $(\tm, \pi)$. 
Expanding the Hamiltonian in Eq.~\eqref{eq:Hopf-less} around $(\Gamma, 0)$ [resp. $(\tm, \pi)$], to quadratic order in $\phi=\Phi-\Phi_c$ and in $\bm{\kappa}=\bk-(\Gamma,0)$ [resp.~$\bm{\kappa}=\bk-(\tm, \pi)$], we get $\zeta(\bm{\kappa})=(\kappa_-, \kappa_z-2i\phi)^\top$ [resp.~$\zeta(\bm{\kappa})=(-\kappa_-, -\kappa_z-2i\phi)^\top$]. 
These spinor functions are topologically equivalent to the form in Eq.~\eqref{eq:kp-sym-spinor} with $\Delta\ell=1$ and $\upsilon=-1$ [resp.~$\upsilon=+1$]. 
According to Eqs.~(\ref{eq:Berry-dipole-hopf-change}) and~(\ref{eq:chi-total-change}) 
the Hopf invariant does not change through this phase transition while from Eq.~(\ref{eq:Berry-dipole-pol-change}) the polarization exhibits a positive (resp.~negative) unit jump at $\Gamma$ (resp.~$\tm$). 
This produces two nonzero topological invariants $\Delta\mathscr{P}_{\Gamma\tx}=-1$ and $\Delta\mathscr{P}_{\Gamma\tm}=-2$, as governed by Eq.~\eqref{eq:rtp-change}. 

Upon further increasing $\Phi$, the model defined by Eq.~(\ref{eq:Hopf-less}) exhibits two more phase transitions (namely at parameter values $\Phi\in\{0,+2\}$), and eventually returns to the trivial phase, as summarized by the phase diagram in Fig.~\ref{fig:Hopfless-MRW-RTP}(a). To illustrate the quantization of the RTP invariants we plot in Fig.~\ref{fig:Hopfless-MRW-RTP}(b) the electric polarization of the valence band [Eq.~\eqref{eq:pol-def}] as a function of rBZ momentum on a path connecting rotation-invariant momenta $\Gamma$, $\tm$ and $\tx$ for a representative parameter value in each topological phase.

\subsection{RTP invariant distinguishes between insulators with the same Hopf invariant \label{sec:same-hopf-diff-rtp}}

We now consider the possibility that the RTP invariant distinguishes
two topological insulators with the same and \emph{nontrivial} value of the
Hopf invariant. 
For this purpose we re-examine the role of symmetries of the canonical Moore-Ran-Wen (MRW) model of the Hopf insulator~\cite{Hopfinsulator_Moore}, 
\begin{equation}\label{eq:hamMRW}
\begin{split}
h(\bm{k}) \eq \left[\zeta^{\dagger}(\bm{k})\bm{\sigma}\zeta(\bm{k})\right]\cdot\bm{\sigma},\;\;\;\;\zeta(\bk)=(\zeta_1,\zeta_2)^\top \\
\zeta_1(\bm{k}) \eq \sin k_x - i\sin k_y, \\
\zeta_2(\bm{k}) \eq \sin k_z + i\left(\sum_{i=x,y,z}\cos k_i + \Phi\right).
\end{split}
\end{equation}
The tight-binding basis vectors coincide with those of the previous model in Sec.~\ref{sec:zero-hopf}, hence the MRW model is classified by the same topological invariants $\Delta\mathscr{P}_{\Gamma\tx}$ and $\Delta\mathscr{P}_{\Gamma\tm}$ in addition to the Hopf invariant $\chi$. 

By applying Eqs.~(\ref{eq:Berry-dipole-hopf-change}) and~(\ref{eq:chi-total-change})
to the phase transitions exhibited by the MRW model at $\Phi\in\{-3,-1,1,3\}$, we get the phase diagram given in Fig.~\ref{fig:Hopfless-MRW-RTP}(c); the polarization as a function of rBZ momentum for all representative parameter values is numerically plotted in Fig.~\ref{fig:Hopfless-MRW-RTP}(d). We observe that two phases with $-3<\Phi<-1$, resp.~$1<\Phi <3$,
have the same value of the Hopf invariant $\chi=-1$. However, the two phases are distinguished by the RTP invariants, i.e., it is impossible to deform one phase into the other without closing the energy gap or breaking P$4$ symmetry. 

In addition to refining the topological classification,  the RTP invariant also explains
the existence of mid-gap surface states at open boundaries with sharply terminated hoppings, as elaborated in Sec.~\ref{sec:BBC_sharp}. 
Such surface states were first numerically observed in Ref.~\cite{Hopfinsulator_Moore}; 
however, since the crystalline symmetries and the RTP invariants of the model have previously not been analyzed, the appearance of these surface states lacked explanation. We are proposing in Sec.~\ref{sec:BBC_sharp}, that  the RTP invariants not only guarantee that surface states interpolate across the bulk gap for sharply-terminated boundaries, but also {put some restrictions on }
the  shape-topology of the Fermi surface (of surface states) over the rBZ. 

From the colored panels in \fig{fig:Hopfless-MRW-RTP}(a) and (c), one can deduce a pattern among the Hopf and RTP invariants that holds for both Hopfless and MRW models [cf.\ Eqs.~(\ref{eq:Hopf-less}) resp.~(\ref{eq:hamMRW})]:
\begin{equation}
    \chi=_42\Delta\mathscr{P}_{\Gamma\tx}+\Delta\mathscr{P}_{\Gamma\tm}.
    \label{eq:hopf-rtp-example}
\end{equation}
This \textit{Hopf-RTP relation} can be rationalized by noticing that all phase transitions in both models occur via Berry dipoles at rotation-invariant momenta; each Berry dipole is responsible for an increment to the  Hopf and RTP invariants given by Eqs.~\eqref{eq:Berry-dipole-hopf-change} and \eqref{eq:Berry-dipole-pol-change}. 
This suggests a more general Hopf-RTP relation that depends on the itinerant-angular-momentum difference between the valence and conduction bands, which reduces to \q{eq:hopf-rtp-example} in the case when the itinerant-angular-momentum difference is unity. 
This general relation is derived in Sec.~\ref{sec:RTP-Hopf} for all space groups P$n$ and all possible combinations of tight-binding-basis orbitals. 
We will further reveal in that section that the Hopf-RTP relation holds whether or not the  band-topology-altering band touchings are of the Berry-dipolar form.

\subsection{Representative models of crystalline Hopf insulators in all \texorpdfstring{$\mathrm{P}n$}{Pn} space groups \label{sec:repr-models}}

Adopting the tight-binding Hilbert spaces described in Sec.~\ref{sec:lattice-ham}, we will construct tight-binding Hamiltonians in the spinor form given by Eq.~\eqref{eq:zszs-hamiltonian} with different choices of the spinor functions $\zeta(\bk,\Phi)$.  All the constructed Hamiltonians below have a trivial first Chern class and a valence band representation which is symmetry-equivalent to\footnote{Symmetry-equivalent to, but \emph{not continuously deformable to}; cf.~end of Secs.~\ref{sec:bandrep} and \ref{sec:trivialchern} for the definition of these two distinct notions.} to a basis band representation, i.e., a band representation induced from either of the tight-binding-basis orbitals. 

As shown in Appendix~\ref{app:sym-cond-spinor}, the rotation-symmetry condition in Eq.~\eqref{eq:hamilt-sym-condition} is equivalent to the following condition on the spinor function:
\begin{equation}
    R_{C_n}(\bk)\zeta(\bk)=e^{i\beta(\bk)}\zeta(C_n\bk),
    \label{eq:sym-condition-spinor}
\end{equation}
with $\beta(\bk)$ being an arbitrary phase factor.
To construct a Hamiltonian that  transforms under reciprocal lattice translations according to Eq.~\eqref{eq:ham_non-periodic}, we choose each spinor component to transform as
\begin{align}
    \zeta_j(\bk+\bG)\eq e^{-i\bG\cdot\br_j}\zeta_j(\bk) 
    \label{eq:non-periodicity-spinor}
\end{align}
for both $j=1,2$.

\begin{table*}
    \begingroup
    \centering
    \caption{Glossary of representative Hamiltonians of crystalline Hopf insulators. The model Hamiltonians, with the indicated space group symmetry P$n$ and basis orbitals $\varphi_{1,2}$ characterized by their angular momentum difference $\Delta\mathcal{L}=\mathcal{L}_2-\mathcal{L}_1$ and real-space position difference $\Delta\br=(\Delta\br_\perp, 0)$, $\Delta\br_\perp=\br_{1,\perp}-\br_{2,\perp}$, are of the form $h(\bk)=[\zeta^\dagger(\bk)\bm{\sigma}\zeta(\bk)]\cdot\bm{\sigma}$ with the spinor components $\zeta(\bk)=(\zeta_1(\bk),\zeta_2(\bk))^\top$ listed in the table. 
    Notations in the table:  $\bm{t}(a)=\sqrt{3}\left(\sin(\pi a/3), -\cos(\pi a/3)\right)^\top$, $\bk_\perp=(k_x,k_y)^\top$.
    The `Hopf-less model' of Sec.~\ref{sec:zero-hopf} is not included in the table, while the `MRW model' of Sec.~\ref{sec:same-hopf-diff-rtp} corresponds to the second listed model.
    }\label{tab:model-zoo}
    \begin{ruledtabular}
    \begin{tabular}{l@{\hspace{0.25cm}}
    l@{\hspace{0.25cm}}l@{\hspace{0.55cm}}l@{\hspace{0.5cm}}l}
        SG & $\Delta\mathcal{L}$ & $\Delta \br_\perp$ & $\zeta_1(\bk)$ & $\zeta_2(\bk)$ \\ 
        \hline
        \rule{0pt}{0.5cm}
          P4 & 0 & $\Big(\tfrac{1}{2},\tfrac{1}{2}\Big)$ & $e^{i\frac{k_x+k_y}{2}}\left[\cos\frac{k_x}{2}\cos\frac{k_y}{2}+i\sin\frac{k_x}{2}\sin\frac{k_y}{2}\sin\frac{k_x+k_y}{2}\sin\frac{k_x-k_y}{2}\right]$ & $e^{i\frac{k_x+k_y}{2}}\left[\sin k_z + i\left(\sum_{i=x,y,z}\cos k_i + \Phi\right)\right]$ \vspace{0.3cm} \\
         & 1 & $(0,0)$ & $\sin k_x -i \sin k_y$ & $\sin k_z + i\left(\sum\limits_{i=x,y,z}\cos k_i + \Phi\right)$ \vspace{0.3cm} \\
         & & $\Big(\tfrac{1}{2},\tfrac{1}{2}\Big)$ & $e^{i\frac{k_x+k_y}{2}}\left[\sin\frac{k_x+k_y}{2}+i\sin\frac{k_x-k_y}{2}\right]$ & $e^{i\frac{k_x+k_y}{2}}\left[\sin k_z + i\left(\sum\limits_{i=x,y,z}\cos k_i + \Phi\right)\right]$ \vspace{0.3cm} \\
         & 2 & $(0,0)$ & $\sin\frac{k_x+k_y}{2}\sin\frac{k_x-k_y}{2} + i\sin k_x\sin k_y$ & $\sin k_z + i\left(\sum\limits_{i=x,y,z}\cos k_i + \Phi\right)$ \vspace{0.6cm} \\
         P3 & 0 & $\Big(\tfrac{1}{3},\tfrac{1}{3}\Big)$ & $1 + 2\exp{-i\tfrac{3}{2} k_x}\cos\left\{\tfrac{\sqrt{3}}{2} k_y\right\}$ & $e^{-ik_x}\left[\sin k_z + i\left(\sum\limits_{a=1}^3\sqrt{3}\sin\left\{\bm{t}(2a)\cdot \bk_\perp\right\} + \tfrac{9}{2}\cos k_z + \Phi\right)\right]$  \vspace{0.3cm} \\
          & 1 & $(0,0)$ & $\sum\limits_{a=1}^3 e^{-i a\frac{2\pi}{3}}\exp{i\bm{t}(2a)\cdot \bk_\perp}$ & $\sin k_z + i\left(\sum\limits_{a=1}^3\sqrt{3}\sin\left\{\bm{t}(2a)\cdot \bk_\perp\right\} + \tfrac{9}{4}\cos k_z + \Phi\right)$ \vspace{0.6cm} \\
          P6 & 1 or 2 & $(0,0)$ & $\sum\limits_{a=1}^{6}e^{-i\Delta\mathcal{L} a\frac{\pi}{3} } \exp{i\bm{t}(a)\cdot\bk_\perp}$ & $\sin k_z + i\left(\sum\limits_{a=1}^{6}\cos\{\bm{t}(a)\cdot\bk_\perp\} + 4\cos k_z + \Phi\right)$ \vspace{0.3cm} \\
          & 3 & $(0,0)$ &  $\sum\limits_{a=1}^{6}e^{i\pi a} \left[\exp{i\bm{t}(a)\cdot\bk_\perp}-i\exp{i\sqrt{3}\bm{t}\big(a+\tfrac{1}{2}\big)\cdot\bk_\perp}\right]$ & as for $\Delta\mathcal{L}=1,2$
        \vspace{0.1cm}  \\
        \end{tabular}    
    \end{ruledtabular}    
    \endgroup
\end{table*}

To find model Hamiltonians that realize non-trivial RTP and Hopf invariants, we look for spinor functions that vanish $[\zeta(\bk^0, \Phi^0)=(0,0)^\top]$ at a rotation-invariant point $\bk^0$ at an isolated parameter value $\Phi^0$; in the vicinity of $\bk^0$, we want the effective $\bk\cdot\bp$ Hamiltonian to have the Berry-dipolar form given by Eq.~\eqref{eq:kp-sym-spinor}. We can then claim with certainty 
that after reopening the gap the values of the Hopf and the RTP invariants will change according to Eqs.~(\ref{eq:Berry-dipole-hopf-change}, \ref{eq:Berry-dipole-pol-change}).

In Table~\ref{tab:model-zoo}, we list representative spinors for several choices of the space group symmetry P$n$ and of the basis orbitals, the latter specified by
the difference in on-site angular momenta $\Delta\mathcal{L}=\mathcal{L}_2-\mathcal{L}_1$ and in the positional centers
$\Delta\br_\perp=\br_{2,\perp}-\br_{1,\perp}$. \medskip 

\noindent A few comments are in order: \medskip 

\noindent{(\emph{i})} A tight-binding Hamiltonian with a fixed value for $\Delta\mathcal{L}$ and   $\Delta \boldsymbol{r}_\perp$ is compatible with multiple choices of the basis orbitals whose on-site angular momenta and Wyckoff positions differ by $\Delta\mathcal{L}$ and   $\Delta \boldsymbol{r}_\perp$ respectively, as long as both Wyckoff positions are $C_n$-invariant. However, in deriving the specific forms for the spinor $\zeta$ in Table~\ref{tab:model-zoo}, we have assumed that the tight-binding-basis orbital $\varphi_1$ has trivial on-site angular momentum $\mathcal{L}_1=0$ and a $C_n$-invariant Wyckoff position with reduced coordinate $\br_{1,\perp}=(0,0)$.\footnote{$\br_{1,\perp}=(0,0)$ implies that $\zeta_1$ is periodic in reciprocal-vector translations orthogonal to the $k_z$-axis, according to  \q{eq:non-periodicity-spinor}. In the presented models the reciprocal lattice vectors in $(k_x,k_y)$ plane are $\bG_{1,\perp}=(1,0)$, $\bG_{2,\perp}=(0,1)$ in the case of $\mathrm{P}2$ and $\mathrm{P}4$ space groups and $\bG_{1,\perp}=\frac{2\pi}{3}(1, \sqrt{3})$, $\bG_{2,\perp}=\frac{2\pi}{3}(1, -\sqrt{3})$ in the case of P3 and P6 space groups.} \medskip \\

\noindent{(\emph{ii})} Table~\ref{tab:model-zoo} does not explicitly present P$2$-symmetric models.   However, a representative P$2$-symmetric model can be obtained simply by perturbatively lowering the symmetry of the explicitly-given $\mathrm{P}4$-symmetric models, e.g., by applying a small uniaxial compression in one of the directions perpendicular to the rotation axis, while maintaining the bulk energy~gap.  \medskip \\

\noindent{(\emph{iii})} In Table~\ref{tab:model-zoo}, we have listed only models with positive values for the difference in on-site angular momenta: $\Delta\mathcal{L}\geq 0$ for a given difference in Wyckoff position $\Delta \br_{\perp}$. 
Model Hamiltonians with the same $\Delta \br_{\perp}$ but with negative $\Delta\mathcal{L}$ can be obtained from the tabulated model Hamiltonians by applying complex conjugation composed with the inversion $\bk \ri -\bk$:
\e{h_{-\Delta\mathcal{L}}(\bk)= h_{\Delta\mathcal{L}}^*(-\bk).\la{timereverse}}
$h_{-\Delta\mathcal{L}}$ and $h_{\Delta\mathcal{L}}$ can be viewed as two Hamiltonians related by time reversal, which inverts angular momenta while preserving real-spatial positions.\footnote{A two-band Hamiltonian $h_{\Delta\mathcal{L}}(\bk)$ is expressed in the basis of orbitals $\varphi_{1,2}$. Adopting the complex-conjugate orbitals $\varphi_{1,2}^*$, the time-reversal operator [understood as a linear map from the Hilbert space spanned by orbitals $(\varphi_1,\varphi_2)$ to the Hilbert space spanned by $(\varphi_1^*,\varphi_2^*)$] is represented as complex conjugation; therefore, the time-reversal-related Hamiltonian is $h^\Theta_{\Delta\mathcal{L}}(\bk) = h^*_{\Delta\mathcal{L}}(-\bk)$. Additionally, time-reversal flips the on-site angular momenta, but preserves the real-space coordinates: if the basis vectors $\varphi_{1,2}$ are characterized by the onsite angular momenta difference $\Delta\mathcal{L}$ and position difference $\Delta\br$, then $\varphi_{1,2}^*$ are characterized by $-\Delta\mathcal{L}$ and $\Delta\br$. In other words, it must be that $h^\Theta_{\Delta\mathcal{L}}(\bk) = h_{-\Delta\mathcal{L}}(\bk)$. The sought Eq.~(\ref{timereverse}) follows by combining the previous two. A more formal discussion is included as Appendix~\ref{app:time-rev-ham}.} 
Because time reversal also inverts the Berry-curvature vector, the polarization and the Hopf invariant of a pair of time-reversal-related Hamiltonians are related as 
\e{\mathscr{P}(\Pi;h_{\Delta\mathcal{L}})=\mathscr{P}(-\Pi;h_{-\Delta\mathcal{L}}),\as \chi[h_{\Delta\mathcal{L}}]=-\chi[h_{-\Delta\mathcal{L}}].\la{timerevinv}}

\section{Bulk-boundary correspondence of RTP}\label{sec:BBC}

The key notion that underlies the bulk-boundary correspondence of RTP insulators is that of an \emph{angular-momentum anomaly} associated to a rotation-symmetric boundary -- roughly speaking, this means the itinerant angular momentum values of surface states are non-identical to the itinerant angular momentum values of bulk-conduction states, as well as of bulk-valence states. 
This anomaly of the surface states derives from a nontrivial RTP invariant in the bulk, as proven in \s{sec:anomaly}.

The surface states having anomalous values of itinerant angular momenta has two implications that we explore subsequently.
First, in \s{sec:BBC_sharp} we show that, in the presence of an open boundary with sharply terminated hoppings, the surface-state energies necessarily interpolate across the bulk gap to connect the bulk-valence and bulk-conduction bands, ensuring that there are gapless excitations no matter where the Fermi level is positioned within the bulk energy gap.  

When the boundary is not sharp, this interpolation across the bulk energy gap is not guaranteed by a nontrivial RTP invariant. 
Nevertheless, \s{sec:BBC_zak} shows that a more general topological feature survives that is encoded not in the energy dispersion but rather in the wave function of surface-localized states. 
Namely, the Berry-Zak phase of surface states  along a certain symmetrically-chosen loop in the rBZ is quantized to rational multiples of $2\pi$, and  moreover these quantized values for the Berry-Zak phase are non-identical to the Berry-Zak phase of bulk-conduction Bloch states (also non-identical to the  Berry-Zak phase of bulk-valence Bloch states) along an equivalent loop in the bulk BZ. 
We call this a \textit{Zak-phase anomaly}.

\subsection{Angular-momentum anomaly}\la{sec:anomaly}

Here, we formulate the angular-momentum anomaly for a two-band, P$n$-symmetric, insulating Hamiltonian with trivial Chern class. Firstly in \s{sec:anomalyrough}, a sketchy formulation is given to quickly convey the essence of the anomaly.
The admitted vagueness in \s{sec:anomalyrough} will be dispelled \s{sec:anomalyprecise} by a detailed exposition of the projected position operator, conjoined with a case study of the $\mathrm{P}4$-symmetric Moore-Ran-Wen model. 
Much of this subsection will rely on model-dependent, pictorial illustrations;  a more formal, model-independent proof is deferred to Appendix~\ref{app:ang-mom-anomaly}.

\subsubsection{Angular-momentum anomaly: sketchy introduction}\la{sec:anomalyrough}

By assumption, there exists a pair of rotation-invariant reduced momenta, labelled $\Lambda$ and $\Xi$, where both mutually-disjoint and iso-orbital conditions are fulfilled [cf.\ \q{twoconds}]. 
This implies that the difference in  polarization between $\Lambda$ and $\Xi$ is quantized to integer values. 
For convenience, the labels for the reduced momenta are chosen such that the RTP invariant $\Delta\mathscr{P}_{\Lambda\Xi}$ [defined in Eq.~(\ref{eq:RTP-def})] is positive. 

Our goal is to elucidate what implications this bulk invariant has for a p\textit{$n$-symmetric facet}, with the p$n$ plane group being a subgroup (of the 3D space group P$n$) that is  generated by  $C_n$ rotation as well as translations in the directions perpendicular to the rotation axis. For concreteness, we consider a RTP insulator on a semi-infinite geometry defined for $z>0$, with the $z$-axis coinciding with the rotation axis; the boundary at $z=0$ forms a `bottom' p$n$-symmetric facet.

The angular-momentum anomaly is a property of the itinerant angular momenta of the Hilbert space associated to this semi-infinite geometry. 
Roughly speaking, a nontrivial RTP invariant guarantees that certain minimal number of bands of states localized to the surface have itinerant angular momenta that are distinct from those of the bulk bands; this minimal number of `anomalous' bands is nothing more than the absolute value of the RTP invariant itself.  

To characterize what exactly makes an `anomalous' band distinct from a bulk band, we will need to define four numbers that count Bloch states with certain itinerant angular momenta. 
Firstly, let us define $\sharp_f\widetilde{\mathcal{L}}_v(\Lambda)$ as the number of (linearly independent) rotation-invariant Bloch states\footnote{Whenever we discuss a 3D Hamiltonian in semi-infinite geometry ($z>0$) or slab geometry ($z_1 < z < z_2$), we use the term \emph{Bloch state} to describe any state with a well-defined in-plane (reduced) momentum $\bk_\perp=(k_x,k_y)$.} 
in our Hilbert space that (\emph{i}) have reduced momentum $\Lambda$, (\emph{ii}) has support in the proximity of the surface, and (\emph{iii}) have itinerant angular momentum coinciding with that of the bulk-valence band [$\widetilde{\mathcal{L}}_v(\Lambda)$]. 
We take a `state' localized near a surface to mean  an eigenstate of the projected position operator, with an eigenvalue that is close to the spatial coordinate of the surface; this will be formalized below in \s{sec:anomalyprecise}. 
Such a state is thus not necessarily an eigenstate of the Hamiltonian. 
We define $\sharp_f\widetilde{\mathcal{L}}_c(\Lambda)$ almost identically as $\sharp_f\widetilde{\mathcal{L}}_v(\Lambda)$, except in (\emph{iii}) we replace `bulk-valence' with `bulk-conduction', and $\widetilde{\mathcal{L}}_v\ri \widetilde{\mathcal{L}}_c$.
Lastly, we define two other numbers $\sharp_f\widetilde{\mathcal{L}}_{c,v}(\Xi)$ for states satisfying (\emph{i}--\emph{iii}) with $\Lambda \ri \Xi$. 

We propose that these counting numbers are related to the bulk RTP invariant as 
\e{ -\sharp_f \widetilde{\mathcal{L}}_v(\Lambda)+\sharp_f\widetilde{\mathcal{L}}_v(\Xi)=\sharp_f \widetilde{\mathcal{L}}_c(\Lambda)-\sharp_f\widetilde{\mathcal{L}}_c(\Xi)=\Delta\mathscr{P}_{\Lambda\Xi}.\la{amanomaly}}
We call this relation an \textit{angular-momentum anomaly}. Our emphasis on `bottom facet' is merely to fix a sign convention; if we had chosen instead a p$n$-symmetric facet at the `top' of a semi-infinite bulk, then \q{amanomaly} holds with replaced $\Delta\mathscr{P}_{\Lambda\Xi}\ri -\Delta\mathscr{P}_{\Lambda\Xi}$ on the right-hand side.

The condition $\Delta\mathscr{P}_{\Lambda\Xi}>0$ implies there are at least $|\Delta\mathscr{P}_{\Lambda\Xi}|$-number of `anomalous' surface-localized bands (\emph{a})~whose itinerant angular momentum at $\Lambda$ is identical to that of the bulk-conduction band, but distinct from that of the bulk-valence band, and (\emph{b})~whose itinerant angular momentum at $\Xi$ is identical to that of the bulk-valence band but distinct from that of the bulk-conduction band. The sense in which the angular-momentum values of surface-localized bands are anomalous is that they could neither derive from bulk conduction states alone, nor could they derive from bulk valence states alone.

\subsubsection{\texorpdfstring{$\!\!\!\!$}{}Angular-momentum anomaly: precise formulation and case study}\la{sec:anomalyprecise}

\begin{figure}
    \centering
    \includegraphics{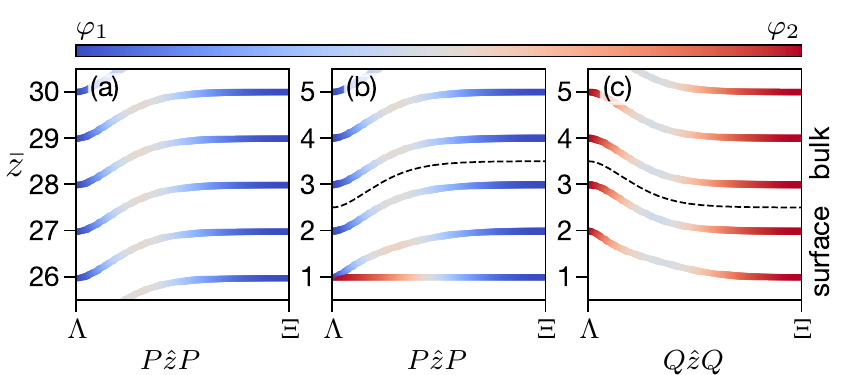}
    \caption{For a P4-symmetric Hamiltonian with the RTP invariant $\Delta\mathscr{P}_{\Lambda\Xi}=\mathscr{P}_\Xi-\mathscr{P}_\Lambda = +1$, we plot the spectrum of the projected position operators $P\hat{z}P$ and $Q\hat{z}Q$, along a straight-line path $\Lambda\Xi$ in rBZ. 
    In panel (a)~we focus only  on the bulk-like polarization bands; whereas in (b)~and (c) we recognize both bulk-like and surface-like polarization bands. 
    The color indicates the orbital decomposition of the wave function into the basis orbitals, with blue and~red tones indicating the wave function overlap with  
    $\varphi_1$ resp.~$\varphi_2$. 
    In this particular model, a rotation-invariant Bloch-Wannier state that overlaps with $\varphi_1$ (resp. $\varphi_2$) has a bulk-valence-like (resp.\ bulk-conduction-like) itinerant angular momentum, thus it is possible to infer the four numbers ($\sharp_f\widetilde{\mathcal{L}}_{c,v}(\Xi), \sharp_f\widetilde{\mathcal{L}}_{c,v}(\Lambda)$) from the above plots. 
    The dashed black line in panels (b) and (c) indicates (an arbitrary) partitioning of the polarization bands into bulk-like vs.~surface-like, and serves an argument presented in the text of Sec.~\ref{sec:anomalyprecise}.
    } 
    \label{fig:PzP}
\end{figure}

Let us elaborate on the meaning of `states localized near a surface' as eigenstates of the projected position operator in a semi-infinite geometry.  

Before tackling the more complicated case of a semi-infinite geometry, it is worth 
reviewing some basic properties of the projected position operator for an infinite-sized, discrete-translation-symmetric lattice. 
Assuming that a bulk energy gap exists, one can define the projector $P_b$ to the bulk-valence band of Bloch states, and $Q_b$ to the bulk-conduction band of Bloch states. 
Consider the projected position operators $P_b\hat{z}P_b$ or $Q_b\hat{z}Q_b$, with $(x,y,z)$ a set of Euclidean coordinates such that the $z$-axis lies parallel to the rotation axis. 
Because the projected position operator is invariant under a subgroup (of $\mathrm{P}n$) containing all translational elements confined within the $(x,y)$-plane, eigenstates of the projected position operator can be labelled by a reduced momentum $\bkp$. 
Such eigenstates which are extended (as a Bloch wave) in the $(x,y)$-plane also turn out to be maximally- and exponentially-localized (as a Wannier function) in the $z$ direction \cite{marzari1997}; we therefore call these eigenstates \textit{hybrid Bloch-Wannier states}. 

Suppose a Bloch-Wannier state is an eigenstate of $P_b\hat{z}P_b$ with eigenvalue $\bar{z}$ at $\bk_\perp$, then it must be that the spectrum of $P_b\hat{z}P_b$ at $\bk_\perp$ contains also  $\bar{z}$ plus any integer,\footnote{This follows from the Hamiltonian being symmetric under translations parallel to the rotation axis. Denoting the operator for such translations as $t_z$, it holds that $t_z\hat{z}t_z^{-1}=\hat{z}+1$. Given that $t_z$ is a symmetry, it follows that $[t_z,P_b]=0$, and hence $t_z P_b\hat{z}P_bt_z^{-1}=P_b(\hat{z}+1)P_b$.} i.e., the spectrum has a ladder-like structure, which we illustrate in Fig.~\ref{fig:PzP}(a) for the MRW model given in Eq.~\eqref{eq:hamMRW}.
Each `rung' of the ladder (indexed by integer $j$) is generically dispersive over the rBZ, and this dispersion is described by a continuous function  $\bar{z}_j(\bkp)$ which covers an interval (or `band') on the $z$-axis. 
In the context of two-band, insulating Hamiltonians, there is only a single Bloch-Wannier state with a given $\bkp$ in any unit-length interval of the $z$-axis, and the dispersive function $\bar{z}_j(\bkp)$ can be identified, modulo one, with the polarization function $\mathscr{P}(\bkp)$ \cite{kingsmith_polarization}. 
We describe all Bloch-Wannier eigenstates corresponding to  $\bar{z}_j(\bkp)$ as belonging to the $j$-th \textit{polarization band}, in analogy with how eigenstates of a discrete-translation-invariant Hamiltonian can be sorted into energy bands. The positive value $\Delta\mathscr{P}_{\Lambda\Xi}=\mathscr{P}_\Xi-\mathscr{P}_\Lambda>0$ implies that the ladder-rung function $\bar{z}_j(\bkp)$ is displaced at $\bk_\perp=\Xi$ relative to $\Lambda$ by a positive number $\Delta\mathscr{P}_{\Lambda\Xi}\in\mathbb{Z}^+$ of unit cells, as illustrated for $\Delta\mathscr{P}_{\Lambda\Xi}=1$ in~Fig.~\ref{fig:PzP}(a).

What implications does this bulk displacement have for surface-localized states in our semi-infinite slab geometry? 
Let us assume the  Hamiltonian matrix elements are terminated near the surface at $z=0$ such that a Fermi level exists within the bulk energy gap and does not intersect any surface-localized energy eigenstates; the specific values of these matrix elements are not essential to our argument. 
Then one can  define a projector $P$ to all occupied states, with the restriction of $P|_{\bkp}$ (to a reduced-momentum sector) being a smooth\footnote{The `smooth-in-$\bkp$' condition would fail if the Fermi level crossed a surface state -- in such a scenario the surface state would discontinuously move from projector $P$ to $Q$ as the state's energy passes through the Fermi level.} function of $\bkp$; likewise, $Q$  projects to all unoccupied states. 
We show in Fig.~\ref{fig:PzP}(b) and~(c) the spectrum of $P\hat{z}P$ and $Q\hat{z}Q$ respectively, with each curved line associated to a polarization~band.\footnote{We continue to call the eigenstates of $P\hat{z}P$ and of $Q\hat{z}Q$ in the semi-infinite geometry (and also in the slab geometry) as Bloch-Wannier states (and their $\bkp$-dependent eigenvalues as polarization bands), just as we did for the case of bulk states in infinite geometry, even though the eigenvalues $\bar{z}_j(\bk_\perp)$ (especially for the eigenstates localized near the surface) are now \emph{not} identified with the bulk polarization as in Ref.~\cite{kingsmith_polarization}.}

Let us sort polarization bands into two categories: \emph{bulk-like} vs.~\emph{surface-like}. 
Being `bulk-like' means to be indistinguishable from bulk polarization bands up to corrections that are exponentially small for large $z>0$. Let us then pick a polarization band of $P\hat{z}P$ that lies sufficiently far away from the surface to be bulk-like. (Which particular band is picked is not essential.) 
This chosen band serves as a cutoff: every polarization band lying above the chosen band (on the $z$-axis) is also categorized as bulk-like; every polarization band lying below the chosen band is then surface-like. We repeat this sorting for the polarization bands of $Q\hat{z}Q$. The net result is illustrated in \fig{fig:PzP}(b,c), with bulk-like and surface-like bands separated by a dashed line.

To reveal the angular-momentum anomaly, we first count how many bulk-like Bloch-Wannier states (coming from both projectors $P$ and $Q$) exist with reduced momentum $\bkp=\Lambda$ and \emph{bulk-valence-like} itinerant angular momentum $\widetilde{\mathcal{L}}_v(\Lambda)$.\footnote{Since $\widetilde{\mathcal{L}}_v(\Lambda)$ is the itinerant angular momentum of bulk valence Bloch states with reduced momentum $\Lambda$, we call this value of itinarant angular momentum in the following discussion as \textit{bulk-valence-like} at $\Lambda$ (and similarly for $\bk_\perp=\Xi$). \label{foot:bulk-valence-like}}
We call this number $\sharp_b \widetilde{\mathcal{L}}_v(\Lambda)$;  
likewise, we define the number $\sharp_b \widetilde{\mathcal{L}}_v(\Xi)$ at reduced momentum $\bkp=\Xi$.
In fact, both numbers $\sharp_b \widetilde{\mathcal{L}}_v(\Lambda)$ and $\sharp_b \widetilde{\mathcal{L}}_v(\Xi)$ are formally infinite, and only their difference is well-defined.
Numbers $\sharp_b \widetilde{\mathcal{L}}_c(\Lambda)$ and $\sharp_b \widetilde{\mathcal{L}}_c(\Xi)$ are defined similarly using the \emph{bulk-conduction-like} itinerant angular momentum.\footnote{
The notion of \textit{bulk-conduction-like} itinerant angular momentum at $\Lambda$ is defined similarly to footnote~\ref{foot:bulk-valence-like} with using the bulk conduction Bloch states with replaced $\widetilde{\mathcal{L}}_v(\Lambda)\to \widetilde{\mathcal{L}}_c(\Lambda)$. 
Note that in our terminology, every Bloch-Wannier state at each $\Lambda$ and $\Xi$ [whether surface-like or bulk like in the partitioning via the black dashed line in Fig.~\ref{fig:PzP}(b,c)] can be characterized as having either bulk-valance like or bulk-conduction-like itinerant angular momentum.} 

In \fig{fig:PzP}(b,c), any blue data-point at $\Pi=\Xi$ or $\Lambda$ corresponds to a Bloch-Wannier state with itinerant angular momentum $\widetilde{\mathcal{L}}_v(\Pi)$. 
Since there are no blue data in panel (c), Fig.~\ref{fig:PzP}(b) should convince the reader that the difference of $\sharp_b \widetilde{\mathcal{L}}_v(\Pi)$ between the two high-symmetry points equals the polarization difference between the two choices of $\Pi$, i.e., 
\begin{subequations}\label{eqn:equalsrtp+analogously}
\e{ \sharp_b \widetilde{\mathcal{L}}_v(\Lambda)-\sharp_b\widetilde{\mathcal{L}}_v(\Xi)=\Delta\mathscr{P}_{\Lambda\Xi}.\la{equalsrtp}}
Similarly, for Bloch-Wannier states with bulk-conduction-like angular momentum we find [by inspecting red data in \fig{fig:PzP}(b,c)] that 
\e{ \sharp_b \widetilde{\mathcal{L}}_c(\Lambda)-\sharp_b\widetilde{\mathcal{L}}_c(\Xi)=-\Delta\mathscr{P}_{\Lambda\Xi}.\la{analogously}}
\end{subequations}
The minus sign on the right-hand side of \q{analogously} reflects that the RTP invariant of the bulk conduction band is opposite in sign to the RTP invariant of the bulk valence band.

Observe that the total number of polarization bands (both bulk-like and surface-like) with angular momentum $\widetilde{\mathcal{L}}_v(\Pi)$ cannot depend\footnote{
This follows from the set of \textit{all} polarization bands (both bulk-like and surface-like, coming from both projects $P$ and $Q$) forming basis vectors of the Hilbert space $\mathscr{H}_{\textrm{semi-}\infty}$ in the semi-infinite geometry. 
Defining $\mathscr{H}_{\textrm{semi-}\infty}(\bkp)$ as the restriction of the Hilbert space to a reduced-momentum sector, continuity over the rBZ demands that the dimension of $\mathscr{H}_{\textrm{semi-}\infty}(\bkp)$ cannot depend on $\bkp$. (Strictly speaking, the dimension is infinite, and what we mean is that $\mathrm{dim}[\mathscr{H}_{\textrm{semi-}\infty}(\bkp)]-\mathrm{dim}[\mathscr{H}_{\textrm{semi-}\infty}(\bkp')]=0$ for any $\bkp$ and $\bkp'$.) 
Under our mutually-disjoint assumption, $\mathscr{H}_{\textrm{semi-}\infty}(\Pi)$, with $\Pi\in \{\Lambda,\Xi\}$, can be decomposed into two itinerant-angular-momentum sectors: 
\e{\mathscr{H}_{\textrm{semi-}\infty}(\Pi)=\mathscr{H}_{\textrm{semi-}\infty}(\Pi;\widetilde{\mathcal{L}}_v)\oplus \mathscr{H}_{\textrm{semi-}\infty}(\Pi;\widetilde{\mathcal{L}}_c).}
Now the dimensions of both Hilbert spaces on the right-hand side are equal by construction. This fact, combined with the constancy of dim[$\mathscr{H}_{\textrm{semi-}\infty}(\bkp)$] over the rBZ, implies \q{cannotdepend}.} on the choice between $\Pi=\Lambda$ and $\Pi=\Xi$:
\e{ \sharp_b \widetilde{\mathcal{L}}_v(\Lambda)+\sharp_f \widetilde{\mathcal{L}}_v(\Lambda)=\sharp_b \widetilde{\mathcal{L}}_v(\Xi)+\sharp_f \widetilde{\mathcal{L}}_v(\Xi).\la{cannotdepend}}
Combining \q{cannotdepend} with Eqs.~(\ref{eqn:equalsrtp+analogously}) gives the angular-momentum anomaly \eqref{amanomaly}.

The mismatch between the number of surface-like Bloch-Wannier bands with bulk-valence-like itinerant angular momentum at $\Lambda$ vs.~at $\Xi$ implies an anomaly, where some of these bands must change their character between $\Lambda$ and $\Xi$.
For instance, assuming $\Delta\mathscr{P}_{\Lambda\Xi}=1$,  \q{amanomaly} implies the existence of \textit{at least} one anomalous surface-like polarization band whose itinerant angular momentum is bulk-conduction-like at $\Lambda$ and bulk-valence-like at $\Xi$. 
In our model, the anomalous surface-like band lies at the bottom of \fig{fig:PzP}(b) and 
exhibits the largest color range. 
(That there exists only one, and that it
lies in the $P$ subspace, are model-dependent details.)
A non-pictorial, model-independent proof of the angular-momentum anomaly is presented in Appendix~\ref{app:ang-mom-anomaly}.

Importantly, the angular-momentum anomaly has implications beyond polarization bands. 
This is because the surface-like polarization band(s) are continuously deformable to surface-localized \emph{energy} band(s) of a Hamiltonian with open boundary conditions, modulo addition or subtraction of bands whose itinerant angular momenta are bulk-valence-like at all $\Pi$, or bulk-conduction-like at all $\Pi$. 
(Examples of analogous deformations were discussed in \ocite{fidkowski_bulkboundary} and \ocite{Cohomological}.) 
As a case in point, the anomalous surface-like polarization band in  \fig{fig:PzP}(b) is continuously deformable to the surface-localized energy band in \fig{fig:mrw_surface}(c). 
Henceforth, any statement that applies equally well to surface-localized energy bands and surface-like polarization bands is said to apply to \textit{faceted bands}, e.g., the angular-momentum anomaly holds for faceted bands. 
The states that make up a faceted band will be referred to as \textit{faceted states}.

Further implications of the anomaly are discussed in the subsequent two subsections.\\

\subsection{Conditionally-robust surface states at sharp boundaries \label{sec:BBC_sharp}}

The RTP implies \textit{conditionally-robust} surface states; namely,
the interpolation of surface-state energies across the bulk energy gap is robust against gap- and rotation-symmetry-preserving deformations of the \textit{bulk} Hamiltonian, assuming that the surface Hamiltonian is rotation-symmetric and satisfies the \textit{sharp boundary condition}.
`Sharpness' of a boundary means that 
all Hamiltonian elements are exactly as they would be in the bulk, except for hopping matrix elements crossing a surface; these latter elements are set to zero.
Examples of such conditionally-robust surface states were numerically observed for the rotation-symmetric Hopf-insulating models in Refs.~\cite{Hopfinsulator_Moore, hopfsurfacestates_deng}, but have thus far lacked explanation.

Here, we explain the relation of these conditionally-robust surface states to the RTP invariants. 
First, in Sec.~\ref{sec:sharp-BBC} we use spectral properties of Toeplitz and circulant matrices to derive a bulk-boundary correspondence for surface states along lines connecting high-symmetry momenta in rBZ. In Sec.~\ref{sec:fermiology} we discus how much information the RTP invariants convey about the actual shape-topology of the Fermi lines formed by the surface states inside rBZ. Finally, in Sec.~\ref{sec:sharpboundaryP4} we compare our theoretical discussion against concrete numerical calculations performed for the MRW model.

\subsubsection{Bulk-boundary correspondence at sharp boundaries}\label{sec:sharp-BBC}

To prove this conditional robustness, we first focus on rotation-invariant lines (denoted $\gamma_\Pi$) that project onto reduced momentum $\Pi\in\{\Lambda,\Xi\}$. 
By restricting the bulk Hamiltonian $h(\bk)$ to any of these lines, we obtain a Hamiltonian $h_{\Pi}(k_z)$ that depends continuously on one parameter $k_z$. 
Because of the mutually-disjoint condition \eqref{eq:mut-disj} at $\Pi$, $h_{\Pi}(k_z)$ is a diagonal two-by-two matrix (for all $k_z$), in the basis that simultaneously diagonalizes the rotation matrix. 
Thus, $h_{\Pi}(k_z)=h^v_{\Pi}(k_z)\oplus h^c_{\Pi}(k_z)$ can be interpreted as a direct sum of Hamiltonians for two non-hybridizing chains, with each chain having one orbital per primitive unit cell. 

Generally, the Hamiltonian of a one-orbital-per-cell chain with sharp boundaries is equal to a Hermitian Toeplitz matrix.\footnote{Elements of a Toeplitz matrix $\mathcal{T}$ are constant along each diagonal, i.e., $\mathcal{T}_{i,j}=\mathcal{T}_{i-j}$. Hermiticity further implies $\mathcal{T}_{i-j}=\mathcal{T}^*_{j-i}$.} 
As was shown in Ref.~\cite{toeplitz_zhu2017} for a sufficiently long chain, the spectrum of a Toeplitz matrix is asymptotically equal to the spectrum of a corresponding circulant matrix\footnote{Circulant matrices form a special class of Toeplitz matrices with $c_{i-j}=c_{-n+i-j}$ for matrices of size $n\times n$. Given a Toeplitz matrix $A_{i,j}=a_{i-j}$ we define the corresponding circulant matrix as \begin{equation}c_{i-j}=\begin{cases}a_{0} & \textrm{for}\;i=j, \\ a_{-i+j}+a_{n-i+j} & \textrm{otherwise.}\end{cases}\end{equation}} which describes the same tight-binding model with \emph{periodic} (instead of open) boundary conditions. 
It is known that a Hamiltonian defined on a periodic chain of length $L$, with one orbital per unit cell,  gives a continuous energy band in the limit $L\rightarrow \infty$. 
In particular, there are no isolated states outside of the (single) energy band for both $h^v_{\Pi}(k_z)$ and $h^c_{\Pi}(k_z)$.
Crucially, per the result of Ref.~\cite{toeplitz_zhu2017}, it follows that the continuity of the spectra of $h^v_{\Pi}$ and $h^c_{\Pi}$ persists in the presence of open boundary conditions on the chain, meaning that the two-orbital-per-cell chain $h_{\Pi}$ has no in-gap states at sharp boundaries.

Here, the angular-momentum anomaly enters the stage.  As argued in \s{sec:anomaly}, \q{amanomaly} with $\Delta\mathscr{P}_{\Lambda\Xi} > 0$  implies the existence of  $|\Delta\mathscr{P}_{\Lambda\Xi}|$-number of  `anomalous'  bottom-surface-localized energy band(s) whose itinerant angular momentum is bulk-conduction-like at $\Lambda$ and bulk-valence-like at $\Xi$. 
Moreover, the absence of any discrete energy eigenvalues within the bulk energy gap at both $\Lambda$ and $\Xi$ (explained in the previous paragraph) ensures that each of the anomalous surface-localized energy bands is energetically attached to the bulk valence band at $\Xi$, and also energetically attached to the bulk conduction band at $\Lambda$.
Thus, if we associate one $E(\bk_\perp)$ (i.e., energy vs.~reduced-momentum) dispersion to each linearly-independent, surface-localized Bloch state,
it is guaranteed that at least $|\Delta\mathscr{P}_{\Lambda\Xi}|$-number of surface-state dispersions interpolate across the bulk gap as the reduced momentum is varied along \textit{any} (possibly curved) line connecting $\Lambda$ and $\Xi$. 

\begin{figure}
    \centering
    \includegraphics{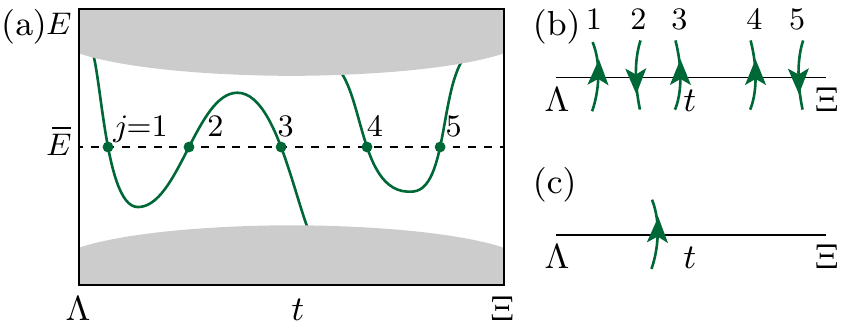}
    \caption{(a) Spectrum of a semi-infinite RTP insulator with an RTP invariant $\Delta\mathscr{P}_{\Lambda\Xi}=1$ along a path in rBZ, that starts at $\Lambda$, ends at $\Xi$ and is parametrized by a real parameter $t$. 
    Given any reference energy $\bar{E}$ inside the bulk energy gap, the value of the RTP invariant is equal, up to a sign, to the sum of $\sign[dE_j/dt]$ (signs of scalar velocities) over all surface-localized states ($j$) that cross $\bar{E}$ (green dots), cf.~Eq.~(\ref{eq:RTP-velocity}).
    (b) The surface-localized states form Fermi lines (green lines) at the Fermi energy with an orientation defined by green arrows. (c) There is an ambiguity in the exact number and position of the Fermi lines which is assigned to (\emph{i}) the possibility to add or remove pairs of Fermi states with opposite signs of scalar velocity, and (\emph{ii}) the possibility to shift a Fermi state along the $\Lambda\Xi$ line without removing it. 
    In this sense panels (b) and (c) are topologically equivalent, and can both arise for $\Delta\mathscr{P}_{\Lambda\Xi}=1$.}
    \label{fig:surf_spec}
\end{figure}

To be more precise, let us parametrize any such a line as $\bk_{\perp}(t)$ with $t\in[0,1]$ [where $\bk_\perp(0)=\Lambda$ and $\bk_\perp(1)=\Xi$]. 
Then for any fixed reference energy $\bar{E}$ within the bulk energy gap, one can associate to each linearly-independent Bloch state (at energy $\bar{E}$) a scalar velocity $dE_j/dt$, with $j$ an index to distinguish between multiple Bloch states crossing $\bar{E}$ [with possibly several such states along a given band, cf.~Fig.~\ref{fig:surf_spec}(a)]. 
These velocities are related to the RTP invariant as
\e{ \Delta\mathscr{P}_{\Lambda\Xi}=-\sum_j \sgn\bigg[\f{dE_j}{dt}\bigg], \label{eq:RTP-velocity}}
where $\sum_j$ sums over all states at the reference energy [denoted by green dots in Fig.~\ref{fig:surf_spec}(a)]. 
Thus, there are minimally $|\Delta\mathscr{P}_{\Lambda\Xi}|$-number of Bloch states at the reference energy, all with the same sign for the velocity. 
In principle, however, this number can be increased by any even number, so long as each added pair of Bloch states comes with opposite signs for $dE_j/dt$. 
Such possible additions are attributed to (\emph{i}) non-anomalous surface-localized energy bands whose itinerant angular momenta are either bulk-conduction-like at both  $\Lambda$ and~$\Xi$, or bulk-valence-like at both  $\Lambda$ and~$\Xi$ or to (\emph{ii}) additional crossings of $\bar{E}$ by the anomalous band when parameter $t$ swipes the interval $[0,1]$.

Figure~\ref{fig:surf_spec}(a) illustrates how the counting in \q{eq:RTP-velocity} is done for a schematic energy spectrum representing the case of $\Delta\mathscr{P}_{\Lambda\Xi}=1$: there are five surface-localized Bloch states at the reference energy $\bar{E}$, of which four have pairwise opposite signs for their scalar velocity. One pair [denoted by $j=4,5$ in Fig.~\ref{fig:surf_spec}(a)] can be associated to a non-anomalous surface-localized energy band `peeled' off from the bulk-conduction valence band. Another pair [denoted by $j=2,3$ in Fig.~\ref{fig:surf_spec}(a)] is attributed to additional crossings of energy $\bar{E}$ by the anomalous band. Contribution from these pairs vanishes in Eq.~\eqref{eq:RTP-velocity} leading to a correct prediction of the RTP invariant.

\subsubsection{Fermi-line topology}\label{sec:fermiology}

The above argument holds for any $\bar{E}$ within the bulk energy gap, and holds in particular for $\bar{E}$ being the Fermi energy of the bulk insulating Hamiltonian for an electronic system. 
This has further implications for the topology of `Fermi surfaces' at sharp boundaries.
Since Fermi surfaces contributed by surface states  form lines in the rBZ, we henceforth call them \textit{Fermi lines}. 

It is useful to equip each Fermi line with an
orientation: for this we define a right-handed (possibly non-orthogonal) system of the following three vectors: (\emph{i}) the first vector equals to $\hat{\boldsymbol{t}}(dE/dt)$ (with $\hat{\boldsymbol{t}}$ a unit two-vector that is locally parallel to the rBZ-line $\Lambda\Xi$, and points in the direction of increasing $t$), (\emph{ii}) the second vector points in the paper-to-sky direction, and (\emph{iii}) the third vector is locally tangential to the Fermi line.
The orientation of the third vector is fixed by the handedness of the frame, thus
defining the orientation of the Fermi line.  
For illustration, Fig.~\ref{fig:surf_spec}(b) presents five segments of the Fermi lines crossing the rBZ-line $\Lambda\Xi$, with orientations determined by applying the right-hand rule to Fig.~\ref{fig:surf_spec}(a). Because the right-hand-rule mapping is bijective, one can also start from the oriented Fermi-line segments in   Fig.~\ref{fig:surf_spec}(b) to deduce the signs of $dE_j/dt$ and therefore the RTP invariant via the rule in  \q{eq:RTP-velocity}.

Let us now consider the inverse problem of determining Fermi lines from the RTP invariant(s): knowing  all the RTP invariants in the rBZ, can one infer the topology of the shape of Fermi lines, as well as their orientations? 
We argue that the answer is `no', i.e., that the Fermi-line topology remains ambiguous, which can be understood as follows.
For each RTP invariant between reduced momenta $\Pi_1$ and $\Pi_2$, one can infer that the sum of all Fermi-level states on a possibly-curved rBZ-line  (connecting $\Pi_1$ and $\Pi_2$) must satisfy the velocity-RTP rule in \q{eq:RTP-velocity}. 
There are two potential sources of ambiguity. First, the oriented segments of Fermi lines crossing the rBZ-line $\Pi_1\Pi_2$ are determined modulo an ambiguity in the exact positioning of segments along  $\Pi_1\Pi_2$; however, such smooth changes cannot induce changes in the topology.
On the other hand, the Fermi lines crossing the rBZ-line $\Pi_1\Pi_2$ are fixed modulo addition/subtraction of pairs of Fermi-line segments with opposite orientations [as illustrated in Fig.~\ref{fig:surf_spec}(b,c)]; as this is a discontinuous change that alters the number of summands in Eq.~(\ref{eq:RTP-velocity}), it should be expected that it implies a change in the Fermi line topology.

The oriented Fermi-line segments, which are in the described sense ambiguously defined over all rBZ-lines $\{\Pi_1\Pi_2\}$ on which an RTP invariant is well-defined, can be continued and connected to form closed Fermi loops that preserve orientation.
We illustrate in Fig.~\ref{fig:Fl-ambig}(b,c) two topologically reconstructed Fermi lines obtained in this way for the same choice of RTP invariants, with the Fermi lines being topologically distinct from one another because of the ambiguity explained here.
The ambiguity due to addition/subtraction of paired Fermi-line segments translates to an ambiguity due to addition/subtraction of `trivial' Fermi loops; if a Fermi loop intersects any rBZ-line for which an RTP invariant is well-defined, the condition of `triviality' means that the intersection does not modify  $\sum_j \text{sgn}[dE_j/dt]$.\footnote{An example of a nontrivial loop is any that encircles only one rotation-invariant reduced momentum $\Pi$, assuming $\Pi$ is mutually-disjoint and belongs to an iso-orbital subset. An example of a trivial loop is any that encircles no such $\Pi$.}

\subsubsection{Sharp-boundary states and Fermi-line topology in \texorpdfstring{$\mathrm{P}4$}{P4}-symmetric model} \la{sec:sharpboundaryP4}

To illustrate the findings of this section, we focus on the same $\mathrm{P}4$-symmetric tight-binding Hilbert space that formed the basis for the 
Moore-Ran-Wen model, presented in Sec.~\ref{sec:same-hopf-diff-rtp}. 
To remind the reader, the mutually-disjoint condition holds at all rotation-invariant reduced momenta $\{\Gamma, \tx, \tm\}$, and the iso-orbital condition holds also for any pair in this set, implying an RTP invariant is well-defined for any pair in this set. There are however only two independent RTP invariants.

\begin{figure}
    \centering
    \includegraphics{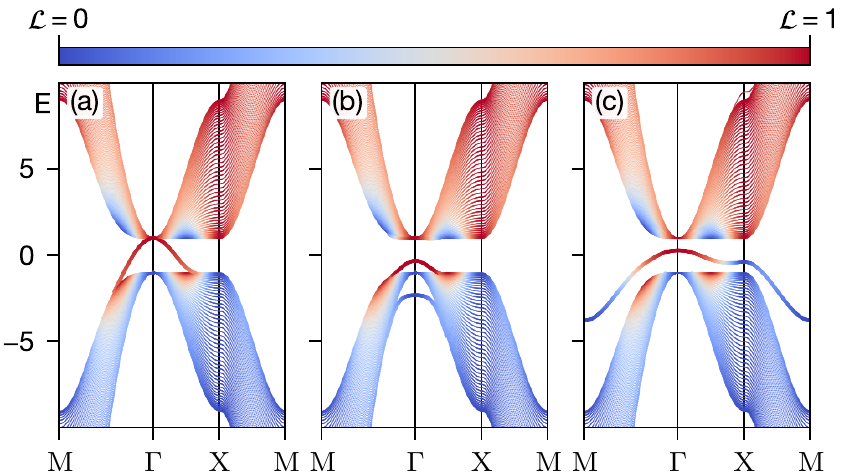}
    \caption{Spectrum of a semi-infinite MRW model [Eq.~\eqref{eq:hamMRW}] defined for $z>0$. The shades of blue vs.~red denote contributions to the eigenstates from the basis orbitals with angular momentum $\mathcal{L}=0$ vs.~$\mathcal{L}=1$. 
    At rotation-invariant reduced momenta the states are purely represented by orbitals with $\mathcal{L}=0$ or $1$. 
    (a) Sharp boundary possessing a gapless surface state. 
    (b) A surface-localized potential $V=-1.3$ added to both orbitals in the bottom layer of the semi-infinite system opens a gap in the surface spectrum. 
    (c) There exist a Hamiltonian perturbation that detaches the surface state from the rest of the spectrum.}
    \label{fig:mrw_surface}
\end{figure}

We start by looking at a semi-infinite Moore-Ran-Wen model Hamiltonian [given in the bulk by Eq.~\eqref{eq:hamMRW}] with $\Phi=-2$, for which the polarization protrudes in the negative $z$ direction at $C_4$-invariant momentum $\Gamma$ [orange line in Fig.~\ref{fig:Hopfless-MRW-RTP}(d)].
In Fig.~\ref{fig:mrw_surface}, we present the energy spectrum of the model, with colors (blue vs.~red) indicating the contribution of basis orbitals (with on-site angular momentum $\mathcal{L}=0$ vs.~$\mathcal{L}=1$) to the eigenstates. 
At rotation-invariant reduced momenta, each valence (resp.~conduction) bulk state  has itinerant angular momentum $\widetilde{\mathcal{L}}_v=0$ (resp.~$\widetilde{\mathcal{L}}_c=1$).
The itinerant angular momentum of the anomalous surface-localized energy band is bulk-valence-like at $\tx$ and $\tm$, and bulk-conduction-like at $\Gamma$.
At a sharp boundary [cf.~Fig.~\ref{fig:mrw_surface}(a)] we indeed see that, in accordance with the above argument, the surface-state energies interpolate from the bulk-conduction band at $\Gamma$ to the bulk-valence band at $\tx$; likewise, a similar interpolation occurs between $\Gamma$ and $\tm$. 

\begin{figure}
    \centering
    \includegraphics{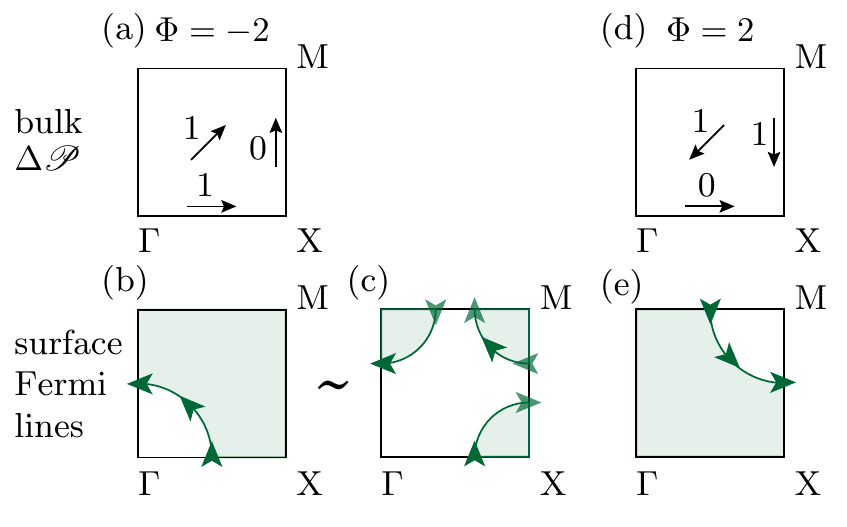}
    \caption{
    (a), (d) RTP invariants in the top right quadrant of the rBZ for the MRW model with parameter values (a) $\Phi=-2$ and (d) $\Phi=2$. 
    Arrows denote the direction along which the polarization is changing by the shown values.
    (b), (c), (e) Fermi lines are denoted by green lines and their orientations by green arrows. 
    Green shaded areas correspond to momenta in rBZ at which the surface-localized states are occupied. 
    In panel (b) we show the minimal number of Fermi lines that is compatible with the RTP invariants for the MRW model with $\Phi=-2$. 
    Another possible topology of Fermi lines is shown in panel (c), which differs from panel (b) by the addition/subtraction of Fermi loops that intersect the $\tx\tm$ rBZ lines twice and with opposite orientations.
    In both (b) and (c), the surface-band state at $\Gamma$ (at $\tx$ and $\tm$) is empty (occupied).
    (e) The areas of occupied surface states enclose different high-symmetry momenta in rBZ when the model parameter is changed to $\Phi=2$ and the RTP invariants acquire different values.
    }
    \label{fig:Fl-ambig}
\end{figure}

Now let us illustrate for the present RTP invariants [shown for the studied model in Fig.~\ref{fig:Fl-ambig}(a) with the arrows denoting the direction in which the polarization difference is calculated and the values equal to $\Delta\mathscr{P}$ in this direction] the possible topology of Fermi lines. 
In Fig.~\ref{fig:Fl-ambig}(b), we draw the minimal number of (oriented) Fermi lines compatible with the RTP invariants $\Delta\mathscr{P}_{\Pi_1\Pi_2}$, $\Pi_1,\Pi_2\in\{\Gamma, \tx, \tm\}$, $\Pi_1\neq\Pi_2$, within one quadrant of rBZ (denoted by green arrows). 
The Fermi lines at other quadrants can be obtained by a four-fold rotation of the displayed pattern around $\Gamma$. 
The area where the surface states are occupied with electrons is denoted by a green shade. 
Importantly, the RTP invariant defines only the minimal amount of Fermi lines that cross a given line connecting two rotation-invariant reduced momenta. 
As we show in Fig.~\ref{fig:Fl-ambig}(c), the connectivity of the Fermi lines can change, for instance, due to the addition of trivial Fermi loops located between $\tx$ and $\tm$, which do not enclose any rotation-invariant reduced momentum and which contribute to the summation in Eq.~(\ref{eq:RTP-velocity}) with a pair of states of opposite scalar velocity;
this new connectivity is presented in Fig.~\ref{fig:Fl-ambig}(c). 
However, the ambiguity in connecting the Fermi lines does not change the fact that for the given RTP invariants, the reduced momenta $\tx$ and $\tm$ are surrounded by areas with one more occupied surface state than in the area around the reduced momentum $\Gamma$. 

Reduced momenta around which the surface states are occupied can change when the values of the RTP invariant change. To see this, consider the MRW model with the parameter value $\Phi=2$.
The corresponding values of the topological invariants can be read from Fig.~\ref{fig:Fl-ambig}(d); in particular, note that the Hopf invariant coincides with the one for the previously considered parameter, i.e.,  $\chi[\Phi=-2]=\chi[\Phi=2]=-1$.
We see in Fig.~\ref{fig:Fl-ambig}(e) that with the new parameter value the areas with occupied surface states surround reduced momenta $\Gamma$ and $\tx$, instead of $\tm$ and $\tx$.

We emphasize that the assumption of a sharp boundary is essential for the robust surface states as it allows to use certain spectral properties of Toeplitz matrices. 
For non-sharp boundaries, the anomalous surface bands can be detached from the bulk bands. 
Because the energy eigenvalues of each anomalous surface band form a sheet over the rBZ, it is in principle possible that none of the sheets intersect the Fermi level, implying the existence of an energy gap to surface-state excitations.\footnote{This however does not rule out the existence of `gapless' excitations due to higher-order topology \cite{zhida_dminus2,Schindler_higherorderTI} -- a concept explored for delicate topological insulators in \ocite{penghaoAA_quantizedmagnetism}.} 
To illustrate  this point for the MRW model, we added to the bottom layer of the sharply-terminated tight-binding Hamiltonian an on-site potential that does not discriminate between $s$ and $p_+$ orbitals.
This leads to an energetic detachment of the anomalous surface-localized energy bands from the bulk-conduction band in  Fig.~\ref{fig:mrw_surface}(b), thus confirming that surface states are not robust for an arbitrary surface termination.\footnote{This is how we actually constructed the projectors $P$ and $Q$ with smooth-in-$\bk_\perp$ restrictions $P|_{\bk_\perp}$ and $Q|_{\bk_\perp}$ when generating the data in Fig.~\ref{fig:PzP}.} 
Finally, we show in Fig.~\ref{fig:mrw_surface}(c) that the anomalous surface-localized energy band can be energetically detached from \emph{both} bulk-conduction and bulk-valence bands for a particular, rotation-invariant surface termination.
However, despite the complete detachment, the itinerant angular momenta of the surface energy bands remain anomalous. This has non-trivial implications for their Zak phase, which we explore in the next subsection.

\subsection{Zak-phase anomaly as boundary signature of the returning Thouless pump \label{sec:BBC_zak}}

The angular-momentum anomaly at two rotation-invariant momenta $\Lambda$ and $\Xi$ implies the existence of `anomalous' faceted band(s) whose itinerant angular momentum is bulk-valence-like at one rotation-invariant, reduced momentum, 
say $\Lambda$, and bulk-conduction-like at another, say $\Xi$.  
This implies that there exists a symmetrically-chosen loop (to be specified below) in the reduced Brillouin zone which passes through both momenta, such that the Zak phase of the faceted band(s) (henceforth called the \textit{faceted Zak phase}) is topologically distinct from both the bulk-conduction Zak phase and the bulk-valence Zak phase on the corresponding BZ loop. 
We call this distinction of Zak phases the \emph{Zak-phase anomaly}. 

Our discussion is structured as follows. First, in Sec.~\ref{sec:fac-zak-def}, we define the faceted Zak phase. 
Then we proceed with Sec.~\ref{sec:bulk-Zak-def} to define the bulk-Zak phases. 
Having all the definitions at hand, in \s{sec:bbc-mrw} we exemplify the Zak-phase anomaly through the $\mathrm{P}4$-symmetric Moore-Ran-Wen model. 
Finally, in \s{sec:zakanomalymain}, we precisely formulate the Zak-phase anomaly  for all $\mathrm{P}n$ space groups. 
We clarify what we mean by two Zak phases being `distinct', and we summarize in Table~\ref{tab:models_summary} all the `symmetrically-chosen rBZ loop' for which the Zak-phase anomaly exists.

\subsubsection{Faceted Zak phase}\label{sec:fac-zak-def}

The faceted Zak phase is simplest to conceptualize in the case where the faceted band is a single surface-localized energy band that is completely detached from both the bulk-conduction and the bulk-valence bands, as illustrated in Fig.~\ref{fig:mrw_surface}(c) for the $\mathrm{P}4$-symmetric MRW model under non-sharp boundary conditions.
Letting $\ket{v(\bk_\perp)}$ denote the intra-cell\footnote{For the semi-infinite geometry, the `cell' is infinite in the $z$-direction.} wave function of the detached band, we define the \textit{faceted Zak phase}  as a line integral of the associated Berry connection over a symmetrically-chosen loop $\mathcal{S}$ in the rBZ ($d\bk_\perp$ denotes a line element along this loop):
\begin{align}
     \mathscr{Z}_f
    \eq
    \int\limits_{\mathcal{S}}d\bk_\perp\cdot\boldsymbol{\mathcal{A}}_\perp, \lin 
    \boldsymbol{\mathcal{A}}_\perp\eq(\mathcal{A}_x,\mathcal{A}_y),\as
    \mathcal{A}_\mu[v(\bk_\perp)]=i\bra{v(\bk_\perp)}\ket{\partial_\mu v(\bk_\perp)}.
    \label{eq:Zak-phase-rbz}
\end{align}
The path $\mathcal{S}$ is chosen so that the faceted Zak phase is quantized to integer multiples of $2\pi/n$ owing to the $\mathrm{P}n$ symmetry.
A few representative examples of symmetrically-chosen loops are illustrated in  \fig{fig:zak-paths}, but a full disclosure on the general form of the loops for any $\mathrm{P}n$ symmetry is deferred to \s{sec:zakanomalymain}. 
We emphasize that the faceted Zak phase inputs the wave functions of surface states independent of whether the states are occupied by fermions.

More generally, one may deal with any number of surface-localized energy bands (as long as they are detached from the bulk energy bands), or any number of surface-like polarization bands. 
Then the faceted Zak phases are obtained as the phases of the eigenvalues of a \textit{Wilson-loop matrix}; this matrix is a path-ordered exponential of the non-Abelian Berry connection:
\begin{equation}\label{eq:wilson-loop-rbz}
\begin{split}
    \mathcal{W}_f&= \hat{\mathrm{T}}\exp\left[i\int\limits_{\mathcal{S}}d\bk_\perp\cdot\boldsymbol{\mathcal{A}}_\perp[v(\bk_{\perp})]\right], \\ \mathcal{A}^{jj'}_\mu&=i\bra{v_j(\bk_\perp)}\ket{\partial_\mu v_{j'}(\bk_\perp)},
\end{split}
\end{equation}
where $\hat{\mathrm{T}}$ is the path-ordering operator, and $j,j'$ are band indices of the surface bands. 
An algorithm to extract the intra-cell wave functions $\ket{v_j(\bk_\perp)}$ for surface-like polarization bands was introduced in \ocite{nband-hopf} and is reviewed pedagogically in Appendix~\ref{app:wannier-cut}.

\subsubsection{Zak phase of the bulk bands}\label{sec:bulk-Zak-def}

The Zak-phase anomaly states that the faceted Zak phase is topologically distinct from the \textit{bulk-conduction Zak phase} as well as from the \textit{bulk-valence Zak phase}, which inputs the wave function of the bulk-conduction resp.~the bulk-valence band:
\begin{align}
     \mathscr{Z}_j
    \eq\int\limits_{\mathcal{S}(k_z)}d\bk_\perp\cdot\boldsymbol{\mathcal{A}}_\perp[u_j(\bk_{\perp},k_z)], \as j=v,c.
    \label{eq:Zak-phase-bulk}
\end{align}
The line integral is over a loop $\mathcal{S}(k_z)$ contained within a plane of fixed $k_z$; if projected in the $k_z$ direction, $\mathcal{S}(k_z)$ coincides with $\mathcal{S}$ defined for the faceted Zak phase. Note that $\mathscr{Z}_{v(c)}$ can depend on $k_z$ through $\mathcal{S}(k_z)$ for a \emph{general} two-band Hamiltonian.
However, the $C_n$ symmetry [in combination with the fact that $\mathcal{S}(k_z)$ is a symmetrically chosen path] imply that the admissible values of $\mathscr{Z}_j$ are quantized to integer multiples of $2\pi/n$, which reduces the prospective $k_z$-dependence to discontinuous jumps.
But such discontinuities are forbidden under the assumption that the two-band Hamiltonian is insulating, hence the bulk Zak phase ultimately has  no $k_z$ dependence.

\subsubsection{Case study: \texorpdfstring{$\mathrm{P}4$}{P4}-symmetric Moore-Ran-Wen model\label{sec:bbc-mrw}}

\begin{figure}
    \centering
    \includegraphics{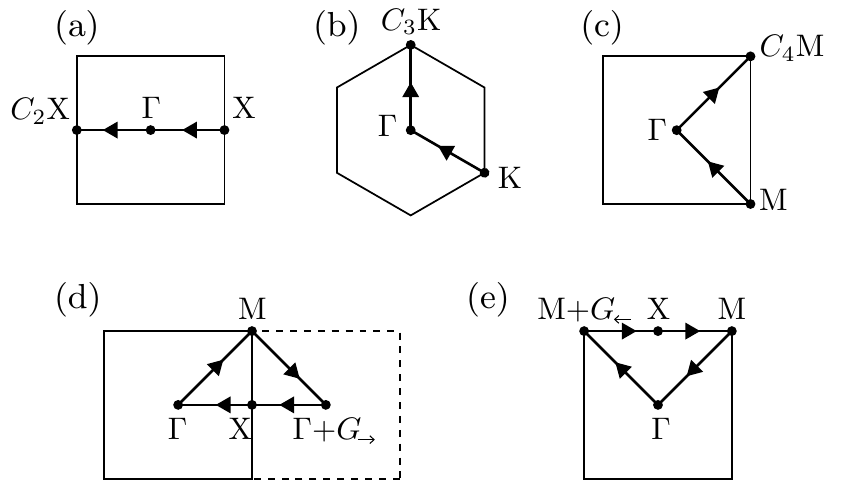}
    \caption{An (incomplete) list of symmetrically-chosen rBZ-loops for which a Zak-phase anomaly may exist. 
    (a) $C_2$-related-, (b) $C_3$-related- and (c) $C_4$-related-two-segment loops. 
    (d) and (e) Two exceptional four-segment loops, described in Sec.~\ref{sec:zakanomalymain}.}
    \label{fig:zak-paths}
\end{figure}

Let us illustrate the Zak-phase anomaly with the $\mathrm{P}4$-symmetric MRW model Hamiltonian given in  Eq.~\eqref{eq:hamMRW}. 
We focus on the parameter range $-3<\Phi<-1$ where the polarization protrudes (in the negative $z$ direction) at $\Gamma$, as illustrated in Fig.~\ref{fig:Hopfless-MRW-RTP}(d); this implies RTP invariants $\Delta\mathscr{P}_{\Gamma\tx}=1$ and $\Delta\mathscr{P}_{\Gamma\tm}=1$.
The corresponding energetically detached anomalous surface band is illustrated in Fig.~\ref{fig:mrw_surface}(c) for a suitable choice of a non-sharp boundary.

According to our formulation of the Zak-phase anomaly, we expect that the faceted Zak phase, computed for the anomalous surface band  in Fig.~\ref{fig:mrw_surface}(c), is distinct from the bulk Zak phases. 
Based on the just-mentioned RTP invariants, we further expect this mismatch of Zak phases for an rBZ-loop $\mathcal{S}_2$ passing through momenta $\Gamma$ and $\tx$, as well as an rBZ-loop $\mathcal{S}_4$ passing through $\Gamma$ and $\tm$.
For this mismatch to be quantized to integer multiples of $2\pi/4$, the rBZ-loops must be chosen symmetrically: we choose 
$\mathcal{S}_2$ (resp. $\mathcal{S}_4$) such that half the loop is mapped to the other half by a two-fold (resp.\ four-fold) rotation in the rBZ, as illustrated in Fig.~\ref{fig:zak-paths}(a, c). 
While $\mathcal{S}_2= \tx\Gamma(C_2\tx)$ can be chosen as a straight loop, $\mathcal{S}_4=\tm\Gamma(C_4\tm)$ necessarily has kinks at the four-fold-invariant momenta.~\footnote{Compare also to similar consideration of Berry phases on $C_3$-related paths by Ref.~\cite{bouhon_wilsonloopapproach} within the context of fragile topology.}
The advantage of these \textit{$C_m$-related loops} ($m=2,4)$ is that the Zak phase for a nondegenerate energy band is quantized to rational multiples of $2\pi$ and fully determined by the itinerant angular momenta at $\Gamma$, $\tx$ and $\tm$, according to a theorem in \ocite{TBO_JHAA} (cf.~Sec.~IV\kern 0.1em D\kern 0.1em 1 therein):
\begin{subequations}\label{eq:zak-wilson-mrw-gx+gm}
\e{
\frac{\mathscr{Z}_{\tx\Gamma(C_2\tx),j}}{2\pi} &=_1 \frac{\widetilde{\mathcal{L}}_j(\tx)-\widetilde{\mathcal{L}}_j(\Gamma)}{2},
\label{eq:zak-wilson-mrw-gx} \\
\frac{\mathscr{Z}_{\tm\Gamma(C_4\tm),j}}{2\pi} &=_1
\frac{\widetilde{\mathcal{L}}_j(\tm)-\widetilde{\mathcal{L}}_j(\Gamma)}{4},
\label{eq:zak-wilson-mrw-gm}
}
\end{subequations}
where $j\in\{v,c,f\}$ respectively denotes the bulk-valence, the bulk-conduction, and the faceted Zak phase. 

Let us offer a hand-waving argument that captures the essential intuition of the above results. 
The Zak phase for $\tm\Gamma(C_4\tm)$ can be expressed as a sum of two contributions: $\mathscr{Z}_{\tm\Gamma(C_4\tm)}=\omega_{\tm\Gamma}+\omega_{\Gamma(C_4\tm)}$, with $\omega_{\tm\Gamma}$ [resp.\ $\omega_{\Gamma(C_4\tm)}$] the contribution from the oriented line $\tm\Gamma$ [resp.\ $\Gamma(C_4\tm)$]. 
Note that $\omega_{\Gamma(C_4\tm)}$ receives as input the wave function along $\Gamma(C_4\tm)$, which is related by four-fold symmetry to the wave function along $\tm\Gamma$. 
However, $\tm\Gamma$ is mapped by four-fold rotation to the inverse of $\Gamma(C_4\tm)$, i.e.,  $\Gamma(C_4\tm)$ with opposite orientation. 
It follows that the contribution to   $\varphi_{\tm\Gamma(C_4\tm)}$ from an infinitesimal line element along $\tm\Gamma$ cancels the contribution from a rotation-related line element along  $\Gamma(C_4\tm)$; the only nontrivial contributions to $\mathscr{Z}_{\tm\Gamma(C_4\tm)}$ arise from applying four-fold rotation to the Bloch states at rotation-invariant momenta, thus giving   \q{eq:zak-wilson-mrw-gm}. 
This sketchy explanation is formalized in  Appendix~\ref{app:two-segm}.

Equations~(\ref{eq:zak-wilson-mrw-gx+gm}) can be applied to calculate the bulk-conduction Zak phase, the bulk-valence Zak phase, as well as the faceted Zak phase. 
For the bulk Zak phases, we input into the formula the values of itinerant angular momenta in the second row of the $\mathrm{P}4$ segment of Table~\ref{tab:itinerant_ell}. 
To compute the faceted Zak phase, we apply that the itinerant angular momentum of the anomalous surface band is bulk-valence-like at $\tx$ and $\tm$ and bulk-conduction-like at $\Gamma$ according to the angular-momentum anomaly [cf.\ \s{sec:anomaly}]. 
The final result of these computations confirms a mismatch between the bulk and faceted Zak phases, namely: 
\begin{equation}\label{eq:zak_mrw}
\begin{split}
    \textrm{for}\;\;\tx\Gamma(C_2\tx):\;\;&\mathscr{Z}_v=0,\;\;\mathscr{Z}_c=0, \;\;\mathscr{Z}_f=\pi, \\
    \textrm{for}\;\;\tm\Gamma(C_4\tm):\;\;&\mathscr{Z}_v=0,\;\; \mathscr{Z}_c=0,\;\; \mathscr{Z}_f=3\pi/2,    
\end{split}
\end{equation}
with all equalities understood to hold modulo $2\pi$. Similar computations can be performed for different choices of the parameter $\Phi$ in the MRW model, with the resultant values of the Zak phases summarized in the second row and last column of Table~\ref{tab:models_summary}.  

\begin{table*}
    \begin{threeparttable}
    \begingroup \centering
    \caption{
    Summary of the Hopf-RTP relation [cf.~Eq.~\eqref{eq:hopf-RTP}] and bulk-boundary correspondence [cf.~Eqs.~\eqref{thesame}, \eqref{eq:zak-bulk-surf} and \eqref{eqn:zakanomaly_all}] for the tight-binding models presented in Table~\ref{tab:model-zoo}. 
    The MRW model corresponds to the second listed model.
    The column `Phase diagram' specifies the values of the Hopf and RTP invariants in each phase of the given tight-binding models, specified by the values of the model parameter $\Phi$. 
    The column `Loop' lists all $\bk$-loops going through at least one pair of rotation-invariant momenta for which the RTP invariant is nontrivial and a Zak-phase anomaly exists. 
    The bulk Zak phases are given in columns `$\mathscr{Z}_v$' and `$\mathscr{Z}_c$'; the faceted Zak phase  is given in the column `$\mathscr{Z}_f(\Phi)$' as a function of the parameter $\Phi$.
    }    
    \label{tab:models_summary}
    \begin{ruledtabular}
    \begin{tabular}{l@{\hspace{0.25cm}}l@{\hspace{0.25cm}}l@{\hspace{0.8cm}}ll@{\hspace{0.2cm}}rccl}
         \multicolumn{3}{l}{Model symmetry} & \multicolumn{2}{c}{Hopf-RTP relation} & \multicolumn{4}{c}{Bulk-boundary correspondence} \vspace{0.2cm}\\
         SG & $\Delta\mathcal{L}$ & $\Delta \br_\perp$ & Hopf-RTP relation & \hspace{8.8 mm} Phase diagram & Loop \hspace{0.2cm} & $\mathscr{Z}_v$ & $\mathscr{Z}_c$ & $\mathscr{Z}_f(\Phi)$ \\ 
         \hline
         $\mathrm{P}4$ & 0 & $(1/2,1/2)$ & $\chi=_42\Delta\mathscr{P}_{\tx\tm}$ & \raisebox{-0.6 cm}
          {\includegraphics{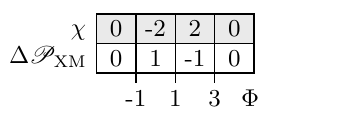}} & $\Gamma\tm(\Gamma+\bG_1)\tx\Gamma^{(1)}$& $0$& $0$ & \raisebox{-0.55 cm}
          {\includegraphics{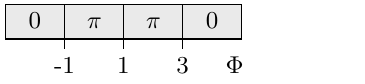}}\\
          & 1 & $(0,0)$ & $\chi=_4\Delta\mathscr{P}_{\Gamma\tm}+2\Delta\mathscr{P}_{\Gamma\tx}$ & \raisebox{-0.75 cm}
          {\includegraphics{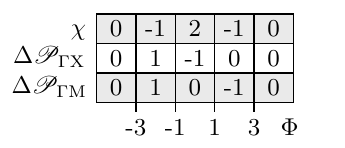}} & $\begin{array}{r} \tx\Gamma(C_2\tx)\\\tm\Gamma(C_4\tm)          \end{array}$
          & $\begin{array}{c} 0\\0\end{array}$ 
          & $\begin{array}{c} 0\\0\end{array}$ 
          & \raisebox{-0.7 cm}
          {\includegraphics{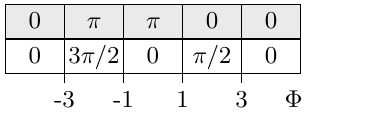}}\\
          & & $(1/2,1/2)$ & $\chi=_4\Delta\mathscr{P}_{\tm\Gamma}$ & \raisebox{-0.6 cm}
          {\includegraphics{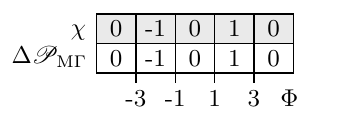}} & $\tm\Gamma(C_4\tm)$& $0$ & $\pi$ & \raisebox{-0.55 cm}
          {\includegraphics{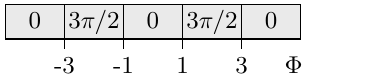}}\\
          & $2$ & $(0,0)$ & $\chi=_4 2\Delta\mathscr{P}_{\tm\Gamma}$ & \raisebox{-0.6 cm}
          {\includegraphics{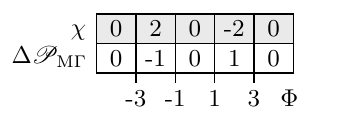}} & $\tm\Gamma(C_4\tm)$ & 0 & 0 & \raisebox{-0.55 cm}
          {\includegraphics{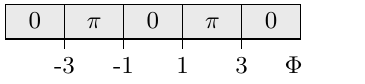}} \vspace{0.8cm}  \\   
         $\mathrm{P}3$ & 0 & $(1/3,1/3)$ & $\chi=_3\Delta\mathscr{P}_{\tkpr\tk}$ & \raisebox{-0.6 cm}
          {\includegraphics{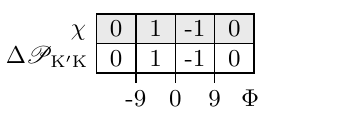}} & $\tkpr\tk(C_3^{-1}\tkpr)$ & 0 & $2\pi/3$ & \raisebox{-0.55 cm}
          {\includegraphics{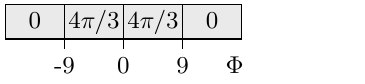}}\\
          & 1 & $(0,0)$ & $\chi=_3\Delta\mathscr{P}_{\Gamma\tk}+\Delta\mathscr{P}_{\Gamma\tkpr}$ & \raisebox{-0.75 cm}
          {\includegraphics{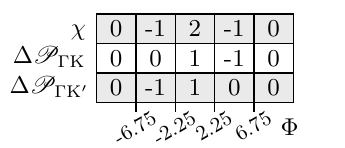}} & $\begin{array}{r} \tk\Gamma(C_3\tk) \\ \tkpr\Gamma(C_3\tkpr) \end{array}$ & $\begin{array}{c} 0\\0\end{array}$ & $\begin{array}{c} 0\\0\end{array}$ & \raisebox{-0.7 cm}
          {\includegraphics{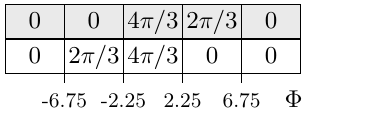}} \vspace{0.8cm} \\
         $\mathrm{P}6$ & $1$ & $(0,0)$ & $\chi=_62\Delta\mathscr{P}_{\Gamma\tk} + 3\Delta\mathscr{P}_{\Gamma\tm}$ & \raisebox{-0.75 cm}
          {\includegraphics{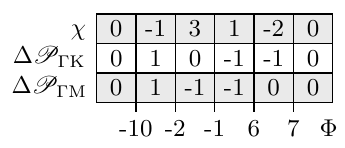}} & $\begin{array}{r} \tk\Gamma(C_3\tk) \\ \tm\Gamma(C_2\tm) \end{array}$ & $\begin{array}{c} 0\\0\end{array}$ & $\begin{array}{c} 0\\0\end{array}$ & \raisebox{-0.7 cm}
          {\includegraphics{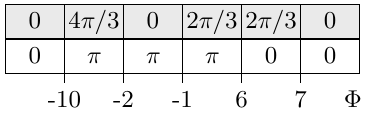}}\\
         & $2$ & $(0,0)$ & $\chi=_6 2\Delta\mathscr{P}_{\tk\Gamma}$ & \raisebox{-0.6 cm}
          {\includegraphics{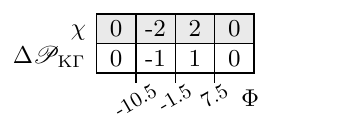}} & $\tk\Gamma(C_3\tk)$ & 0 & 0 & \raisebox{-0.55 cm}
          {\includegraphics{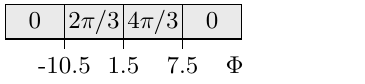}}\\
         & $3$ & $(0,0)$ & $\chi=_6 3\Delta\mathscr{P}_{\tm\Gamma}$ & \raisebox{-0.6 cm}
          {\includegraphics{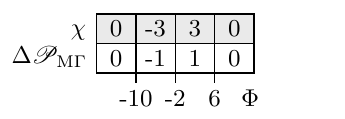}} & $\tm\Gamma(C_2\tm)$ & 0 & 0 & \raisebox{-0.55 cm}
          {\includegraphics{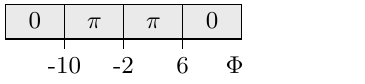}} \\
    \end{tabular}
    \end{ruledtabular}
    \endgroup
    \begin{tablenotes}
    \item[$^{(1)}$] $\bG_1=(2\pi, 0)$ is one of the reciprocal lattice vectors in the reduced BZ. 
    \end{tablenotes}
    \end{threeparttable}
\end{table*}

\subsubsection{Zak-phase anomaly: main result}\la{sec:zakanomalymain}

Having obtained some intuition about the Zak-phase anomaly in specific models, we now prove it in a more general setting.
Let us first specify the symmetrically-chosen loops for which a Zak-phase anomaly exists in correspondence with a nontrivial RTP invariant in general P$n$ symmetric models.
We assume that the RTP invariant $\Delta\mathscr{P}_{\Lambda\Xi}>0$ for two rotation-invariant reduced momenta $\Lambda$ and $\Xi$. The associated little groups are denoted by $\mathcal{G}_{\Lambda}$ and $\mathcal{G}_{\Xi}$, with both groups being  subgroups of the cyclic, order-$n$ group $C_n$, which is isomorphic to the point group  of $\mathrm{P}n$. 

Suppose there is an integer $m$ that satisfy the following two properties: \\

\noindent (\emph{i}) $m>1$ and divides the order of $\mathcal{G}_\Pi$, for both $\Pi\in \{\Lambda,\Xi\}.$\\

\noindent (\emph{ii}) The itinerant angular momenta of bulk-conduction and bulk-valence bands are distinct modulo $m$:
\e{\widetilde{\mathcal{L}}_c(\Pi) \neq_{m} \widetilde{\mathcal{L}}_v(\Pi),\label{eq:mut-disj-lx}}
for both $\Pi\in \{\Lambda,\Xi\}.$ \\

\noindent If there are multiple values of $m$ satisfying (\emph{i}--\emph{ii}), for concreteness we choose the largest such value for $m$. 
Then the Zak-phase anomaly exists for a suitably chosen oriented loop in rBZ. 
This loop is defined as a (possibly bent) path that concatenates two segments, $\Xi\Lambda$ and $\Lambda(\Xi+\bG)$, such that one segment is, up to a translation by a reciprocal lattice vector and with reversed orientation, equal to the $C_m$ rotation of the other segment:
$\Lambda (\Xi+\bG) = C_m(\Lambda\Xi) + \bG'$.
For brevity, we call such path as a \emph{$C_m$-related-two-segment loop.}
In Fig.~\ref{fig:zak-paths}(a--c), we illustrate a representative $C_m$-related-two-segment loop for all $m=2,3,4$.\footnote{A $C_6$-related-two-segment loop does not exist because there is only a single $C_6$-invariant reduced momentum -- the rBZ center. For Hamiltonians with six-fold symmetry, a nontrivial RTP can arise solely from a subgroup (e.g., $\mathrm{P}2$ or $\mathrm{P}3$) of the full space group.}

Let us formulate the Zak-phase anomaly for these two-segment loops, supposing that the faceted-band subspace consists of $\mathcal{N}$ linearly-independent Bloch states that are smoothly defined over the rBZ. We emphasize that this discussion equally applies to a set of surface-localized energy bands detached from the rest of the spectrum or of surface-like polarization bands defined by the partitioning illustrated in Fig.~\ref{fig:PzP}(b,c).

Parallel transport (within the faceted-band subspace, along a rBZ loop) is encoded in $\mathcal{N}$ faceted Zak phases $\{\mathscr{Z}^{i}_f\}_{i=1}^\mathcal{N}$, which are computed by diagonalizing the $\mathcal{N}$-by-$\mathcal{N}$ Wilson-loop matrix in \q{eq:wilson-loop-rbz}. 
For a generically chosen rBZ-loop, the $\mathcal{N}$ phases take generic values without preference for rational multiples of $2\pi$. 
However, for the $C_m$-related-two-segment rBZ-loops, we propose that $\Delta\mathscr{P}_{\Lambda\Xi}$-number of the faceted Zak phases (denoted $\{\mathscr{Z}^{i}_f \}_{i=1,\ldots,\Delta\mathscr{P}_{\Lambda\Xi}}$) are symmetry-fixed to a rational multiple of $2\pi/n$ that depends on the itinerant angular momenta of the bulk-valence and bulk-conduction bands:
\begin{equation}
    \frac{\mathscr{Z}^i_f}{2\pi} =_1\f{\widetilde{\mathcal{L}}_v(\Xi) - \widetilde{\mathcal{L}}_c(\Lambda)}{m}.\la{thesame}
\end{equation}
This result derives ultimately from the angular-momentum anomaly, which guarantees the existence of at least  $\Delta\mathscr{P}_{\Lambda\Xi}$-number of anomalous faceted bands with itinerant angular momenta that are bulk-conduction-like at $\Lambda$ and bulk-valence-like at $\Xi$:
\begin{align}
    \widetilde{\mathcal{L}}_f(\Lambda) = \widetilde{\mathcal{L}}_c(\Lambda); \as
    \widetilde{\mathcal{L}}_f(\Xi) = \widetilde{\mathcal{L}}_v(\Xi).
    \label{eq:surface_ell}
\end{align}
Application of a theorem from Sec.~IV\kern 0.1em D\kern 0.1em 1 in Ref.~\cite{TBO_JHAA} (which relates itinerant angular momenta to Zak phases) gives us  \q{thesame}, as elaborated in Appendix~\ref{app:two-segm}. 
Equation~\eqref{thesame} should be contrasted with the Zak phases of the bulk-valence and bulk-conduction bands:
\begin{equation}\label{eq:zak-bulk-surf}
\begin{split}
\f{\mathscr{Z}_v}{2\pi} &=_1 \f{ \widetilde{\mathcal{L}}_v(\Xi) - \widetilde{\mathcal{L}}_v(\Lambda) }{m}, \\
\f{\mathscr{Z}_c}{2\pi} &=_1 \f{\widetilde{\mathcal{L}}_c(\Xi) - \widetilde{\mathcal{L}}_c(\Lambda)}{m},
\end{split}
\end{equation}
which also follow directly from applying the just-stated theorem to the itinerant angular momenta of each individual bulk band. 
Both bulk Zak phases are distinct from each of the anomalous faceted Zak phase, as can be proven by the method of contradiction.\footnote{Assuming that $\mathscr{Z}_f=_{2\pi}\mathscr{Z}_v$,  \q{eq:zak-bulk-surf} implies that $\widetilde{\mathcal{L}}_v(\Lambda)=_m\widetilde{\mathcal{L}}_c(\Lambda)$ which contradicts the initial assumption \eqref{eq:mut-disj-lx}. Similarly we can prove that $\mathscr{Z}_f=_{2\pi}\mathscr{Z}_v$ leads to a contradiction.} 

By going through all possible RTP invariants in all possible $\mathrm{P}n$-symmetric two-band Hamiltonians, we find only two exceptions to the rule that an integer $m$ exists satisfying (\emph{i}--\emph{ii}). 
In both these exceptions (to be elaborated below), we propose that the Zak-phase anomaly exists for a \emph{four-segment} loop going through four rotation-invariant reduced momenta  $\{\Lambda, \Xi, \Omega,\Omega+\bG\}$, with $\Omega$ and $\bG$ specified below.  \\

\noindent \textit{Exception 1:} A $\mathrm{P}4$-symmetric tight-binding Hilbert space that is a sum of two basis band representations, induced from two orbitals with onsite angular momentum difference $\Delta\mathcal{L}=0$ and centered at $C_4$-invariant Wyckoff positions with relative distance $\Delta\br_\perp=(1/2, 1/2)$ [for example at positions $1a$ and $1b$ defined in \fig{fig:wp-rotinvmom}(a)].
According to Table~\ref{tab:itinerant_ell}, the mutually-disjoint condition in this case is fulfilled only at reduced momenta $\tx$ and $\tm$, hence there is only one independent RTP invariant $\Delta\mathscr{P}_{\tx\tm}$. In the first row of Table~\ref{tab:model-zoo}, we present a model Hamiltonian that realizes
Exception 1, with a one-parameter phase diagram summarized in the first row of Table~\ref{tab:models_summary}. 
The itinerant angular momenta of the bulk bands satisfy
\begin{equation}\label{eq:mut-disj-C4s1as1b-xm}
\begin{split}
    \widetilde{\mathcal{L}}_v(\tx)&\neq_2\widetilde{\mathcal{L}}_c(\tx), \\
    \widetilde{\mathcal{L}}_v(\tm)&\neq_4\widetilde{\mathcal{L}}_c(\tm), \\
    \textrm{but}\;\;\widetilde{\mathcal{L}}_v(\tm)&=_2\widetilde{\mathcal{L}}_c(\tm),
\end{split}
\end{equation}
implying that the choice $m{=}4$ violates the property (\emph{i}) at $\tx$, while the choice $m{=}2$ violates the property (\emph{ii}) at $\tm$. 
One may verify that a Zak-phase anomaly does not exist for a $C_2$-related-two-segment loop $\tm\tx(\tm+\bG)$, but it exists for four-segment loop going through $\{ \Lambda=\tm,\Xi=\tx,\Omega=\Gamma, \Gamma+(2\pi, 0)\}$, as illustrated in~\fig{fig:zak-paths}(d).\\

\noindent \textit{Exception 2:} A $\mathrm{P}4$-symmetric tight-binding Hilbert space that is identical to Exception 1, except the onsite angular momenta difference is $\Delta\mathcal{L}=2$. 
In analogy with the first exception, we find the appropriate four-segment rBZ-loop that goes through $\{ \Lambda=\Gamma,\Xi=\tx,\Omega=\tm, \tm+(-2\pi,0)\}$, as illustrated  in \fig{fig:zak-paths}(e). \\

\noindent In both Exceptions, the mutually-disjoint condition holds at $\Lambda$ and $\Xi$, but not at $\Omega$. [As a reminder, the last statement means that the  itinerant angular momentum at $\Omega$ is identical for both bulk-valence and bulk-conduction states: $\widetilde{\mathcal{L}}_v(\Omega)=\widetilde{\mathcal{L}}_c(\Omega)\equiv \widetilde{\mathcal{L}}(\Omega).$] 

By generalizing techniques developed in \ocite{TBO_JHAA}, we derive in \app{app:four-segm} a relation between the Zak phases (for these four-segment rBZ-loops) and the itinerant angular momenta of the bulk bands.
We find that $\Delta\mathscr{P}_{\Lambda\Xi}$-number of the faceted Zak phases (denoted $\{\mathscr{Z}^i_f \}_{i=1,\ldots,\Delta\mathscr{P}_{\Lambda\Xi}}$) are symmetry-fixed to a rational multiple of $2\pi/n$ that depends on the bulk itinerant angular momenta as: 
\begin{subequations}
\label{eqn:zakanomaly_all}
\begin{equation}
\label{eq:zakanomaly_foursegment}
\f{\mathscr{Z}^i_f}{2\pi} =_{1} \frac{\widetilde{\mathcal{L}}(\Omega)-\widetilde{\mathcal{L}}_c(\Lambda)}{m_\Lambda} + \frac{\widetilde{\mathcal{L}}(\Omega)-\widetilde{\mathcal{L}}_v(\Xi)}{m_\Xi},
\end{equation}
where $m_{\Lambda(\Xi)}$ denotes the maximal order of rotation symmetry of the little group $\mathcal{G}_{\Lambda(\Xi)}$.
This rational value is distinct (mod $2\pi$) from the symmetry-fixed values of the bulk-valence and bulk-conduction Zak phases:
\begin{equation}
\begin{split}
\f{\mathscr{Z}_v}{2\pi} =_1\frac{\widetilde{\mathcal{L}}(\Omega)-\widetilde{\mathcal{L}}_v(\Lambda)}{m_\Lambda} + \frac{\widetilde{\mathcal{L}}(\Omega)-\widetilde{\mathcal{L}}_v(\Xi)}{m_\Xi}, \\
\f{\mathscr{Z}_c}{2\pi} =_1 \frac{\widetilde{\mathcal{L}}(\Omega)-\widetilde{\mathcal{L}}_c(\Lambda)}{m_\Lambda} + \frac{\widetilde{\mathcal{L}}(\Omega)-\widetilde{\mathcal{L}}_c(\Xi)}{m_\Xi}.
\label{eq:bulkzak_foursegment}
\end{split}
\end{equation}
\end{subequations}
Once again, we can show by contradiction that the faceted Zak phase is always distinct from the bulk-valence and bulk-conduction Zak phases. 

To illustrate the Zak-phase anomaly for a four-segment loop, we employ a model Hamiltonian presented in the first row of Table~\ref{tab:model-zoo}, and substitute the bulk values of itinerant angular momenta into Eqs.~(\ref{eqn:zakanomaly_all}), to derive the faceted and bulk Zak phases detailed in the first row of Table~\ref{tab:models_summary}. 
The remaining entries in  Table~\ref{tab:models_summary} summarize the Zak-phase anomaly for various two-segment loops, and for all remaining tight-binding models in Table~\ref{tab:model-zoo}.

\section{Relating the RTP and Hopf invariants} \label{sec:RTP-Hopf}

By considering the $\bk\cdot\bp$ description near Berry-dipolar band-touching points, we have previously argued in Sec.~\ref{sec:Hopf-RTP-correspondence} that the RTP and Hopf invariants are not completely independent; a specific relation [cf.~\q{eq:hopf-rtp-example}] between the two invariants was extracted heuristically  for a class of Hamiltonians with $\mathrm{P}4$ symmetry. 
Here, we derive this $\mathrm{P}4$ result more systematically, and further generalize the Hopf-RTP relation to all $\mathrm{P}n$ space groups, with the main result encapsulated in Eqs.~(\ref{hopfrtp346}--\ref{eq:hopf-RTP}) of \s{sec:mainresult}.
Representative case studies from all $\mathrm{P}n$ space groups are presented in \s{sec:modelRTPHopf} to corroborate the Hopf-RTP relation. 

We prove the Hopf-RTP relation in three different ways, with each offering distinct insights:\\  

\noindent (\emph{i}) In \s{sec:hopf-rtp-motivation}, we prove the relation assuming that all phase transitions are mediated by Berry dipoles (by generalizing the arguments in Sec.~\ref{sec:Hopf-RTP-correspondence}). \\

\noindent (\emph{ii}) In Sec.~\ref{sec:proof-RTP-Hopf}, we relax the just-mentioned assumption, and prove the Hopf-RTP relation more generally by relating the angular-momentum anomaly of RTP invariants (formulated in Sec.~\ref{sec:anomaly})s
to the bulk-boundary correspondence of Hopf insulators.\\

\noindent (\emph{iii}) In \s{sec:linking-RTP-Hopf}, we  prove the relation purely from a bulk perspective, by utilizing an interpretation of  the Hopf invariant in terms of linking numbers (as was briefly introduced in Sec.~\ref{sec:Berry-transitions}).

\subsection{Main result: Hopf-RTP relation}\label{sec:mainresult}

Let us first formulate the RTP invariants in the most general context of a two-band, $\mathrm{P}n$-symmetric, insulating  Hamiltonian with trivial first Chern class.

Consider a subset of rotation-invariant reduced momenta $\{\Pi_j\}_{j=1,\ldots,J_\alpha}$ in the rBZ. 
For the polarization difference between any pair $j,j'$ in this subset to be quantized to integer values, it is necessary that: \\

\noindent (\emph{i}) The itinerant angular momenta of both bulk-conduction and bulk-valence bands satisfy the mutually-disjoint condition for all $j$:
\e{ \Delta\widetilde{\mathcal{L}}(\Pi_j)=\widetilde{\mathcal{L}}_c(\Pi_j)-\widetilde{\mathcal{L}}_v(\Pi_j) \;\text{mod}\; m_{\Pi_j}
\neq 0.\la{mutualdisj}}

\noindent (\emph{ii}) For all $j$, the restriction of the valence band representation ($\mathrm{VBR}$) to ${\gamma_{\Pi_j}}$ is identical to the restriction of a basis band representation to $\gamma_{\Pi_j}$:
\e{ \exists\,\alpha\in \{1,2\}\;\, \textrm{s.t.} \,\;\forall j\in \{1,\ldots, J_\alpha\},\;\, \mathrm{VBR}\bigg|_{\gamma_{\Pi_j}}\!\!\!\!\!\!=\mathrm{BBR}[\varphi_\alpha]\bigg|_{\gamma_{\Pi_j}}\!\!\!\!\!\!.\la{mutualdisj2}}
Here, $\mathrm{BBR}[\varphi_\alpha]$ denotes the basis band representation induced from tight-binding-basis orbital $\varphi_\alpha$ with positional center $\br_\alpha=(\br_{\alpha,\perp},z_\alpha)$.  
Note the equality in \q{mutualdisj2} holds with the same basis band representation (specified by $\alpha=1,2$) for all $j$.\\

\noindent One may view condition (\emph{ii}) as a generalization of the iso-orbital condition in \q{isoorbital} for any number of rotation-invariant reduced momenta. 
Conditions (\emph{i}--\emph{ii}) jointly imply that the polarizations at all $\Pi_j$ are symmetry-fixed, modulo integers, to~$z_\alpha$. 

We take the order ($J_\alpha$) of the subset to be maximal, in the sense that every rotation-invariant reduced momentum satisfying the iso-orbital condition for BBR$[\varphi_\alpha]$ and the mutually-disjoint condition should be included in the subset. Consequently, if the subset contains a $C_m$-invariant reduced momentum $\Pi_j$, it also contains all momenta from the corresponding $n/m$-plet obtained from $\Pi_j$ by a $C_n$ rotation.
Because there are only two basis band representations in a two-band Hamiltonian, all `mutually-disjoint' rotation-invariant reduced momenta in the rBZ are sortable into at most two maximal `iso-orbital' subsets.

\subsubsection{\texorpdfstring{$\mathrm{P}n$}{Pn}-symmetric Hopf-RTP relation for \texorpdfstring{$n=3,4,6$}{n=3,4,6}}

For space groups $\mathrm{P}n$ ($n=3,4,6$), a
combination of symmetry ($\mathrm{P}n$), topology (triviality of Chern class) and Hilbert-space constraint (two-band Hamiltonian) ensures that all reduced momenta in the rBZ satisfying the mutually-disjoint condition belong to a \textit{single} iso-orbital subset $\{\Pi_j\}_{j=1\dots J}$, as shown in Appendix~\ref{app:no-band-inv-from-Chern}. 
Then the RTP invariants relate to the Hopf invariant and the itinerant-angular-momenta differences through
\e{
\ins{For $\mathrm{P}n$} (n=3,4,6), \as \chi&=_n   \sum_{j=2}^{J} \Delta\widetilde{\mathcal{L}}(\Pi_{j})\;\Delta\mathscr{P}_{\Pi_1\Pi_{j}}.\la{hopfrtp346} 
}
Within the momenta subset, the choice of $\Pi_1$ can be made arbitrarily without  affecting  the above relation modulo $n$.

\subsubsection{\texorpdfstring{$\mathrm{P}2$}{P2}-symmetric Hopf-RTP relation}

For space group $\mathrm{P}2$, depending on the equivalence class\footnote{Two Hamiltonians are equivalent if they  can be continuously deformed into one another while preserving gap and $\mathrm{P}2$ symmetry.} of the Hamiltonian, either \\

\noi{P2-I} the rBZ contains only a single maximal momentum subset satisfying (\emph{i}--\emph{ii}), in which case one may also apply the formula \eqref{hopfrtp346} with $n=2$, or\\

\noi{P2-II} the four rotation-inequivalent rotation-invariant momenta of the rBZ can be divided into two maximal subsets, with one subset $\{\Pi_{j_1}^1\}_{j_1=1,\ldots,J_1}$ satisfying (\emph{i}--\emph{ii}) for BBR$[\varphi_1]$, and the other $\{\Pi_{j_2}^2\}_{j_2=1,\ldots,J_2}$ satisfying (\emph{i}--\emph{ii}) for BBR$[\varphi_2]$. 
It is shown in Appendix~\ref{app:no-band-inv-from-Chern} that neither $J_1$ nor $J_2$ can be odd, lest the first Chern invariant in the $(k_x,k_y)$-plane is also odd.
Focusing on Hamiltonians with trivial first Chern class restricts our consideration then to $J_1=J_2=2$, with the following, modified Hopf-RTP relation:
\e{
\text{For P2-II}, \as \chi&=_2  \sum_{\alpha=1}^2
\Delta\mathscr{P}_{\Pi^\alpha_1\Pi^\alpha_2}. \la{eq:hopf-RTP-P2II}
}
Within any maximal iso-orbital momenta subset indexed by $\alpha$, the choice of $\Pi_1^\alpha$ can be made arbitrarily without  affecting  the above relation modulo $2$. 
A representative $\mathrm{P}2$-symmetric model of class P2-I can be obtained by lowering the symmetry of either of the 
$\mathrm{P}4$-symmetric models 
{presented in the second and third rows of} \tab{tab:model-zoo}, while an example realization of a class P2-II Hamiltonian is presented in Sec.~\ref{sec:P2-example}.

\subsubsection{Summarizing Hopf-RTP relation for all P$n$ space groups}

The previously-presented Hopf-RTP relations can be viewed as particular instances of a summarizing Hopf-RTP relation: 
\e{
 \text{For any P}n, \as \chi&=_n  \sum_{\alpha=1}^2 \sum_{j_\alpha=2}^{J_\alpha} \Delta\widetilde{\mathcal{L}}(\Pi_{j_\alpha}^\alpha)\;\Delta\mathscr{P}_{\Pi^\alpha_1\Pi^\alpha_{j_\alpha}}, \la{eq:hopf-RTP}
 }
 with one of $\{J_1,J_2\}$ possibly being zero.

\subsection{Model Hamiltonians exemplifying the Hopf-RTP relation}\la{sec:modelRTPHopf}

\subsubsection{Examples for \texorpdfstring{$\mathrm{P}n$}{Pn} (\texorpdfstring{$n=3,4,6$}{n=3,4,6}) models}

For each $\mathrm{P}n$-symmetric ($n=3,4,6$) model Hamiltonian presented in \tab{tab:model-zoo}, we display in Table \ref{tab:models_summary} the particular forms of the Hopf-RTP relation, derived from \q{hopfrtp346} by inputting the model-dependent itinerant angular momenta listed in Table~\ref{tab:itinerant_ell}. 
Additionally, Table \ref{tab:models_summary} provides a phase diagram with the values of the Hopf and RTP invariants for each phase, corroborating the relation between the invariants.

While the Tables only explicitly describe models whose tight-binding-basis orbitals differ in on-site angular momenta by $0\leq\Delta\mathcal{L}\leq n/2$, models with $\Delta\mathcal{L}<0$ can be obtained from the tabulated models by time reversal. 
One may verify that if the Hopf-RTP relation [\q{hopfrtp346}] holds for a model $h_{\Delta \mathcal{L}}(\bk)$, it holds consistently for the time-reversed model $h_{-\Delta \mathcal{L}}(\bk)$ [related by time reversal to $h^*_{\Delta \mathcal{L}}(-\bk)$ via \q{timereverse}].\footnote{
For this consistency check, we utilize that the Hopf invariant and the polarization transform under time reversal as in \q{timerevinv}, while the difference in itinerant angular momentum changes as  
\begin{equation}
\Delta\widetilde{\mathcal{L}}(-\Pi;h_{-\Delta\mathcal{L}})=m-\Delta\widetilde{\mathcal{L}}(\Pi;h_{\Delta\mathcal{L}})\label{eqn:DelL-transfo}
\end{equation}
for a $C_m$-invariant momentum $\Pi$. This equality can be derived from \q{eq:itin-am-from-basis} with the observation that $\bG_{-\Pi}=C_m(-\Pi)+\Pi=-\bG_{\Pi}$, and that time reversal inverts on-site angular momenta while preserving Wyckoff positions.  Combining these time-reversal relations into the statements of the Hopf-RTP relation [\q{hopfrtp346}], we get
\begin{align}
    -\chi[h_{\Delta\mathcal{L}}]&\stackrel{\textrm{\q{hopfrtp346}}}{=_n} -  \sum_{j=2}^{J} \Delta\widetilde{\mathcal{L}}(\Pi_{j};h_{\Delta\mathcal{L}})\;\Delta\mathscr{P}_{\Pi_1\Pi_{j}}[h_{\Delta\mathcal{L}}]\lin
    &\stackrel{\textrm{Eqs.~(\ref{eqn:DelL-transfo},\ref{timerevinv})}}{=_n}-\sum_{j=2}^{J} \left[m_j-\Delta\widetilde{\mathcal{L}}(-\Pi_{j};h_{-\Delta\mathcal{L}})\right]\;\Delta\mathscr{P}_{-\Pi_1,-\Pi_{j}}[h_{-\Delta\mathcal{L}}] \lin
    &=_n\sum_{j=2}^{J} \Delta\widetilde{\mathcal{L}}(-\Pi_{j};h_{-\Delta\mathcal{L}})\;\Delta\mathscr{P}_{-\Pi_1,-\Pi_{j}}[h_{-\Delta\mathcal{L}}] \lin
    &=_n\sum_{i=2}^{J} \Delta\widetilde{\mathcal{L}}(\Pi_{i};h_{-\Delta\mathcal{L}})\;\Delta\mathscr{P}_{\Pi_1,\Pi_{i}}[h_{-\Delta\mathcal{L}}]\stackrel{\textrm{\q{timerevinv}}}{=_n}\chi[h_{-\Delta\mathcal{L}}]. \label{eq:hoptrtpTR}
\end{align}
To derive the third row, we used that $C_m$-invariant momenta always appear in $n/m$-plets with equal values of polarization, and therefore we deduced, that the term $\sum m_j\Delta\mathscr{P}_{-\Pi_1,-\Pi_j}$ will contribute multiples of $n$ to the sum. 
To derive the fourth row, we used that a maximal subset $\{\Pi_j\}_{j=1\dots J}$ satisfying mutually disjoint and iso-orbital conditions in a $h[\Delta\mathcal{L}]$ model is in one-to-one correspondence with a maximal iso-orbital subset $\{-\Pi_j\}_{j=1\dots J}$ in a $h[-\Delta\mathcal{L}]$ model. 
Therefore one may relabel reduced momenta in the sum and arrive at the fourth row of Eq.~\eqref{eq:hoptrtpTR}. 
Using the fact that the Hopf invariant changes sign under timer reversal \eqref{timerevinv}, we conclude that if the Hopf-RTP relation holds for $h_{\Delta\mathcal{L}}$ then it must also hold for $h_{-\Delta\mathcal{L}}$.}

\subsubsection{Model Hamiltonian in class \texorpdfstring{$\mathrm{P}2\mathrm{-II}$}{P2-II} \label{sec:P2-example}}

To exemplify the Hopf-RTP relation in Eq.~\eqref{eq:hopf-RTP-P2II}, we construct an explicit model with a valence band representation that is not symmetry-equivalent to a basis band representation. 
To do so, we start with the $\mathrm{P}4$-symmetric MRW model [Eq.~\eqref{eq:hamMRW}] and add a term that reduces the symmetry to~$\mathrm{P}2$:
\begin{equation}
    h^\mathrm{inv.}=f(1-\cos k_x)\sigma_z,
\end{equation}
which closes the bulk energy gap along two out of the four rotation-invariant lines ($\gamma_\tx$ and $\gamma_\tm$) for $0.5\leq f \leq 12.5$. 
After the gap reopens, the valence band representation satisfies
\begin{equation}
\begin{split}
    \mathrm{VBR}\bigg|_{\gamma_\Gamma\cup\,\gamma_{\ty}}&=\mathrm{BBR}[\varphi_1],\bigg|_{\gamma_\Gamma\cup\,\gamma_{\ty}}, \\
    \mathrm{VBR}\bigg|_{\gamma_{\tm}\cup\,\gamma_{\tx}}&=\mathrm{BBR}[\varphi_2],\bigg|_{\gamma_{\tm}\cup\,\gamma_{\tx}},
\end{split}
\end{equation}
while the triviality of the Chern class is preserved.
Since the mutually-disjoint condition is fulfilled at all rotation-invariant momenta, the complete set of RTP invariants is $\Delta\mathscr{P}_{\Gamma\mathrm{Y}},\Delta\mathscr{P}_{\tx\tm}\in\mathbb{Z}$. For the presented model, the Hopf-RTP relation reads
\begin{equation}
    \chi=_2 \Delta\mathscr{P}_{\Gamma\ty}+\Delta\mathscr{P}_{\tm\tx}.
\end{equation}

\subsection{Hopf-RTP relation from Berry dipoles\label{sec:hopf-rtp-motivation}}

\subsubsection{Berry-dipole transitions lead to Eq.~(\ref{hopfrtp346})} 

Our goal here is to prove the Hopf-RTP relation in \q{hopfrtp346}, assuming that one begins from a canonical trivial phase and that all subsequent band touchings are Berry-dipolar. 
By `canonical trivial', we mean a phase with only on-site energies and zero `hopping'  matrix elements of the real-space tight-binding Hamiltonian. 
In such a phase, the Hopf and RTP invariants vanish, and both bulk-valence and bulk-conduction bands are basis band representations. 
In the following proof we assume the valence band representation to be induced from basis orbital $\varphi_1$, and the conduction band representation from $\varphi_2$. Then the itinerant angular momenta of the valence and conduction bands are compatible with the Berry dipole $\bk\cdot\bp$ Hamiltonian of the form given by Eqs.~\eqref{eq:kp-sym-spinor} and \eqref{eq:zszs-hamiltonian} (cf. discussion in footnote~\ref{foot:spinor2}). We comment on the opposite choice of the orbitals inducing valence and conduction bands at the end of this section.

Let us modify the Hamiltonian and transit to a distinct gapped phase via a Berry-dipole critical point. 
We assume that a Berry-dipole band touching occurs at an $(n/m)$-multiplet of $C_m$-invariant momenta $\{\bk^\textrm{t}=(\Pi^\textrm{t},k_z^\textrm{t})\}_{\textrm{t}=1,\ldots,n/m}$. 
As explained in \s{sec:Hopf-RTP-correspondence} [Eq.~(\ref{eq:chi-total-change})], the Hopf invariant changes by an integer value:  
\begin{subequations}
\e{\delta \chi=-\sum_{\textrm{t}=1}^{n/m}\upsilon(\bk^\textrm{t}) \Delta\ell(\bk^\textrm{t})=-\f{n}{m}\upsilon(\bk^1) \Delta\ell(\bk^1), \label{eq:chi-change-1}} 
with helicity $\upsilon(\bk^\textrm{t})=\pm 1$ and spin $\Delta\ell(\bk^\textrm{t})\in \Z$ being parameters of the Berry-dipole $\bk\cdot \bp$ Hamiltonian [cf.\ \q{eq:kp-sym-spinor}] and $\bk^1$ being a representative momentum of the $n/m$-multiplet. 
At the same time, according to Eq.~\eqref{eq:Berry-dipole-pol-change}, the Zak phase at reduced momentum $\Pi_t$ must change by $-2\pi \upsilon(\bk^\textrm{t})$, implying (by the geometric theory of polarization) that the polarization at $\Pi^\textrm{t}$ changes by $\delta\mathscr{P}(\Pi^\textrm{t})= -\upsilon(\bk^\textrm{t})$. 
Combining the last two equalities relates the change in the Hopf invariant to the change in polarization:
\e{   \delta\chi \eq \f{n}{m} \Delta{\ell}(\bk^1)\;\delta\mathscr{P}(\Pi^1). \la{combinelast}}
The integer-valued $\Delta{\ell}(\bk^1)$ defined through the continuous-rotation-invariant, Berry-dipole Hamiltonian can be identified, modulo $m$, as the difference in  itinerant angular momenta 
between conduction and valence states at $\bk^1$ [cf.\ \q{eqn:l-L-mod-m}].

One may therefore substitute  $\Delta{\ell}\ri \Delta \widetilde{\mathcal{L}}$ in \q{combinelast} if we relax the equality to an equivalence modulo $n$:
\e{  \delta\chi =_n \f{n}{m} \Delta\widetilde{\mathcal{L}}(\Pi^1)\;\delta\mathscr{P}(\Pi^1)=\sum_{\Pi}\Delta\widetilde{\mathcal{L}}(\Pi)\;\delta\mathscr{P}(\Pi). \label{eq:chi-change-3}}
In the last equality, we 
extended the expression to sum over all $C_{m_{\Pi}}$-invariant ($m_{\Pi}>1$) momenta $\Pi$, with the understanding that values of $\Pi$ not belonging to the $(n/m)$-multiplet do not contribute to the sum, as the change in polarization $\delta\mathscr{P}(\Pi)$ is zero for these non-belongers. 

Next, we apply that 
differences $\delta\mathscr{P}$ in polarization are unaffected by the choice of spatial origin, hence one may as well redefine the polarization $\mathscr{P}\ri \mathscr{P}'=\mathscr{P}-z_1$, with $z_1$ the z-coordinate of the positional center of basis orbital $\varphi_1$: 
\begin{align}
    \delta\chi=_n \sum_{\Pi}\Delta\widetilde{\mathcal{L}}(\Pi)\,\delta[\mathscr{P}(\Pi)-z_1].
    \label{eq:deltachi-deltap}
\end{align}
Beginning from the canonical trivial phase and tuning a Hamiltonian parameter, the net change of the Hopf/RTP invariants after a series of band touchings is simply a sum of changes attributed to each band touching, and moreover $\delta \chi$ is linearly related to $\delta[\mathscr{P}(\Pi)-z_1]$ [cf.\ Eq.~\eqref{eq:deltachi-deltap}]. This implies the Hopf invariant is related to the polarization as 
\begin{equation}
    \chi=_n\sum_{\Pi}\Delta\widetilde{\mathcal{L}}(\Pi)[\mathscr{P}(\Pi)-z_1].
    \label{eq:chi-pol-z1}
\end{equation}
By our assumption that the Hamiltonian has a RTP, there exists at least one rotation-invariant momentum (call it $\Pi'$) where the mutually-disjoint condition is satisfied: $\Delta\widetilde{\mathcal{L}}(\Pi')\neq 0$; the polarization $\mathscr{P}(\Pi')$ must then equal $z_1$ modulo integer, according to an argument presented below \q{eq:pol-def}. By replacing  $z_1$ in \q{eq:chi-pol-z1} by $\mathscr{P}(\Pi')$ (mod 1), one arrives 
at
\begin{equation}
    \chi=_n \!\sum_{\Pi}\!\Delta\widetilde{\mathcal{L}}(\Pi) [\mathscr{P}(\Pi)-\mathscr{P}(\Pi')]+ (\text{integer})\!\sum_{\Pi}\Delta\widetilde{\mathcal{L}}(\Pi).
    \label{eq:chi-pol-z12}
\end{equation}
The sum $\sum_{\Pi}\Delta\widetilde{\mathcal{L}}(\Pi)$ in the last term may be recognized, modulo $n$, as the first Chern invariant over a 2D subtorus (of the BZ) parametrized by $k_x$ and $k_y$ [cf.~Appendix~\ref{app:Chern-via-ell} for details]. But this Chern invariant vanishes by our assumption that the first Chern class is trivial.
Therefore, we can simply drop the last term in \q{eq:chi-pol-z12} under the equivalence modulo $n$:
\begin{equation}
    \chi=_n\sum_{\Pi}\Delta\widetilde{\mathcal{L}}(\Pi)\Delta \mathscr{P}_{\Pi' \Pi}; \as \Delta \mathscr{P}_{\Pi' \Pi} = \mathscr{P}(\Pi)-\mathscr{P}(\Pi').
    \label{eq:chi-pol-via-Berrydip}
\end{equation}
\end{subequations}
Assuming that there is only one maximal iso-orbital subset $\{\Pi_j\}_{j=1\dots J}$, the above equation is equivalent to \q{hopfrtp346}. To recognize this equivalence, note in the summation over all rotation-invariant reduced momenta, that $\Delta\widetilde{\mathcal{L}}(\Pi)\neq 0$ only when $\Pi$ belongs to the subset $\{\Pi_j\}_{j=1\dots J}$, and one may as well choose $\Pi'=\Pi_1$ (a representative member of this subset) without changing the relation modulo $n$. 

Let us finally show, that the derived result does not depend on the choice of the orbital inducing the valence band representation. Indeed, if instead of $\varphi_1$ we choose it to be $\varphi_2$, the spinor of a Berry dipole $\bk\cdot\bp$ Hamiltonian compatible with the itinerant angular momenta of the valence and conduction bands is given by Eq.~\eqref{eq:kp-sym-spinor-2}. As was discussed in footnote~\ref{foot:hopf-change-2}, the change in the Hopf invariant in this case is of an opposite sign than in Eq.~\eqref{eq:chi-change-1}. Simultaneously, the differences in itinerant angular momenta between the basis Bloch functions and between conduction and valence bands are related with an extra minus sign $\Delta\widetilde{\mathcal{L}}(\Pi)=_{m_\Pi}-\left(\widetilde{\mathcal{L}_2}(\Pi)-\widetilde{\mathcal{L}_1}(\Pi)\right)$. Therefore, for a $\bk\cdot\bp$ expansion around $\Pi$, $\Delta\ell=_{m_\Pi} -\Delta\widetilde{\mathcal{L}}(\Pi)$, and Eq.~(\ref{eq:chi-change-3}) holds without modifications.
 
Note that so far we didn't reproduce Eq.~(\ref{eq:hopf-RTP-P2II}) which applies for the exceptional $\mathrm{P}2$-symmetric case, which leads us to raise the following question: could there have been two maximal iso-orbital subsets within the described construction? 
The answer is that \emph{not under our assumption that all band touchings are Berry-dipolar}, as we proceed to clarify.

\subsubsection{Beyond the Berry-dipole assumption}\la{sec:relaxdipole}

We have just derived the Hopf-RTP relation \q{hopfrtp346} based on an assumption that the only band touchings are Berry dipoles. What Hamiltonian phases are missed under this assumption? Could it be that the Hopf-RTP relation takes a different form in these `missed phases'? The short answer is that there are no `missed phases' for space groups $\mathrm{P}(n=3,4,6)$; however, there exists `missed phases' for $\mathrm{P}2$, for which a more general Hopf-RTP relation in \q{eq:hopf-RTP} holds.

What exactly distinguishes the `missed phases' from the phases that are not missed? To pin this down, one needs to appreciate that a Berry-dipole band touching does not transform the band representation of the  bulk-valence subspace into a distinct band representation. In more detail, if one begins from the canonical trivial phase with the bulk-valence band being a basis band representation  BBR$[\varphi_j]$, with $\varphi_j$ a basis orbital with on-site angular momentum $\mathcal{L}_j$ and Wyckoff position $\br_j$, then any subsequent Berry-dipole band touchings will not change the fact that the bulk-valence band is a band representation induced from an $(\mathcal{L}_j,\br_{j,\perp})$-orbital; the $z$-coordinate of the Wyckoff position will, however, generically differ. This follows from three facts:\\

\noindent (\emph{i}) Recall that the  Berry dipole is, by definition, a dipolar source of Berry curvature with zero monopole charge. Hence \textit{Fact 1 (topological formulation): a Berry dipole cannot change the first Chern class of the bulk-valence band.}
A theorem proven in \ocite{nogo_AAJH} tells that a unit-rank $\mathrm{P}n$-symmetric band  with trivial first Chern class is simply a band representation of $\mathrm{P}n$. 
By `unit-rank', we mean that the band is given by a single linearly-independent Bloch function $\psi_{\bk}$ over the Brillouin torus, as is certainly true for the bulk-valence band of a two-band Hamiltonian. 
Thus an equivalent statement is \textit{Fact 1 (crystallographic formulation): a Berry-dipole transition preserves the band-representability of  a unit-rank bulk-valence band.}\\

\noindent (\emph{ii}) Because a Berry-dipole band touching occurs at an isolated $\bk$-point, while the itinerant angular momentum is a property of a $\bk$-line, a Berry-dipole touching cannot change the itinerant angular momenta of the bulk-valence band.
Hence \textit{Fact 2: the bulk valence band has exactly the same itinerant angular momenta as a basis band representation induced from $\varphi_1$.} \\

\noindent (\emph{iii}) \textit{Fact 3 (crystallographic formulation): two unit-rank band representations of $\mathrm{P}n$ are equivalent as band representations if and only if they have  exactly the same itinerant angular momenta.} 
To be equivalent as unit-rank band representations of $\mathrm{P}n$ means that their corresponding representative Wannier orbitals have identical on-site angular momenta $\mathcal{L}$ and $xy$-projected Wyckoff positions $\br_{\perp}$; we have been denoting equivalence classes of band representations by BR$(\mathcal{L},\br_{\perp})$. 
Applying the same theorem in Ref.~\cite{nogo_AAJH}, an equivalent restatement is \textit{Fact 3 (topological formulation): given a $\mathrm{P}n$-symmetric, unit-rank band with trivial Chern class, its itinerant angular momenta are in a one-to-one correspondence with an
equivalence class of band representations.} 
Fact 3 is proven in \app{app:onetoone}, and contrasts with the other known fact that if the Chern class is nontrivial, the Chern invariants are \textit{not} in one-to-one correspondence with the itinerant angular momenta~\cite{Chen_bulktopologicalinvariants}.\\
 
\noindent Synthesizing these three facts, we find that in any Hamiltonian phase connected to the canonical trivial phase via Berry-dipole transitions, the valence band is a band representation in the same equivalence class [BR$(\mathcal{L}_j,\br_{j,\perp})]$ as for BBR$[\varphi_j]$. 
Thus we must admit the possibility of `missed phases' where the bulk-valence band is a band representation but  is not symmetry-equivalent to  either of BBR$[\varphi_j]$.

Actually, for two-band Hamiltonians with space groups P$n$  ($n=3,4,6$),  the triviality of the first Chern class and the Hilbert-space constraint of a two-band Hamiltonian altogether imply that the valence band representation is symmetry-equivalent to one of BBR$[\varphi_j]$, as proven in  Appendix~\ref{app:no-band-inv-from-Chern}. 
This suggests that \q{hopfrtp346} holds without exceptions for the P$n$  ($n=3,4,6$) space groups. 
This is true, and the most direct way to prove this is by the methods of Sec.~\ref{sec:proof-RTP-Hopf} or Sec.~\ref{sec:linking-RTP-Hopf}. 
On the other hand, a trivial Chern class in the $\mathrm{P}2$ space group does \textit{not} imply the bulk-valence band is symmetry-equivalent to one of BBR$[\varphi_j]$. 
Thus the more general Hopf-RTP relation in \q{eq:hopf-RTP} is only needed for the $\mathrm{P}2$ space group. For models in class P2-I this general relation coincides with Eq.~\eqref{hopfrtp346}, as for P$n$ ($n=3,4,6$), while for models in class P2-II it takes the form in Eq.~\eqref{eq:hopf-RTP-P2II}. This latter relation is proven separately in Sec.~\ref{sec:hopf-RTP-proof-non-BBR}.

\subsection{Hopf-RTP relation from angular-momentum anomaly}\label{sec:proof-RTP-Hopf}

Having discussed the bulk-boundary correspondence of the RTP insulator, we employ the angular-momentum anomaly [introduced in \s{sec:anomaly}] to prove the general, mod-$n$ Hopf-RTP relation in Eq.~\eqref{eq:hopf-RTP}. 
This proof will use the known bulk-boundary correspondence of the Hopf insulator~\cite{AA_teleportation}, which we first review in Sec.~\ref{sec:BBC-Hopf}. 
To make the proof self-contained, we recapitulate in Sec.~\ref{sec:BBC-rtp} a few salient features of the bulk-RTP-boundary correspondence [discussed at length in Sec.~\ref{sec:BBC}]. 
Subsequently, in Secs.~\ref{sec:BBC-Hopf-RTP-BBR} and \ref{sec:hopf-RTP-proof-non-BBR}, we relate the surface signatures of both Hopf and RTP invariants, and then derive the relation between their bulk counterparts.
Throughout this section, we consider a semi-infinite geometry where the insulator occupies a space with $z>0$; our boundary will be the bottom $(z=0)$ facet.

\subsubsection{Review: Bulk-boundary correspondence of the Hopf insulator}\label{sec:BBC-Hopf}

Hopf insulators do not necessarily have conducting surface states whose energies interpolate between the bulk valence and conduction bands; any mid-gap surface states can be removed from the gap by a gap-preserving Hamiltonian deformation [cf.~Fig.~\ref{fig:mrw_surface}(b)]. 
Despite being removable from the energy gap, the totality of all surface-localized states (occupied and unoccupied) possesses a  Chern number (which we call a \textit{faceted Chern number}) that is equal to the bulk Hopf invariant~$\chi$~\cite{AA_teleportation}. 
We note that this equality holds in the convention of the surface normal vector $\hat{\bm{n}}$ pointing outside of the sample. For the bottom facet this means that the normal vector is oriented opposite to the unit vector along the $z$-axis: $\hat{\bm{n}}=-\hat{\bm{e}}_z$.

To rationalize the stated equality between the faceted Chern number and the Hopf invariant, consider that the total surface anomalous Hall conductance on a semi-infinite geometry must vanish, when the contributions from all states (occupied or not) in the Hilbert space are added up. 
The bulk-conduction and bulk-valence bands each contribute a surface Hall conductance of $-\chi/2$ with the same sign \cite{rauch_sahc} [that the Hopf invariant is equal (in magnitude and in sign) for both the bulk-conduction and the bulk-valence bands follows from a fundamental antisymmetry of all Hamiltonians that are a sum of Pauli matrices: $\sigma_y h(\bk)^* \sigma_y=-h(\bk)$].
For the net surface Hall conductance to vanish, there must exist faceted states with a nontrivial, compensating Chern number: $\mathscr{C}_f=\chi$.

As with the faceted Zak phases described in Sec.~\ref{sec:BBC_zak}, there are two approaches to compute the faceted Chern number $\mathscr{C}_f$. First, one can find a Hamiltonian deformation that detaches surface-localized bands from the rest of the spectrum [e.g., cf.~Fig.~\ref{fig:mrw_surface}(c)],  and then compute the Chern number by standard numerical algorithms. Alternatively, one can compute the Chern number of the surface-like polarization bands, employing a method introduced in \ocite{nband-hopf} and reviewed in Appendix~\ref{app:wannier-cut}.

\subsubsection{Recapitulation: Bulk-boundary correspondence of the RTP insulator\label{sec:BBC-rtp}}

There exists at most two maximal subsets  of rotation-invariant reduced momenta in the rBZ satisfying the mutually-disjoint \eqref{mutualdisj} and iso-orbital \eqref{mutualdisj2} conditions. 
We denote them as $\{\Pi_{j_\alpha}^\alpha\}_{j_\alpha=1,\ldots, J_\alpha}$, with $\alpha{=}1,2$ a subset index. 
Picking $\Pi_1^\alpha$ as a representative member in each subset, a complete set of independent RTP invariants is given by $\{\Delta\mathscr{P}_{\Pi_1^1\Pi_{j_1}^1}\}_{j_1=2\dots J_1}\bigcup\{\Delta\mathscr{P}_{\Pi_1^2\Pi_{j_2}^2}\}_{j_2=2\dots J_2}.$\footnote{According to Appendix~\ref{app:no-band-inv-from-Chern}, a maximal subset containing only one element (i.e., one reduced momentum) is incompatible with the Hamiltonian having trivial Chern class, therefore there is at least one well-defined RTP invariant for each non-empty maximal subset.}

An angular-momentum anomaly is associated to  each of the independent RTP invariant $\Delta\mathscr{P}_{\Pi_1^\alpha\Pi_{j_\alpha}^\alpha}$ [cf.\ Sec.~\ref{sec:anomaly}]. To recapitulate what this means: defining $\mathcal{N}_c(\Pi)$ [resp.\ $\mathcal{N}_v(\Pi)$] as the total number of faceted states whose itinerant angular momenta are bulk-conduction-like (resp.\ bulk-valence-like) at $\Pi$, then 
\begin{equation}\la{anomaly2}
\begin{split} 
\mathcal{N}_v(\Pi_{j_\alpha}^\alpha) &= \mathcal{N}_v(\Pi_1^\alpha)+\Delta\mathscr{P}_{\Pi_1^\alpha\Pi_{j_\alpha}^\alpha}, \\
\mathcal{N}_c (\Pi_{j_\alpha}^\alpha) &= \mathcal{N}_c(\Pi_1^\alpha)-\Delta\mathscr{P}_{\Pi_1^\alpha\Pi_{j_\alpha}^\alpha}, 
\end{split}\end{equation}
with the total number of faceted bands being independent~of~$\Pi$:
\e{
\mathcal{N}_\mathrm{tot} \eq \mathcal{N}_v(\Pi)+\mathcal{N}_c(\Pi).
\la{ntot}}

\subsubsection{Hopf-RTP relation for a valence band that is symmetry-equivalent to a basis band representation\label{sec:BBC-Hopf-RTP-BBR}}

Assuming that the valence band is symmetry-equivalent to a basis band representation, there is only one maximal subset of reduced momenta $\{\Pi_{j}\}_{j=1\dots J}$ fulfilling the mutually-disjoint and iso-orbital conditions, and therefore a complete set of independent RTP invariants is given by $\{\Delta\mathscr{P}_{\Pi_1\Pi_{j}}\}_{j=2\dots J}$.

Our next step is to relate the surface signatures of the RTP and Hopf invariants, by utilizing a relation between Chern numbers and itinerant angular momenta derived in  \ocite{Chen_bulktopologicalinvariants} for p$n$-symmetric Hamiltonians. 
A simple manipulation of their relation (detailed in Appendix~\ref{app:Chern-via-ell}) allows to express  the faceted Chern number, modulo $n$, as the sum of itinerant angular momenta of all faceted states at \textit{all} rotation-invariant reduced momenta:
\begin{subequations}
\e{
    \mathscr{C}_f &=_n-\sum_{i=1}^{\mathcal{N}_\mathrm{tot}}
    \sum_{\Pi\in \text{all}}\widetilde{\mathcal{L}}_i(\Pi), \as \widetilde{\mathcal{L}}_i(\Pi) \in \{0,1,\ldots, m_{\Pi}-1\}.
    \label{eq:C-facet-via-ell}
}
By `all rotation-invariant reduced momenta', we mean not just $C_n$-invariant momenta, but also any $C_{m_{\Pi}}$-invariant momentum with $1<m_{\Pi}<n$, where $m_{\Pi}$ is the order of the little group of $\Pi$. Note that each $C_{m_\Pi}$-invariant momentum appears as a member of an $n/m_\Pi$-plet of symmetry-related momenta; in Eq.~(\ref{eq:C-facet-via-ell}) we include all members of each multiplet. Each of $\widetilde{\mathcal{L}}_i(\Pi)$ should be understood as mod-$m_{\Pi}$ variables taking values in $\{0,1,\ldots, m_{\Pi}-1\}$. 
The minus sign in \q{eq:C-facet-via-ell} [which is absent in Eq.~(\ref{eq:app:Chern-via-ell}) in the appendix] reflects the fact that the facet normal is oriented opposite to the rotation axis $\hat{\bm{z}}$. 
One may then apply \q{ntot} to separate faceted states into two categories, namely with bulk-valence-like vs.~bulk-conduction-like itinerant angular momenta:\footnote{For $\Pi$ not belonging to the maximal subset $\widetilde{\mathcal{L}}_v(\Pi)=\widetilde{\mathcal{L}}_c(\Pi)$, meaning that partitioning of $\mathcal{N}_\textrm{tot}$ into $\mathcal{N}_v(\Pi)$ and $\mathcal{N}_c(\Pi)$ is not unique. However, the equality of itinerant angular momenta ensures that the derivation goes through \emph{irrespective} of the chosen partitioning. A convenient choice which helps to derive the subsequent \q{eq:C-facet-separrate-bulk} is to set $\mathcal{N}_{v(c)}(\Pi)=\mathcal{N}_{v(c)}(\Pi_1)$ for all $\Pi$ not belonging to the maximal subset and $\Pi_1$ being the chosen representative number of the maximal subset.
}
\e{
    \mathscr{C}_f     &=_n -\sum_{\Pi\in \text{all}}\left[\mathcal{N}_v(\Pi)\widetilde{\mathcal{L}}_v(\Pi)+\mathcal{N}_c(\Pi)\widetilde{\mathcal{L}}_c(\Pi)\right].
    \label{eq:C-facet-via-ell2}
}
Inserting the angular-momentum-anomaly [Eq.~(\ref{anomaly2})] in \q{anomaly2}, we obtain
\begin{equation}\label{eq:C-facet-separrate-bulk}
\begin{split}
    \mathscr{C}_f     &=_n-\bigg[\mathcal{N}_v(\Pi_1)\left(\sum_{\Pi\in \text{all}}\widetilde{\mathcal{L}}_v(\Pi)\right)+\mathcal{N}_c(\Pi_1)\left(\sum_{\Pi\in \text{all} } \widetilde{\mathcal{L}}_c(\Pi)\right) \\
    &\qquad\quad+\sum_{j=2}^{J}\Delta\mathscr{P}_{\Pi_1\Pi_{j}}\left(\widetilde{\mathcal{L}}_v(\Pi_{j})-\widetilde{\mathcal{L}}_c(\Pi_{j})\right)\bigg].
\end{split}
\end{equation}
The first two sums 
on the right-hand side individually vanish modulo $n$ as they are related to the Chern numbers of the bulk-valence resp.~bulk-conduction bands, which are assumed to be zero in studied Hamiltonians.\footnote{To elaborate, consider that $(-)\sum_{\Pi\in\mathrm{all}}\widetilde{\mathcal{L}}_v(\Pi)$ [resp.\ $(-)\sum_{\Pi\in\mathrm{all}}\widetilde{\mathcal{L}}_c(\Pi)$] equals (modulo $n$) the Chern number in the $(k_x,k_y)$-plane of the bulk-valence [resp.\ bulk-conduction] band, as derived in \app{app:Chern-via-ell}. Our assumption that the Hamiltonian has trivial Chern class implies that the first two sums in Eq.~\eqref{eq:C-facet-separrate-bulk} are zero modulo $n$.}
The third and only remaining term on the right-hand side of \q{eq:C-facet-separrate-bulk} represents the effect of the anomaly in replacing  angular momenta which are bulk-conduction-like with those that are bulk-valence-like, or vice versa (depending on the sign of the RTP invariant.) 
A simple rewriting gives: 
\begin{equation}
    \mathscr{C}_f=_n\sum_{j=2}^{J}\Delta\mathscr{P}_{\Pi_1\Pi_{j}}\Delta\widetilde{\mathcal{L}}(\Pi_{j}),
    \label{eq:C-facet-ell-relation}
\end{equation}
\end{subequations}
with $\Delta\widetilde{\mathcal{L}}(\Pi) = \widetilde{\mathcal{L}}_c(\Pi)-\widetilde{\mathcal{L}}_v(\Pi)\pmod{m_\Pi}$.\footnote{The fact that all $C_{m_\Pi}$-invariant momenta from an $n/m_\Pi$-plet are present in the sum ensures that a possible multiples-of-$m_\Pi$ difference between $\Delta\widetilde{\mathcal{L}}(\Pi)$ and $\widetilde{\mathcal{L}}_c(\Pi)-\widetilde{\mathcal{L}}_v(\Pi)$ does not modify Eq.~\eqref{eq:C-facet-ell-relation} modulo $n$.}
Using the equality between the faceted Chern number and the bulk Hopf invariant $\chi=\mathscr{C}_f$ [cf.\ \s{sec:BBC-rtp}], we arrive at the mod-$n$ Hopf-RTP relation in \q{hopfrtp346}. 
We also emphasize that, in contrast to Sec.~\ref{sec:hopf-rtp-motivation}, the present derivation made no assumption whatsoever concerning the phase transitions encountered when reaching the final Hamiltonian from the `canonical trivial' phase.

\subsubsection{Hopf-RTP relation for a valence band that is not symmetry-equivalent to a basis band representation\label{sec:hopf-RTP-proof-non-BBR}}

We show in Appendix~\ref{app:no-band-inv-from-Chern} that of the $\mathrm{P}n$ groups, it is only the $\mathrm{P}2$ group that allows for crystalline Hopf insulators whose valence band representation is not  symmetry-equivalent to a basis band representation.
This case arises only when the mutually-disjoint condition is satisfied for all rotation-invariant momenta:
\begin{equation}\forall \Pi_j, \as \Delta\widetilde{\mathcal{L}}(\Pi_j)=1, \label{eq:mut-disj-P2}
\end{equation}
and when the two maximal subsets of reduced momenta (fulfilling the iso-orbital condition) each have exactly two elements, which we denote as $\{\Pi^\alpha_1,\Pi^\alpha_2\}_{\alpha=1,2}$; there are then only two RTP invariants $\{\Delta\mathscr{P}_{\Pi^\alpha_1\Pi^\alpha_2}\}_{\alpha=1,2}$.

To prove the Hopf-RTP relation \eqref{eq:hopf-RTP-P2II}, we again express the faceted Chern number in terms of the itinerant angular momenta of the faceted states, i.e., we utilize \q{eq:C-facet-via-ell} with $n=m_{\Pi}=2$ for all $\Pi$. 
Because each term is integer-valued, one may freely alter plus and minus signs [in \q{eq:C-facet-via-ell}] without modifying the equality modulo 2. 
Next we insert into \q{eq:C-facet-via-ell}  the angular-momentum-anomaly [Eqs.~\ref{anomaly2}], obtaining:
\begin{equation}\la{toinsert}
\begin{split}
    \mathscr{C}_f    &=_2  \sum_{\alpha=1}^2 \bigg[\mathcal{N}_v(\Pi_1^\alpha)\sum_{j=1}^2 \widetilde{\mathcal{L}}_v(\Pi_j^\alpha)+(v\ri c)\bigg]\\
    &\as +\sum_{\alpha=1}^2\Delta\mathscr{P}_{\Pi^\alpha_1\Pi^\alpha_2}\Delta\widetilde{\mathcal{L}}(\Pi^\alpha_2).
    \end{split}
\end{equation}
The square-bracketed terms in the first line of \q{toinsert} sum to an even integer\footnote{To reveal that this sum is indeed even, we rewrite the bracketed terms as
\begin{subequations}
    \begin{align}
        \sum_{\alpha=1}^2 \sum_{j=1}^2 \bigg\{\mathcal{N}_v(\Pi_1^1) \widetilde{\mathcal{L}}_v(\Pi_j^\alpha) + \left[\mathcal{N}_v(\Pi_1^\alpha)-\mathcal{N}_v(\Pi_1^1)\right]\widetilde{\mathcal{L}}_v(\Pi_j^\alpha)+(v\ri c)\bigg\}.
    \end{align}
The first term is proportional to $\sum_{\Pi\in \text{all}}\widetilde{\mathcal{L}}_v(\Pi)$, which has the same parity as the first Chern number of the bulk-valence band; hence this term vanishes (modulo $2$) under our assumption that the Hamiltonian has trivial first Chern class.
Likewise, the analogous term with $[v\ri c]$ vanishes modulo $2$, and we are left with
    \begin{align}
        &\sum_{j=1}^2\left[\left\{\mathcal{N}_v(\Pi_1^2)-\mathcal{N}_v(\Pi_1^1)\right\}\widetilde{\mathcal{L}}_v(\Pi_j^2)+\left\{\mathcal{N}_c(\Pi_1^2)-\mathcal{N}_c(\Pi_1^1)\right\}\widetilde{\mathcal{L}}_c(\Pi_j^2)\right] \lin
        &=_2\!\!\sum_{j=1}^2\left[\left\{\mathcal{N}_v(\Pi_1^2)-\mathcal{N}_v(\Pi_1^1)\right\}\!\widetilde{\mathcal{L}}_v(\Pi_j^2)+\left\{-\mathcal{N}_v(\Pi_1^2)+\mathcal{N}_v(\Pi_1^1)\right\}\!\widetilde{\mathcal{L}}_c(\Pi_j^2)\right].\la{tocom}
    \end{align}
In obtaining the last row, we used that $\mathcal{N}_c(\Pi)=\mathcal{N}_\mathrm{tot}-\mathcal{N}_v(\Pi)$ and that $\mathcal{N}_\mathrm{tot}$ appearing twice and with opposite signs in the second curly-bracketed term is cancelled out in the expression.
Combining both terms in \q{tocom}, we get
    \begin{align}
        \sum_{j=1}^2\left\{\mathcal{N}_v(\Pi_1^1)-\mathcal{N}_v(\Pi_1^2)\right\}\Delta\widetilde{\mathcal{L}}(\Pi_j^2) &=_2\left\{\mathcal{N}_v(\Pi_1^1)-\mathcal{N}_v(\Pi_1^2)\right\}\sum_{j=1}^2 1=_20.
    \end{align}
    \end{subequations}
} 
and therefore can be dropped from the modulo $2$ equality.
Then, bearing in mind that $\chi=\mathscr{C}_f$ and $\Delta\widetilde{\mathcal{L}}=1$ [cf.\ \q{eq:mut-disj-P2}], the desired Hopf-RTP relation \eqref{eq:hopf-RTP-P2II} follows from \q{toinsert}.

\subsection{Hopf-RTP relation from linking numbers}\label{sec:linking-RTP-Hopf}

An alternative proof of the Hopf-RTP relation derives from the linking-number interpretation of the Hopf invariant, which we briefly introduced in \s{sec:Berry-transitions} [cf.~Fig.~\ref{fig:Hopf-linking}]. 
Namely, the Hamiltonian with a flattened spectrum $h_\mathrm{flat}(\bk)=\boldsymbol{h}_\mathrm{flat}(\bk)\cdot\bsigma$ [defined for the spinor-form Hamiltonian in \q{eq:zszs-hamiltonian-normalized}] is viewed as a continuous map from the BZ three-torus to the Bloch two-sphere. The corresponding Hopf invariant equals the linking number of the (oriented) preimages of \textit{any two} distinct points on the Bloch two-sphere~\cite{Hopf:1931,wilczekzee_linkingnumbers,kennedy_hopfchern}. 

For instance, let us arbitrarily pick a point $b$ on the Bloch two-sphere [green dot in Fig.~\ref{fig:whitehead-min}(b)]. 
Because the Hamiltonian maps between spaces whose dimensions differ by unity, the preimage of $b$, denoted $\gamma(b)$, is generically a one-dimensional submanifold of the BZ. We will make a few assumptions to simplify the first-round discussion, deferring the general treatment to \app{app:hopf-rtp-whitehead}. 
The first assumption is that $\gamma(b)$ is topologically equivalent to a contractible circle, as illustrated with green path in Fig.~\ref{fig:whitehead-min}(a).
Then, $\gamma(b)$ is equipped with a positive-integer-valued multiplicity $\mu[\gamma(b)]$ and an orientation $\bm{\tau}[\gamma(b)]$, 
which is a vector tangential to $\gamma(b)$ at each $\bk_0\,{\in}\,\gamma(b)$.

We proceed by defining these quantities. 
Let $\tinyloop$ [blue in Fig.~\ref{fig:whitehead-min}(a)] be an infinitesimal loop linking with $\gamma(b)$ and centered at $\bk_0$ [black dot in Fig.~\ref{fig:whitehead-min}(a)]; as $\bk$ is varied around $\tinyloop$, the image of $\bk$ on the Bloch sphere will make $\mu[\gamma(b)]$ number of revolutions (or `orbits') around $b$, as illustrated in Fig.~\ref{fig:whitehead-min}(b). 
Given that the direction of these revolutions is governed by a right-hand rule with a thumb pointing inside of the Bloch sphere, the orientation $\bm{\tau}[\gamma(b)]$ at $\bk_0$ is determined using the right-hand rule by the direction in which $\bk$ is varied around $\tinyloop$, as illustrated with a green arrow in Fig.~\ref{fig:whitehead-min}(a).

\begin{figure}
    \centering
    \includegraphics{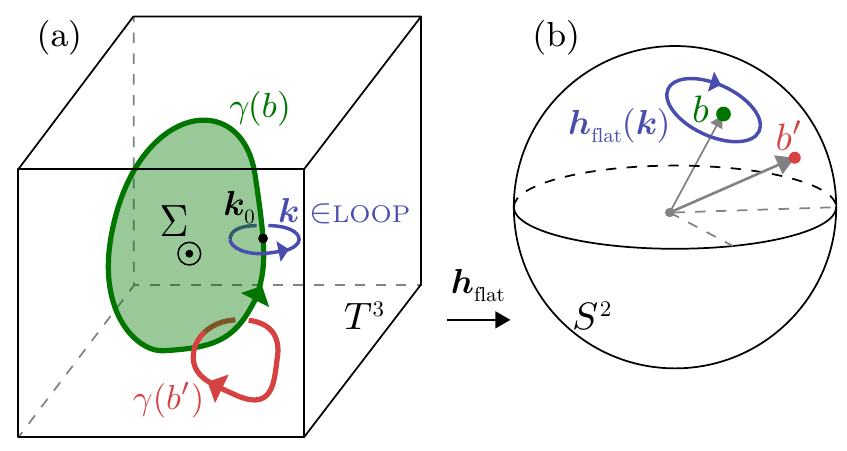}
    \caption{
    Two-band, insulating Hamiltonian with flattened spectrum is represented as a map from a three-torus to a two-sphere, $h_\mathrm{flat}(\bk)=\boldsymbol{h}_\mathrm{flat}(\bk)\cdot\bsigma$ with $\boldsymbol{h}_\mathrm{flat}\in S^2$. 
    The preimage of a point $b\in S^2$ [green dot in (b)] is generically a one-dimensional BZ-submanifold $\gamma(b)$  [green line in (a)] with the orientation defined by the winding of $\boldsymbol{h}_\mathrm{flat}(\bk)$ around the point $b$ [blue line in (b)] as $\bk$ varies around infinitesimal ${\tinyloop}$ [blue line in (a)] centered at $\bk_0\in\gamma(b)$ [black dot in (a)].
    The preimage line $\gamma(b)$ is a boundary to an oriented surface $\Sigma$ [green shaded area in (a)]. 
    Preimage $\gamma(b')$ of another point $b'\in S^2$ [red dot in (b)] intersects $\Sigma$ a number of times equal to $|\chi|$.
    }
    \label{fig:whitehead-min}
\end{figure}

It is useful to introduce a Gaussian surface $\Sigma$ bounded by $\gamma(b)$; any surface bounded by $\gamma(b)$ will do [cf.~Fig.~\ref{fig:whitehead-min}(a) for an example of such a surface]. 
The surface $\Sigma$ inherits an orientation from the orientation of $\gamma(b)$ via the right-hand rule: if the fingers are aligned along the positively oriented direction of $\partial\Sigma = \gamma(b)$, then the thumb gives the positive orientation of $\Sigma$. 
The linking-number interpretation states that for any point $b'\neq b$ on the Bloch two-sphere, the corresponding preimage $\gamma(b')$ will link with $\gamma(b)$ with a \emph{linking number}\footnote{
The notion of `linking number' must be generalized to contexts beyond the Hopf map from $S^3$ to $S^2$. For crystallographic-symmetry-constrained maps, it is useful to adopt a \textit{generalized linking number} which encodes not just that two loops link but also the multiplicity $\mu[\gamma(b)]$ of each loop.  
(In the Hopf map, $\mu=1$  and is a non-issue.) 
The generalized linking number of  two loops [$\gamma(b)$ and $\gamma(b')$] is then $\mu[\gamma(b)]\times\mu[\gamma(b')]$
multiplied with the usual notion of linking number. A subsequent use of terminology (`intersection number') in the main text should also be understood in this generalized sense.\label{foot:Hopf-linking} } that is equal to the Hopf invariant $\chi$~\cite{Hopf:1931,wilczekzee_linkingnumbers,kennedy_hopfchern}. 
This means that if we parametrize  $\gamma(b')$ by a variable that we interpret as `time', then $\gamma(b')$ will intersect the Gaussian surface $\Sigma$ a total of $|\chi|$ `times'; the sign of $\chi$ distinguishes two topologically distinct ways in which two oriented loops can link [cf.~Fig.~\ref{fig:link-sign}]. 
Because this linking holds for any choice~\cite{kennedy_hopfchern} of $b'$ over the Bloch sphere (for fixed $b$), it follows
that the image of $\Sigma$ covers the Bloch sphere $\chi$ times~\cite{Whitehead:1947}. 

Furthermore, the map from $\Sigma$ to the Bloch sphere can also be viewed as a map between a pair of two-spheres; namely, because the boundary $\partial\Sigma$ of $\Sigma$ is mapped to a single point on the Bloch sphere, one can identify $\partial \Sigma \sim \textrm{pt.}$ with a single point, such that the quotient $\Sigma /{\sim} \simeq S^2$ is topologically identified as a two-sphere. 
But the wrapping number of a map between two two-spheres is nothing more than the (first) Chern number. 
Hence, we identify $\chi$ as the Chern number of the map from $\Sigma$ to the Bloch two-sphere; we call this the \textit{Hopf-Chern relation} (not to be confused with the fact that the Hamiltonian map has trivial first Chern class). 
The Hopf-Chern relation was first espoused by Whitehead for the Hopf map: $S^3 \ri S^2$~\cite{Whitehead:1947}; we are proposing here that the Hopf-Chern relation holds also for crystallographic-symmetry-constrained maps from  $T^3 \ri S^2$ with trivial first Chern class.

\begin{figure}
    \centering
    \includegraphics{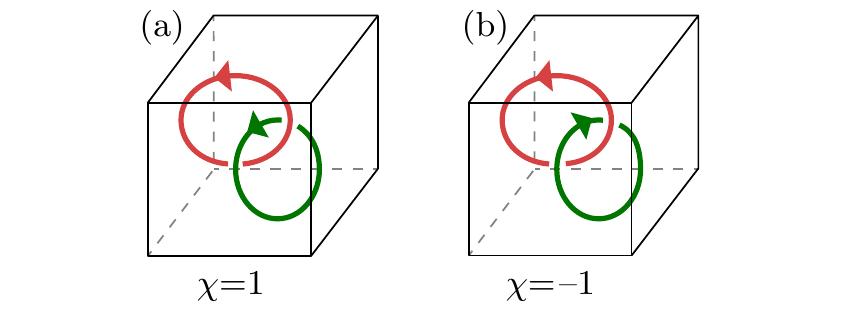}
    \caption{Geometrical interpretation of the Hopf invariant $\chi$ as a linking number of two oriented preimage lines (green and red loops).
    The sign of $\chi$ [positive in (a), negative in (b)] distinguishes between two topologically distinct ways to link the oriented loops.}
    \label{fig:link-sign}
\end{figure}

The Hopf-Chern relation turns out to be very useful in proving the Hopf-RTP relation. 
Deferring the general proof (which holds for any $\mathrm{P}n$ symmetry) to Appendix~\ref{app:hopf-rtp-whitehead}, we focus here on proving the relation for a particular class of $\mathrm{P}4$-symmetric Hamiltonians for which only a single RTP invariant $\Delta\mathscr{P}_{\Gamma \tm}$ is well-defined, that being the polarization difference between the only two four-fold-invariant reduced momenta: $\Gamma$ and $\tm$. 
An example of such a model is given in the third row [i.e., $\textrm{SG} = \textrm{P}
4$, $\Delta\mathcal{L}=1$, $\Delta \br_\perp = (1/2,1/2)$] of Table~\ref{tab:model-zoo}. 

Due to $\mathrm{P}4$ symmetry and the mutually disjoint condition, the wave function of the valence (and also of the conduction) band is constant along the four-fold-invariant $\bk$-line specified by $\Gamma$; in other words, the Hamiltonian maps this $\bk$-line to a single point on the Bloch sphere. 
This single point is either the north pole or the south pole, in a basis that simultaneously diagonalizes the rotation matrix. 
These statements that apply to $\Gamma$ apply equally well to $\tm$. 
In the presented argument we restrict ourselves to models with the valence (resp.~conduction) band representation being symmetry-equivalent to the basis band representation BBR$[\varphi_1]$ (resp.~BBR$[\varphi_2]$). Then both four-fold-invariant $\bk$-lines map to the south pole (cf.~Fig.~\ref{fig:whitehead-c4-min} for an illustration).

\begin{figure}
    \centering
    \includegraphics{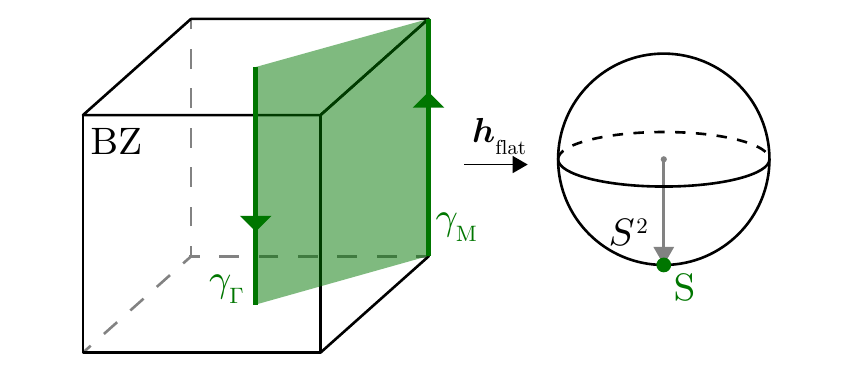}
    \caption{Preimage manifold of the south pole $\textrm{S}\sim -\sigma_z$ consists of two vertical rotation-invariant lines $\gamma_\Gamma$ and $\gamma_\textrm{M}$, illustrated by dark-green oriented lines. Light-green sheet indicates a 
    surface bounded by $\gamma_\Gamma\cup\gamma_\textrm{M}$. }
    \label{fig:whitehead-c4-min}
\end{figure}

Let us suppose for the moment that the preimage of the south pole comprises only these two $\bk$-lines. 
Defining orientations of these preimage lines using the above-described procedure is not straightforward; to overcome this complication, we present in Appendix~\ref{app:hopf-rtp-whitehead} a symmetry-based method to determine preimage orientations. 
To make the discussion in this section more brief, we defer application of this method to the general proof in Appendix~\ref{app:hopf-rtp-whitehead} and assume no knowledge of the preimage orientations in the current discussion. This will leave a sign ambiguity in the derived
Hopf-RTP relation.

A Gaussian surface that is bounded by the preimage of the south pole is given by a vertical rectangular sheet with the two vertical sides equal to the four-fold-invariant $\bk$-lines, and the two horizontal sides identified due to the periodicity of the Brillouin torus, cf.~Fig.~\ref{fig:whitehead-c4-min}. 
For an orientation of the constructed surface to be well-defined from orientations of the preimages, it must be
that the two vertical lines have opposite orientations; one choice of possible orientations is illustrated in Fig.~\ref{fig:whitehead-c4-min}.
Then the Hopf-Chern relation tells us that the Hopf invariant equals, up to a sign determined by the preimage orientation, to the Chern number defined over this sheet. 
Application of Stokes' theorem further gives us that the same Chern number equals the difference in Berry-Zak phase between the two four-fold-invariant $\bk$-lines, divided by $2\pi$ \cite{AA_wilsonloopinversion}. 
But this difference (divided by $2\pi$) is simply the polarization difference ($\Delta\mathscr{P}_{\Gamma \tm}$) between these lines. 
Hence we arrive at a Hopf-RTP relation $\chi=\pm\Delta\mathscr{P}_{\Gamma \tm}$. The sign in this relation can be further specified after the preimage orientations are fixed. This additional step is completed for any P$n$ symmetry in Appendix~\ref{app:hopf-rtp-whitehead}.

This particular instance of the Hopf-RTP relation was na\"ively derived under the assumption that the preimage of the south pole comprises no preimages besides
the two four-fold-invariant $\bk$-lines. 
Thus, suppose the preimage of the south pole involves a $\bk$-loop centered at a generic $\bk$-point. 
There must then be three other $\bk$-loops related by four-fold symmetry. 
The Chern numbers of all four $\bk$-loops must be identical, because the Berry curvature transforms as a pseudovector under crystallographic operations. 
It follows that the na\"ive Hopf-RTP relation can be modified by additions of integer multiples of four:
\e{ \chi =_4 \pm\Delta\mathscr{P}_{\Gamma \tm}.\label{eq:hopfrtp-min-ex}}
We note that other choices of additional preimages are possible; they are all analyzed in Appendix~\ref{app:hopf-rtp-whitehead}, and result in the same mod-4 reduction of the Hopf-RTP relation.
It would be beneficial to points readers' attention to the result indicated in row $\textrm{SG} = \textrm{P}4$, $\Delta\mathcal{L}=1$, $\Delta \br_\perp = (1/2,1/2)$] in Table~\ref{tab:models_summary}, which specifies the sign in the relation \eqref{eq:hopfrtp-min-ex} to be minus.
This is a particular instance of the $\mathrm{P}n$-symmetric Hopf-RTP relation in \q{hopfrtp346}, which is proven more generally in Appendix~\ref{app:hopf-rtp-whitehead}.

\section{{Delicate crystalline insulators beyond two-band Hamiltonians}} \label{sec:multi-band}

In discussing possible extensions of crystalline Hopf insulators beyond two-band Hamiltonians, we ask: \textit{Is the classification by the Hopf and RTP invariants naturally generalizable when the Hilbert space is expanded to include more trivial bands?} 
(The meaning of `triviality' will be sharpened below.)
Speaking colloquially, does the topological nontriviality `survive' an expansion of the Hilbert space to include strictly-localized degrees of freedom? Formalizing what `survival' means for the crystalline Hopf insulator leads directly to the related notions of \emph{delicate topology} and \emph{symmetry-protected delicate topology}. 
The categorization of topological insulators into delicate vs.~fragile vs.~stable is the first step (or an essential premise) in any formal, algebraic-topological classification~\cite{kitaev_periodictable,nogo_AAJH,bouhon_wilsonloopapproach,crystalsplit_AAJHWCLL,AA_teleportation,DeNittis_classifyAI,shiozaki_review,read_compactwannier}.

A different motivation for the above, italicized question comes from attempting to realize delicate topology in the laboratory. 
Any solid-state/meta-material/cold-atomic realization of topological band structure must contend with  the number of bands -- in any single experimental platform -- being large (even  infinite, formally). 
Thus, an effective few-band Hamiltonian is only an approximate description. 
It is then essential, when confronting nature, to explore (and hopefully break down) the limitations of this approximation.

Our analysis of the postulated question is organized as follows. 
Firstly, in Sec.~\ref{sec:procedure}, we introduce the procedure by which we expand a topologically-nontrivial two-band Hamiltonian to a $({>}2)$-band Hamiltonian.
The implications of the Hopf invariant being delicate-topological are explored  in \s{sec:hopfbeyond2}; the implications of the RTP invariant being symmetry-protected delicate-topological are explored in \s{sec:rtpbeyond2}. 
Both the Hopf and the RTP invariants necessitate a topological obstruction for the localization of Wannier orbitals which is explained in the subsequent Sec.~\ref{sec:CHI-multicel}; going beyond Hopf and RTP insulators, it will be argued therein that this Wannier obstruction is generally a delicate-topological property.

\subsection{\texorpdfstring{$({>}2)$}{>2}-band Hamiltonians from two-band Hamiltonians}\la{sec:procedure}

We consider a specific procedure to produce a $({>}2)$-band Hamiltonian from a two-band Hamiltonian: we expand the two-band Hilbert space by adding to it a basis band representation (of $\mathrm{P}n$) induced from a third, representative tight-binding-basis orbital $\varphi_3$. 
If $\varphi_3$ is centered at a $C_n$-invariant Wyckoff position $\br_3$ and forms a certain one-dimensional representation of the site-stabilizer group $\mathrm{P}n_{\br_3}$, then the Hilbert space is expanded to three bands. 
(The expansion is to four or more bands if $\br_3$ is not $C_n$-invariant.)

By construction, a basis band representation is a \textit{tightly-bound band representation} -- a band representation in which each Wannier orbital is localized to a single site, or \textit{one-site-localized}.  
(A `site' is a notion existing only in the tight-binding formalism, and represents the physical location of a tight-binding-basis orbital. 
If $\br_3$ is $C_n$-invariant, the allowed sites are $\br_3$ plus any Bravais-lattice vector.)

A choice is made whether to add this tightly-bound band representation to the conduction or valence subspace; in either case, we preserve the notion of the bulk energy gap inherited from the two-band Hamiltonian.
(The choice between conduction vs.~valence is made by tuning the on-site energies of the tightly-bound band representation as desired.) 
After this addition procedure, we allow any continuous, $\mathrm{P}n$-symmetric deformation of the $({>}2)$-band Hamiltonian that preserves the bulk gap. 
This procedure can be repeated to produce a Hamiltonian with any number of conduction and valence bands. 

Given that a two-band Hamiltonian is characterized by a nontrivial Hopf or RTP invariant, one may ask if this special topological property of the conduction and valence subspace is lost under the expansion procedure detailed above. 
At this point, the notion of delicate topology emerges. The Hopf and RTP invariants, defined respectively by Eqs.~\eqref{eq:hopfinvar} and~\eqref{eq:RTP-def}, are examples of  delicate topological invariants~\cite{Nelson:2021}. 
This means their status as topological invariants is lost under the addition of certain `trivial' bands to either the conduction or to the valence subspace. 
In the classification of stable/fragile crystalline topological insulators, being `trivial' means being a band representation. 
In contrast, in the more refined classification of delicate crystalline topological insulators, being `trivial' means being a \textit{tightly-bound band representation}.\footnote{For the mathematically-oriented reader, further motivation of this definition of `triviality' stems from observing that a tightly-bound band representation is a space-group-equivariant analog of a product bundle, as explained in Appendix~G of \ocite{crystalsplit_AAJHWCLL}.}

However, the Hopf and RTP invariants differ in that the Hopf invariant becomes ill-defined 
under the addition of \emph{any} tightly-bound band representation (as elaborated in \s{sec:hopfbeyond2}), while the RTP invariant loses its invariant status only under the addition of \emph{particular} tightly-bound band representations (cf.\ \s{sec:rtpbeyond2}). 
This difference, and its implications, are elaborated subsequently.

\subsection{Hopf insulators beyond two-band Hamiltonians}\la{sec:hopfbeyond2}

For two-band Hamiltonians with trivial first Chern class, a bulk-boundary correspondence equates the bulk Hopf invariant with the faceted Chern number. 
We will separately address whether the bulk or boundary invariant has a natural extension to $({>}2)$-band Hamiltonians.

\subsubsection{Delicacy of the bulk Hopf invariant}\la{sec:delicatehopf}

As originally formulated~\cite{Hopf:1931}, the Hopf invariant classifies maps from a three-sphere ($S^{\! 3}$) to a two-sphere ($S^{\! 2}$). 
In our context, the two-sphere is the Bloch sphere -- the space of two-band Hamiltonians. 
Our assumption that the Chern class is trivial allows to continuously deform the map such that its restriction to a 2D subtorus (of the Brillouin torus) is a constant map. 
In particular, it is possible to have the entire boundary of the (first) BZ map to the same point on the Bloch sphere. 
This allows one to compactify the boundary to a single point, $\partial\textrm{BZ}\sim\textrm{pt.}$, which deforms the BZ into a three-sphere, $\textrm{BZ}/{\sim}\simeq S^{\! 3}$ (where `$/{\sim}$' indicates a quotient modulo the equivalence relation~\cite{Quotient-topology}).
It follows that maps from the Brillouin three-torus to the Bloch sphere, under the constraint of a trivial first Chern class, are topologically equivalent to maps from a three-sphere to a two-sphere.
For this reason, we (following other authors~\cite{Hopfinsulator_Moore}) have taken the liberty of using the term `Hopf invariant' to classify Chern-trivial maps from the Brillouin three-torus to the Bloch two-sphere. 
By this definition, any enlargement of the two-band Hamiltonian renders the Hopf invariant to be ill-defined.

It was shown in \ocite{Hopfinsulator_Moore} that for two-band Hamiltonians, the Hopf invariant has an equivalent expression as an integral of the Chern-Simons three-form of the Abelian Berry gauge field [cf.~Eq.~\eqref{eq:hopfinvar}].
Unlike the Hopf invariant, the Chern-Simons formula remains well-defined and gauge-invariant for higher-than-two-band Hamiltonians with trivial Chern class, assuming we limit ourselves to expanding only the conduction subspace. 
(The Chern-Simons formula is not directly applicable to the case where the valence subspace is expanded due to the Abelian nature of the Berry gauge field in the formula.) 

One may then ask if the Chern-Simons formula remains quantized to integers for (${N_c}{\,{+}\,}1$)-band Hamiltonians with a $({N_c}{\,{>}\,}1)$-band conduction subspace and trivial first Chern class. 
The negative answer can be rationalized from algebraic-topological considerations: the space of (${N_c}+1$)-band Hamiltonians with a $({N_c}\,{>}\,1)$-band conduction subspace is the complex Grassmannian $Gr_\mathbb{C}({N_c},1)=U({N_c}+1)/[U({N_c})\times U(1)]$~\cite{kitaev_periodictable}, where $U({N_c})$
is the group of unitary matrices of dimension 
${N_c}$, and it is known~\cite{kennedy_thesis} that the first Chern class gives the complete classification of maps from the BZ three-torus to 
$Gr_\mathbb{C}({N_c},1)$ for ${N_c}>1$.\footnote{A different extension was proposed in Ref.~\cite{nband-hopf} for the \emph{$N$-band Hopf insulators}, where $N$ non-degenerate bands are separated by $N-1$ energy gaps; then a $\mathbb{Z}$-valued topological invariant corresponding to the third homotopy group of the \emph{complex flag manifold}~\cite{Tiwari:2020,Bouhon:2020} ${F\!\ell}_\mathbb{C}=U(N)/U(1)^{\times N}$ is computed as the quantized sum of the non-quantized integrals of the Chern-Simons $3$-form over the 
$N$ individual bands. 
This quantity remains invariant as long as no band degeneracy is formed between any pair of bands. 
In our present work, we instead focus on the `conventional' scenario where topological invariants change only if one closes the energy gap separating the valence bands from the conduction bands.
}

\subsubsection{Delicacy of the faceted Chern number} \label{sec:chern-delicacy}

We just showed that the Hopf invariant becomes ill-defined after enlargement of the two-band Hamiltonian. 
Can we say something about the faceted Chern number {when the rank of either the valence or the conduction subspace is larger than one}? 
As one of us showed in Ref.~\cite{penghaoAA_quantizedmagnetism} [cf.~Sec.~II.A.1 and Fig.~3(b) therein], in the absence of any crystalline symmetry except the lattice translation, the presence of multiple bands in a given subspace allows to modify the value of the faceted Chern number by an arbitrary integer number, depending on the number of considered faceted bands $\mathcal{N}_f$. 
This is due to the fact that the bulk bands can \emph{individually} carry non-zero Chern numbers even when the total Chern number of the corresponding subspace is zero. 
Now, if we enlarge the set of surface-like polarization bands by including one polarization band with a bulk-like non-zero Chern number, the value of the faceted Chern number will be modified by this bulk Chern number.

While both the Hopf invariant and the faceted Chern number lose their invariance for $({>}2)$-band Hamiltonians, there is a special case of $\mathrm{P}n$-symmetric Hamiltonians for which ``$\mathscr{C}_f$ mod-$n$'' remains a topological invariant. 
For this to be true, the valence band representation (and also the conduction band representation) must be symmetry-equivalent to a composite band representation that is decomposable into $N_v$ ($N_c$) identical band representations $\mathrm{BR}(\mathcal{L}_v, \br_{v,\perp})$ [$\mathrm{BR}(\mathcal{L}_c, \br_{c,\perp})$]. 
In this case the rotation symmetry guarantees that the Chern number of each individual bulk band is a multiple of $n$ and therefore it can modify the faceted Chern number by multiples of $n$ which keeps the value of $\mathscr{C}_f\mod n$. In Sec.~\ref{sec:rtpbeyond2} we define this invariant in terms of the angular-momentum anomaly.

\subsection{RTP insulators beyond two-band Hamiltonians}\la{sec:rtpbeyond2}

For two-band Hamiltonians, the bulk-boundary correspondence of RTP insulators equates the bulk RTP invariant with an anomaly associated to the angular momentum of faceted states. We will separately address whether the bulk invariant or the boundary anomaly has a natural extension to $({>}2)$-band Hamiltonians.

\subsubsection{Symmetry-protected delicacy of the bulk RTP invariant}\la{sec:delicatertp}

The RTP invariant $\Delta\mathscr{P}_{\Pi\Pi'}$ has been defined as an integer-valued polarization difference of the valence band between two rotation-invariant reduced momenta $\Pi$ and $\Pi'$. 
A natural way to extend the RTP invariant to higher-than-two-band Hamiltonians is to ask whether the valence-subspace polarization difference is fixed by symmetry to integer values if the valence/conduction subspace contains more than one band. 
We find that the polarization difference remains quantized  if the enlarged Hamiltonian obeys two conditions which are generalizations of \q{mutualdisj} and \q{mutualdisj2}. \\

\noindent (\emph{i}) The enlarged Hamiltonian satisfies the \emph{mutually-disjoint condition} at $\Pi$ and $\Pi'$: collecting the itinerant angular momenta of bulk-valence and bulk-conduction states (at $\Pi$) into the respective sets $\{\widetilde{\mathcal{L}}_{v,j}(\Pi)\}_{j=1}^{N_v}$ and $\{\widetilde{\mathcal{L}}_{c,j}(\Pi)\}_{j=1}^{N_c}$, we require that the two sets are disjoint:
\begin{equation}
\label{eqn:mutually-disjoint-multiband}
\begin{split}
\textrm{\emph{mutually-disjoint condition at $\Pi$:}} \\
\left\{\widetilde{\mathcal{L}}_{v,j}(\Pi)\right\}_{j=1}^{N_v}\cap\left\{\widetilde{\mathcal{L}}_{c,j}(\Pi)\right\}_{j=1}^{N_c}=\varnothing.
\end{split}
\end{equation}
This condition is required to hold for both $\Pi$ and $\Pi'$. \\

\noindent (\emph{ii}) There exists a composite basis band representation $\mathrm{CBR}$, such that the restriction of the valence band representation $\mathrm{VBR}$  to $\Pi$ equals the restriction of $\mathrm{CBR}$ to $\Pi$, and likewise $\mathrm{VBR}|_{\Pi'}=\mathrm{CBR}|_{\Pi'}$.\footnote{ A composite basis band representation is defined as 
a Whitney sum of elementary basis band representations: $\mathrm{CBR}=\mathrm{EBR}_1 \oplus \mathrm{EBR}_2\oplus\ldots \oplus \mathrm{EBR}_{N_v}$. 
Each of $\mathrm{EBR}_j$ is induced from a basis orbital with a certain site-stabilizer symmetry representation and a certain positional center $(\br_{j,\perp},z_j)$. \label{foot:CBR}} \\ 
    
In general, along the rotation-invariant $\bk$-line specified by $\Pi$, the tight-binding Hamiltonian $h(\Pi,k_z)$ is block-diagonal in the basis where the itinerant rotation matrix is diagonal; what the mutually-disjoint condition (\emph{i}) gives additionally is that each block (corresponding to one angular-momentum sector) belongs either to the conduction or to the valence subspace. 
This implies that the Hilbert space spanned by the valence eigenstates of $h(\Pi,k_z)$ is independent of $k_z$. (\emph{i}--\emph{ii}) jointly imply that the multi-band polarization of the valence subspace of $h(\Pi,k_z)$ is equal, modulo integers, to the multi-band polarization of the just-mentioned composite band representation.\footnote{The multi-band polarization can be computed via the non-Abelian Wilson loop as $\mathscr{P}(\Pi)=_1\frac{1}{2\pi}\Im\log\det\mathcal{W}$ (see, e.g., \ocite{AA_wilsonloopinversion}). 
The multi-band polarization of the composite band representation simply equals $\sum_{j=1}^{{N_v}}z_j$, with $z_j$ defined in the previous footnote \ref{foot:CBR}. 
To prove the mod-one equality of $\mathscr{P}(\Pi)$ and $\sum_{j=1}^{{N_v}}z_j$, one can apply an equivalent expression of the Wilson loop in terms of the $z_j$ and the projector $p(\Pi,k_z)$ to the valence subspace of $h(\Pi,k_z)$, as shown in Appendix~C of \ocite{AA_wilsonloopinversion}; the mutually-disjoint condition gives that $p(\Pi,k_z)=\diag(\mathbbold{1},\mathbbold{0})$ is a $k_z$-independent diagonal matrix, with $\mathbbold{1}$ a ${N_v}\times {N_v}$ identity matrix and  $\mathbbold{0}$ a  ${N_c}\times {N_c}$ zero matrix.} 
Since this statement holds also for $\Pi \leftrightarrow \Pi'$ with the \textit{same} composite basis band representation, we deduce that the valence polarization of $h(\Pi,k_z)$ is equal, modulo integers, to the valence polarization of $h(\Pi',k_z)$. 
Thus one concludes the polarization difference between $\Pi$ and $\Pi'$ is quantized to integer values. This holds true whether or not the composite basis band representation is identical to the valence band at  rotation-invariant $\bk$-lines other than $\Pi$ or $\Pi'$.

Working within the symmetry constraints of (\emph{i}--\emph{ii}), one can generalize the RTP invariant to Hamiltonians of arbitrary number of bands.
The stability of the RTP invariant under symmetry-constrained Hilbert-space enlargements is the defining property of \textit{symmetry-protected delicate topology}.

To illustrate how the RTP invariant can be nullified, one can add to the Hilbert space any basis band representation $\mathrm{BBR}[\varphi_3]$ of $\mathrm{P}n$ with a $C_m$-invariant Wyckoff position $\br_3$ that is \textit{not} $C_n$-invariant, in other words $1<m<n$, and with $\varphi_3$ forming a one-dimensional representation of $\mathrm{P}n_{\br_3}$. 
Such a band representation necessarily has rank $r$ greater than one, meaning that the Hilbert space is spanned by $r$ independent Bloch states at each $\bk$.
However, BBR$[\varphi_3]$ cannot be split/decomposed into multiple band representation of smaller rank.  For such band representations, one may verify (e.g., from the Bilbao crystallographic server \cite{elcoro_EBRinBilbao}) that for every pair of a rank-1 elementary band representation BBR$[\varphi_1]$ and a rank-$(r\!>\!\!1)$ elementary band representation BBR$[\varphi_3]$
\begin{align}
    \exists! \,\Pi, \textrm{ s.t.} \as \forall j=1\dots r,\as
    \widetilde{\mathcal{L}}_1(\Pi)\neq\widetilde{\mathcal{L}}_3^j(\Pi),
\end{align}
where $\exists!$ means `exists one and only one', $\widetilde{\mathcal{L}}_1(\Pi)$ is itinerant angular momentum of BBR$[\varphi_1]$ and $\widetilde{\mathcal{L}}_3^j(\Pi)$, $j=1\dots r$ are itinerant angular momenta of BBR$[\varphi_3]$.
This implies that the addition of   BBR$[\varphi_3]$, whether it is to the conduction or valence subspace, makes the mutually-disjoint condition unfulfillable at (at least) two rotation-invariant reduced momenta, and hence, nullifies all RTP invariants.

\subsubsection{Angular-momentum anomaly of \texorpdfstring{$({>}2)$}{>2}-band Hamiltonians}\la{sec:anganomalymore2}

For $({>}2)$-band Hamiltonians with an RTP invariant $\Delta\mathscr{P}_{\Pi_1\Pi_2}$ that is well-defined and greater than zero, the angular-momentum anomaly [encapsulated by \q{amanomaly} and derived in Appendix~\ref{app:ang-mom-anomaly}] holds just as well as for two-band Hamiltonians, implying the existence of at least $\Delta\mathscr{P}_{\Pi_1\Pi_2}$-number of anomalous faceted bands whose itinerant angular momenta are bulk-conduction-like at $\Pi_1$ and bulk-valence-like at $\Pi_2$. 

We remind the reader that for \textit{two}-band Hamiltonians, an itinerant angular momentum is bulk-valence-like at $\Pi$ if it is identical to the itinerant angular momentum of the bulk-valence band at $\Pi$. 
By assumption of a two-band Hamiltonian, the bulk-valence subspace comprises a single band, hence the itinerant angular momentum at $\Pi$ is uniquely defined. 

In contrast, for an $N_v$-band bulk-valence subspace, there is not a unique angular momentum, but a set $\{\widetilde{\mathcal{L}}^j_{v} (\Pi)\}_{j=1,\ldots, N_v}$ of possibly non-identical angular momenta. 
An itinerant angular momentum that is \textit{bulk-valence-like} is then equal to any of the $N_v$ values in this set (which potentially contains repetitions). 
Likewise, one may generalize the notion of \textit{bulk-conduction-like} to an $N_c$-band bulk-conduction subspace. 
It is with these generalized notions that the angular-momentum anomaly [encapsulated by \q{amanomaly}] holds for $(N_c+N_v)$-band Hamiltonians. 

We see that specifying $\Delta\mathscr{P}_{\Pi_1\Pi_2}$ does not uniquely specify the subspace of faceted bands, owing to the non-uniqueness in itinerant-angular-momentum values implicit in (the generalized notions of) `bulk-conduction-like' and `bulk-valence-like'. 
In other words, there are multiple, topologically-distinct faceted subspaces consistent with a single value of  $\Delta\mathscr{P}_{\Pi_1\Pi_2}$. 
This suggests there exists a finer topological classification [of P$(n{>}2)$-symmetric $(N_c+N_v)$-band Hamiltonians] than is afforded by the RTP invariant $\Delta\mathscr{P}_{\Pi_1\Pi_2}$ alone.
We believe this to be a promising research direction.

Why did we restrict the above conjecture to $\mathrm{P}n$ with $n\,{>}\,2$? 
The $\mathrm{P}2$ space group is special in that there are only two irreducible representations of the $C_2$ point group, thus being bulk-valence-like (or bulk-conduction-like) uniquely specifies an itinerant angular momentum.\footnote{This is because the mutually-disjoint condition at $\Pi$ holds only if (\emph{i}) all members of the bulk-valence-like set  $\{\widetilde{\mathcal{L}}^j_{v} (\Pi)\}_{j=1,\ldots, N_v}$ of itinerant angular momenta are equal, (\emph{ii}) all members of the bulk-conduction-like set $\{\widetilde{\mathcal{L}}^j_{c} (\Pi)\}_{j=1,\ldots, N_c}$ of itinerant angular momenta are also equal, and (\emph{iii}) $\widetilde{\mathcal{L}}^1_{v}(\Pi)\neq \widetilde{\mathcal{L}}^1_{c} (\Pi)$.} 
It follows that the full set of RTP invariants uniquely determines the itinerant angular momenta of the faceted subspace.

In the general setting of ($N_v+N_c$)-band P$n$-symmetric Hamiltonians the angular-momentum anomaly manifests in conditionally-robust surface states that interpolate between valence and conduction band in the same way as for the 2-band Hamiltonians. 
Our proof of the conditional-robustness of the surface states, presented in Appendix~\ref{app:Toeplitz-bounds}, uses certain spectral properties of block-Toeplitz matrices and generalizes the $2$-band proof presented in \s{sec:BBC_sharp}. 

For another physical manifestation of the angular-momentum anomaly -- the Zak-phase anomaly -- to hold, the sufficient (but not necessary) condition is to have identical itinerant angular momenta of all valence bands at $\Pi_1$, and also at $\Pi_2$; the same must be true for the conduction itinerant angular momenta. 
In this case, the Zak-phase anomaly tells that $\mathcal{N}$ faceted bands have a set of $\mathcal{N}$ Zak phases, with $|\Delta\mathscr{P}_{\Pi_1\Pi_2}|$ of them being symmetry-fixed to a rational multiple of $2\pi/n$, and that these symmetry-fixed Zak phases are distinct from the bulk-valence and also from the bulk-conduction Zak phase. 
This can be shown by applying the theorem from Sec.~IV\kern 0.1em D\kern 0.1em 1 in Ref.~\cite{TBO_JHAA} in the same way as it was done for the two-band case in Sec.~\ref{sec:zakanomalymain} and Appendix~\ref{app:Zak-phase-anomaly}. 

Let us finally consider the faceted Chern number of the $({>}2)$-band RTP insulators. 
As we already mentioned in Sec.~\ref{sec:chern-delicacy}, in the special case of identical itinerant angular momenta of the valence bands $\widetilde{\mathcal{L}}_{v,j}(\Pi)=\widetilde{\mathcal{L}}_v(\Pi)$ for all $j=1\dots N_v$ at each $\Pi$ [and same for the conduction bands $\widetilde{\mathcal{L}}_{c,j}(\Pi)=\widetilde{\mathcal{L}}_c(\Pi)$ for all $j=1\dots N_c$], the angular-momentum anomaly leads to a faceted Chern number defined modulo $n$. 
Given the itinerant angular momenta of the $\mathcal{N}_f$ faceted bands, the mod-$n$-Chern number is given by
\begin{align}
    \mathscr{C}_f&=_n-\sum_{i=1}^{\mathcal{N}_f}\sum_\Pi\widetilde{\mathcal{L}}_i(\Pi)\lin 
    &=_n \sum_{\alpha=1}^2 \sum_{j_\alpha=2}^{J_\alpha} \Delta\widetilde{\mathcal{L}}(\Pi_{j_\alpha}^\alpha)\;\Delta\mathscr{P}_{\Pi^\alpha_1\Pi^\alpha_{j_\alpha}},
    \label{eq:C-facet-via-ell-multi-band}
\end{align}
where $\{\Pi^\alpha_{j_\alpha}\}_{j_\alpha=1\dots J_\alpha}$, $\alpha=1,2$, are two maximal subsets of rotation-invariant reduced momenta satisfying conditions \eqref{mutualdisj} and \eqref{mutualdisj2}. By `maximal' we mean that each set contains all reduced momenta, satisfying these conditions, and the size of each subset is given by an integer $J_\alpha$.
This relation is proven in exactly the same way as the Hopf-RTP relation in Sec.~\ref{sec:proof-RTP-Hopf}.
This similarity with the Hopf-RTP relation [cf.\ \q{eq:hopf-RTP}] of two-band Hamiltonians is not accidental. 
For two-band Hamiltonians with trivial Chern class, the bulk-boundary correspondence of Hopf insulators equates the Hopf invariant with the faceted Chern number: $\chi=\mathscr{C}_f$. 
For $({>}2)$-band Hamiltonians with identical valence (conduction) itinerant angular momenta, the Hopf invariant loses its invariant meaning, yet the faceted Chern number remains well-defined modulo $n$ due to the anomalous values of itinerant angular momenta.

\section{Multicellularity of delicate topological insulators}\label{sec:CHI-multicel}

The Hopf and the RTP invariants topologically obstructs the valence subspace (and also the conduction subspace) from being continuously deformable to a tightly-bound band representation~\cite{Nelson:2021}. 
It is worth emphasizing that the RTP invariant \textit{alone} would readily impose such an obstruction, assuming the continuous deformation preserves $\textrm{P}n$ symmetry throughout. 
For instance, this would hold for RTP insulators with trivial Hopf invariant, as was modelled in \s{sec:zero-hopf}. 

The obstruction against being tightly-bound holds despite both the conduction as well as the valence bands being band representations.
A different way to state this is that despite the Wannier orbitals being exponentially-localized and locally-symmetric, some of the Wannier orbitals necessarily have support across multiple primitive unit cells. Such obstructed band representations are described as  \textit{multicellular}.

Our exposition of multicellularity is organized in the present section as follows.
We first demonstrate in Sec.~\ref{sec:multicell} that multicellularity is generally a delicate property, meaning that the topological obstruction against being tightly bound can be nullified by adding a certain trivial band to either the conduction or the valence subspace. 
Then in Sec.~\ref{sec:graded-multicell} we introduce a grading of multicellularity, defined roughly to be the smallest number of primitive unit cells that a representative Wannier orbital can be confined to, if allowing for any continuous deformation of the two-band Hamiltonian that preserves the energy gap and symmetry. 
Such grading can meaningfully distinguish between different crystalline Hopf insulators, as we exemplify in Sec.~\ref{sec:grading-examples}. 
Finally, we extend the notion of multicellularity beyond two-band Hamiltonians in Sec.~\ref{sec:grading-multiband}.

\subsection{Multicellularity is a delicate property}\label{sec:multicell}

From the discussion in Sec.~\ref{sec:multi-band}, we see that the specific forms of multicellularity -- associated with the Hopf and RTP invariants -- are delicate-topological. 
In fact, a stronger statement holds true, namely that \emph{any} form of multicellularity existing in principle is a delicate property. 
In other words, for any fixed-band, insulating Hamiltonian (inclusive of, but going beyond, crystalline Hopf insulators) whose valence subspace is a band representation with a topological obstruction to being tightly-bound, this obstruction is removable by adding a certain tightly-bound band representation to either the conduction or the valence subspace.

This follows from a few facts in space-group-equivariant vector bundle theory that we will just state, referring the reader to  \ocite{crystalsplit_AAJHWCLL} for the detailed argument. 
For a rapid initiation to the terminology:\footnote{See Ref.~\cite{Fruchart:2013} for an introductory discussion of band topology from the perspective of vector bundles.} an energy band can be viewed as a vector bundle~\cite{Hatcher:2003,FreedMoore_twistedequivariantmatter,crystalsplit_AAJHWCLL,Bouhon:2020} with the base space equal to the Brillouin zone and with fibres in direct correspondence with the $\bk$-dependent Bloch functions; while  `space-group-equivariance' means, in our context, that applying the rotation operator to a Bloch function at $\bk$ produces a Bloch function at $C_n\bk$ within the same band.

The first fact we need, proven in Appendix~G of~\ocite{crystalsplit_AAJHWCLL}, states that a tightly-bound band representation and a non-tightly-bound representation (with the same space group $\mathcal{G}$, Wyckoff position, and site-stabilizer representation) are isomorphic as $\mathcal{G}$-equivariant vector bundles.\footnote{The standard notion of isomorphism in vector bundles is explained in~\ocite{Hatcher:2003}; the notion of a space-group-equivariant isomorphism is explained in~\ocite{crystalsplit_AAJHWCLL}.} 
The second fact follows from the universal space-group-equivariant vector bundle theorem,\footnote{The universal bundle theorem in~\ocite{Atiyah1989} applies to $\calp$-equivariant vector bundles with finite group $\calp$. For  $\mathcal{G}$-equivariant vector bundles with $\mathcal{G}$ a space group, we observe that the translation subgroup of $\mathcal{G}$ acts trivially on the bundle, hence we  may directly apply the universal bundle theorem with $\calp$ identified as the (finite) point group of $\mathcal{G}$.} which states that a $\mathcal{G}$-equivariant isomorphism between two vector bundles ($\mathfrak{E}$ and $\mathfrak{E}'$) corresponds bijectively to a $\mathcal{G}$-equivariant homotopy between $\mathfrak{E}$ and $\mathfrak{E}'$. 
Applied to our context, this means there exists a continuous, $\mathcal{G}$-symmetric deformation between a non-tightly-bound band representation and a tightly-bound representation (with the same data specified above). 
If the band representation in question is the valence band of a tight-binding Hamiltonian, the existence of this $\mathcal{G}$-equivariant isomorphism (Fact 1) is guaranteed only if the rank of the conduction subspace is unconstrained. 
Equivalently stated (due to Fact 2), if one can arbitrarily expand the conduction subspace by adding tightly-bound band representations, a valence band representation can always be continuously deformed to a tightly-bound band representation while preserving space-group symmetry.
This is equivalent to stating that multicellularity of the valence Wannier orbitals is a delicate property, exactly as advertised in the previous paragraph.

\subsection{Graded multicellularity: definition and motivation}\label{sec:graded-multicell}

The RTP invariant suggests a refinement of multicellularity through the introduction of \emph{grading}.
We will first define graded multicellularity in the simplified case of an insulator whose valence subspace is composed of a single band with trivial first Chern class (i.e., a single, analytic Bloch function over the Brillouin torus): the \emph{grade} $\mathfrak{g}\in\mathbb{Z}^+\cup \{\infty\}$ is the lower bound on the number of primitive unit cells that make up a non-primitive cell,\footnote{A non-primitive cell is a path-connected volume of space that, when translated through all vectors of a Bravais lattice, does \textit{not} cover all of space without either overlapping itself or leaving voids~\cite{ashcroft_mermin}.} which is able to contain a representative Wannier orbital; the lower bound is defined given that one is allowed to continuously deform a tight-binding Hamiltonian $h(\bk)$ of fixed matrix dimension while preserving the energy gap and crystalline symmetry.\footnote{We allow to continuously and symmetrically deform the Wyckoff positions of the tight-binding-basis orbitals, as detailed in the Supplemental Material of \ocite{Nelson:2021}.} 
(For crystalline topological insulators in Wigner-Dyson classes A, AI and AII, the notion of grading is typically more useful when one imposes a local-symmetry condition~\cite{nogo_AAJH,crystalsplit_AAJHWCLL} on the Wannier orbitals.)
In particular, an insulator with $\mathfrak{g}=1$ is referred to as \emph{unicellular}, while all $\mathfrak{g}\geq 2$ (including $\mathfrak{g}=\infty$, in principle) correspond to different grades of multicellularity.
\smallskip

The notion of grading depends on whether the Wannier orbitals form an orthonormal basis for the valence band. 
To formalize this discussion, it is useful 
to distinguish between three kinds of Wannier orbitals: (1) the conventional \emph{orthonormal (`$\mathpzc{on}$-') Wannier orbitals}, (2) the generalized \emph{orthogonal (`$\mathpzc{og}$-') Wannier 
orbitals}, and (3) the further relaxed \emph{Wannier-$\mathpzc{type}$ orbitals}. 

As is well-known in the conventional case, assuming that the Chern class is trivial, $\mathpzc{on}$-Wannier orbitals are obtained by inverse Fourier transform  of Bloch states [cf.~\q{blochwannierfourier}], with the (\emph{i}) Bloch states  depending smoothly on $\bk$, and with (\emph{ii}) the intra-cell components  of the Bloch states being orthonormal at each~$\bk$:
\begin{subequations}
\e{\braket{u_{j}^{\mathpzc{on}}(\bk)}{u_{j'}^{\mathpzc{on}}(\bk)}=\delta_{jj'}, \as j,j'=1,\ldots, N_v{+}N_c.}
As respective consequences of (\emph{i}--\emph{ii}), the $\mathpzc{on}$-Wannier orbitals are (\emph{i}) exponentially localized, and (\emph{ii}) obey the orthonormality relation  \begin{equation}
\Big\langle\!\braket{W_{j,\bR}^{\mathpzc{on}}}{W_{j',\bR'}^{\mathpzc{on}}}\!\Big\rangle=\delta_{jj'}\delta_{\bR,\bR'},\la{orthonormalwannier}
\end{equation}
\end{subequations}
where $\bR,\bR'$ are Bravais-lattice vectors, and our bra-ket notation for Bloch and Wannier states is explained in Appendix~\ref{app:tb-formalism}. 
In contrast, one may drop the normalization of the intra-cell components of the Bloch states
but retain the weaker condition of orthogonality and completeness at each $\bk$, resulting in $\left\{\big|{u_{j}^\mathpzc{og}(\bk)}\big\rangle\right\}_{j=1}^{N_v{+}N_c}$.
Then the inverse Fourier transform generates  $\mathpzc{og}$-Wannier orbitals, 
which obey
\begin{equation}
\Big<\!\braket{W_{j,\bR}^{\mathpzc{og}}}{W_{j',\bR'}^{\mathpzc{og}}}\!\Big>=\mathscr{N}_{j,\bR}^2\delta_{jj'}\delta_{\bR,\bR'},
\label{eqn:og-wannier}
\end{equation}
with $\mathscr{N}_{j,\bR}\in \mathbb{R}{\setminus}\{0\}$ the respective norms.
Similar to the $\mathpzc{on}$- case, the number of $\mathpzc{og}$-Wannier orbitals matches the dimensions of the valence resp.~conduction subspace.
Finally, the Wannier-$\mathpzc{type}$ orbitals $\left\{\left|W^\mathpzc{type}_{j,\bR}\right>\!\Big>\right\}_{j=1}^{N}$ considered by Ref.~\cite{read_compactwannier} constitute a potentially \emph{overdetermined} basis with $N \geq N_v + N_c$, i.e., either the valence subspace or the conduction subspace (or both) are potentially represented by more Bloch states (and consequently, after applying the inverse Fourier transform, by more Wannier orbitals) than is the actual dimension of the respective band subspace. 
In this case, the only retained property is the \emph{completeness} of the intra-cell components $\left\{\big|{u_{j}^\mathpzc{type}(\bk)}\big\rangle\right\}_{j=1}^{N}$ of the Bloch states at each~$\bk$.

Either kind of the the Wannier representations is often required to be \emph{exponentially-localized}, meaning that the amplitudes of the Wannier orbital can be bounded by $\mathrm{e}^{-\norm{R}/r_w}$ with some width $r_w>0$. A more stringent requirement is for the Wannier representation to be \emph{compactly supported}, meaning that each Wannier orbital $\left|W_{j,\bR}\right>\!\Big>$ has non-vanishing amplitude on a finite number of sites that appear within a bounded region in space. Clearly, the grade $\mathfrak{g}$ could only be finite for the compactly supported Wannier orbitals.
\smallskip

One may ask if the graded multicellularity is more fruitfully defined in terms of the $\mathpzc{on}$-Wannier orbitals, the $\mathpzc{og}$-Wannier orbitals, or the Wannier-$\mathpzc{type}$ orbitals. 
To place this question in a wider context, let us first review several basic relations between the existence of the various kinds of Wannier orbitals. In subsequent paragraphs, we review non-trivial results of three works~\cite{Dubail:2015,read_compactwannier,Schindler:2021} that have explored related themes.

First, note that having an exponentially-localized $\mathpzc{on}$-Wannier representation is a more stringent condition than having an exponentially-localized $\mathpzc{og}$-Wannier representation, and having either of them is more stringent that having an exponentially-localized Wannier-$\mathpzc{type}$ representation.
The analogous relations between the existence of $\mathpzc{on}$-Wannier/$\mathpzc{og}$-Wannier/Wannier-$\mathpzc{type}$ representations also hold if we replace in the previous sentence `exponentially-localized' by `compactly supported'. 
More interestingly, if the valence subspace of a $\bk$-periodic tight-binding Hamiltonian is diagonalized by (unnormalized) intra-cell functions $\left\{\big|u_j^{\mathpzc{og}}(\bk)\big\rangle\right\}_{j=1,\ldots,N_v}$ that are vector-valued polynomial functions of $\{e^{ik_\xi}\}_{\xi=x,y,z}$ with no zeros, the inverse Fourier transform of the intra-cell functions constitute representative compactly-supported $\mathpzc{og}$-Wannier orbitals.
Since the normalization $\norm{\big|{u_{j}^\mathpzc{og}(\bk)}\big\rangle}$ of a vector-valued polynomial is also guaranteed to be analytic with no zeros, the inverse Fourier transform of $\big|{u_j^\mathpzc{og}(\bk)}\big\rangle/\norm{\big|{u_{j}^\mathpzc{og}(\bk)}\big\rangle} = \big|{u_j^\mathpzc{on}(\bk)}\big\rangle$ is an exponentially-localized $\mathpzc{on}$-Wannier orbital.
Conversely, the non-existence of an exponentially-localized $\mathpzc{on}$-Wannier representation implies the non-existence of a compactly supported $\mathpzc{og}$-Wannier representation.

Let us next summarize the non-trivial results found by Refs.~\cite{Dubail:2015,read_compactwannier,Schindler:2021}. 
Assuming that the valence bands support an exponentially-localized Wannier-$\mathpzc{type}$ (resp.~exponentially-localized $\mathpzc{on}$-Wannier) representation, Ref.~\cite{read_compactwannier} (Ref.~\cite{Schindler:2021}) analyzed within the context of topological band theory whether the representation is \emph{compactly supported}. 
In both cases, a non-trivial value of a topological invariant was found to obstruct the existence of a compactly supported Wannier representation (of the respective kind). 

Namely, N.~Read (generalizing an earlier work with J.~Dubail~\cite{Dubail:2015}) showed in Ref.~\cite{read_compactwannier} that stable topological bands (i.e., bands belonging to a non-trivial stable equivalence class in algebraic $K$-theory) 
in all Altland-Zirnbauer symmetry classes and in real-spatial dimensions $d\geq 2$~\cite{kitaev_periodictable,Ryu:2010}  do not admit a (symmetric) compactly-supported Wannier-$\mathpzc{type}$ representation, 
unless the nontriviality of the topological band is due solely to weak topological invariants existing already for $d=1$.
This obstruction is equivalent to saying that not all representative Wannier-$\mathpzc{type}$ orbitals 
(respecting an Altland-Zirnbauer symmetry) can have finite extent, i.e., the `grading' is infinite, if the notion of grade were generalized to apply to any band subspace (not just band representations), and if the lower bound on the number of primitive cells is defined in the stable-equivalence class (as opposed to an equivalence class of $h(\bk)$ with fixed matrix dimension).

More recently, Schindler and Bernevig in Ref.~\cite{Schindler:2021} showed that certain obstructed atomic insulators~\cite{TQC} (whose valence subspaces are band representations but whose conduction subspaces are fragile) do not admit a compactly-supported $\mathpzc{on}$-Wannier representation. 
This statement holds (in this case for the conventional $\mathpzc{on}$-Wannier
orbitals) 
even if one allows for the matrix dimension of $h(\bk)$ to increase, so long as the conduction subspace remains fragile.

We thus find that in all the previous studies, the categorization into compact vs.~non-compact has been used as a real-space-centric diagnostic for trivial vs.~nontrivial band structures in the stable or fragile topological classification. (Here we view the Schindler-Bernevig model as a generalized example of fragile topology.)
\smallskip

In our work, we opt for the `middle ground' with respect to the three kinds of Wannier orbitals, by adopting the $\mathpzc{og}$- case. 
Namely, what we propose
is that: \medskip \\
\noindent (\emph{a}) delicate topological insulators can have compactly-supported $\mathpzc{og}$-Wannier functions, and that  \medskip \\ 
\noindent (\emph{b}) the grade $\mathfrak{g}$, as defined for $\mathpzc{og}$-Wannier orbitals,
offers a finer-grained notion of compactness, with different finite values of $\mathfrak{g}$ distinguishing different categories of delicate topological insulators. \medskip \\
Presently, we do not know if the analog of $\mathfrak{g}$, defined for conventional $\mathpzc{on}$-Wannier orbitals, can take finite values for delicate topological insulators. For the models presented here, the $\mathpzc{on}$-Wannier orbitals are infinite in extent.

\subsection{Graded multicellularity: case studies}\label{sec:grading-examples}

To exemplify Proposals (\emph{a}--\emph{b}) formulated in the last paragraph of Sec.~\ref{sec:graded-multicell}, we proceed to briefly investigate bounds on the grade $\mathfrak{g}_\textrm{Hopf}$ of the Hopf insulator if we allow for continuous Hamiltonian deformation that may break any point group symmetry, and on the grade $\mathfrak{g}_\textrm{RTP}$ of an RTP insulator in the presence of rotation symmetry. 

We begin with the case of the Hopf insulator. On one hand, based on the homotopic inequivalence of the Hopf-insulating Hamiltonian map to a constant map, it was shown in Ref.~\cite{Nelson:2021} that the Hopf insulator is not unicellular, implying a lower bound on the grade: $\mathfrak{g}_\textrm{Hopf}\geq 2$.
On the other hand, one may consider the concrete MRW model of the Hopf insulator~\cite{Hopfinsulator_Moore} [corresponding to the entry ``$\textrm{P}4$, $\Delta\mathcal{L}=1$, $\Delta\br_\perp=(0,0)$'' in Table~\ref{tab:model-zoo}] with non-normalized valence intra-cell function:
\begin{equation}
\ket{u_v^{\mathpzc{og}}(\bk)}=i\sigma_y\left(\begin{array}{c}
\sin k_x + i\sin k_y \\
\sin k_z - i\left(\sum\limits_{\xi=x,y,z}\cos k_\xi + \Phi\right)
\end{array}\right).\label{eqn:MRW-valence-state}
\end{equation}
Its inverse Fourier transform yields an $\mathpzc{og}$-Wannier
orbital supported over seven primitive unit cells: 
a `central' cell, together with six `adjacent' cells that are attributable to the appearance of $e^{\pm i k_{x,y,z}}$ in Eq.~(\ref{eqn:MRW-valence-state}). Combining both results,
\begin{subequations}
\begin{equation}
2 \leq \mathfrak{g}_\textrm{Hopf} \leq 7. \label{eqn:Hopf-grade}
\end{equation}

The second example that we use to illustrate these ideas is the crystalline Hopf insulator characterized by RTP invariants $\Delta \mathscr{P}_{\Pi\,\Pi'}$ and by Hopf invariant $\chi$ [recall that the invariants are constrained through Eq.~(\ref{eq:hopf-RTP})]. 
For simplicity, let us assume that all rotation-invariant reduced momenta are encompassed in a single iso-orbital subset -- which per Table~\ref{tab:itinerant_ell} implies that $\Delta\br_\perp = (0,0)$, i.e., both basis orbitals lie along the same rotation axis.
Because the polarization $\mathscr{P}_z(k_x,k_y)$ extends over at least $[1{+}\max_{\Pi,\Pi'}\{\abs{\Delta \mathscr{P}_{\Pi\,\Pi'}}\}]$  
rotation-invariant primitive unit layers, the $\mathpzc{og}$-Wannier orbital
extends over at least the same number of primitive unit cells, giving the (rather generous) lower bound
\begin{equation}
\mathfrak{g}_\textrm{RTP}\geq 1+\max_{\Pi,\Pi'}\{\abs{\Delta \mathscr{P}_{\Pi\,\Pi'}}\}\label{eqn:CHI-grade}
\end{equation}
for the grade of RTP insulators. 

In fact, a more restrictive lower bound than Eq.~(\ref{eqn:CHI-grade}) can be deduced by further considering the spread of the $\mathpzc{og}$-Wannier orbital in directions \emph{perpendicular} to the rotation axis. 
Let us allude to the fact that, owing to the non-trivial value of any RTP invariant, (1) the valence band representation ($\textrm{VBR}$) of an RTP insulator \emph{is not continuously deformable into} a basis band representation ($\textrm{BBR}$). 
However, it follows from the rotation eigenvalues of the valence band (under the stated assumption of a single all-encompassing iso-orbital subset) that (2) the $\textrm{VBR}$ \emph{is symmetry-equivalent to} a $\textrm{BBR}$ spanned by one of the basis orbitals (let's say $\textrm{BBR}[\varphi_1]$).

We now argue as follows. 
First, the fact (1) implies that $\textrm{VBR}$ must contain some contribution from orbital $\varphi_2$ in some unit cells.
However, the symmetry equivalence of $\textrm{VBR}$ to $\textrm{BBR}[\varphi_1]$, stated by fact (2), implies that the $\mathpzc{og}$-Wannier orbital transforms under $C_n$ with the angular momentum $\mathcal{L}_1$ of orbital $\varphi_1$. 
This transformation rule precludes the $\varphi_2$ contribution from orbitals along the rotation axis. 
It can then be seen that the minimal way for the $\mathpzc{og}$-Wannier orbital to acquire a contribution from $\varphi_2$ orbitals while exhibiting the angular momentum $\mathcal{L}_1$ is due to an $n$-plet of sites that transform cyclically into each other under $C_n$ in at least one of the $[1{+}\max_{\Pi,\Pi'}\{\abs{\Delta \mathscr{P}_{\Pi\,\Pi'}}\}]$ rotation-invariant primitive unit layers. Therefore, under the stated assumptions (namely: space group $\textrm{P}n$, an all-encompassing iso-orbital subset, and at least one RTP invariant with non-zero value), we find
\begin{equation}
\mathfrak{g}_\textrm{RTP} \geq 1 + n + \max_{\Pi,\Pi'}\{\abs{\Delta \mathscr{P}_{\Pi\,\Pi'}}\}.\label{eqn:stricter-RTP-grade}
\end{equation}
In particular, for the non-trivial phases of the MRW model (which fulfill the stated assumption with $n=4$ and $\max_{\Pi,\Pi'}\{\abs{\Delta \mathscr{P}_{\Pi\,\Pi'}}\}=1$), in combination with the argument presented before Eq.~(\ref{eqn:Hopf-grade}), we obtain that
\begin{equation}
6 \leq \mathfrak{g}_\textrm{MRW} \leq 7.
\end{equation}
\end{subequations}

\subsection{Grading of models with \texorpdfstring{${\,{>}{2}\,}$}{>2} bands}\label{sec:grading-multiband}

We finally generalize the definition of graded multicellularity to delicate topological insulators with an $N_v$-band valence subspace, i.e., the valence subspace is spanned by ${N_v>1}$ 
linearly independent Bloch states at each $\bk$. 
For this purpose, it is convenient to adopt in the present subsection certain bundle-theoretic language: we view the valence subspace as a rank-$N_v$ vector bundle, with an  $N_v$-{(complex)-}dimensional vector space associated to each $\bk$ of the base space (our Brillouin torus).
In constructing the $\mathpzc{og}$-Wannier
orbitals,
we must first decompose the rank-$N_v$ bundle into $N_v$ unit-rank bundles (i.e., line bundles); this decomposition amounts to choosing $N_v$ orthogonal (but not necessarily normalized)
intra-cell functions $\left\{\big|{u_j^{\mathpzc{og}}(\bk)}\big\rangle\right\}_{j=1,\ldots,N_v}$ which span the $N_v$-dimensional vector space at each $\bk$. 

The inverse Fourier transform of  $\big|{u_j^{\mathpzc{og}}(\bk)}\big\rangle$ defines representative (i.e., translation-inequivalent) $\mathpzc{og}$-Wannier orbitals
$\big|{W_{j,\bR}^{\mathpzc{og}}}\big\rangle\big\rangle$, and we define the \emph{partial grade} $\mathfrak{g}_j$ of the $j$-th line bundle to be the number of primitive unit cells in which $\big|{W_{j,\bR}^{\mathpzc{og}}}\big\rangle\big\rangle$ has a finite support. 
The \textit{total grade} $\mathfrak{g}$ is then defined as the lower bound of the \emph{sum} of the partial grades of all representative $\mathpzc{og}$-Wannier orbitals,
\begin{equation}
\mathfrak{g} - 1 = \textrm{min}\sum_{j=1}^{N_v} \left(\mathfrak{g}_j - 1\right),\label{eqn:higher-rank-grade}
\end{equation}
given that one is allowed to continuously deform the tight-binding Hamiltonian $h(\bk)$ of fixed matrix dimension while preserving the energy gap and crystalline symmetry.

For a concrete example, consider the case of an RTP insulator with an $N_v$-rank valence subspace and  $\max_{\Pi,\Pi'}\{\abs{\Delta \mathscr{P}_{\Pi\,\Pi'}}\}=M$, thereby generalizing the discussion of a crystalline Hopf insulator in Sec.~\ref{sec:grading-examples}. 
We argue that the total grade $\mathfrak{g}$ obeys again the inequality~(\ref{eqn:CHI-grade}).
Indeed, when one constructs the $\mathpzc{og}$-Wannier orbitals, 
the decomposition of the rank-$N_v$ bundle into $N_v$ line bundles leaves us a choice as to how to distribute the $M$ quanta of $2\pi$ Zak phases among the $N_v$ line bundles; effectively, we must choose an integer decomposition $M = \mathrm{n}_1 + \mathrm{n}_2 + \ldots + \mathrm{n}_{N_v}$.
This limits each partial grade to $\mathfrak{g}_j \geq 1 + \abs{\mathrm{n}_j}$, and the total grade to $\mathfrak{g} \geq 1 + \sum_{j=1}^{N_v} \abs{\mathrm{n}_j} \geq 1 + M$, irrespective of the rank of the valence subspace. 
We see that no matter how one performs the distribution, the inequality (\ref{eqn:CHI-grade}) holds. 

In addition, by adopting similar assumptions and arguments that led to Eq.~(\ref{eqn:stricter-RTP-grade}), we anticipate that each bundle with $\mathrm{n}_j\neq 0$ must have support on at least $n$ additional unit cells within some rotation-invariant primitive unit layer. Optimally, one can imagine these five unit cells to be the same ones for each bundle (alternatively, one could imagine assigning all the $M$ quanta of $2\pi$ Zak phase into a single bundle), allowing us to  generalize to the $({>}2)$-band case also the \emph{stricter} bound in~Eq.~(\ref{eqn:stricter-RTP-grade}).

Determining the exact value of the grade for Hopf and RTP insulators in space groups $\textrm{P}n$, as well as of other multicellular topological insulators, constitutes an interesting question that we leave open for future studies. 
We further anticipate the grade as defined by Eq.~(\ref{eqn:higher-rank-grade}) to be invariant under the addition of unicellular band representations that do not trivialize the delicate topological invariants, but further research is needed to firmly establish such a property of graded multicellularity.

\section{Topological semimetals intermediate crystalline Hopf insulators}\label{sec:semi-metallic-transition-region}

We began this paper by extolling the virtues of the Berry dipole [Eqs.~(\ref{eq:kp-sym-spinor}--\ref{eq:zszs-hamiltonian})], namely that (\emph{a})~it affords a clean, analytic derivation of the simultaneous change in Hopf and RTP invariants across a Berry-dipole phase-transition point [Sec.~\ref{sec:berry-dipole}], and (\emph{b})~it gives a practical strategy for a theorist to cook up tight-binding-model Hamiltonians of crystalline Hopf insulators [\s{sec:models}]. 

The price to pay for these virtues is that a Berry-dipole phase transition lies in a region of the phase diagram that is `hard to reach', i.e., to interpolate from an insulating, $\mathrm{P}n$-symmetric Hamiltonian to a Berry-dipole Hamiltonian requires fine-tuning  more Hamiltonian parameters than would otherwise be needed to interpolate to a generic, non-insulating, $\mathrm{P}n$-symmetric Hamiltonian. 
How many more parameters exactly? This we answer in \s{sec:notgeneric}, wherein we also attribute the additional fine-tuning to an `emergent' mirror symmetry (of the Berry-dipole $\bk\cdot\bp$ Hamiltonian) that is absent from the tight-binding Hamiltonian. 

The implications of being `hard to reach' are explored in the subsequent \s{sec:generictransition}, namely that a $\mathrm{P}n$-symmetric interpolation\footnote{If $h(\bk;t)$ interpolates between a trivial insulator $h(\bk;0)$ to a crystalline Hopf insulator $h(\bk;1)$, we say the interpolation is $\mathrm{P}n$-symmetric if the Hamiltonian is $\mathrm{P}n$-symmetric for all $t\in [0,1]$.} from trivial to crystalline Hopf insulator is typically intermediated by an `easier-to-reach' topological-semimetallic phase that involves either Weyl-point or nodal-line degeneracies. 
The basic properties of these intermediate topological-semimetallic phases are illustrated by a model $\bk\cdot \bp$ Hamiltonian in \s{sec:casestudyTSM}, which is a first step toward more naturalistic/material-realistic models of delicate topology.

The emergent mirror symmetry offers a symmetrically-distinct route to a crystalline Hopf insulator. 
This stems from the observation that the mirror-symmetric Berry-dipole band-touching \textit{point} can be continuously expanded (in $\bk$-space) to a mirror-symmetric topological nodal-\textit{loop}. 
Thus, an alternative strategy to realize a crystalline Hopf insulator is to  begin from a mirror-symmetric topological-nodal-line semimetal and break the mirror symmetry in a manner to be ascertained. 
Unlike previously-studied nodal-loop semimetals~\cite{Burkov:2011}, the presently described semimetal realizes the first-known example of a symmetry-protected delicate topological semimetal, as we explain in \s{sec:delicateSM}.

\subsection{The Berry dipole is `hard to reach'}\la{sec:notgeneric}

Let us construct a $C_n$-symmetric $\bk\cdot\bp$ Hamiltonian exhibiting a Berry dipole. 
For concreteness, we pick the case of $n=4$, with the two basis Bloch states differing by unity for their itinerant angular momenta; this is the case for both the $\Gamma$ and the $\tm$ line of both the Hopfless and the MRW models discussed in Sec.~\ref{sec:models}. 
A general $C_4$-symmetric $\bk\cdot\bp$ Hamiltonian then has the form
\begin{equation}\label{generalC4ham}
\begin{split}
h(\bk)\eq f(\bk)\sigma_++f(\bk)^*\sigma_-+g(\bk)\sigma_z, \;\; \sigma_{\pm}=\f{\sigma_x\pm i\sigma_y}{2}, \\
f(\bk)\eq \alpha k_{-} +\beta k_{-}(k_z-b) \;\in \mathbb{C}; \as k_{\pm}=k_x\pm ik_y, \\
g(\bk)\eq  a(k_z-b)^2+ck_+k_-+d  \;\in \mathbb{R},
\end{split}
\end{equation}
with $\alpha,\beta$ being complex constants, and $a,b,c,d$ real constants. 
Equation (\ref{generalC4ham}) reduces to the Hamiltonian of a Berry dipole centered at the momentum origin [Eqs.~(\ref{eq:kp-sym-spinor},\ref{eq:zszs-hamiltonian}) with $\phi=0$] if the above $\bk\cdot\bp$ parameters take the particular values:
\e{\text{Berry dipole:}\as b=d=\alpha=0,\as c=-a=\beta/2=1.\label{eqn:parameters}}
Even after allowing for simple generalizations of the Berry-dipole Hamiltonian,\footnote{Namely, the sign of the Hamiltonian is of little significance; the momenta can be rescaled $C_4$-symmetrically; the bands are allowed to touch at an arbitrary point along the $k_z$-axis. The resulting more general Berry-dipole Hamiltonian has the form
\begin{align}
    h(\bk)\eq\pm\zeta^\dagger(\bk)\bsigma\zeta(\bk)\cdot \bsigma, \lin
    \zeta(\bk)\eq\begin{pmatrix}
        Zk_- \\
        A(k_z-b)
    \end{pmatrix},
    \label{eq:Berry-dipole-gen}
\end{align}
with parameters $Z, A\in \mathbb{C}$ and $b\in\mathbb{R}$.
} we find it necessary to tune at least three real parameters to attain a Berry-dipole Hamiltonian from \q{generalC4ham}.\footnote{Indeed, to tune the Hamiltonian given by \q{generalC4ham} to the Berry-dipole Hamiltonian \eqref{eq:Berry-dipole-gen}, one needs to tune two real parameters to $d=\alpha=0$. Parameters $\beta$ and $c$, related to the free parameters of the Berry-dipole Hamiltonian as $\beta=\pm2ZA^*\in\mathbb{C}$ and $c=\pm|Z|^2\in\mathbb{R}$, are kept free. However, the real parameter $a$ must be tuned to the value fixed by the free parameters $\beta$ and $c$, namely $a=-|\beta/2|^2/c$. Thus, in total, we require , which in total results in three fine tuned parameters.}

The necessity of fine-tuning can be partially justified by  certain symmetries of the Berry-dipole $\bk\cdot \bp$ Hamiltonian that are generically absent in the tight-binding Hamiltonian. 
Since the $\bk\cdot \bp$ Hamiltonian is obtained by truncating a Taylor expansion of the tight-binding Hamiltonian, these `certain symmetries' apply to low-order terms in the expansion but are generically broken at higher order. 
To the extent that the low-energy physics of insulators is determined by Bloch states in a small $\bk$-region containing the band-touching point, one may refer to these `symmetries' as \textit{emergent}. 

One of these emergent symmetries constrains the Berry-dipolar Hamiltonian [Eqs. (\ref{eq:kp-sym-spinor}--\ref{eq:zszs-hamiltonian})] as
\e{ \sigma_z h(\bk_{\perp},k_z;\phi) \sigma_z= h(\bk_{\perp},-k_z;-\phi),\la{genmirror}}
which might be viewed as a $\pi$-rotation symmetry in the  4D $(\bk,\phi)$-space.\footnote{To see that this constraint holds, note that the spinor-form Hamiltonian \eqref{eq:zszs-hamiltonian} can be written in terms of spinor components as
\begin{equation}
    h(\bk; \phi)=\begin{pmatrix}
        |\zeta_1(\bk; \phi)|^2-|\zeta_2(\bk; \phi)|^2 & 2\zeta_1(\bk; \phi)\zeta_2^*(\bk; \phi) \\
        2\zeta_1^*(\bk; \phi)\zeta_2(\bk; \phi)& - |\zeta_1(\bk; \phi)|^2+|\zeta_2(\bk; \phi)|^2
    \end{pmatrix},
    \label{eq:ham-via-z1z2}
\end{equation}
where $\zeta(\bk; \phi)=(\zeta_1(\bk; \phi),\zeta_2(\bk; \phi))^\top$. Components of the spinor \eqref{eq:kp-sym-spinor} satisfy the relations $\zeta_{1(2)}(\bk_\perp, -k_z; -\phi)=(-)\zeta_{1(2)}(\bk_\perp, k_z; \phi)$. Plugging the spinor components evaluated at $(\bk_\perp, -k_z;-\phi)$ into \q{eq:ham-via-z1z2}, we see that \q{genmirror} holds.} 
At the phase-transition point ($\phi=0$), the $\pi$-rotation reduces to a  mirror symmetry that inverts $k_z$ in $\bk$-space, and correspondingly inverts the rotational axis in real space.\footnote{For $C_n$-symmetric $\bk\cdot \bp$ Hamiltonians with even $n$, composing a $z$-mirror reflection with a two-fold rotation (about $z$) gives a spatial-inversion operation. \label{foot:PT-sym}} (In short, we call this symmetry over the 3D $\bk$-space a \textit{$z$-mirror symmetry}.)  
An alternative view of \q{genmirror} is that $\phi$ parametrizes a Hamiltonian perturbation that breaks the $z$-mirror symmetry. One may further verify that the $z$-mirror symmetry fixes the complex $\bk\cdot \bp$ parameter $\alpha$ in \q{generalC4ham} to zero, as is true for the Berry-dipole Hamiltonian. 

Identifying the 4D $\pi$-rotation symmetry confers two other advantages: the first is that it inspires a strategy to realize crystalline Hopf insulators that is alluded to in \s{sec:generictransition}; the second is that it explains why the  continuum Hopf number [\q{eq:hopfinvar-cont}] and continuum Zak phase [\q{eq:pol-continuum}] attain the half-integer values:
\begin{subequations}
\e{&\chi^\textrm{cont.}[h(\bk;\phi)]= \f1{2} \Delta \ell \text{sgn}[\phi], \la{halfchi}\\    &\f{\mathscr{Z}^{\mathrm{cont.}}}{2\pi}[h(\bk;\phi)] =\f1{2} \text{sgn}[\phi].\la{halfzak}}
\end{subequations}
Equation (\ref{halfchi}) was argued pictorially in Sec.~\ref{sec:Berry-transitions} and is proven in Appendix~\ref{app:change-hopf}, while \q{halfzak} can be proven in analogy with the argument in Sec.~\ref{sec:Berry-transitions} and footnote~\ref{foot:Zak-change}, with the modified assumption that the intra-cell function is symmetric under the $\pi$-rotation.  
The 4D $\pi$-rotation symmetry constrains the continuum topological numbers as:
\begin{subequations}\la{inversionconstraint}
\begin{eqnarray}
\chi^\textrm{cont.}[h(\bk;\phi)]&=-\chi^\textrm{cont.}[h(\bk;-\phi)], \\
\f{\mathscr{Z}^{\mathrm{cont.}}}{2\pi}[h(\bk;\phi)]&=-\f{\mathscr{Z}^{\mathrm{cont.}}}{2\pi}[h(\bk;-\phi)],
\end{eqnarray}
\end{subequations}
where the square brackets indicate arguments of $\chi^\textrm{cont.}$ and of  $\mathscr{Z}^\textrm{cont.}/2\pi$, here treated as functionals.
To rationalize Eqs.~(\ref{inversionconstraint}), we equate the Hopf number and the Zak phase of two equal Hamiltonians $h(\bk;\phi)$ and $\sigma_z h(\bk_\perp,-k_z;-\phi)\sigma_z$ [related by Eq.~\eqref{genmirror}]. 
Interpreting the conjugation of the Hamiltonian with $\sigma_z$ matrices (and the associated reversal of sign of $k_z$) as the $z$-mirror crystallographic operation, we obtain Eqs.~\eqref{inversionconstraint}
by ($\textrm{a}$) the interpretation of the Zak phase as a positional $z$-polarization which is odd under $z$-mirror operation, and by ($\textrm{b}$) the fact (proven in Appendix~\ref{app:incompatibility}) that the Hopf invariant is odd under improper crystallographic spacetime operations -- of which $z$-mirror is a case in point.\footnote{Any crystallographic spacetime operation has a four-dimensional matrix action on the four-vector of spacetime, with the matrix $M$ a direct sum $O(3)\oplus O(1)$; $O(1)=\{1,-1\}$ simply distinguishes if the operation reverses time or not. We say the spacetime operation is proper if det$[M]=1$, and improper if det$[M]=-1$. } 
Coupling \q{inversionconstraint} with the guaranteed integer-quantization of \textit{changes} to the continuum topological numbers across the phase-transition point ($\phi=0$) (which were discussed in Sec.~\ref{sec:Berry-transitions}), one derives that the continuum topological numbers are indeed quantized to half-integer values as $\phi$ approaches zero from both positive and negative directions.

\subsection{Topological semimetals are `easier to reach'}\la{sec:generictransition}

To understand the generic phase transition between distinct crystalline Hopf insulators, note that the considered symmetry setting (Wigner-Dyson class $\textrm{A}$ with $\mathrm{P}n$ symmetry) imposes no local-in-$\bk$ constraint on the $\bk\cdot\bp$/tight-binding Hamiltonian $h(\bk)$ at generic $\bk$.\footnote{On the rotation-invariant $\bk$-line, the mutually-disjoint condition ensures that $h(\bk)$ and the $\bk$-independent rotation matrix are simultaneously diagonal, implying that only a single, tunable real parameter is needed for a band touching on this $\bk$-line. 
Such a one-parameter-tuned band touching is a (multi-)Weyl point, whose existence precludes a crystalline-Hopf-insulating phase.
See related discussion in the previous \s{sec:notgeneric}.} This implies~\cite{vonNeumann:1929,Bzdusek:2017} that the generic number of real parameters needed to tune to an energy degeneracy (or two-band touching) is three; the number ${\mathfrak{D}}=3$ is sometimes referred to as the \emph{codimension of an eigenvalue degeneracy}. One implication is that the generic band degeneracy in three momentum dimensions is a  Weyl point~\cite{wan2010}, i.e., a band-touching point in $\bk$-space acting as a monopole source of Berry curvature~\cite{berry_quantalphase}.

If the three-dimensional $\bk$-space is supplemented with a tuning parameter $\phi$, the generic band-degeneracy in  the four-dimensional $(\bk,\phi)$-space is a one-dimensional line.
Without additional symmetries and without fine-tuning, such a line generically forms a contractible loop with a nonzero projection $[\phi_1,\phi_2]$ on the $\phi$-axis; such loop can be interpreted as  the concatenation of world-lines  of at-least-two Weyl $\bk$-points (of opposite monopole charge) which are created at an earlier `time' $\phi_1$ and annihilated at a later $\phi_2$. 
In other words, for any $\phi$ in the (open) interval $(\phi_1,\phi_2)$, the $\bk$-space Hamiltonian corresponds to a Weyl semimetal (WSM). We will refer to the union of these world-lines as a \emph{Weyl-point trajectory} (WPT). 
By enclosing the trajectory with a momentum 3-sphere $S^{\! 3}$, the WPT can be ascribed a topological charge given by an element of the third homotopy group of the classifying space of two-band Hamiltonians: $\pi_3(S^{\! 2})=\mathbb{Z}$, i.e., the Hopf invariant on the enclosing sphere. 
If this invariant is nontrivial, then the two insulating phases for $\phi<\phi_1$ resp.~$\phi>\phi_2$ exhibit a different value of the Hopf invariant. 

Reading the above argument `backwards', one finds that two $\mathrm{P}n$-symmetric Hamiltonians with distinct Hopf invariants are generically intermediated by an extended WSM region. The exceptions require to either:\medskip \\
\noindent (\emph{i}) continuously shrink the WPT to a single point (a Berry dipole!)  on a rotation-invariant $\bk$-line, or \medskip \\
\noindent (\emph{ii}) continuously deform the WPT such that it lies entirely within the $\bk$-space at fixed `time' $\phi=\phi_0$.\medskip \\
\noindent Observe that a WPT occurring `instantaneously' is simply a topological nodal-line degeneracy. While both options generically require to fine-tune the $\mathrm{P}n$-symmetric Hamiltonian in Wigner-Dyson symmetry class A, option (\emph{ii}) can also arise in a non-generic/higher-symmetry scenario where the Hamiltonian is $z$-mirror-symmetric at $\phi=\phi_0$ but $z$-mirror-asymmetric at $\phi \neq \phi_0$; this is consistent with a previous observation that a non-trivial Hopf invariant requires breaking all improper crystallographic symmetries, including all mirror symmetries.

\subsection{Case study of intermediate topological semimetals }\la{sec:casestudyTSM}

We illustrate options (\emph{i}--\emph{ii}) with a model $\bk\cdot \bp$ Hamiltonian:
\begin{equation}\label{eqn:H-almost-dipole}
\begin{split}
h(\bk;\mathbb{m}) &= 2\alpha k_z \left( k_x \sigma_x - k_y \sigma_y\right) - 2\beta k_z \left( k_x \sigma_y + k_y \sigma_x \right) \\
&\phantom{=} + \left[\gamma(k_x^2 + k_y^2) - \delta k_z^2 - \mathbb{m}
\right]\sigma_z,
\end{split}
\end{equation}
where $\mathbb{m}$ is a real-valued tuning parameter which we will shortly relate to the parameter $\phi$ in the Berry-dipole Hamiltonian. 
The Hamiltonian in Eq.~(\ref{eqn:H-almost-dipole}) is invariant under mirror symmetry $M_z: (k_x,k_y,k_z)\mapsto (k_x,k_y,-k_z)$ represented by $\sigma_z$, and rotation symmetry $C_{4z}:(k_x,k_y,k_z)\mapsto (-k_y,k_x,k_z)$ represented by $R_{C_{4z}}=\diag(1,-i)$. 
The mirror symmetry implies that the $\sigma_{x,y}$ ($\sigma_z$) term is odd (even) in $k_z$, and the $C_{4z}$-symmetry then guarantees that (up to unimportant terms proportional to the identity matrix) the expression in Eq.~(\ref{eqn:H-almost-dipole}) with free coefficients $\alpha$, $\beta$, $\gamma$, $\delta$ and $\mathbb{m}$ 
is the most general symmetry-compatible Hamiltonian up to quadratic order in the momentum components. 

By employing a rotation of the in-plane coordinates $(k_x,k_y)$, and after rescaling $k_z$ relative to the in-plane coordinates, we can generally achieve $\alpha = \abs{\gamma} = \abs{\delta} = 1$ and $\beta =0$. We further assume $\gamma,\delta > 0$, leading to 
\begin{equation}
h_0(\bk;\mathbb{m}) = 2 k_z \left( k_x \sigma_x  - k_y \sigma_y\right)+ \left(k_x^2 + k_y^2 - k_z^2 - 
\mathbb{m}
\right)\sigma_z,\label{eqn:H-almost-dipole-3}
\end{equation}
in which case the model is a Weyl semimtal (WSM) with a pair of Weyl points located at $\bk = (0,0,\pm\sqrt{\abs{\mathbb{m}}})$ for $\mathbb{m}<0$ [Fig.~\ref{fig:semimetal}(a)], and a nodal-line semimetal (NLSM) with a nodal ring of radius $\sqrt{\mathbb{m}}$ located at the $k_z = 0$ plane for $\mathbb{m}>0$~\cite{Sun:2018} [Fig.~\ref{fig:semimetal}(b)]. 
Notably, the critical point $\mathbb{m}=0$ is a Berry dipole, described by the spinor $\zeta(\bk)^\top = (k_x + i k_y \, , \, k_z)$, suggesting that the outlined model is close to the Hopf-insulating phase.

We now consider breaking of the $M_z$ and $C_{4z}$ symmetry, and study the fate of the semi-metallic phases and of the Berry dipole. 
Multiple additional terms can be added to the Hamiltonian; for simplicity, we consider only two,
\begin{equation}
h_\textrm{perp.}(\bk;A,B) = + 2 A \left(k_x \sigma_y + k_y \sigma_x \right) + 2 B \sigma_x,\label{eqn:H-pert}
\end{equation}
which illustrate the generic behavior of the Berry dipole under perturbations. 
The $B$ term breaks both mirror and rotation symmetries, while the $A$ term breaks only the mirror symmetry. 
As a function of parameters $(\mathbb{m},A,B)$, the combined Hamiltonian $h_0(\bk;\mathbb{m})+h_\textrm{perp.}(\bk;A,B)$ has a rich phase diagram illustrated in Fig.~\ref{fig:Berry-dipole-perturbation}, which we now discuss in detail. 

\begin{figure}
    \centering
    \includegraphics[width=0.48\textwidth]{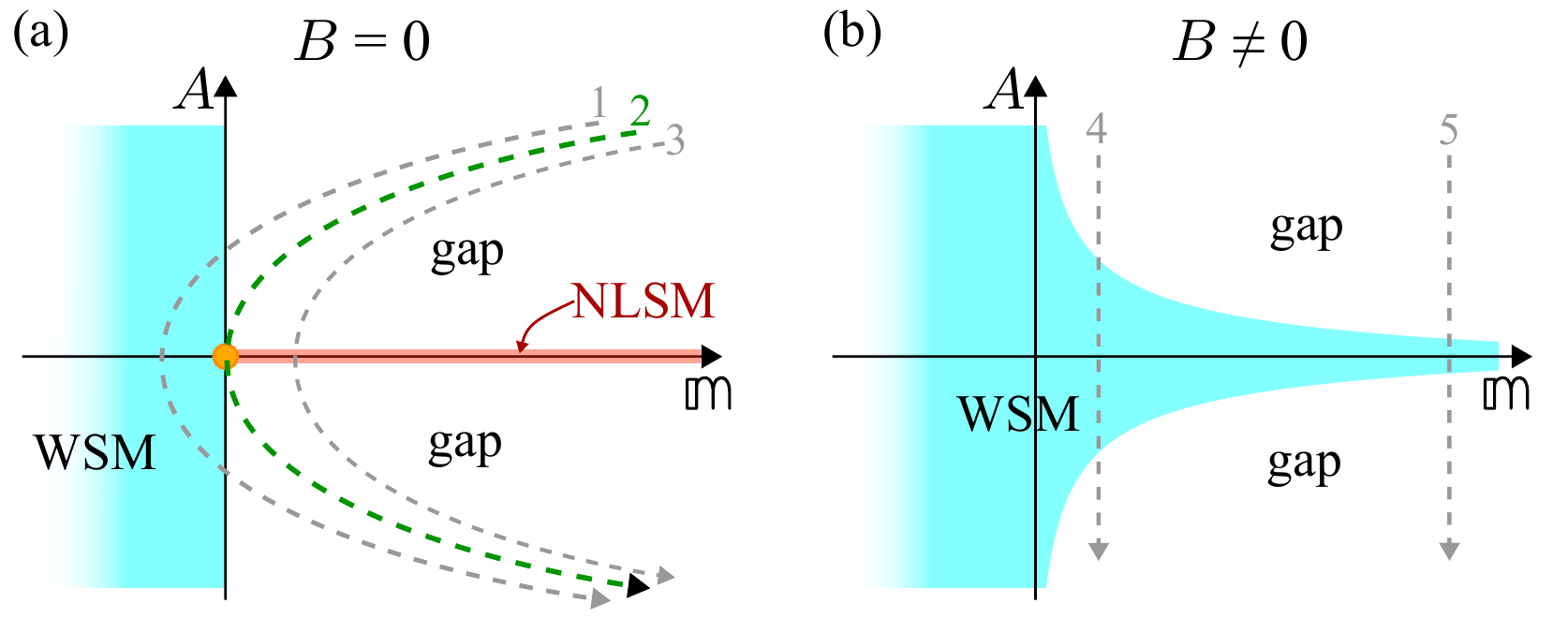}
    \caption{
    Phase diagram of the Hamiltonian in Eq.~(\ref{eqn:H-almost-dipole-3}) with the perturbations in Eq.~(\ref{eqn:H-pert}) that break both $M_z$ and $C_{4z}$ symmetry. 
    (a) 
    For $B=0$ (rotation symmetry preserved) there is a Weyl semimetal (WSM, cyan) phase at $\mathbb{m}<0$, and a nodal-line semimetal (NLSM, red) phase at $A=0,\mathbb{m}>0$. The gapless phases separate two gapped regions (white) that have different values of the Hopf invariant. 
    The trajectory `2' (dashed green line), described by $A = \phi$ and $\mathbb{m}=\phi^2$, reproduces the spinor-form Hamiltonian in Eqs.~(\ref{eq:zszs-hamiltonian}) and~(\ref{eq:kp-sym-spinor}). 
    However, a more generic trajectory connecting the two gapped regions either exhibits an extended WSM phase (path `1') or passes through the NLSM phase (path `3'). 
    (b) For $B\neq 0$, the NLSM phase expands into a Weyl semimetal. 
    In this case, arbitrary trajectory connecting the two gapped regions passes through an extended WSM phase. 
    Paths `4' and `5' are discussed in the text of Sec.~\ref{sec:casestudyTSM}. 
    The change in the Hopf invariant along the trajectories `1'--`5' amounts to $\delta\chi=-1$.}
    \label{fig:Berry-dipole-perturbation}
\end{figure}

First, for $B=0$ (when rotation symmetry is preserved) there is a WSM phase at $\mathbb{m}<0$, and a NLSM phase at $A=0,\mathbb{m}>0$. 
The union of both gapless phases separates two gapped regions, which we now argue to be topologically distinct. Observe that the trajectory marked by `2' (dashed green line), parametrized as $A = \phi$ and $\mathbb{m}=\phi^2$, corresponds \emph{exactly} to the Hamiltonian in Eq.~(\ref{eq:zszs-hamiltonian}) with spinor $\zeta(\bk)^\top = (k_x + i k_y\, , \, k_z + i \phi)$, which exhibits a single gapless point: a Berry dipole. 
It is stated in Eq.~(\ref{eq:Berry-dipole-hopf-change})  (and derived in Appendix~\ref{app:change-hopf}) that the Berry dipole carries a non-trivial Hopf charge; therefore, the two insulating phases are characterized by a different value of the Hopf invariant, which changes by $-1$ along the trajectory `2'. 
However, the Berry-dipole transition is fine-tuned: small deviations (indicated by path `1' resp.~`3') result in an intermediate semi-metallic phase that is either a WSM or NLSM. 
In these cases, the change of the Hopf invariant is facilitated by the non-trivial Hopf charge carried by the WPT in the four-dimensional $(\bk,\phi)$-space, representing an ``expanded'' analog of the Berry dipole.

The inclusion of finite $B\neq 0$ transforms the NLSM phase into a WSM region, bounded by lines $\abs{A}<\abs{B}/\!\sqrt{\mathbb{m}}$. 
Therefore, when moving from the gapped phase at $A>0$ to the gapped phase at $A<0$, one always needs to pass through an extended WSM phase where a pair of Weyl points are first created and later annihilated.
[One rationalization is that the mirror symmetry needed to protect the NLSM is broken throughout the phase diagram in \fig{fig:Berry-dipole-perturbation}(b).] 
The corresponding WPT is nearly vertical (i.e, along the $k_z$-axis) in the $\bk$-space for small $\mathbb{m}/B$ (path `4') and nearly horizontal [i.e., inside the $(k_x,k_y)$-plane] for large $\mathbb{m}/B$ (path `5'). 
Setting $\mathbb{m}=\cos t$ and $B=\sin t$ with $t\in[0,\tfrac{\pi}{2}]$ allows us to `rotate' the WPT trajectory in the $(\bk,A)$ space from lying in $(k_x,k_y)$-plane (phase transition via NLSM) to lying in $(k_z,A)$-plane (phase transition via extended WSM region with Weyl points along the rotation axis). 
We have thus shown on a concrete model how the Berry-dipole transition is generically altered into a NLSM or an extended WSM phase.

\subsection{Intermediate nodal-line semimetal is delicate-topological}\la{sec:delicateSM}

\begin{figure*}
    \centering
    \includegraphics[width=0.87\textwidth]{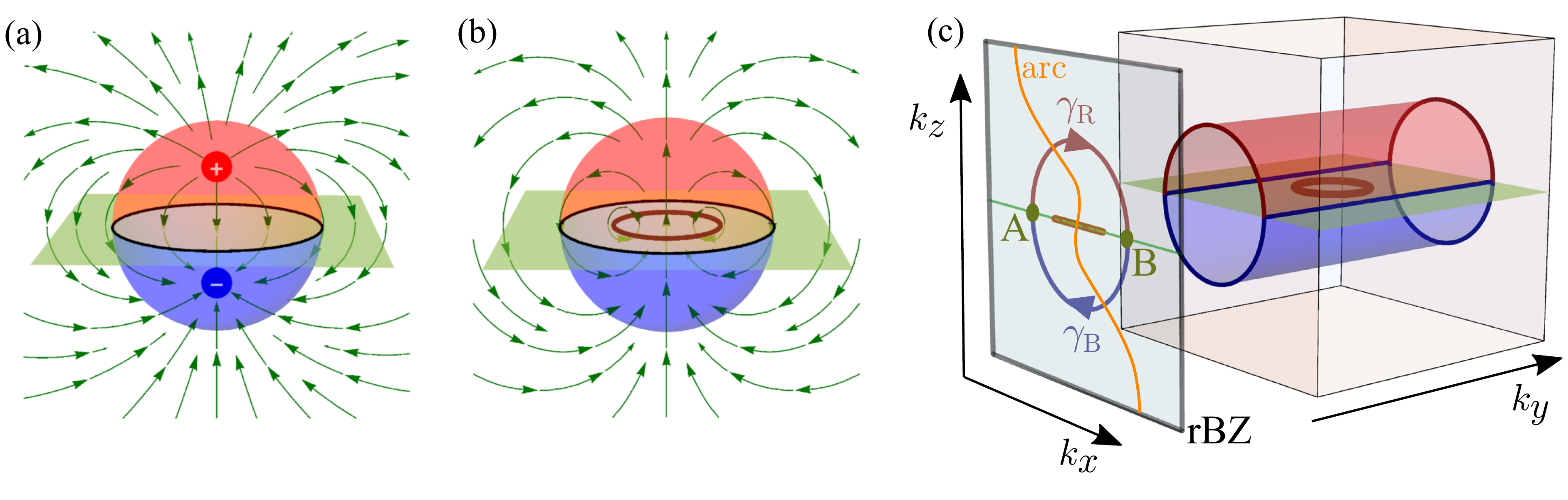}
    \caption{
    (a,b) The band degeneracy of the model in Eq.~(\ref{eqn:H-almost-dipole-3}) for (a) $\mathbb{m}<0$, and (b) $\mathbb{m}>0$. 
    The red/blue dots indicate a Weyl node with positive/negative chirality, and the brown loop represents a nodal-line ring. 
    The green arrows display the Berry curvature field, whereas the bright green sheet indicates the $z$-mirror-invariant plane. 
    The colored hemispheres each exhibit a non-trivial and quantized value of the Chern number. 
    This Chern number is a symmetry-protected delicate topological invariant, which obstructs the annihilation of the Weyl points when they collide for $\mathbb{m}=0$. 
    (c) For $\mathbb{m}>0$ (when nodal-line ring is present), we enclose the band degeneracy with a $z$-mirror-symmetric cylinder that projects onto a circular path $\gamma_\textrm{R}{\cup}\gamma_\textrm{B}$ in the reduced Brillouin zone (rBZ). 
    The valence band has the same orbital character at both points $A$ and $B$; however the polarization increases by $1$ along the red path $\gamma_\textrm{R}$ (projection of the semi-cylinder with positive Chern number) and drops by $1$ along $\gamma_\textrm{B}$, thus exhibiting a returning Thouless pump (RTP). 
    The mirror-eigenvalue anomaly on the circular path $\gamma_\textrm{R}{\cup}\gamma_\textrm{B}$ (with mirror-invariant points $\textrm{A}$ and $\textrm{B}$) in the presence of sharp boundary condition enforces the existence of two surface Fermi arcs (orange lines) which both terminate at the rBZ projection of the nodal-line ring (brown segment located between points $\textrm{A}$ and $\textrm{B}$.
    The Fermi arcs terminate at projections of band nodes elsewhere in rBZ (of which only a part in the $(k_x,k_z)$-plane is shown).
    }
    \label{fig:semimetal}
\end{figure*}

With few exceptions~\cite{Sun:2018,Lim:2017}, previous analyses of the robustness of Weyl points has been based on topological invariants which are stable topological, i.e., characteristic of equivalence classes defined in $K$ theory~\cite{Horava_stabilityofFSandKtheory}. 
One may therefore expect that two Weyl points (with opposite chirality) in our model Hamiltonian [Eq.~(\ref{eqn:H-almost-dipole-3}) with $\mathbb{m}<0$, illustrated in Fig.~\ref{fig:semimetal}(a)] to pairwise annihilate upon collision at $\mathbb{m}=0$; after all, the stable-topological invariant associated to the union of the two Weyl points, which one may view as the Chern number of a Gaussian surface enclosing both Weyl points, vanishes.

In spite of this vanishing Chern number, the two Weyl points do not annihilate but instead `convert' into a nodal-line ring for $\mathbb{m}>0$ [Fig.~\ref{fig:semimetal}(b), see also the phase diagram in \fig{fig:Berry-dipole-perturbation}(a)]
This is suggestive that the two Weyl points/nodal loop are associated to a \textit{non-stable} topology which prevents bands from untouching. 
A priori, logic permits this topology  to be either fragile or delicate; we confirm here that the topology is delicate.
What we will present has some overlap with an older work (\ocite{Sun:2018}) by one of us, where the model Hamiltonian in Eq.~(\ref{eqn:H-almost-dipole-3}) was first investigated in detail. 
The novelty of the following exposition lies in our embedding of a known model into the unifying framework of delicate topological (semi)metals, while also demonstrating that the returning Thouless pump, angular-momentum anomaly, conditionally-robust surface states, and Zak-phase anomaly are generalizable concepts with a wider applicability beyond rotation-symmetric Hamiltonians.

\subsubsection{Mirror-protected delicate topological invariant in the bulk}

To identify the delicate-topological invariant associated to the band touching of Eq.~(\ref{eqn:H-almost-dipole-3}), it is worth recognizing that both the conduction and the valence band carry opposite eigenvalues of the $z$-mirror symmetry (represented by $\hat{M_z} = \sigma_z$) at the reflection-invariant plane ($k_z=0$) -- this may be viewed as the reflection-symmetric analog of the mutually-disjoint condition [Eq.~(\ref{eq:mut-disj})], which we previously formulated for rotation-invariant $\bk$-lines. 
One can then  consider a 2-sphere ($S^{\! 2}$) that encloses the nodes and is symmetric under $M_z$, and argue as follows: (1) because of the mirror symmetry, the information about the continuous Hamiltonian on $S^{\! 2}$ is fully encoded on one [e.g., the upper ($k_z >0 $)] hemisphere, and (2) due to the mutually-disjoint condition, the spectrally-flattened Hamiltonian [cf.~Eq.~\eqref{eqn:flatten-nonspinor}] 
over the sphere's equator (a circle in the $k_z=0$ plane) is simply $\sigma_z$; this implies that the equator maps to a single point on the Bloch sphere and allows to treat the hemisphere as if it were a closed $2$-sphere for the purpose of topological classification. 
The Hamiltonian on this abstract $2$-sphere exhibits an integer `wrapping' number around the Bloch sphere, which translates into an integer-valued Chern number on the upper hemisphere. 
One easily verifies that for the situation with a pair of $M_z$-related Weyl points, this wrapping number simply equals the net Berry-monopole charge of Weyl points in the upper hemisphere [cf.~Fig.~1(a) in Ref.~\cite{Sun:2018}]. 
As a consequence, the colliding Weyl points cannot annihilate: opening of an energy gap inside the $2$-sphere would contradict the existence of a non-trivial topological invariant on the $2$-sphere.

It was further proven in the Supplemental material of Ref.~\cite{Sun:2018}, using relative homotopy groups~\cite{Hatcher:2002}, that the quantization of the Chern number persists in $({>}2)$-band Hamiltonians -- if and only if the conduction/valence bands obey the mutually-disjoint condition of $M_z$ eigenvalues.\footnote{The term `mutually-disjoint' was not used in the previous work~\cite{Sun:2018}.}
In the language we have formulated here, the Chern number on the upper hemisphere, which enforces the conversion of the Weyl-point pair into a nodal-line ring, is a symmetry-protected delicate topological invariant. 
One can thus interpret such $({>}2)$-band gapless models as a (symmetry-protected) \emph{delicate topological (semi)metal}.\footnote{In the case of a metal, the Gaussian two-spherical surface used to define the topological invariant is continuously deformable to the Fermi surface, with anticipated implications for topo-fermiology~\cite{topofermiology}.} 

The bulk Chern number over the hemisphere is equivalent to a reflection-symmetric analog of the returning Thouless pump, as we now demonstrate. 
For this purpose, instead of enclosing the band degeneracy with an $M_z$-symmetric 2-sphere, here we enclose it with an $M_z$-symmetric cylinder with the cylindrical axis in a direction (say, $y$) orthogonal to $z$.
One may view the cylinder as an $M_z$-symmetric `inflation' of the 2-sphere, as illustrated in \fig{fig:semimetal}(c). 
Because of the mutually-disjoint condition, the bulk valence band (restricted to the mirror-invariant $\bk$-plane) has overlap with one and only one basis orbital. 
This implies that an analog of the iso-orbital condition [Eq.~(\ref{isoorbital})] is satisfied for pair of mirror-invariant $\bk$-lines within this mirror-invariant $\bk$-plane -- including those two $\bk$-lines that lie at the intersection of the cylinder with the $(k_z=0)$-plane, and project onto mirror-invariant points [$\textrm{A}$ and $\textrm{B}$ in \fig{fig:semimetal}(c)] in the 2D $\textrm{rBZ}$. 
The projection of the cylinder forms a circular path $\gamma_\textrm{R}{\cup}\gamma_\textrm{B}$ in rBZ, which we decompose into two $z$-mirror-related semicircular segments: $\gamma_\textrm{R}$ and $\gamma_\textrm{B}$.

Because of the non-trivial first Chern number on the red semi-cylinder (equalling the net Berry-monopole charge of all Weyl points in the upper hemisphere), the $\bkp$-dependent polarization grows by one primitive lattice period (in the direction parallel to the surface normal) as $\bkp\in\textrm{rBZ}$ is continuously advanced along $\gamma_\textrm{R}$.
The polarization subsequently reverts to the original value after further moving along $\gamma_\textrm{B}$ -- manifesting a returning Thouless pump (RTP) along the closed path $\gamma_\textrm{R}{\cup}\gamma_\textrm{B}$.

\subsubsection{Bulk-boundary correspondence 
of mirror-protected returning Thouless pump}

Let us explain the boundary correspondents of the bulk delicate topological invariant associated with these band degeneracies. 
For this purpose, it is convenient to interpret the cylinder as the Brillouin zone ($\textrm{BZ}'$) of a 2D-translation-invariant, mirror-symmetric Hamiltonian $h'$; $h'$ is obtained by restricting the 3D Hamiltonian [\q{eq:zszs-hamiltonian-normalized}] to the cylinder, and interpreting the azimuthal coordinate of the cylinder as a linear momentum coordinate of $\textrm{BZ}'$; $\textrm{BZ}'$ projects in the $y$-direction to a 1D reduced Brillouin zone denoted $\textrm{rBZ}'$. 
Much like how the angular-momentum anomaly [Eq.~(\ref{amanomaly})] is the boundary correspondent of the bulk rotation-protected RTP, we claim that {the restricted Hamiltonian} $h'$ exhibits an analogous \textit{mirror-eigenvalue anomaly} as the boundary correspondent to the mirror-protected RTP, under open boundary conditions that are $z$-mirror-symmetric. 

For concreteness, we assume a half-infinite geometry with facet normal in the $y$-direction -- which is also the direction of the cylindrical axis. 
What the mirror-eigenvalue anomaly means is that for at least one faceted band, the mirror eigenvalues at mirror-invariant points of $\textrm{rBZ}'$ [${\textrm{A}}$ and ${\textrm{B}}$ in Fig.~\ref{fig:semimetal}(c)] are opposite in sign; this can be proven with essentially the same techniques as those in \app{app:ang-mom-anomaly}. 
Applying a Zak-phase theorem proven in \ocite{AA_wilsonloopinversion}, this alternation of mirror eigenvalues guarantees the existence of at least one symmetry-quantized $\pi$ Berry-Zak phase in the spectrum of the (possibly non-Abelian) Wilson loop, for the (possibly multi-band) faceted subspace -- contrasting to the vanishing Berry-Zak phase of the bulk conduction and bulk valence band. 
This is the mirror-symmetric analog of the Zak-phase anomaly [Eqs.~(\ref{thesame}) and \eqref{eq:zak-bulk-surf}] that we have previously formulated for rotation symmetry.

A further implication of the mirror-eigenvalue anomaly is the existence of conditionally-robust edge states over $\textrm{rBZ}'$; this means that edge-state energies of $h'$ necessarily interpolate across the bulk energy gap of $h'$ under sharp boundary conditions; a proof of this would utilize spectral properties of Toeplitz matrices [cf.~Appendix~\ref{app:Toeplitz-bounds}], in close analogy with the demonstration of rotation-symmetric conditionally-robust surface states presented in \s{sec:BBC_sharp}.

Conditionally-robust edge states of $h'$ imply conditionally-robust \emph{surface Fermi arcs} [orange lines in~Fig.~\ref{fig:semimetal}(c)] of the 3D Hamiltonian subject to sharp boundary conditions, with one arc crossing each of the semicircles $\gamma_\textrm{R}$ and $\gamma_\textrm{B}$ in the 2D $\textrm{rBZ}$; this holds for both the Weyl semimetal and the nodal-ring semimetal phase of the Hamiltonian in Eq.~(\ref{eqn:H-almost-dipole-3}).
This clarifies the origin of Fermi arcs that were observed numerically (but remained unexplained) in Fig.~2 of Ref.~\cite{Sun:2018} under sharp boundary conditions.  
It is worth emphasizing that the surface Fermi arcs are not topologically guaranteed under non-sharp boundary conditions -- a point not articulated in \ocite{Sun:2018}.

\section{Discussion, material applications, and outlook for delicate topology}\label{sec:conclusion}

If we ask, what is the minimal dimension of a tight-binding Hamiltonian matrix (at each momentum) that models a topological insulator with vanishing first Chern numbers, the answer is four in the case of stable topology, and three in the case of fragile topology~\cite{nogo_AAJH}.
In contrast, delicate topological insulators are the \textit{simplest} topological insulators, in the sense that they are realizable by two-band Hamiltonians. 

In this work, we study two-band Hamiltonians  whose  $n$-fold rotation symmetry protects a new topological invariant -- the returning Thouless pump (RTP) -- in addition to the celebrated Hopf invariant that requires no point-group symmetry. 
The RTP describes how the Zak phase advances by an integer multiple of $2\pi$ across half a reciprocal-lattice vector, then retracts by the same quantity over the next half of the same reciprocal-lattice vector. 
Such a Zak-phase evolution could conceivably be measured in interferometric experiments~\cite{atala_measureszak} with ultra-cold atoms/molecules~\cite{schuster_realizinghopf, schuster_blueprint}.  

The RTP invariant extends Pontrjagin's seminal classification of two-band insulating Hamiltonians in three spatial dimensions~\cite{pontrjagin_classification,kennedy_hopfchern}, which, in the case of trivial first Chern class, is given by the Hopf invariant.
As proof of principle, we have found a tight-binding model with a trivial Hopf invariant but a nontrivial RTP invariant.
We have exhaustively established mod-$n$ relations between the Hopf and RTP invariants for all possible $n$-fold symmetric Hamiltonians $(n=2,3,4,6)$; such relations allow to determine the Hopf invariant modulo $n$ by a measurement of the Zak phase 
given the \textit{a priori} knowledge of the itinerant angular momentum of bulk-conduction and bulk-valence bands. 
Presently we know of no other method, even existing in principle, to measure the Hopf invariant in condensed-matter systems. 

We have proposed a bulk-boundary correspondence for RTP insulators: on rotation-symmetric crystal facets, the surface-localized states of RTP insulators carry anomalous values of the itinerant angular momenta and of the Berry-Zak phase, which are non-identical to the itinerant angular momenta resp.~to the 
Berry-Zak phase of the bulk-conduction states, and also non-identical to corresponding quantities of the bulk-valence states.
This anomaly of the itinerant angular momentum guarantees the existence of surface states at any energy within the bulk energy gap, if the bulk Hamiltonian matrix elements are sharply terminated at a rotation-symmetric facet. Such surface states are robust against continuous deformations of the \textit{bulk} Hamiltonian that preserve both the bulk energy gap and $\mathrm{P}n$ symmetry; however the same surface states can in principle be removed from any energy within the bulk energy gap, if the sharp boundary condition is relaxed. 
We expect the sharp boundary to be a good approximation to the surface termination of layered materials with weak, inter-layer, van-der-Waals coupling, such as $\mathrm{Bi}_2\mathrm{Se}_3$ \cite{zhang2009,zhang2009a}. 

With the goal of materializations (or experimental simulations) of crystalline Hopf insulators, one presently-achieved advancement is an elucidation of the topological-semimetallic phases that generically intermediate the transition from trivial to crystalline Hopf insulators. 
In particular, we have found that the crystalline-Hopf phase lies in proximity to a mirror-symmetry-protected, nodal-line semimetal previously studied in Ref.~\cite{Sun:2018}; this mirror-protected semimetal is the only delicate topological semimetal that we know presently, and demonstrates that the returning Thouless pump (and its corresponding boundary anomalies) are generalizable concepts with a wider applicability beyond rotation-symmetric Hamiltonians.\footnote{While not directly applicable to metals, a three-band touching in the phonon bandstructure was recently argued 
to be delicate topological~\cite{park_acoustictriplepoint}.} 
One may then conceive the strategy of beginning from an existing Weyl or nodal-line semimetal, then inducing an energy gap by material engineering to  enter the Hopf/RTP insulator phase. 
In this regard, we remark that a delicate-topological semimetallic Hamiltonian similar to our case study in Eq.~(\ref{eqn:H-almost-dipole-3}) has been previously applied to model the ferromagnetic compound \ce{HgCr2Se4}~\cite{Xu:2011}. 
Additionally, strain-induced conversions of mirror-related Weyl points to mirror-protected nodal loops were recently predicted~\cite{Bouhon:2020b} for a family of two-element hexagonal compounds, including \ce{ZrTe}, \ce{MoC} and \ce{WC}, with the initial motivation of materializing band degeneracies with non-Abelian topological charges~\cite{Wu:2019}; it would be interesting to revisit these materials with a focus on mirror-protected delicate topology.

A possible complication in materializing Hopf/RTP insulators is that in most known tight-binding models, the next-nearest-neighbour hoppings have larger magnitude than the nearest-neighbour hoppings. 
We do not know if this is required by principle or a present reflection of our ignorance. If it is ignorance, then the theoretical development of tight-binding Hamiltonians with larger nearest-neighbor hoppings would expand the material options for experimental realization. 
However, it is worth remarking that having larger next-nearest-neighbour hoppings is not physically unrealistic -- one would look for materials which lie further from the tight-binding limit and closer to the nearly-free-electron regime~\cite{ashcroft_mermin}; it may also be necessary to look for basis Wannier orbitals which are anisotropic hybridizations of multiple chemical orbitals. 
Beyond solid-state realizations, a recent proposal~\cite{schuster_realizinghopf, schuster_blueprint} suggests a way to realize long-ranged hoppings in periodically driven ultra-cold dipolar molecules trapped in an optical lattice. 

A second difficulty in Hopf/RTP materializations is that realistic materials often require more than a two-band description. 
A generalization of the Hopf invariant to arbitrarily many (say $N$) bands exists, assuming that all $N$ bands are energetically detached from one another~\cite{nband-hopf}. 
Whether or not such generalization is easier to materialize than the two-band Hopf insulator remains to be seen. 
The RTP invariant also allows generalization to more-than-two-band Hamiltonians, assuming that the bulk valence and conduction bands transform differently under rotation, as discussed in Sec.~\ref{sec:rtpbeyond2}. 
Such higher-band generalization with a symmetry constraint is a hallmark of  \textit{symmetry-protected delicate topology}.

Going beyond our assumption of rotational symmetry and a trivial Chern class, one can ask whether the discussed results can be extended to other crystallographic symmetries and/or to nontrivial Chern classes. 
Although these questions are beyond the scope of this work, we would like to mention two possible directions. 
Firstly, one can study delicate topological invariants protected by other crystalline symmetries.
For instance, mirror symmetry can also protect an RTP invariant (as demonstrated by our case study in \s{sec:delicateSM}), although it forces the Hopf invariant to be zero for 3D-translation-invariant Hamiltonians; a systematic study of mirror-protected delicate topology in all spatial dimensions is still lacking and would be desirable. 
Secondly, the extension of the Pontrjagin classification in the presence of rotation symmetry could be generalized to models with a non-trivial first Chern class. 
In this case, the two-band models in the absence of point-group symmetries are classified by four invariants: three Chern numbers computed on the two-dimensional cuts of the BZ and the fourth invariant that takes integer values modulo twice the greatest common divisor of the three Chern numbers~\cite{Hopf:1931,kennedy_hopfchern}. 
It is not inconceivable that some analog of the topological invariants discussed here will likewise acquire an ambiguity depending on the Chern class.

\section{Acknowledgments}
We thank Bastien Lapierre and Luka Trifunovic for valuable comments on the topological boundary signatures at open boundary conditions with smoothly terminated hoppings. 
We further acknowledge suggestions by Frank Schindler concerning the concept of graded multicellularity, and an advice from Zhihui Zhu on the spectral properties of block Toeplitz matrices. 
We extend our gratitude to Nicholas Read for a helpful clarification of the notion of Wannier-$\mathpzc{type}$ functions. 
The paper has benefited from on-going discussions with Penghao Zhu and Taylor Hughes on delicate topology in other symmetry classes.  
A.~N.~was supported by the Swiss National Science Foundation (SNSF) grant No.~176877, and by Forschungskredit of the University of Zurich, grant No.~FK-20-098. T.~N.~acknowledges support from the European Research Council (ERC) under the European Union’s Horizon 2020 research and innovation programm (ERC-StG-Neupert-757867-PARATOP) and from NCCR MARVEL funded by the SNSF. A.~A.~was initially supported by the Gordon and Betty Moore Foundation EPiQS Initiative through Grants No.~GBMF 4305 and GBMF~8691
at the University of Illinois. T.~B. was supported by NCCR MARVEL and by the SNSF Ambizione grant No.~185806. 


\appendix

\section{Glossary of symbols and abbreviations\label{app:glossary}}

All important symbols that are used in the manuscript are summarized in Table \ref{tab:symbols}.

    {\renewcommand{\arraystretch}{1.15}
    \begin{longtable}{p{2.5cm}l}
    \caption{Table of symbols. 
    Some employed abbreviations, listed alphabetically: 
    `ang.' -- angular,
    `inv.' -- \emph{invariant},
    `itin.' -- \emph{itinerant},
    `red' -- \emph{reduced},
    `repr.' -- \emph{representation},
    `rot.' -- \emph{rotation},
    `val.' -- \emph{valence}, 
    }\label{tab:symbols}  
    \\ \hline\hline
    symbol & meaning  \\ \hline
    $\simeq$ \dotand homeomorphism of topological spaces \\ 
    $\mathbbold{0},\mathbbold{1}$ \dotand zero matrix, identity matrix \\
    $a=_n b$ \dotand equivalence modulo $n$: $\!a\!\!\pmod{n}\!=\!b\!\!\pmod{n}$\!\! \\
    $\alpha$ \dotand label of tight-binding basis orbital \\
    $\mathcal{A} = \mathcal{A}_\alpha \,dk_\alpha$ \dotand Berry-connection $1$-form \\
    $\boldsymbol{\mathcal{A}}=(\mathcal{A}_x,\mathcal{A}_y,\mathcal{A}_z)$ \dotand Berry-connection vector \\
    $\mathcal{A}_\mu^{jj'}$ \dotand non-Abelian Berry connection \\
    $b,b'$ \dotand points on the Bloch sphere \\
    $\mathcal{B}^{r}_\kappa$ \dotand $r$-dimensional ball (with radius $\kappa$) \\
    $\beta(\bk)$, $e^{i\beta(\bk)}$ \dotand phase (factor), $U(1)$ gauge transformation \tabreak of a Bloch/continuum state \\
    BR$(\mathcal{L},\br_\perp)$ \dotand band repr. induced  from an $(\mathcal{L},\br_\perp)$-orbital \\
    BBR$[\varphi_\alpha]$ \dotand basis band repr. induced from an orbital $\varphi_\alpha$ \\
    $\abs{\textrm{BZ}}$ \dotand volume of Brillouin zone \\
    $c,v$ \dotand label for conduction and valence band \\
    $C_2,C_3,C_4,C_6,C_n$ \dotand rotation symmetry of order $2,3,4,6,n$ \\
    $\mathscr{C}$ \dotand Chern number \\
    $\textrm{ch}_2(\mathcal{F})$ \dotand second Chern character \\
    $\mathrm{CBR}$ \dotand composite band repr. \\
    $d$ \dotand differential \\
    $\mathbbm{d}$ \dotand order of Berry dipole \\
    $\delta_{\alpha\beta},\delta_{ij}$ \dotand Kronecker symbol \\
    $\updelta(\bk)$ \dotand Dirac delta function (in $\bk$-space) \\
    $\partial$ \dotand boundary operator \\
    $\mathcal{D}$ \dotand 2D domain of integration\\ 
    $\mathfrak{D}$ \dotand codimension of band node formation \\
    $\hat{\bm{e}}_j$ \dotand unit vector \\
    $E$ \dotand energy \\
    $\bar{E}$ \dotand reference energy \\
    $dE_j/dt$ \dotand a scalar velocity of a $j$-th Bloch state \\
    $\mathfrak{E}$ \dotand vector bundle \\
    $\varepsilon>0$ \dotand infinitesimally small real number \\
    $\epsilon_{\alpha\beta\gamma}$ \dotand fully antisymmetric Levi-Civita symbol \\
    $\mathcal{F} = \mathcal{F}_{\alpha\beta}\, dk_\alpha \wedge dk_\beta$ \dotand Berry-curvature 2-form \\
    $\boldsymbol{\mathcal{F}}=(\mathcal{F}_x,\mathcal{F}_y,\mathcal{F}_z)$ \dotand Berry-curvature vector \\
    $\varphi_{\boldsymbol{R,\alpha}}=\ket{\boldsymbol{R},\alpha}\rangle$ \dotand tight-binding-basis orbitals of $\mathscr{H}$\\ 
    $\varphi_{\alpha}=\ket{\alpha}$ \dotand tight-binding-basis orbitals of $\mathscr{H}_\textrm{cell}$\\ 
    $\ket{R_z,\alpha}\drangle$ \dotand tight-binding basis of $\mathscr{H}_{\textrm{column}}$ \\
    $\gamma$ \dotand path, 1D domain of integration\\
    $\gamma_\Pi$ \dotand rotation-invariant line in the BZ \\
    $\Gamma,\textrm{K},\textrm{K}',\textrm{M},\textrm{X}$ \dotand high-symmetry points in rBZ\\
    $\boldsymbol{G}$, $\bG_\Pi=C_m\Pi-\Pi$ \dotand reciprocal lattice vector \\
    $\mathcal{G}$ \dotand space group \\
    $\mathcal{G}_{\br}$, \dotand site-stabilizer group at $\br=(\br_\perp,z)$\\
    $\mathcal{G}_\Pi$ \dotand little group of $\Pi\in\textrm{rBZ}$ \\
    $g\in\mathcal{G}$ \dotand element of the space group \\
    $\check g$ \dotand matrix repr. of $g$ acting on spatial coordinates \\
    $\mathfrak{g}$\dotand grade of multicellular ($\mathpzc{og}$-)Wannier orbital
    \\
    $\phi$, $\mathbb{m}$ \dotand tuning parameter in continuum Hamiltonians \\
    $\Phi$ \dotand tuning parameter in lattice Hamiltonians \\
    $h, h(\bk)$ \dotand Hamiltonian (in momentum space) \\
    $h_\mathrm{flat}(\bk)$ \dotand spectrally flattened Hamiltonian \\
    $\tilde h, \tilde h(\bk)$ \dotand Hamiltonian in the periodic-in-$\bk$ convention \\
    $\boldsymbol{h} = (h_x,h_y,h_z)$ \dotand $2$-band Hamiltonian decomposed to  $\sigma_{x,y,z}$ \\
    $\boldsymbol{h}_\mathrm{flat} \in S^2$ \dotand $2$-band Hamiltonian with flattened spectrum \tabreak decomposed to  $\sigma_{x,y,z}$ \\
    $H\subset \mathcal{G}$ \dotand subgroup of a space group \\
    $\mathscr{H}$ \dotand tight-binding Hilbert space \\
    $\mathscr{H}_\textrm{cell}$ \dotand intra-cell Hilbert space\\
    $\mathscr{H}_\textrm{column}$ \dotand intra-column Hilbert space \\ 
    $\mathcal{H}_{\mathcal{N}}$ \dotand finite chain Hamiltonian with $\mathcal{N}$ unit cells\\
    $\mathfrak{H}$ \dotand homotopy function/deformation\\
    $J_\alpha$ \dotand size of a maximal iso-orbital subset \tabreak fulfilling conditions \eqref{mutualdisj} and \eqref{mutualdisj2}\\
    $\mathcal{J}$ \dotand antisymmetry of Pauli-matrix Hamiltonians \\
    $\chi$ \dotand Hopf invariant of lattice Hamiltonian  \\
    $\chi^\textrm{cont.}$ \dotand Hopf number of continuum Hamiltonian  \\
    $\bk=(k_x,k_y,k_z)$ \dotand (crystal) momentum vector \\
    $\bk_\perp=(k_x,k_y)$ \dotand reduced momentum\\
    $\boldsymbol{\mathbbm{k}}=(\bk;\phi)$ \dotand generalized four-momentum space \\
    $\boldsymbol{\kappa}=(\kappa_x,\kappa_y,\kappa_z)$ \dotand small momentum around band touching point \\
    $\kappa$ \dotand radius of considered spheres in $\bk$-space \\
        $\Delta \ell\in\mathbb{Z}$ \dotand Berry-dipole spin \\
    $\lambda$ \dotand matrix eigenvalue \\
    $\tinyloop$ \dotand infinitesimal loop to determine orientation \tabreak and multiplicity of a preimage \\
    $\mathcal{L},\Delta\mathcal{L}$ \dotand (difference of) on-site angular momentum\\
    $\widetilde{\mathcal{L}},\Delta\widetilde{\mathcal{L}}$ \dotand (difference of) itinerant angular momentum\\
    $\sharp_{b(f)}\widetilde{\mathcal{L}}_{v(c)}(\Pi)$ \dotand number of bulk (faceted) states at $\Pi$ with \tabreak val.-(~cond.-)band-like itin.~ang.~momentum \\
    $m\in\{2,3,4,6\}$ \dotand order of little-group rotation symmetry \\
    $m_\Pi$ \dotand order of little-group at red. momentum $\Pi$ \\
    $\mathcal{M}$ \dotand closed surface \\
    $M_z$ \dotand $z$-mirror symmetry \\
    $\mu(\gamma)$, $\mu(\Pi)$ \dotand multiplicity of preimage $\gamma$, $\gamma_\Pi$ \\
    $n\in\{2,3,4,6\}$ \dotand order of crystalline rotation symmetry\\
    $\hat{\bm{n}}$ \dotand unit vector normal to a surface \\
    $N_v,N_c$ \dotand number of conduction resp.~valence bands\\
    $\mathcal{N}_v(\Pi),\mathcal{N}_c(\Pi)$ \dotand number of faceted states \tabreak with $\widetilde{\mathcal{L}}_f(\Pi)=\widetilde{\mathcal{L}}_v(\Pi),\widetilde{\mathcal{L}}_c(\Pi)$ \\
    $\mathscr{N}_{j,\bR}$ \dotand norm of $\mathpzc{og}$-Wannier orbital \\
    $\upsilon\in\{+1,-1\}$   \dotand Berry-dipole helicity \\
    $\textrm{pt.}$ \dotand point (as a topological space) \\
    $\mathfrak{P}$ \dotand stereographic projection $S^{\! 3}\to \mathbb{R}^3\cup\infty$ \\
    $\mathscr{P}$ \dotand electric polarization \\
    $\Delta\mathscr{P}_{\Pi_1\Pi_2}$ \dotand RTP invariant, $\mathscr{P}_{\Pi_2}-\mathscr{P}_{\Pi_1}$ \\
    $\mathrm{P}n$, $n\in\{2,3,4,6\}$ \dotand considered space groups \\
    $\mathrm{P}n_{\br}$ \dotand presently considered site-stabilizer groups \\
    $\mathrm{p}n\subset\mathrm{P}n$ \dotand plane group generated by $C_n$ rotation \tabreak and two out of three translations \\
    $P$ \dotand projector onto filled (valence) states \\
    $P_b$ \dotand projector onto bulk valence bands \tabreak in the infinite geometry \\
    $P_{(b)}\hat{z}P_{(b)}$, $Q_{(b)}\hat{z}Q_{(b)}$ \dotand projected position operators in $z$-direction \\
    $P_{\mathrm{b}},(P_{\mathrm{t\&b}})$ \dotand faceted projector onto bottom facet \tabreak (onto top and bottom facets) \\
    $P_\downarrow$ \dotand projector onto lower half of a finite slab \\
    $P_{v}^{\mathrm{cell}}$ \dotand projector onto valence subspace in $\mathscr{H}_\mathrm{cell}$ \\
    $P_{v}^{\mathrm{column}}$ \dotand projector onto valence subspace in $\mathscr{H}_\mathrm{column}$ \\
    $\mathcal{P}$ \dotand $k_z=0$ punctured plane with removed $\bk=0$ \\
    $\Pi,\Xi,\Lambda,\Omega$ \dotand rotation-invariant reduced momenta in rBZ\\
    $\Psi$ \dotand Berry flux through a hemisphere \\
    $\ket{\psi_{j,\bk}}\!\Big\rangle$ \dotand Bloch (eigen)state \\
    $\psi_{j,\bk}(\bR,\alpha)$ \dotand tight-binding Bloch function \\ 
    $\psi^{(\mathcal{N})}$ \dotand eigenstate of an $\mathcal{N}$-unit-cell Hamiltonian \\
    $(\psi,\theta,\varphi)$ \dotand 3-spherical coordinates \\
    $Q$ \dotand projector onto empty (conduction) states \\
    $Q_b$ \dotand projector onto bulk conduction bands \tabreak in the infinite geometry \\
    $\boldsymbol{r}\!=\!(x,y,z)\!=\!(\boldsymbol{r}_\perp,z)$ \dotand position vector \\
    $R_\theta = e^{-i\theta\Delta\ell/2\cdot\sigma_z}$ \dotand  continuous rotation operator  \tabreak for Berry-dipole Hamiltonian\\ 
    $R_{C_m}$ \dotand discrete rot.~operator of crystalline $h(\bk)$ \\
    $\widetilde{R}_{C_m}(\Pi)$ \dotand itinerant rot.~operator of $h(\Pi+\bk)$ \\
    $\widetilde{\rho}_{m,\alpha}(\Pi)$ \dotand eigenvalues of $\widetilde{R}_{C_m}(\Pi)$ \\
    $\boldsymbol{R}$ \dotand Bravais vector \\
    $\boldsymbol{\mathcal{R}}=(\hat{x},\hat{y},\hat{z})$ \dotand discrete position operator  \\
    $Rev$ \dotand number and direction of revolutions \tabreak of a vector $\boldsymbol{h}(t)$ when $t$ advances by 1 \\
    $\textrm{rBZ}(\,\ni \!\bk_\perp)$ \dotand (momentum in) reduced Brillouin zone \\
    $\textrm{rRL}$ \dotand reduced reciprocal lattice \\
    RTP \dotand returning Thouless pump \\
    $\textrm{S}\sim -\sigma_z$ \dotand south pole of the Bloch sphere\\
    $S^{\!2}$, $S^{\!3}$, $S^{\!3}_\kappa$ \dotand 2-sphere, 3-sphere (with radius $\kappa$) \\
    $\mathcal{S}$ \dotand momentum-loop in rBZ \\
    $s_g=\pm1$ \dotand reversal of time component of symmetry $g$ \\
    $\boldsymbol{\sigma}=(\sigma_x,\sigma_y,\sigma_z)$ \dotand vector of Pauli matrices \\
    $\varsigma\in\mathfrak{S}_r$ \dotand permutation (group) of $r$ elements\\
    $\Sigma$ \dotand Seifert surface stretched over preimage line \\
    $t\in[0,1]$ \dotand parametrization of homotopy \& closed loops \\
    $\hat{\boldsymbol{t}}$ \dotand unit vector tangent to parametrized loop \\
    $\mathrm{t}\,{\in}\,\{1{,}{.}{.}{.}{,}T\}$ in $\bk^\mathrm{t},\Pi^\mathrm{t}$ \dotand label of consecutive Berry-dipole transitions \\
    $\hat{\mathrm{T}}$ \dotand path-ordering operator \\
    $T^3$ \dotand 3-torus \\
    $\mathcal{T}, \{\mathcal{T}_r\}_{r\in\mathbb{Z}}$ \dotand blocks within a block Toeplitz matrix \\
    $\mathscr{T}$ \dotand transfer matrix \\
    $\boldsymbol{\tau}(\bk_0)$ \dotand orientation of a preimage path 
    at $\bk_0$  \\
    $\boldsymbol{\tau}(\Pi)$ \dotand orientation of rot.-inv. line $\gamma_\Pi$ \\
    $\tau_z$ \dotand orientation along $k_z$-axis for\tabreak non-contractible along $k_z$ loops \\
    $\Theta$ \dotand time-reversal symmetry \\
    $\ket{u(\bk)}$, $\ket{u^\mathpzc{on}(\bk)}$ \dotand (normalized) intra-cell part of Bloch state \\
    $\ket{u^\mathpzc{og}(\bk)}$ \dotand orthogonal but not orthonormal \tabreak intra-cell part of Bloch state \\
    $v^{(\mathcal{N})}(k)$ \dotand Fourier transform of $\psi^{(\mathcal{N})}$ \\
    $\big|\mathcal{V}_{j,\bk_\perp,R_z}\big\rangle\big\rangle\big\rangle$ \dotand intra-column state in $\mathscr{H}_{\mathrm{column}}$\\
    $[V(\bk)]_{\alpha\beta}=e^{i\bk\br_\alpha}\delta_{\alpha\beta}$ \dotand matrix relating periodic and non-periodic\tabreak tight-binding Hamiltonians \\
    VBR \dotand valence band repr. \\
    $\ket{W_{j,\bR}}\!\Big\rangle, \ket{W_{j,\bR}^\mathpzc{on}}\!\Big\rangle$ \dotand conventional (orthonormal) Wannier
    orbital
    \\
    $\ket{W_{j,\bR}^\mathpzc{og}}\!\Big\rangle$ \dotand orthogonal ($\mathpzc{og}$-)Wannier orbital \\
    $\ket{W_{j,\bR}^\mathpzc{type}}\!\Big\rangle$ \dotand Wannier-$\mathpzc{type}$ orbital \\
    $\mathcal{W}_\mathcal{S}$ \dotand Wilson loop computed along $\mathcal{S}$  \\
    WPT \dotand Weyl point trajectory \\
    $\mathscr{Z}$ \dotand Berry-Zak phase \\
    $\mathscr{Z}^\textrm{cont.},\delta\mathscr{Z}^\textrm{cont.}$ \dotand (change of) continuum Zak phase \\
    $\bar{z}_j$, $\bar{z}_j(\bk)$ \dotand eigenvalues of projected position operator \\
    $\zeta = (\zeta_1 , \zeta_2)^\top$ \dotand two-component spinor \\ \hline\hline 
    \end{longtable}
    }

\section{Summary of key terms and notions}
\label{app:index-of-notions}

{\parindent0pt

\setlength{\parskip}{0.6em}

\emph{Angular-momentum anomaly} [cf.~Sec.~\ref{sec:anomaly}, in particular Eq.~\eqref{amanomaly}] means that itinerant angular momenta of faceted bands are non-identical to the itinerant angular momenta of bulk-conduction bands, as well as of bulk-valence bands.

\emph{Balanced orientation rule} [cf.~Eq.~\eqref{eq:app:balanced-pairing}] ensures that the number of upward oriented preimage lines is equal to the number of downwards oriented preimage lines for a given image point on a 2-sphere.

\emph{Band representation} is a representation of a space group $\mathcal{G}$ induced from a representation of a site-stabilizer group $\mathcal{G}_{\br}$ at site $\br$. 
Throughout this work, $\mathcal{G}=\mathrm{P}n$ with $n\in\{2,3,4,6\}$.   

\emph{Band representation induced from an ($\mathcal{L},\br^{\perp}$)-orbital} is a representation of the space group $\mathrm{P}n$ induced from a one-dimensional representation of the site-stabilizer group $\mathrm{P}n_{\br}$. 
The representation is specified by the on-site angular momentum ($\mathcal{L}$) and by the in-plane component of the Wyckoff position ($\br_\perp$). 
The $z$ component of the Wyckoff position remains arbitrary (and is continuously tunable). Such band representations for all values of $z$ form an equivalent class BR$(\mathcal{L},\br_{\perp})$. We note that orbitals generating  BR$(\mathcal{L},\br_{\perp})$ may or may not be one-site-localized; however, each basis band representation is an element of some BR$(\mathcal{L},\br_{\perp})$.

\emph{Basis band representation} BBR$[\varphi_\alpha]$ is a band representation of the space group P$n$ induced from a one-dimensional representation of the site-stabilizer group P$n_{\br}$, specified by a tight-binding-basis (i.e., one-site-localized) orbital $\varphi_\alpha$. 
The orbital is specified by its
on-site angular momentum $\mathcal{L}$ and (all components of) its Wyckoff position $\br$. Sometimes we use \emph{tightly-bound band representation} for the same notion.

\emph{Berry dipole} [cf.~Sec.~\ref{sec:Berry-definition}] is a point-like touching of two energy bands, which acts as a dipolar source of Berry curvature. 
It is characterized by a $2\pi$-quantized Berry flux $\Psi$ through a hemisphere centered at the touching point and with the bounding `equator' lying in a plane perpendicular to the dipole axis. (The compensating flux $-\Psi$ flows through the complementary hemisphere).

\emph{Berry-dipole Hamiltonian}, defined by Eqs.~(\ref{eq:kp-sym-spinor},\ref{eq:zszs-hamiltonian}) is a continuum Hamiltonian that exhibits a band touching point at $(\bk,\phi)=(\boldsymbol{0},0)$, and that facilitates a quantized change of the continuum Hopf number and of the continuum Zak phase given, respectively, by Eqs.~(\ref{eq:Berry-dipole-hopf-change}) and~(\ref{eq:Berry-dipole-pol-change}). 

\emph{Berry-dipole spin} $\Delta\ell$ determines how the Berry-dipole Hamiltonian transforms under rotations [cf.~Eq.~\eqref{eq:rot-sym-Berry-dipole-contin}].

\emph{Bulk-valence Zak phase} $\mathscr{Z}_v$ and \emph{bulk-conduction Zak phase} $\mathscr{Z}_c$ [cf.~Eq.~\eqref{eq:Zak-phase-bulk}] are Berry-Zak phases computed for the bulk valence resp.~the bulk conduction band along a path [$\mathcal{S}(k_z)$] in the BZ; the path has a fixed value of $k_z$ and projects onto a path ($\mathcal{S}$) in rBZ that is used for the computation of the corresponding faceted Zak phase.

\emph{Bulk-like} vs.~\emph{surface-like} polarization bands are bands formed by eigenstates of projected position operators ($P\hat{z}P$ and $Q\hat{z}Q$)
that lie above vs.~below a chosen cutoff band in a semi-infinite geometry with a bottom surface. 
The cutoff band is chosen such that all bulk-like polarization bands of a semi-infinite geometry are (up to exponentially small corrections) indistinguishable from the bulk polarization bands in infinite geometry.

\emph{$C_m$-related-two-segment loop} [cf.~Fig.~\ref{fig:zak-paths}] is a two-segment loop in rBZ, with two segments mapped to each other by a $C_m$ rotation (up to a translation by a reciprocal lattice vector). 
We employ such loops (with a few exceptions listed in Sec.~\ref{sec:zakanomalymain}) as the path $\mathcal{S}$ in rBZ to compute the faceted Zak phase (and to define the Zak-phase anomaly).

\emph{Conditionally robust surface states} [cf.~Sec.~\ref{sec:BBC_sharp}] are surface-localized states associated with delicate topological invariants, which energetically connect the bulk valence and the bulk conduction bands under the condition of open boundaries with sharply terminated hopping. 
The `conditional robustness' means that the energetic connection of the surface states to bulk bands can be lost if the sharp boundary condition is lifted.

\emph{Continuum Hopf number} [cf.~Eq.~\eqref{eq:hopfinvar-cont}]of a continuum model is the integral of the Chern-Simon 3-form for the valence band of a two-band insulating Hamiltonian. 
It is a (non-quantized) analog of the Hopf invariant for continuum Hamiltonians, i.e., with $\bk\in\mathbb{R}^3$.

\emph{Continuum Zak phase} [cf.~Eq.~\eqref{eq:pol-continuum}]is an analog of (electric) polarization of a 1D band defined for continuum 1D Hamiltonians, i.e., with $k\in \mathbb{R}$.

\emph{Crystalline Hopf insultors} are rotation-symmetry-enriched two-band insulating  Hamiltonians with trivial first Chern class. They are characterized by the Hopf invariant \eqref{eq:hopfinvar} and the RTP invariants \eqref{eq:RTP-def}.

\emph{Delicate topological insulators} exhibit bands with topological invariants that can be trivialized through the addition of 
a particular (potentially an arbitrary) basis band representation BBR$[\varphi_\alpha]$
to either the valence or the conduction subspace. Examples of delicate topological insulators include the Hopf and the RTP insulator. 

\emph{Facet} is a crystal surface that retains the discrete translational symmetry of a crystal in two directions parallel to the surface.

\emph{Faceted bands} is an umbrella term for surface-localized energy bands and surface-like polarization bands.

\emph{Faceted Chern number} is a Chern number computed for faceted bands. 
For Hopf insulators, bulk-boundary correspondence equates the faceted Chern number (for arbitrary facet) to the Hopf invariant.
For RTP insulators (with identical itinerant angular momenta for all valence  bands, and similarly for all the conduction bands), there is a mod-$n$ relation [cf.~Eq.~(\ref{eq:C-facet-via-ell-multi-band})] relating the faceted Chern number to the RTP invariant.

\emph{Faceted Zak phase} $\mathscr{Z}_f$ [cf.~Eq.~(\ref{eq:Zak-phase-rbz})] is the Berry-Zak phase computed along a path [$\mathcal{S}$] in reduced BZ for a faceted band.

\emph{Fermi lines} are one-dimensional Fermi surfaces contributed by surface states at crystal facets.

\emph{Grade $\mathfrak{g}$ of a Wannier orbital} [cf.~Sec.~\ref{sec:graded-multicell}] is the smallest number of primitive unit cells that a representative Wannier orbital can be confined to. 
The lower bound is defined given that one is allowed to continuously deform the tight-binding Hamiltonian of fixed matrix dimension (including the Wyckoff positions) while preserving the energy gap and the space group symmetry. The notion of \emph{grading} refines the notion of multicellularity.

\emph{Helicity of a Berry dipole $\upsilon$} is one of the discrete-valued characteristics of a Berry dipole [cf.~Eq.~(\ref{eq:kp-sym-spinor})]. 
Its physical meaning is related to the geometry of preimages of points on the Bloch sphere, as clarified in footnote~\ref{foot:helicity}.

\emph{Hopf-Chern relation} [cf.~Eq.~(\ref{eq:whitehead})] relates the Hopf invariant of a two-band, insulating Hamiltonian to the Chern number computed on an oriented surface stretched over the preimage of any point on the Bloch sphere. 

\emph{Hopf-RTP relation} [cf.~Sec.~\ref{sec:RTP-Hopf}, in particular Eq.~\eqref{eq:hopf-RTP}] is a mod-$n$ constraint on the admissible values of the Hopf and the RTP invariants that has to be obeyed by any two-band, insulating Hamiltonian with P$n$ symmetry.

\emph{Hybrid Bloch-Wannier states} are eigenstates of the projected position operator; they are extended as a Bloch wave in the $(x,y)$-plane and are exponentially-localized as a Wannier function in the $z$ direction. 

\emph{Intra-cell component $\ket{u_j^\alpha(\bk)}$ of a Bloch (eigen)state}  is a vector in the intra-cell Hilbert space per the parameterization of Bloch eigenstates in Eqs.~\eqref{eq:app:Bloch-eigenstate}, (\ref{eqn:Bloch-decompost}). 

\emph{Intra-cell Hilbert space} is spanned by basis vectors 
$\{\varphi_\alpha\equiv\ket{\alpha}\}_{\alpha=1,\ldots, N_c+N_v}$, where $\alpha=1,\ldots, N_v+N_c$ is an index distinguishing linearly-independent orbitals within a primitive unit cell.

\emph{Intra-column function} [cf.~Eq.~\eqref{eq:app:intra-column-orbitals}] is a 1D analogue of Wannier function, defined in the intra-column Hilbert space. It can be viewed as the Wannier component of the hybrid Bloch-Wannier functions.

\emph{Intra-column Hilbert space} is defined at each reduced momentum $\bk_\perp$.
It is spanned by basis vectors $\ket{R_z,\alpha}\drangle$, where $R_z$ runs over all Bravais lattice vectors projected onto the $z$ axis, and $\alpha=1,\ldots, N_v+N_c$ is an index distinguishing linearly-independent orbitals within a primitive unit cell.

\emph{Iso-orbital condition} [cf.~Eq.~\eqref{isoorbital}] holds for a pair of rotation-invariant reduced momenta $\Pi_1$ and $\Pi_2$, if there exist a single basis band representation BBR, such that the restriction of the valence band representation to the two rotation-invariant $\bk$-lines $\gamma_{\Pi_1}$ and $\gamma_{\Pi_2}$ equals the restriction of BBR to these lines.

\emph{Iso-orbital subset} $\{\Pi_j\}_{j=1\dots J}$ is the set of rotation-invariant reduced momenta in rBZ that all fulfill the iso-orbital condition with respect to the same basis band representation BBR$[\varphi_\alpha]$ and also the mutually-disjoint condition.
The set is described as \emph{maximal} if every rotation-invariant reduced momentum satisfying the displayed conditions is included in the subset.

\emph{Itinerant angular momentum} $\widetilde{\mathcal{L}}(\Pi)$ [cf.~Eq.~(\ref{eq:itin-am-from-basis})] determines an eigenvalue $\exp(i2\pi\widetilde{\mathcal{L}}(\Pi)/m)$ of the itinerant rotation matrix of a Bloch Hamiltonian at a $C_m$-rotation-invariant $\bk$-line $\gamma_\Pi$.

\emph{Itinerant rotation matrix} $\widetilde{R}_{C_m}(\Pi)$ [cf.~Eq.~\eqref{eq:itinerant_rotation}] commutes with a Bloch Hamiltonian at a $C_m$-rotation-invariant $\bk$-line~$\gamma_\Pi$.
 
\emph{Multicellularity} [cf.~Sec.~\ref{sec:CHI-multicel}] is an obstruction that forbids valence bands to be representable using Wannier orbitals with support on a single primitive unit cell.

\emph{Mutually-disjoint condition} [cf.~Eq.~(\ref{eqn:mutually-disjoint-multiband})] at rotation-invariant reduced momentum $\Pi$ is fulfilled when the set of itinerant angular momenta of bulk-valence states $\{\widetilde{\mathcal{L}}_{v,j}(\Pi)\}_{j=1}^{N_v}$ is disjoint with the set of itinerant angular momenta of bulk-conduction states $\{\widetilde{\mathcal{L}}_{c,j}(\Pi)\}_{j=1}^{N_c}$. 
For two-band Hamiltonians, $N_v=N_c=1$, the mutually-disjoint condition reduces to Eq.~\eqref{eq:mut-disj}.

$n/m$\emph{-fold multiplet} (or $n/m$\emph{-plet} in short) is a set of $n/m$ number of $C_m$-invariant $\bk$-lines related by $C_n$ rotations.

\emph{On-site angular momentum} $\mathcal{L}$ determines the one-dimensional representation $e^{i2\pi\mathcal{L}/n}$ of the $n$-fold rotation symmetry of a tight-binding-basis orbital.
Note that in this work we only consider placing the basis orbitals at Wyckoff positions with multiplicity one [cf.~Fig.~\ref{fig:wp-rotinvmom}(a)], such that their site-stabilized group $\mathrm{P}n_{\br}$ contains an $n$-fold rotation symmetry.

\emph{Polarization band} in rBZ is a collection of eigenstates of the projected position operator such that their eigenvalues $\bar{z}_j(\bkp)$ depend smoothly on the reduced momentum $\bk_\perp$.
We call the corresponding eigenstates as the \emph{(hybrid) Bloch-Wannier states} (for both the infinite and the faceted geometry).

\emph{Reduced BZ} (rBZ) is the projection of the 3D BZ onto the plane perpendicular to the rotation axis. The 2-component elements $\bk_\perp=(k_x,k_y)$ of rBZ are called \emph{reduced momenta}.

\emph{Representative tight-binding-basis orbitals} $\{\varphi_\alpha\}_{\alpha=1}^{N_v+N_c}$ specify the Hilbert space of a tight-binding Hamiltonian. The complete basis of the Hilbert space is given by applying all Bravais-lattice translations to the representative basis orbitals.

\emph{Returning Thouless pump} (RTP) is an adiabatic evolution of a single-momentum Hamiltonian, characterized by an integer increase of the electric polarization within one half of the adiabatic cycle, which is reverted on the second half of the cycle. 
In the present work, the role of the single momentum (adiabatic parameter) is played by the $k_z$ ($\bk_\perp$) component of the 3D momentum, and the integer quantization is enforced by the mutually-disjoint and the iso-orbital conditions [cf.~Eq.~(\ref{twoconds})].

\emph{RTP invariant} [cf.~Eq.~(\ref{eq:RTP-def})] is an invariant associated with the returning Thouless pump.
It is defined as the difference in the electric polarization between two rotation-invariant $\bk$-lines that obey  the mutually-disjoint and the iso-orbital condition. 

\emph{RTP insulator} is a topological insulators with two or more bands that is characaterized by
a non-zero value of at least one
RTP invariant.

\emph{Site-stabilizer group} $\mathcal{G}_{\br}$ is a subgroup of the space group $\mathcal{G}$ of the Hamiltonian, which preserves a given Wyckoff position~$\br$. 
Throughout this work, $\mathcal{G}=\mathrm{P}n$ with $n\in\{2,3,4,6\}$. 

\emph{Sharp boundary condition} is an open boundary condition with the additional property that all
the hopping matrix elements are either bulk-like (if the hopping process does not cross the boundary) or set to zero (if they cross the boundary).

\emph{Spectrally flattened Hamiltonian} is obtained from a given Hamiltonian by setting all energies above (below) the Fermi energy to $+1$ ($-1$) while keeping its eigenstates unchanged. 
For the case of a two-band insulating Hamiltonian, the spectral flattening is performed by Eqs.~(\ref{eqn:flatten-nonspinor}) or~(\ref{eq:zszs-hamiltonian-normalized}).

\emph{Spinor-form Hamiltonian} is defined via a two-component spinor function, $\zeta(\bk)=\left(\zeta_1(\bk), \zeta_2(\bk)\right)^\top \in \mathbb{C}^2$ per
Eq.~\eqref{eq:zszs-hamiltonian}.

\emph{Symmetry-equivalent band representations} are band representations that belong to the same equivalence class BR$(\mathcal{L},\br_{\perp})$. Note that two band representations that are symmetry-equivalent to each other may not be continuously deformable into each other without increasing the dimension of the Hilbert space. For example, the VBR for crystalline Hopf insulator is a band representation but is not deformable to a BBR of the same equivalence class.

\emph{Symmetry-protected delicate topological insulators} are delicate topological insulators whose topological invariant can be trivialized 
through the addition of some (but is robust against adding other) basis-band representations to either the valence or the conduction subspace. 
The presently discussed paradigm example corresponds to the RTP insulator (but excludes Hopf insulators without rotation symmetry). 

\emph{Tight-binding Hilbert space} is the Hilbert space of a tight-binding Hamiltonian spanned by tight-binding basis vectors $\varphi_{\boldsymbol{R,\alpha}}\equiv\ket{\boldsymbol{R},\alpha}\rangle$, with $\bR$ in the Bravais lattice, and $\alpha=1,\ldots,N_v+N_c$ an index distinguishing linearly-independent orbitals within a primitive unit cell.

\emph{Tightly-bound band representation} is synonymous with \emph{basis band representation} (cf.~corresponding entry above). This alternative  term emphasizes that an orbital, inducing this bands representation, is localized to a single site, in other words is \textit{one-site-localized}. 

\emph{Valence band representation} VBR is a band representation corresponding to
the valence band. More precisely, it is the band representation induced by Wannier orbitals placed at an appropriate Wyckoff position.

\emph{Wannier cut} [cf.~Appendix~\ref{app:wannier-cut}] is a tool to construct projector onto faceted bands.

\emph{Wannier orbitals}  $\big|{W_{j,\bR}}\big\rangle\big\rangle$ [cf.~Eq.~(\ref{blochwannierfourier})], are obtained by inverse Fourier transform of Bloch states $\big|{\psi_{j,\bk}(\bR,\alpha)}\big\rangle\big\rangle$, with the Bloch states depending smoothly and periodically
on $\bk$, and the intra-cell components of the Bloch states \Big($\big|{u_j(\bk)}\big\rangle_{j=1}^{N_v+N_c}$\Big) being orthonormal at each $\bk$.
They are also dubbed \emph{$\mathpzc{on}$-Wannier orbitals} \big(labelled $ \big|{W_{j,\bR}^\mathpzc{on}}\big\rangle\big\rangle$\big) in the context of Sec.~\ref{sec:CHI-multicel}, where we distinguish them from \emph{$\mathpzc{og}$-Wannier orbitals} $ \big|{W_{j,\bR}^\mathpzc{og}}\big\rangle\big\rangle$; in constructing the latter, one drops the normalization condition on the intra-cell components of the Bloch states [cf.~Eq.~(\ref{eqn:og-wannier})].

\emph{Wannier-$\mathpzc{type}$ orbitals} are obtained by inverse Fourier transform of a potentially overdetermined set of Bloch states. In contrast to $\mathpzc{on}$-Wannier orbitals, here one drops both the normalization and the orthogonality condition for the intra-cell components of the Bloch states; instead, we only require $\big|{u_j(\bk)}\big\rangle_{j=1}^{N}$ (with $N \geq N_v + N_c$) to be a complete basis of the intra-cell Hilbert space.

\emph{Wyckoff position} is a point in the unit cell at which the representative basis orbital is centered. While Wyckoff position may come in symmetry related multiplets (per primitive unit cell), in the present work we only need to consider Wyckoff positions with multiplicity one [cf.~Fig.~\ref{fig:wp-rotinvmom}(a)].

\emph{Zak-phase anomaly} [cf.~Sec.~\ref{sec:BBC_zak}] is a signature associated with a non-triviail RTP invariant at open boundaries, namely, there exist faceted states for which the Berry-Zak phase (quantized by rotation symmetry to certain rational multiples of $2\pi$) along a certain symmetrically-chosen loop ($\mathcal{S}$) is distinct from the corresponding Berry-Zak phase of both the bulk-valence and the bulk-conduction bands. 

}

\section{Preliminaries on strict tight-binding formalism and notations}\label{app:preliminaries}

We start a series of pedagogical appendices to familiarize the reader with basic band-theoretic concepts in the strict tight-binding formalism, as well as establish our band-theoretic notations. 
After introducing the formalism in Appendix~\ref{app:tb-formalism}, we comment in Appendix~\ref{app:nonperiodic-convention} on differences of the periodic vs.~non-periodic conventions for the $\bk$-dependent tight-binding Hamiltonian $h(\bk)$.
In Appendix~\ref{app:geom-theo-pol}, we derive the tight-binding analogue of the geometric theory of polarization, namely that the (geometric phase/$2\pi$) computed for a given band is equal, modulo a Bravais lattice vector, to the real-spatial positional center of the corresponding Wannier orbital within a unit cell.
In Appendix~\ref{app:Hopf-periodic-nonperiodic} we argue that the Hopf invariant is independent of whether $h(\bk)$ is expressed in the periodic or in the non-periodic convention.

\subsection{Strict tight-binding formalism\label{app:tb-formalism}}

In the tight-binding formalism, the \textit{tight-binding Hilbert space} $\mathscr{H}$ is spanned by tight-binding-basis orbitals 
$\varphi_{\boldsymbol{R,\alpha}}\equiv\ket{\boldsymbol{R},\alpha}\rangle$, with $\bR$ in the Bravais lattice, and $\alpha=1,\ldots,N_v+N_c$ an index distinguishing linearly-independent orbitals within a primitive unit cell. 
It is also useful to introduce a smaller, \textit{intra-cell Hilbert space} $\mathscr{H}_\textrm{cell}$ spanned by basis orbitals 
$\{\varphi_\alpha\equiv\ket{\alpha}\}_{\alpha=1,\ldots,N_c+N_v}$. 
To distinguish the states living in the larger vs.~the smaller Hilbert space, we denote them with different bra/ket symbols, i.e.,
\begin{equation}
\ket{A}\rangle \in \mathscr{H}\qquad\textrm{vs.}\qquad\ket{A}\in\mathscr{H}_\textrm{cell}
\end{equation}
[compare also to Eq.~(\ref{eqn:column-elements})].
The corresponding inner products are similarly denoted as $\langle\braket{A}{B}\rangle$ vs.~$\braket{A}{B}$.

To define the Bloch eigenstates in the tight-binding formalism, we first introduce a $\bk$-dependent basis
that spans the same Hilbert space $\mathscr{H}$:
\begin{equation}
    \ket{\bk,\alpha}\rangle=\sum_{\bR} e^{i\bk\cdot(\bR+\br_\alpha)}\ket{\bR,\alpha}\rangle,
    \label{eq:app:Bloch-basis} 
\end{equation}
with coefficients that input the spatial positions of the basis orbitals $\br_\alpha$ within a representative unit cell. 
We elaborate more on this choice of the $\bk$-dependent basis and on its consequences in a later Sec.~\ref{app:nonperiodic-convention} of this Appendix. 
In the introduced $\bk$-dependent basis we can define the Bloch eigenstates to be 
\begin{equation}
    \big|{\psi_{j,\bk}}\big\rangle\big\rangle=\sum_\alpha u_j^\alpha(\bk)\ket{\bk,\alpha}\rangle,
    \label{eq:app:Bloch-eigenstate}
\end{equation}
where the coefficients in front of the basis states form the eigenvectors of the Hamiltonian matrix 
\begin{equation}
    h(\bk)\big|{u_j(\bk)}\big\rangle=E_j(\bk) \big|{u_j(\bk)}\big\rangle.\label{eqn:Bloch-decompost}
\end{equation}
Note that both the tight-binding Hamiltonian and its eigenvectors $\big|{u_j(\bk)}\big\rangle$ are defined in the intra-cell Hilbert space $\mathscr{H}_\textrm{cell}$ and, as we discuss later in Appendix~\ref{app:nonperiodic-convention}, are $\bk$-aperiodic for generic values of $\br_\alpha$. 
When projecting the Bloch eigenstate onto a tight-binding-basis orbital, we get a tight-binding Bloch function
\begin{equation}
    \psi_{j,\bk}(\bR,\alpha)=\Big\langle\!\bra{\bR,\alpha}\ket{\psi_{j,\bk}}\!\Big\rangle= e^{i\bk\cdot(\bR+\br_\alpha)}u_j^\alpha(\bk), \label{eqn:tight-binding-Bloch-function}
\end{equation}
which is the product of  a plane wave with an intra-cell wave function $u_j^\alpha(\bk)$, in close analogy with the continuum Bloch function. 
The Bloch function is $\bk$-periodic, as the non-periodicity of the cell-periodic function and the plane wave cancel each other. 

It is also useful to define a \emph{Wannier function} (also dubbed \emph{Wannier orbital}, which we in particular adopt in Sec.~\ref{sec:CHI-multicel}), which is the inverse Fourier transform of a Bloch function: 
\begin{align}
    W_{j,\bar{\bR}}(\bR,\alpha)\eq\int\frac{d^3\bk}{|\textrm{BZ}|}e^{-i\bk\cdot\bar{\bR}}\psi_{j,\bk}(\bR,\alpha) \lin
    \eq \int\frac{d^3\bk}{|\textrm{BZ}|}e^{-i\bk\cdot(\bar{\bR}-\bR-\br_\alpha)}u^\alpha_j(\bk),\la{blochwannierfourier}
\end{align}
where $|\textrm{BZ}|$ denotes the volume of the Brillouin zone. 
We assume throughout that the first Chern class is trivial, hence the Bloch function can be chosen to be both analytic and periodic, and correspondingly the Wannier function would be exponentially-localized\footnote{
This is true whether or not the Hopf invariant is nontrivial in the particular case of two-band Hamiltonians.} in real space, with a positional center in the vicinity of unit cell labelled $\bar{\bR}$. 
With the normalization\footnote{
For the discussion of graded multicellularity in Sec.~\ref{sec:graded-multicell}, we found it useful to introduce several refined version of Wannier functions, with the ones introduced in the present paragraph dubbed therein as \emph{orthonormal (`$\mathpzc{on}$-') Wannier orbitals}, $\big|{W_{j,\bR}^\mathpzc{on}}\big\rangle\big\rangle$.
If one instead adopts intra-cell wave functions $\big|{u_j^\mathpzc{og}(\bk)}\big\rangle$ that are orthogonal but not normalized to unity, their inverse Fourier transform results in \emph{orthogonal (`$\mathpzc{og}$-') Wannier orbitals}, $\big|{W_{j,\bR}^\mathpzc{og}}\big\rangle\big\rangle$.
A further level of relaxation corresponds to taking a (potentially overdetermined) set $\left\{\big|{u_j^\mathpzc{type}(\bk)}\big\rangle\right\}_{j=1}^{N}$ with $N\geq N_v+N_c$ of intra-cell wave-functions that vary smoothly with $\bk$ and constitute a basis of $\mathscr{H}_\textrm{cell}$; their inverse Fourier transform generates \emph{Wannier-$\mathpzc{type}$ orbitals}, $\big|{W_{j,\bR}^\mathpzc{type}}\big\rangle\big\rangle$~\cite{read_compactwannier}.
} $\big\langle{u_j(\bk)}\big|{u_{j'}(\bk)}\big\rangle=\delta_{jj'}$ it follows that $\big\langle\big\langle{W_{j,\bR}}\big|{W_{j',\bR'}}\big\rangle\big\rangle=\delta_{jj'}\delta_{\bR,\bR'}$ (a product of Kronecker delta functions).

The Wannier functions are not uniquely defined and depend on the gauge choice of the Bloch function. 
The allowed gauge transformations are $u_j^\alpha(\bk)\to e^{i\beta_j(\bk)}u_j^\alpha(\bk)$, where $\beta_j(\bk)$ is a real-analytic function that transforms under reciprocal lattice translations as $\beta_j(\bk+\bG)=\beta_j(\bk)+2\pi l$ with $l\in\mathbb{Z}$. 
When $\beta(\bk)$ does not wind around the BZ, meaning $l=0$, the spread of the Wannier function may change, but its positional center remains fixed. 
If, on contrary, $\beta(\bk)$ winds in the direction of $\bG_1$ with $l\neq0$ (so-called \emph{large guage transformation}), the center of the Wannier function will shift in the direction of the corresponding lattice vector $\bR_1$ (with $\bG_1\cdot\bR_1=1$) by $l$ unit cells.

Sometimes, e.g. in Appendices \ref{app:pol-position} and \ref{app:ang-mom-anomaly}, it is convenient to consider reduced momenetum $\bk_\perp$ as an external parameter and work in
an \emph{intra-column}\footnote{The term `column' indicates that we restrict ourselves to a specific value of reduced momentum and define the basis indexed by all lattice vectors in the $z$ direction.} \emph{Hilbert space} 
$\mathscr{H}_{\mathrm{column}}$, which is defined at each reduced momentum $\bk_\perp$.
Elements of this space are indicated with a triple angled bracket, i.e.,
\begin{equation}
 \ket{A}\drangle \in \mathscr{H}_\textrm{column}, \label{eqn:column-elements}  
\end{equation} 
and we label the orthonormal basis states in the intra-column Hilbert space as $\ket{R_z,\alpha}\drangle$. 
At each $\bk_\perp$, the 1D analog of a Wannier function in the intra-column Hilbert space is (an \textit{intra-column function}, for brevity) is given by 1D Fourier transform of the Bloch function and subsequent stripping of the plane-wave factor $e^{ik_x(R_x+x_\alpha)+ik_y(R_y+y_{\alpha})}$:
\begin{align}
    \mathcal{V}_{j,\bk_\perp,\bar{R_z}}(R_z,\alpha)&=\Big\langle\Big\langle\!\braket{R_z,\alpha}{\mathcal{V}_{j,\bk_\perp,\bar{R_z}}}\!\Big\rangle\Big\rangle \lin
    &=\int\!\frac{dk_z}{2\pi}e^{-ik_z\bar{R}_z}e^{ik_z(R_z+z_\alpha)}u_j(\bk).
    \label{eq:app:intra-column-orbitals}
\end{align}
Because there is no obstruction to find a global smooth gauge on a base space that is topologically
equivalent to a circle, the Bloch function can always be chosen to be periodic and analytic in $k_z$, such that the intra-column function decays exponentially with $|R_z-\bar{R}_z|$. 
The intra-column function can be viewed as the Wannier component of the hybrid Bloch-Wannier functions (defined in the full Hilbert space $\mathscr{H}$) 
that is exponentially localized in one direction ($z$), but remains extended as a Bloch function in the other two directions.
To exemplify an application of the hybrid Bloch-Wannier functions, 2D-translation-invariant eigenstates of the projected position operators $P_b\hat{z}P_b$ (cf.~Sec.~\ref{sec:anomalyprecise}) have wave functions that are hybrid Bloch-Wannier functions.

\subsection{Periodic vs.~non-periodic convention for tight-binding Hamiltonian\label{app:nonperiodic-convention}}

In this part of the Appendix we comment on the choice of prefactor that we made to construct the $\bk$-dependent basis in Eq.~\eqref{eq:app:Bloch-basis} of the tight-binding Hilbert space. 
In the chosen convention the Hamiltonian is generally aperiodic-in-$\bk$ for generic basis orbitals positions $\br_\alpha$; in particular, it transforms under translation by a reciprocal lattice vector $\bG$ as
\e{ h(\bk+\bG) = V(\bG)^{-1}h(\bk)V(\bG),}
where 
\begin{equation}
[V(\bG)]_{\alpha\beta}=\exp(i\bG\cdot\br_\alpha)\delta_{\alpha\beta}.
\end{equation} 
From this we conclude aperiodicity of the corresponding eigenvectors
\e{ u_j^\alpha(\bk+\bG) =e^{-i\bG\cdot \br_{\alpha}}u_j^\alpha(\bk),
\label{eq:app:Bloch-aperiodicity} }
assuming a periodic gauge of the Bloch functions $\psi_{j,\bk}$.
As we show in the next Appendix~\ref{app:geom-theo-pol}, only in this convention the positions of the basis orbitals are properly taken into account.

An alternative convention deals with strictly $\bk$-periodic Hamiltonians and intra-cell wave functions by choosing the $\bk$-dependent basis to be
\begin{equation}
    |\widetilde{\bk,\alpha}\rangle\rangle=\sum_{\bR} e^{i\bk\cdot\bR}\ket{\bR,\alpha}\rangle.
    \label{eq:app:Bloch-basis-periodic}
\end{equation}
In this case the Hamiltonian and the intra-cell wave functions are periodic in $\bk$:
\e{ &\tilde{h}(\bk)\tilde{u}_j(\bk)=E_{j}(\bk)\tilde{u}_j(\bk), \lin
&\tilde{h}(\bk+\bG)=\tilde{h}(\bk), \as \tilde{u}_j(\bk+\bG)=\tilde{u}_j(\bk). 
\label{eq:app:Hamilt-periodic}}
The Hamiltonians and the intra-cell wave functions in the two conventions are related by
\begin{align}
    h(\bk)\eq V(\bk)^{-1}\tilde{h}(\bk)V(\bk), \label{eq:app:conv-rel-ham} \\
    u_j(\bk)\eq V(\bk)^{-1} \tilde{u}_j(\bk), \label{eq:app:conv-rel-fun}
\end{align}
where in Eq.~\eqref{eq:app:conv-rel-fun} we assumed a periodic gauge of the Bloch functions $\psi_{j,\bk}$.

\subsection{{Tight-binding analog of the geometric theory of polarization}\label{app:geom-theo-pol}}

We show here that the normalized, Brillouin-zone-integral of the Berry connection is related, modulo a Bravais-lattice vector, to the expectation value of the discrete position operator: 
\e{\int \f{d^3\bk}{|\textrm{BZ}|} \braket{u(\bk)}{i\nabk u(\bk)}=_{\bR}
\big\langle\braopket{W_{\bze}}{\boldsymbol{\mathcal{R}}}{W_{\bze}}\big\rangle. \la{berrypol}}
where `$=_{\bR}$' means `equal modulo Bravais lattice vector.
In the tight-binding basis, the discrete position operator can be decomposed as: 
\e{ \boldsymbol{\mathcal{R}}=\sum_{\bR,\alpha}|\bR,\alpha\rangle\rangle(\bR+\br_{\alpha})\langle\langle\bR,\alpha|, }
with $\br_{\alpha}$ an intra-cell position coordinate of the orbital $\alpha$. 
To simplify the notation, we dropped the band index $j$ from both the Bloch and the Wannier state.

The analog of \q{berrypol} for Bloch/Wannier functions defined over continuous real space is well-known \cite{Zak_bandcenter,kingsmith_polarization}, but \q{berrypol} itself carries certain subtleties inherited from the tight-binding formalism that we want to spell out in detail. 
In particular, it is not well-appreciated that \q{berrypol} only holds with $u^\alpha(\bk)$ satisfying the transformation relation given by \eqref{eq:app:Bloch-aperiodicity},
which implies that $u$ is aperiodic in reciprocal-lattice translations for generic values of $\br_{\alpha}$. 

It is instructive to compare our result with the common practice of working  with strictly $\bk$-periodic Hamiltonians, defined in Eq.~\eqref{eq:app:Hamilt-periodic}.
Then the analog of \q{berrypol} 
is 
\e{\int \f{d^3\bk}{|\textrm{BZ}|} \braket{\tilde{u}_{\bk}}{i\nabk \tilde{u}}=_{\bR}
\big\langle\braopket{W_{\bze}}{\tilde{\boldsymbol{\mathcal{R}}}}{W_{\bze}}\big\rangle \la{berrypol2}}
with a modified version of the position operator,
\e{ \tilde{\boldsymbol{\mathcal{R}}}=\sum_{\bR,\alpha}\ket{\bR,\alpha}\rangle\bR\langle\bra{\bR,\alpha}, }
The crucial difference is that 
$\tilde{\boldsymbol{\mathcal{R}}}$ carries no information of where a basis Wannier orbital is located within a unit cell. 
Alternatively stated, to develop a tight-binding analog of the geometric theory of polarization [as encapsulated by its key equation \q{berrypol}] that properly accounts for the (possibly distinct) spatial embedding of the orbitals within a unit cell, it is necessary to utilize the (generically) $\bk$-aperiodic tight-binding Hamiltonian.

Plugging in the definition of the Wannier function [Eq.~\eqref{blochwannierfourier}] and explicitly allowing for a gauge transformation $u^\alpha(\bk)\to e^{i\beta(\bk)}u^\alpha(\bk)$, the proof of \q{berrypol} follows:
\begin{align} 
&\phantom{=}\big\langle\braopket{W_{\bze}}{\boldsymbol{\mathcal{R}}}{W_{\bze}}\big\rangle
= \sum_{\bR,\alpha}\overline{W_{\boldsymbol{0}}(\bR,\alpha)}(\bR+\br_\alpha) W_{\boldsymbol{0}}(\bR,\alpha) = \lin
\eq  \int \f{d^3\bk d^3\bk'}{|\textrm{BZ}|^2}\sum_{\bR,\alpha}
\overline{u^{\alpha}(\bk)}e^{-i\beta(\bk)-i\bk(\bR+\br_\alpha)}(\bR+\br_{\alpha}) \lin
& \qquad\qquad\qquad\times e^{i\bk'(\bR+\br_\alpha)+i\beta(\bk')}u^\alpha(\bk')
\lin \eq  \int \f{d^3\bk d^3\bk'}{|\textrm{BZ}|^2}\sum_{\bR,\alpha}\overline{u^\alpha(\bk)}e^{-i\beta(\bk)-i\bk \cdot(\bR+\br_{\alpha})}[-i\boldsymbol{\nabla}_{\bk'}e^{i\bk' \cdot(\bR+\br_{\alpha})}] \lin
& \qquad\qquad\qquad\times e^{i\beta(\bk')}u^\alpha(\bk').\label{eqn:long-derivation-1}
\end{align}
Because $\psi_{\bk'}$ and any gauge function $e^{i\beta(\bk')}$ are periodic-in-$\bk'$ by construction, one may integrate by parts without acquiring a boundary term:
\e{(\mathrm{\ref{eqn:long-derivation-1}})&=\int \f{d^3\bk d^3\bk'}{|\textrm{BZ}|^2}\sum_{\bR,\alpha}\overline{u^\alpha(\bk)}e^{-i\bk(\bR+\br_{\alpha})-i\beta(\bk)}e^{i\bk' (\bR+\br_{\alpha})+i\beta(\bk')} \lin 
& \qquad\qquad\quad\times\left[i\boldsymbol{\nabla}_{\bk'}u^\alpha(\bk')-u^\alpha(\bk')\boldsymbol{\nabla}_{\bk'}\beta(\bk')\right].\label{eqn:long-derivation-still-goes-on}
}
After carrying out the sum over $\bR$ with the Dirac-delta-function identity:
\e{ \sum_{\bR}e^{i\bk\cdot \bR} = |\textrm{BZ}|\,\updelta(\bk),}
we get 
\begin{align}
    (\mathrm{\ref{eqn:long-derivation-still-goes-on}}) =\int \f{d^3\bk}{|\textrm{BZ}|}\left[\braket{u(\bk)}{i\boldsymbol{\nabla}_\bk u(\bk)}-\boldsymbol{\nabla}_\bk\beta(\bk)\right],
\end{align}
where the first term gives the desired relation \eqref{berrypol} and the second term introduces a modulo-$\bR$ ambiguity, as the gauge phase $\beta(\bk)$ can (in principle) wind nontrivially over the Brillouin torus.
[The proof of \q{berrypol2} is closely 
analogous to what we have just presented.]

\subsection{Equivalence of the Hopf invariant in periodic and non-periodic tight-binding Hamiltonian\label{app:Hopf-periodic-nonperiodic}}

In this Appendix we show that the Hopf invariant [Eq.~\eqref{eq:hopfinvar} in the main text] is 
quantized in both the periodic and the non-periodic tight-binding Hamiltonian convention, and that it takes the same value in both conventions: $\chi[u(\bk)]=\chi[\tilde{u}(\bk)]$. 

As was shown by Pontryagin~\cite{pontrjagin_classification}, the Hopf invariant classifies continuous maps from a three-torus $T^3$ to a two-sphere $S^2$ when the Chern numbers computed on all 2D sub-tori are zero. 
As was later shown by Moore \emph{et al.}~in~\cite{Hopfinsulator_Moore}, such a map can be naturally realized as a two-band insulating Hamiltonian in three dimensions and with a trivial Chern class. 
To ensure the continuity and periodicity of the map over the BZ, the Hamiltonian was considered in the periodic Bloch convention. 
In this case, the authors of Ref.~\cite{Hopfinsulator_Moore} proposed a formula for the Hopf invariant as integral of the Chern-Simons three-form of the Berry gauge field [cf.~Eq.~\eqref{eq:hopfinvar}]. 

In order to appreciate that the same formula is valid in the non-periodic convention, we consider a continuous interpolation between Hamiltonians in the periodic and non-periodic conventions and show that the Hopf invariant is quantized throughout the interpolation, and therefore does not change. 
The desired interpolation parametrized by a real variable $t\in[0,1]$ is given by 
\begin{equation}
    h(\bk,t)=V(t\bk)^{-1}\tilde{h}(\bk)V(t\bk),    \label{eq:app:per-nonper-interpol}
\end{equation}
which satisfies $h(\bk,0)=\tilde{h}(\bk)$ (periodic) and  $h(\bk,1)=h(\bk)$ (non-periodic) with $[V(t\bk)]_{\alpha\beta}=\exp(it\bk\cdot\br_\alpha)\delta_{\alpha\beta}$. 
One can view this interpolation as a continuous shift of the basis orbitals from the spatial origin within the unit cell to their actual positions $\br_{\alpha}$. 
Throughout this interpolation, both the energy spectrum and the translational symmetry are invariant; the spectral invariance implies that the assumed gap does not~close.

Crucially, the Chern-Simons form $\boldsymbol{\mathcal{F}}\cdot\boldsymbol{\mathcal{A}} \,d^3\bk$ remains $\bk$-periodic for any $t\in [0,1]$.
To see this, take into account the relation between intra-cell wave functions throughout the interpolation
\begin{equation}
    u(\bk,t) =  V(t\bk)^{-1} \tilde{u}(\bk).
\end{equation}
This implies that the Berry connections are related as
\begin{eqnarray}
\boldsymbol{\mathcal{A}}[u(\bk,t)] &=&  \braket{u(\bk,t)}{i\nabk u(\bk,t)} \nonumber \\
&=& \boldsymbol{\mathcal{A}}[\tilde{u}(\bk)]+ t\braopket{\tilde{u}(\bk)}{\hat{\br}}{\tilde{u}(\bk)}, \end{eqnarray}
where $[\hat{\br}]_{\alpha\beta}=\br_\alpha\delta_{\alpha\beta}$ is the matrix with positions of the basis orbitals on the diagonal. The periodicity of $\ket{\tilde{u}(\bk)}$ implies that $\boldsymbol{\mathcal{A}}[u(\bk,t)]$ is periodic-in-$\bk$ for all $t\in[0,1]$, and hence also its curl $\boldsymbol{\mathcal{F}}$ (the Berry curvature). 
Therefore the Chern-Simons form is periodic throughout interpolation.

To derive the difference in the Hopf invariant (computed as an integral of the Chern-Simons form) between the two conventions, we adopt the notation of differential forms:
we take
$\mathcal{A}=\boldsymbol{\mathcal{A}}\cdot d\bk$ to be the Berry connection one-form, and its exterior derivative $\mathcal{F}=d\mathcal{A}$ to be the Berry curvature two-form. 
The Chern-Simons three-form in these notations can be simply written as $\mathcal{F}\!\wedge\! \mathcal{A}$ [see Eq.~(\ref{eqn:Hopf-on-S3}) and footnote~\ref{foot:Berry-forms} for a componentwise definition of the wedge product]. 
Now we can write the difference in the Hopf invariants as
\begin{align}
    &\chi[u(\bk)]-\chi[\tilde{u}(\bk)]\lin 
    &\propto\int\limits_{\textrm{BZ}} \mathcal{F}(\bk,1)\!\wedge\! \mathcal{A}(\bk,1)- \int\limits_{\textrm{BZ}} \mathcal{F}(\bk,0)\!\wedge\! \mathcal{A}(\bk,0) \lin 
    &= \int\limits_{\partial \left(\textrm{BZ}\times[0,1]\right)}\mathcal{F}(\bk,t)\!\wedge\! \mathcal{A}(\bk,t) \lin 
    &= \int\limits_{\textrm{BZ}\times[0,1]}\mathcal{F}(\bk,t)\!\wedge\! \mathcal{F}(\bk,t),\la{justproven}
\end{align}
where in the second row we used the just-proven fact that the Chern-Simons form is periodic-in-$\bk$, and hence its integration over two BZs at $t=0$ and $t=1$ taken with opposite orientations can be enlarged to integration over the boundary of $\textrm{BZ}\times[0,1]$. 
In the third row we applied the Stokes' theorem and expressed the external derivative of the Chern-Simons form $d(\mathcal{F}\wedge\mathcal{A})=\mathcal{F}\wedge\mathcal{F}$, which is proportional to
the \emph{second Chern character} $\textrm{ch}_2(\mathcal{F})=(1/2!)(i/2\pi)^2 \mathcal{F}\wedge \mathcal{F}$~\cite{stone_mathematicsforphysics}. Importantly, for two-band Hamiltonians the second Chern character vanishes (the proof of this statement is postponed to Appendix~\ref{sec:secondchern}), which allows us to conclude that the Hopf invariant does not change and hence remains quantized throughout the interpolation. 
We conclude that
Eq.~\eqref{eq:hopfinvar} can be used to compute the Hopf invariant in both Bloch conventions.

\section{Spinor function \texorpdfstring{$\zeta(\bk)$}{z(k))} in 
\texorpdfstring{$\bk{\cdot}\bp$}{kp} 
expansion around a rotation-invariant momentum}\label{app:kp-spinor}

In this Appendix, we derive a rotation-symmetric $\bk{\cdot}\bp$ Hamiltonian of the spinor form specified by Eq.~\eqref{eq:zszs-hamiltonian}. 
This is achieved in two ways.

First, we derive the $\bk{\cdot}\bp$ Hamiltonian assuming the continuous rotation symmetry, which restricts the Hamiltonian as specified in Eq.~\eqref{eq:rot-sym-Berry-dipole-contin}. 
This discussion is organized into three parts. 
First, in Appendix~\ref{app:sym-cond-spinor} we translate the symmetry requirement of the Hamiltonian matrix $h^{\bk\cdot\bp}(\bk)$ to a condition on the spinor $\zeta(\bk)$. 
In Sec.~\ref{sec:general-spinor-form} we derive the general symmetry-compatible form of the spinor $\zeta(\bk)$ to the lowest order in each momentum component and in the tunable parameter $\phi$. 
Then, in Appendix~\ref{app:spinor-cont-deform} we show how the obtained general form can always be continuously deformed into the simple form shown in Eq.~(\ref{eq:kp-sym-spinor}) of the main text. 

Second, we derive the symmetry-allowed $\bk{\cdot}\bp$ Hamiltonian as a Taylor expansion of a lattice Hamiltonian around a rotation-invariant momentum. 
For this, in Appendix~\ref{app:kp-sym-nonuniax}, we show that the $\bk{\cdot}\bp$ Hamiltonian is symmetric under discrete rotation symmetry, specified by the itinerant rotation matrix of the lattice Hamiltonian at the considered rotation-invariant momentum.
Finally, in Appendix~\ref{app:spinor-lattice-sym}, we derive the general form of the spinor with discrete rotation symmetry.

\subsection{Equivalent symmetry condition on the spinor function\label{app:sym-cond-spinor}}

We first show that if the symmetry condition in Eq.~\eqref{eq:rot-sym-Berry-dipole-contin} is applied to spinor-form Hamiltonians [Eq.~\eqref{eq:zszs-hamiltonian}], then this is equivalent to the following condition on the spinor function:
\begin{equation}
    R_\theta\zeta(\bk)=e^{i\beta(\bk)}\zeta(\theta\bk),\quad R_\theta=e^{-i(\theta\Delta\ell/2)\sigma_z},
    \label{eq:app:sym-condition-spinor}
\end{equation}
where $\beta$ is an arbitrary $\bk$-dependent phase. 
To show this, recall that conjugation of a Pauli matrix $\sigma_i$ with $\exp(i\omega\sigma_z)$ gives
\begin{align}
    &\exp(i\omega\sigma_z)\sigma_i\exp(-i\omega\sigma_z)= T(2\omega)_{ij}\sigma_j \lin 
    &\textrm{with}\quad T(\omega)=
    \begin{pmatrix}
        \cos\omega & -\sin\omega & 0 \\
        \sin\omega & \cos\omega & 0 \\
        0 & 0 & 1
    \end{pmatrix}. \label{eq:app:rotation-on-sigma-0}
    \end{align}
By choosing the particulate value $\omega = -\theta\Delta\ell/2$ we obtain $R_\theta\sigma_i R_\theta^{-1}=T(-\theta \Delta\ell)_{ij}\sigma_j$. 
The application of this relation to the spinor-form Hamiltonian in Eq.~(\ref{eq:zszs-hamiltonian}) gives
\begin{align}
    R_\theta h(\bk)R_\theta^{-1} \eq [\zeta^\dagger\sigma_i\zeta]T(-\theta\Delta\ell
    )_{ij}\sigma_j \lin
    \eq [\zeta^\dagger  T(\theta\Delta\ell
    )_{ji}\sigma_i\zeta]\sigma_j \lin
    \eq [\zeta^\dagger R_\theta^{-1}\boldsymbol{\sigma}R_\theta\zeta]\cdot\boldsymbol{\sigma}\Big\rvert_{\bk}, 
    \label{eq:app:ham-condition-rewrite}
\end{align}
where the last line reminds us that the function $\zeta$ (in all of the consecutive expressions) is to be evaluated at $\bk$. 
Separately, the symmetry condition in Eq.~\eqref{eq:rot-sym-Berry-dipole-contin} applied to the spinor-form Hamiltonian in Eq.~(\ref{eq:zszs-hamiltonian}) implies 
\begin{equation}
R_\theta h(\bk)R_\theta^{-1} = [\zeta^\dagger\boldsymbol{\sigma}\zeta]\cdot\boldsymbol{\sigma}\Big\rvert_{\theta\bk},\label{eq:rewrite-split}
\end{equation}
i.e., evaluated at $\theta\bk$.
The relation in Eq.~\eqref{eq:app:sym-condition-spinor} follows from comparing Eqs.~(\ref{eq:app:ham-condition-rewrite}) and~(\ref{eq:rewrite-split}) after considering that 
\begin{equation}
    \zeta^\dagger\bm{\sigma}\zeta=\xi^\dagger\bm{\sigma}\xi\;\Leftrightarrow\;\zeta=\xi e^{i\beta}
    \label{eq:app:ham-to-spinor}
\end{equation}
with $\beta$ an arbitrary phase factor (compare also to the first part of footnote~\ref{foot:rescale-spinor}). 
The equivalence in Eq.~(\ref{eq:app:ham-to-spinor}) follows because $\zeta^\dagger \bm{\sigma}\zeta$ can be interpreted as the expectation value of spin for spinor $\zeta$, and the expectation value of all spin components can agree for a pair of spinors only if  they differ by an unobservable phase factor.

\subsection{General form of the spinor function}\label{sec:general-spinor-form}

To ensure that $h(\bk)$ is analytic in $\bk$, we impose that $\zeta(\bk)$ is also analytic and hence admits convergent Taylor expansions in $\bk$.
Given the symmetry condition \eqref{eq:app:sym-condition-spinor} on the spinor function $\zeta(\bk)$, we aim to find the most general symmetry-compatible choice of the spinor up to linear order in $(k_x,k_y,k_z)$, with the motivation that $\bk=\bze$ marks a prospective
band-touching point. 
The zeroth-order-in-$\bk$ term is parametrized by an additional real parameter $\phi$, such that  $\zeta(\bk=\bze,\phi=0)=0$:
\begin{equation}
    \zeta^{(0)}(\bk,\phi)\!=\!\begin{pmatrix}
		A_1k_+ + B_1k_- + C_1(k_z+i\phi)+ D_1(k_z-i\phi) \\
		A_2k_+ + B_2k_- + C_2(k_z+i\phi)+ D_2(k_z-i\phi)
	\end{pmatrix}\!,\;\hspace{-1.5mm}
	\label{eq:app:spinor-ansatz-linear}
\end{equation}
where $k_\pm=k_x\pm ik_y$, and $A_j$, $B_j$, $C_j$, $D_j$ ($j=1,2$) are complex coefficients. We derive constraints on these coefficients by plugging the ansatz into Eq.~\eqref{eq:app:sym-condition-spinor}. Note that one side of the equation is multiplied by the phase $e^{i\beta(\bk)}$, which is required to be analytic for $\zeta(\bk)$ to be also analytic. In the following, we assume it to be $\bk$-independent, $\beta(\bk)=\beta$, as all $\bk$-dependent terms will contribute to higher orders in $(\bk,\phi)$. Since the variables $k_\pm$, $k_z\pm i\phi$ are independent of each other, we proceed by equating the coefficients that multiply them on the left vs.~the right side of the equation.\bigskip  \\
We proceed in steps as follows.\medskip \\
\noindent{(\emph{i})} Equate the coefficients in front of $k_z+i\phi$ and $k_z-i\phi$, get:
\begin{align}
	& C_1e^{-i\theta\Delta\ell/2}=C_1e^{i\beta},
	& C_2e^{i\theta\Delta\ell/2}=C_2e^{i\beta}. \label{eq:kzm_cond} \\
	& D_1e^{-i\theta\Delta\ell/2}=D_1e^{i\beta},
	& D_2e^{i\theta\Delta\ell/2}=D_2e^{i\beta}. \label{eq:kz-m_cond}
\end{align}
If we assume that
$\beta\neq\pm\theta\Delta\ell/2$, then $C_{1,2} = D_{1,2} = 0$,  
implying that the spinor dependence on $k_z$ and $\phi$ is of higher order. We ignore such a possibility as it requires fine tuning,
leaving us with the following two possibilities: either $\beta=\theta\Delta\ell/2$ or $\beta=-\theta\Delta\ell/2$. We opt for the first of the two options, and comment on the other choice at the end of this section. In that case, we find that
\begin{equation}
    D_1=C_1=0,
    \label{eq:app:CD-result}
\end{equation}
while $C_2$ and $D_2$ are arbitrary.\medskip
\\
\noindent{(\emph{ii})} Equate the coefficients in front of $k_+$:
\begin{align}
	& A_1\left[e^{-i\theta\Delta\ell/2}-e^{i\theta\Delta\ell/2}e^{i\theta}\right]=0, \label{eq:k+_cond_1}\\
	& A_2\,e^{i\theta\Delta\ell/2}\left[1-e^{i\theta}\right]=0. \label{eq:k+_cond_2}
\end{align}
From this we conclude that 
\begin{equation}\label{eq:app:A-result}
\begin{split}
    A_2&=0, \\
    A_1&\neq 0, \quad\mathrm{if}\; \Delta\ell=-1.
\end{split}
\end{equation}
Both coefficients vanish if $\Delta \ell \neq -1$, (cf.~implications below).\medskip \\
\noindent{(\emph{iii})} Similarly equate the coefficients in front of $k_-$ to get:
\begin{equation}\label{eq:app:B-result}
\begin{split}
    B_2&=0, \\
    B_1&\neq 0, \quad\mathrm{if}\; \Delta\ell=1.
\end{split}
\end{equation}
It follows from Eqs.~(\ref{eq:app:A-result}) and~(\ref{eq:app:B-result}) that if $|\Delta\ell|\neq1$ then all linear terms in $k_\pm$
vanish. 
In that case, we are led to include higher-order terms in the ansatz \eqref{eq:app:spinor-ansatz-linear}.
\bigskip

\indent The 
possible second-order contributions in $k_\pm$ are 
\begin{equation}
    \zeta^{(1)}(k_+,k_-)=\begin{pmatrix}
		L_1k_+^2 + M_1k_-^2 + N_1k_+k_- \\
		L_2k_+^2 + M_2k_-^2 + N_2k_+k_-
	\end{pmatrix}.
	\label{eq:app:spinor-ansatz-quadratic}
\end{equation}
with six complex coefficients. We deduce the coefficients in front of the quadratic terms in the same way as for the linear terms.\medskip \\
\noindent{(\emph{i})} First, equate coefficients in front of $k_+k_-$:
\begin{align}
	& N_1e^{-i\theta\Delta\ell/2}=N_1e^{i\theta\Delta\ell/2}, \label{eq:k+k-_cond_1}\\
	& N_2e^{i\theta\Delta\ell/2}=N_2e^{i\theta\Delta\ell/2}. \label{eq:k-k-_cond_2}
\end{align}
From this we conclude that $N_1=0$, while $N_2$ can take arbitrary values. Therefore, independent of the Berry-dipole spin $\Delta\ell$, the second component of the spinor generically depends on the quadratic form $k_+k_- = k_x^2+k_y^2$.\medskip \\ 
\noindent{(\emph{ii})} Equate coefficients in front of $k_+^2$:
\begin{align}
	& L_1\left[e^{-i\theta\Delta\ell/2}-e^{i\theta\Delta\ell/2}e^{i2\theta}\right]=0, \label{eq:k+2_cond_1}\\
	& L_2\,e^{i\theta\Delta\ell/2}\left[1-e^{i2\theta}\right]=0. \label{eq:k+2_cond_2}
\end{align}
This gives
\begin{equation}\label{eq:app:L-result}
\begin{split}
    L_2&=0, \\
    L_1&\neq 0, \quad\mathrm{if}\; \Delta\ell=-2.
\end{split}
\end{equation}
\smallskip \\
\noindent{(\emph{iii})} Similarly for coefficients in front of $k_-^2$ we get:
\begin{equation}\label{eq:app:M-result} 
\begin{split}
    M_2&=0, \\
    M_1&\neq 0, \quad\mathrm{if}\; \Delta\ell=2.
\end{split}\
\end{equation}
For $|\Delta\ell|= 2$ Eqs.~\eqref{eq:app:L-result} and \eqref{eq:app:M-result} define quadratic dependence on $k_\pm$ in the first component of the spinor. 
However, for $|\Delta\ell|> 2$ these contributions vanish and we need to consider still higher orders in $k_\pm$.\bigskip 

\indent In the following we derive the lowest-in-$k_\pm$ term that is allowed by a given Berry-dipole spin $\Delta\ell$ in the first component of the spinor. 
To achieve this, we first derive all allowed terms of the form $k_+^pk_-^q$ with $p,q$ being non-negative integers (the set $\mathbb{N}\cup\{0\}$), and then among these terms find the one with the lowest order $p+q$. \medskip \\
\noindent{(\emph{i})}
Equate coefficients in front of $k_+^pk_-^q$ in the first component of the spinor:
\begin{equation}
    e^{-i\theta\Delta\ell/2}=e^{i\theta\Delta\ell/2}e^{i\theta(p-q)}.
\end{equation}
This equation is 
fulfilled if $q-p=\Delta\ell$, and hence, the allowed terms have the form $k_+^pk_-^{p+\Delta\ell}$ with both $p,p+\Delta\ell\in\mathbb{N}\cup\{0\}$. 
This pair of conditions can be expressed compactly as $p \geq \textrm{max}\{0,-\Delta\ell\}$. \medskip \\
\noindent{(\emph{ii})} To find among all allowed terms the one with the lowest order,  
we need to find
\begin{equation}
p_*=\underset{p\geq \textrm{max}\{0,-\Delta \ell\}}{\mathrm{arg \,min}} (2p+\Delta\ell),
\end{equation}
which, depending on the sign of $\Delta\ell$, is
\begin{equation}
    p_*=\begin{cases}
    0, & \Delta\ell>0 \\
    -\Delta\ell, & \Delta\ell<0
    \end{cases}.
\end{equation}
Therefore, given the value $\Delta\ell$, the first component of the spinor to the lowest order in $k_\pm$ is $\left(k_{-\sign(\Delta\ell)}\right)^{|\Delta\ell|}$.\bigskip 

By combining all the constraints derived for a given Berry-dipole spin $\Delta\ell$, we get the general form of the rotation-symmetric spinor:
\begin{equation}
    \zeta(\bk,\phi)=\begin{pmatrix}
		Z_{\Delta\ell}\left(k_{-\sign(\Delta\ell)}\right)^{|\Delta\ell|} \\
		C_2(k_z+i\phi)+ D_2(k_z-i\phi) + N_2 k_+k_-
	\end{pmatrix},
	\label{eq:app:general-sym-spinor}
\end{equation} 
where $Z_{\Delta\ell}$, $C_2$, $D_2$ and $N_2$ are arbitrary complex coefficients.\bigskip

Let us conclude the discussion by analyzing how the result changes if we instead adopt the second option of the phase, 
$\beta=-\theta\Delta\ell/2$. 
Then, following the steps that were described in this section, we derive that the spinor function has the form 
\begin{equation}
    \zeta(\bk,\phi)=\begin{pmatrix}
		C_1(k_z+i\phi)+ D_1(k_z-i\phi) + N_1 k_+k_- \\
		Z_{\Delta\ell}'\left(k_{\sign(\Delta\ell)}\right)^{|\Delta\ell|}
	\end{pmatrix}.
	\label{eq:app:general-sym-spinor-2}
\end{equation}
To understand the difference between the two Hamiltonians, respectively defined through spinors \eqref{eq:app:general-sym-spinor} and \eqref{eq:app:general-sym-spinor-2}, we note that the conduction band of the spinor-form Hamiltonian [Eq.~\eqref{eq:zszs-hamiltonian}] is proportional to the spinor $\ket{u_c(\bk)}\propto\zeta(\bk)$, while the valence band can be expressed as $\ket{u_v(\bk)}\propto\sigma_y\zeta^*(\bk)$.
We can define an angular momentum $\ell_{v(c)}$ of the valence (conduction) eigenstate at the rotation axis as
\begin{equation}
R_\theta\ket{u_{v(c)}(0,0,k_z)}=e^{i\theta\ell_{v(c)}}\ket{u_{v(c)}(0,0,k_z)}.
\end{equation}
Both the Hamiltonian $h(\bk)$ and the rotation matrix $R_\theta$ are expressed in a basis given by the orbitals $\varphi_1$, $\varphi_2$; since the rotation matrix is diagonal in this basis, these orbitals get a factor $e^{-i\theta\Delta\ell/2}$ resp.~$e^{i\theta\Delta\ell/2}$ under rotation by an angle $\theta$ and therefore have angular momenta $\ell_{1(2)}=-(+)\Delta\ell/2$.
Since valence and conduction bands at rotation axis diagonalize the rotation matrix, the set $\{\ell_v, \ell_c\}$ coincides with the set $\{\pm\Delta\ell/2\}$, meaning that at rotation axis the valence and conduction bands are of a single orbital character. 
Computing the valence and conduction angular momenta for the spinor \eqref{eq:app:general-sym-spinor}, we observe that a valence (resp.~conduction) band has angular momentum $\ell_v=-\Delta\ell/2$ (resp.~$\ell_c=\Delta\ell/2$), meaning that it is of the $\varphi_1$ (resp.~$\varphi_2$) character. For the spinor \eqref{eq:app:general-sym-spinor-2}, the valence (resp.~conduction) angular momentum has opposite sign and the corresponding band is of the $\varphi_2$ (resp.~$\varphi_1$) character.
 
To relate the valence (resp.~conduction) band given by spinor \eqref{eq:app:general-sym-spinor-2} to the valence (resp.~conduction) band given by spinor \eqref{eq:app:general-sym-spinor}, it is enough to flip the order of the basis orbitals and change the sign of the Berry-dipole spin in the rotation matrix written in the new basis. This deformation allows us to extend the results presented for the spinor~\eqref{eq:kp-sym-spinor} in Sec.~\ref{sec:berry-dipole} to the spinor \eqref{eq:kp-sym-spinor-2}, which are the simplified forms of the spinors \eqref{eq:app:general-sym-spinor} and \eqref{eq:app:general-sym-spinor-2}, respectively.

\subsection{Continuous deformation of the spinor to simplified form\label{app:spinor-cont-deform}}

Consider a Hamiltonian $h(\bk,\phi)$ [cf.\  Eq.~\eqref{eq:zszs-hamiltonian}] encoded by
the spinor in Eq.~\eqref{eq:app:general-sym-spinor}. Let us first argue what are the allowed continuous deformations of this Hamiltonian that preserve its rotation symmetry and `topology', in a sense that will be made increasingly precise.

Let us assume the energy bands of $h(\bk,\phi)$ touch at $(\bk,\phi)=(\bze,0)$ and nowhere else in the $(\bk,\phi)$-space, which we take to be isomorphic to $\mathbb{R}^4$. 
For any constant-$\phi$ slice of  $(\bk,\phi)$-space, one can define a three-parameter Hamiltonian $h(\bk)$ and its associated continuum Hopf and Zak numbers [cf.~Eqs.~\eqref{eq:hopfinvar-cont} and \eqref{eq:pol-continuum}]. 
These numbers are the $\bk\cdot \bp$ analogs of the Hopf and RTP-Zak invariants for tight-binding Hamiltonians. Assuming that $h(\bk)$ is a truncated Taylor expansion of a tight-binding Hamiltonian, changes in the continuum Hopf/Zak numbers (as $\phi$ is tuned through zero) equal changes in the lattice  Hopf/RTP invariants (of the tight-binding Hamiltonian which regularizes the $\bk\cdot \bp$ Hamiltonian). 
In particular, this implies 
\begin{equation}
\delta\chi^\textrm{cont.} = \chi^{\textrm{cont.}}(\phi=0^+)-\chi^{\textrm{cont.}}(\phi=0^-)\in \Z.
\end{equation}
This integer quantization holds under continuous deformation of the $\bk\cdot\bp$ Hamiltonian $h(\bk,\phi;t)$ parametrized by $t\in [0,1]$, if we impose for all $t$ that the energy gap closes only at~$(\bk,\phi)=(\bze,0)$. 

This freedom given to us by the $t$ parameter allows us to bring the spinor derived in Eq.~(\ref{eq:app:general-sym-spinor}) to a simpler form. 
The condition that the gap closes only at $(\bk,\phi)=(\bze,0)$ can be enforced by insisting that $|Z_{\Delta\ell}|\neq0$ and $|C_2|\neq|D_2|$.
Indeed, the gap closes if and only if $\zeta=0$.
Imposing that the first component of $\zeta$ vanishes implies $k_{-\sign(\Delta\ell)}=0$, and hence also $N_2k_+k_-=0$. 
Then the second component of $\zeta$ equal to zero is equivalent to $C_2(k_z+i\phi)=-D_2(k_z-i\phi)$, which implies that $(|C_2|-|D_2|)\sqrt{k_z^2+\phi^2}=0$. This means that when $|C_2|\neq|D_2|$, the spinor turns to zero only at $(\bk,\phi)=(\bze,0)$.

Hence, we are allowed to continuously adjust the coefficients to certain preferred values as long as we avoid  $|Z_{\Delta\ell}| = 0$ (and also avoid $|C_2|=|D_2|$) throughout the interpolation. 
Since this constraint on our interpolation does not involve $N_2$, we can always set this parameter to zero. 
We may as well take $Z_{\Delta\ell}=1$. 
The selected values of $C_2$ and $D_2$ depend on the ratio between their magnitudes. 
We choose $C_2=1$, $D_2=0$ if $|C_2|>|D_2|$ and $C_2=0$ $D_2=-1$ if $|C_2|<|D_2|$.
There exists a continuous deformation of the spinor-form Hamiltonian from arbitrary initial values of the coefficients to the specified ones without closing the energy gap and without breaking the rotation symmetry (and thus without altering the topological invariants of the corresponding Berry dipole). 
The final form of the spinor is
\begin{equation}
    \zeta(\bk,\phi)=\begin{pmatrix}
		\left(k_{-\sign(\Delta\ell)}\right)^{|\Delta\ell|} \\
		\upsilon\,k_z+i\phi
	\end{pmatrix},
	\label{eq:app:general-sym-simple-spinor}
\end{equation}
where 
\begin{equation}
\upsilon=\sign(|C_2|-|D_2|)\in\{-1, 1\}.
\end{equation}
This is the form of the spinor utilized throughout the main text, beginning from Eq.~\eqref{eq:kp-sym-spinor}.
A similar deformation performed on the spinor in Eq.~\eqref{eq:app:general-sym-spinor-2} leads to
\begin{equation}
    \zeta(\bk,\phi)=\begin{pmatrix}
		\upsilon\,k_z+i\phi \\
		\left(k_{\sign(\Delta\ell)}\right)^{|\Delta\ell|}
	\end{pmatrix}.
	\label{eq:app:general-sym-simple-spinor-2}
\end{equation}

\subsection{Symmetry of \texorpdfstring{$\bk\cdot\boldsymbol{p}$}{k.p}-expansion of a lattice Hamiltonian}\label{app:kp-sym-nonuniax}

In this Appendix we derive the rotation-symmetry constraint of a lattice Hamiltonian expanded for small $\bk$ around a rotation-symmetric point in the BZ. 
A lattice Hamiltonian obeys the following symmetry condition:
\e{
R_{C_m}h(\bk)R_{C_m}^{-1}=h(C_m\bk).
}
We now expand the Hamiltonian around point $\bk_\Pi=(\Pi,k_z)$ where the reduced momentum $\Pi$ is shifted by a reciprocal lattice vector after $m$-fold rotation, $C_m\Pi=\Pi+\bG_\Pi$. 
Adopting the local momentum $\bkappa=\bk-\bk_\Pi$, we obtain
\e{
R_{C_m}h(\bk_\Pi+\bkappa)R_{C_m}^{-1}\eq h(C_m\bk_\Pi+C_m\bkappa)\lin
\eq h(\bk_\Pi+\bG_\Pi+C_m\bkappa).
\label{eq:app:kp-sym-derivation}
}
We use that the Hamiltonians at momenta shifted by a reciprocal lattice vector are related by
\begin{equation}
    h(\bk+\bG)=V_{\bG}^{-1}h(\bk)V_{\bG}^{\phantom{-1}},
\end{equation}
with $[V_{\bG}]_{\alpha\beta}=\exp(i\br_{\alpha}\cdot\bG)\delta_{\alpha\beta}$.
Applying this relation to Eq.~\eqref{eq:app:kp-sym-derivation} we get
\e{
V_{\bG_\Pi}R_{C_m}h(\bk_\Pi+\bkappa)(V_{\bG_\Pi}R_{C_m})^{-1}\eq h(\bk_\Pi+C_m\bkappa).
}
We now expand the Hamiltonian at small $\bkappa$ around $\bk_\Pi$ and denote this expansion as $h^{\bk\cdot\bp}(\bkappa)$. 
Using the itinerant rotation matrix defined in Eq.~\eqref{eq:itinerant_rotation}, we get the symmetry condition on the $\bk\cdot\bp$-Hamiltonian
\begin{equation}
    \widetilde{R}_{C_m}(\Pi)h^{\bk\cdot\bp}(\bkappa)\widetilde{R}_{C_m}^{-1}(\Pi)=h^{\bk\cdot\bp}(C_m\bkappa).
    \label{eq:app:kp-sym-condition}
\end{equation}
The itinerant rotation matrix can be written as
\begin{equation}
    \widetilde{R}_{C_m}(\Pi)=\exp(-i\pi\frac{\widetilde{\mathcal{L}}_2(\Pi)-\widetilde{\mathcal{L}}_1(\Pi)}{m}\sigma_z).\label{eqn:widetilde-R-C-m}
\end{equation}
Therefore, the Berry-dipole spin $\Delta\ell$ in the $\bk\cdot\bp$ expansion around 
$\bk_\Pi$ is given modulo the order of rotation symmetry $m$ by the itinerant angular momentum difference $\widetilde{\mathcal{L}}_2(\Pi)-\widetilde{\mathcal{L}}_1(\Pi)$ between the basis Bloch states.

\subsection{Spinor function of \texorpdfstring{$\bk{\cdot}\boldsymbol{p}$}{k.p}-expansion of a lattice Hamiltonian} \label{app:spinor-lattice-sym}

In this Appendix we derive a spinor in the lowest order in $(\bkappa, \phi)$ that is allowed by the discrete rotation symmetry \eqref{eq:app:kp-sym-condition} of a $\bk\cdot\boldsymbol{p}$ Hamiltonian around momentum $\bk_\Pi$ and of the spinor form~\eqref{eq:zszs-hamiltonian}. To fix the notations, we assume that at rotation-invariant line $\gamma_\Pi$ [line projected onto $\Pi$ in the $(k_x,k_y)$ plane] the valence band coincides with a band representation induced from the first basis orbital $\varphi_1$. (We will comment on the opposite choice, when the orbital $\varphi_2$ induces the valence band representation, at the end of this Appendix.) This means that the valence (conduction) itinerant angular momentum $\widetilde{\mathcal{L}}_{v(c)}(\Pi)$ is equal to the itinerant angular momentum of the first (second) basis orbital $\widetilde{\mathcal{L}}_{1(2)}(\Pi)$, and the itinerant rotation matrix in Eq.~(\ref{eqn:widetilde-R-C-m}) can be rewritten as
\begin{equation}
    \widetilde{R}_{C_m}(\Pi)=\exp(-i\pi\frac{\Delta\widetilde{\mathcal{L}}(\Pi)}{m}\sigma_z),
\end{equation}
where $\Delta\widetilde{\mathcal{L}}(\Pi)=\widetilde{\mathcal{L}}_c(\Pi)-\widetilde{\mathcal{L}}_v(\Pi) \pmod{m}$. The sought spinor is then continuously deformable to
\begin{equation}
    \zeta(\boldsymbol{\kappa},\phi)=\begin{pmatrix}
        A\kappa_-^{\Delta\widetilde{\mathcal{L}}(\Pi)}+B\kappa_+^{m-\Delta\widetilde{\mathcal{L}}(\Pi)} \\
        \upsilon \kappa_z+i\phi
    \end{pmatrix},
    \label{eq:app:spinor-sym-lat}
\end{equation}
where $\upsilon\in\{-1,1\}$ is defined as in Sec.~\ref{app:spinor-cont-deform}. Depending on the value of $\Delta\widetilde{\mathcal{L}}(\Pi)$, we can further deform the first component of the spinor \eqref{eq:app:spinor-sym-lat} to
\begin{align}
    \zeta_1(\boldsymbol{\kappa})=\begin{cases}
    \kappa_-^{\Delta\widetilde{\mathcal{L}}(\Pi)} & \textrm{if}\;\; \Delta\widetilde{\mathcal{L}}(\Pi)<m/2  \as\textrm{or}\;\; \left\{\begin{array}{l}
         \Delta\widetilde{\mathcal{L}}(\Pi)=m/2  \\
         |A|>|B| 
    \end{array}\right. \\
    \kappa_+^{m-\Delta\widetilde{\mathcal{L}}(\Pi)} & \textrm{if}\;\; \Delta\widetilde{\mathcal{L}}(\Pi)>m/2  \as\textrm{or}\;\; \left\{\begin{array}{l}
         \Delta\widetilde{\mathcal{L}}(\Pi)=m/2  \\
         |A|<|B| 
    \end{array}\right.
    \end{cases}.
    \label{eq:app:spinor1-sym-lat}
\end{align}
Throughout these deformations the gap does not close and therefore the value of the continuum Hopf and Zak phase numbers do not change. By identifying the power of the $\kappa_\pm$ in \q{eq:app:spinor1-sym-lat} with the absolute value of the Berry-dipole spin $|\Delta\ell|$ [as in \q{eq:app:general-sym-simple-spinor}] and the sign of $\kappa_\pm$ with minus the sign of the Berry-dipole spin, we conclude that to the lowest order in $\bkappa$, the Berry-dipole spin $\Delta\ell$ obtained from a Taylor expansion of a lattice model is the lowest-in-absolute-value integer that obeys the constraint
\begin{equation}
    \Delta\ell\mod m=\Delta\widetilde{\mathcal{L}}(\Pi).
\end{equation}

\noindent In the following we proceed to prove the deformability of a symmetry-constrained spinor 
\begin{equation}
    \widetilde{R}_{C_m}(\Pi)\zeta(\bkappa,\phi)=e^{i\beta(\bkappa)}\zeta(C_m\bkappa,\phi)
\end{equation}
to the form given in Eqs.~(\ref{eq:app:spinor-sym-lat}) and (\ref{eq:app:spinor1-sym-lat}). The main difference from the proof given in Appendix~\ref{sec:general-spinor-form} is that the symmetry is lowered from SO(2) to $C_m$ and as a consequence the Berry-dipole spin is changed from being an integer $\Delta\ell\in\mathbb{Z}$ to $\Delta\widetilde{\mathcal{L}}(\Pi)\in\mathbb{Z}_m$. As we show in the following, this modulo $m$ ambiguity is a source of a modified expression for the first component of the spinor in Eq.~\eqref{eq:app:spinor-sym-lat}.\medskip

\noindent (\emph{i}) As the first step, we determine the dependence of the spinor on $\kappa_z$ and $\phi$.
We start as in Appendix~\ref{sec:general-spinor-form}, by assuming the linear dependence on $(\kappa_z\pm i\phi)$ and equating the coefficients in front of independent variables. 
We again conclude, that to have a non-constant dependence on $\kappa_z$ and $\phi$, the phase $\beta$ must be equal to $\pm \pi\Delta\widetilde{\mathcal{L}}(\Pi)/m$. 
The initial assumption that the valence band of the lattice model at rotation-invariant line $\gamma_\Pi$ is purely of the $\varphi_1$ character trivially implies 
that the valence (conduction) band of the $\bk\cdot\bp$ expansion at $\bkappa=\boldsymbol{0}$ is of a $\varphi_1$ ($\varphi_2$) character. The conduction band of a spinor-form Hamiltonian is proportional to $\zeta$, therefore, to ensure the aforementioned orbital characters, the dependence on $\kappa_z$ and $\phi$ must be present only in the second component of the spinor. This is true when $\beta=\pi\Delta\widetilde{\mathcal{L}}(\Pi)/m$, and then the second spinor component is $\zeta_2(\bkappa,\phi)=C\cdot(\kappa_z+i\phi)+D\cdot(\kappa_z-i\phi)+f(\kappa_+,\kappa_-)$. \\

\noindent (\emph{ii}) We proceed to determine the spinor dependence on $\kappa_\pm$. 
First, with the chosen $\beta$, the second component of the spinor must not change when the momentum is rotated, 
\begin{equation}
    \zeta_2(\bkappa,\phi)=\zeta_2(C_m\bkappa,\phi).
\end{equation} 
Therefore, to the lowest order in $\bkappa$, $\zeta_2(\bkappa, \phi)=C(\kappa_z+i\phi)+D(\kappa_z-i\phi)+N\kappa_+\kappa_-$.\\

\noindent(\emph{iii}) The first component of the spinor, which to the lowest order in $(\bkappa,\phi)$ depends only on $\kappa_\pm$, transforms under rotation~as~
\begin{equation}
    \zeta_1(\kappa_+,\kappa_-)=e^{i2\pi\Delta\widetilde{\mathcal{L}}(\Pi)/m}\zeta_1(C_m\kappa_+,C_m\kappa_-).    
    \label{eq:app:zeta1-sym-lat}
\end{equation}
We want to find the lowest powers of $\kappa_+^p$ and $\kappa_-^q$ which are allowed by the symmetry to constitute $\zeta_1$. For the former term, Eq.~\eqref{eq:app:zeta1-sym-lat} has the form
$\kappa_+^p=\kappa_+^p\exp(i2\pi[\Delta\widetilde{\mathcal{L}}(\Pi)+p]/m)$.
By solving it for the lowest possible $p>0$
we get $p=m-\Delta\widetilde{\mathcal{L}}(\Pi)$. Similarly, the lowest possible $q=\Delta\widetilde{\mathcal{L}}(\Pi)$.\\

\noindent From the steps (\emph{i}-\emph{iii}), we get the following spinor
\begin{equation}
    \zeta(\boldsymbol{\kappa},\phi)=\begin{pmatrix}
        A\kappa_-^{\Delta\widetilde{\mathcal{L}}(\Pi)}+B\kappa_+^{m-\Delta\widetilde{\mathcal{L}}(\Pi)} \\
        C(\kappa_z+i\phi) + D(\kappa_z-i\phi) + N\kappa_+\kappa_-
    \end{pmatrix}.
\end{equation}

\noindent{(\emph{iv})} With the same argument as presented in Appendix~\ref{app:spinor-cont-deform}, we can continuously deform this spinor to the form given in Eq.~\eqref{eq:app:spinor-sym-lat}, without altering the Hopf and Zak phase numbers.\\

\noindent{(\emph{v})} Let us finally show, that the spinor \eqref{eq:app:spinor-sym-lat} can be further deformed to have the first component given by \q{eq:app:spinor1-sym-lat}. We first consider $\Delta\widetilde{\mathcal{L}}(\Pi)\neq m/2$, and see that one of the terms $A\kappa_-^{\Delta\widetilde{\mathcal{L}}(\Pi)}$ and $B\kappa_+^{m-\Delta\widetilde{\mathcal{L}}(\Pi)}$ is of higher
order than another, and therefore can be neglected.
After that, the coefficient in front of the remaining term can be deformed to be $A=1$ resp.~to $B=1$, without closing the gap. \\

\noindent{(\emph{vi})} Alternatively,
when $\Delta\widetilde{\mathcal{L}}(\Pi)= m/2$, the two just discussed terms are of the same order $m/2$ and none of them can be neglected. 
However, we can continuously deform these terms to a simpler form without closing the gap. 
Importantly, if $|A|\neq|B|$, the gap closes only at $\kappa_+=\kappa_-=\kappa_z=\phi=0$. 
Therefore, we can deform the coefficients $A$ and $B$, while keeping the sign of their relative magnitude $\sgn(|A|-|B|)$ constant. 
This means that for $|A|>|B|$ we can deform the coefficients to be $A=1$ and $B=0$, while for $|A|<|B|$ the resulting coefficients are $A=0$ and $B=1$. 
Thus, we arrive to the first component of the spinor given in Eq.~\eqref{eq:app:spinor1-sym-lat}. \medskip

Let us finally discuss what happens if we choose the valence band to coincide at rotation-invariant line $\gamma_\Pi$ with the band representation induced from the second orbital $\varphi_2$. In this case the itinerant rotation matrix can be written as
\begin{equation}
    \widetilde{R}_{C_m}(\Pi)=\exp(i\pi\frac{\Delta\widetilde{\mathcal{L}}(\Pi)}{m}\sigma_z),
\end{equation}
and the Berry-dipole spin fulfills the constraint
\begin{equation}
    \Delta\ell\mod m=m-\Delta\widetilde{\mathcal{L}}(\Pi).
\end{equation}
To ensure the correct orbital characters of the valence and conduction bands, the dependence on $\kappa_z$ and $\phi$ must enter only the first component of the spinor, for which the phase $\beta$ must again be fixed to $\pi\Delta\widetilde{\mathcal{L}}(\Pi)/m$. Similar derivation as presented above shows that the spinor which defines a $\bk\cdot\bp$ Hamiltonian can be deformed to $\sigma_x\zeta(\bkappa, \phi)$ with $\zeta(\bkappa, \phi)$ given in Eqs.~\eqref{eq:app:spinor-sym-lat} and \eqref{eq:app:spinor1-sym-lat}.

\section{Analytic results for spinor-form Berry dipole}\label{app:spinor-form}

In this Appendix we derive and list several analytic results related to the Berry-dipole Hamiltonian specified by Eqs.~\eqref{eq:kp-sym-spinor} and \eqref{eq:zszs-hamiltonian} of the main text. 
Our discussion is organized into several parts as follows. 
First, in Appendix~\ref{app:berry-flux} we analyze the flux through a hemisphere associated with the Berry-dipole Hamiltonian, deriving in particular the result in Eq.~(\ref{eqn:Berry-dipole-flux}) of Sec.~\ref{sec:Berry-definition}. 
We continue in Appendix~\ref{app:Berry-curv-formulas} with listing for the same class of Hamiltonians the expressions for the Berry connection, Berry curvature, and Chern-Simons form. 
This discussion is followed in Appendix~\ref{app:change-hopf} with a derivation 
of the change of the continuum Hopf number across a Berry-dipole gap closing, which depends on the Berry-dipole spin as captured by Eq.~(\ref{eq:Berry-dipole-hopf-change}) in Sec.~\ref{sec:Berry-transitions}. 
In the course of the latter derivation, we utilize that the second Chern character of arbitrary two-band Hamiltonians necessarily vanishes. 
This property, which we also used in the previous Appendix~\ref{app:Hopf-periodic-nonperiodic}, is finally derived in Appendix~\ref{sec:secondchern}.

\subsection{Berry flux associated with a Berry dipole}\label{app:berry-flux}

In this Appendix we derive that, for a rotation-symmetric Berry dipole given by the spinor function in Eq.~\eqref{eq:kp-sym-spinor} with $\phi=0$, the Berry flux $\Psi$ through the upper hemisphere is quantized, and given by Eq.~(\ref{eqn:Berry-dipole-flux}), i.e., it is negatively proportional to the Berry-dipole spin $\Delta\ell$. 
The discussion is organized into two parts: we first explain why the flux through the hemisphere is quantized for $\phi=0$ at all, and afterwards we proceed with the actual integration.

To understand the quantization of the flux through the hemisphere, note that the second component of the spinor in Eq.~(\ref{eq:kp-sym-spinor}) [which is the same as Eq.~\eqref{eq:app:general-sym-simple-spinor} in the previous Appendix] \emph{vanishes} for $k_z = \phi = 0$. Setting $\phi$ to $0$ [which is assumed throughout the present discussion], 
we find that the valence-band state is, up to an unimportant phase factor, constant throughout the  
plane $k_z=0$ with the momentum origin $\bk=\boldsymbol{0}$ removed (we refer to this plane as the \emph{punctured equatorial plane} $\mathcal{P}$):
\begin{equation}
\ket{u_v(k_x,k_y,0)} = \frac{i\sigma_y\zeta^*(k_x,k_y,0)}{\norm{\zeta(k_x,k_y,0)}} \propto \left(\begin{array}{c} 0 \\ 1\end{array}\right).\label{eqn:u_v-gauge}
\end{equation}
It follows that the spectrally flattened 
Hamiltonian
\begin{align}
\frac{h(k_x,k_y,0)}{\sqrt{-\!\det h(k_x,k_y,0) }}
&=\mathbbold{1} -2 \ket{u_v(k_x,k_y,0)}\!\bra{u_v(k_x,k_y,0)}\lin 
&= \sigma_z 
\end{align}
is constant inside $\mathcal{P}$~\cite{Sun:2018}. This implies that any two-dimensional sheet with boundary in $\mathcal{P}$ wraps around the Bloch sphere integer number of times.\footnote{
Mathematically, one defines an equivalence `$\sim$' that identifies all the points on the boundary of the sheet as the same point; this operation preserves the continuity of the Hamiltonian $h$. 
The quotient space ``$\textrm{hemisphere}/{\sim}$'' is homeomorphic to $S^{\! 2}$. 
The wrapping of sphere (the compactified sheet) around another (the Bloch) sphere is known to be characterized by the homotopy group $\pi(S^{\! 2}) = \mathbb{Z}$.}
The integer-valued `wrapping number' around the Bloch sphere is known to be in a one-to-one correspondence with the quantized Berry-flux through the sheet~\cite{Xiao:2010}; it thus follows that the considered Berry-flux through the sheet with boundary in $\mathcal{P}$ must also be quantized.

We next show that the quantized flux is specifically equal to $-2\pi \Delta\ell$ if the Berry-dipole spin equals $\Delta\ell$. 
A particularly simple choice of a sheet with boundary at the equatorial plane, that we considered in the main text [cf.~Eq.~(\ref{eq:Berry-dipole-flux})], is the \emph{upper hemisphere}.  
However, for the purpose of analytic integration, another choice is more convenient, namely the closed surface on which the spinor function $\zeta(\bk)$ has norm equal to $1$, i.e., 
\begin{equation}
\mathcal{M} = \left\{\bk \;\Big|\; \left(k_x^2 + k_y^2\right)^\abs{\Delta\ell} + k_z^2 = 1, k_z > 0\right \}.
\end{equation}
We choose again the gauge of the the valence state $\ket{u_v(\bk)}$ to be governed by the left-side of the proportionality symbol in Eq.~(\ref{eqn:u_v-gauge}), where the norm $\norm{\zeta(\bk)}=1$ now drops from the expressions.
We parametrize the points on the surface $\mathcal{M}$ with spherical coordinates 
$\theta\in[0,\pi/2]$ and $\varphi\in[0,2\pi]$ as
\begin{equation}
k_x=(\sin\theta)^{\frac{1}{|\Delta\ell|}}\cos\varphi,\;\; k_y=(\sin\theta)^{\frac{1}{|\Delta\ell|}}\sin\varphi,\;\;
k_z=\cos\theta.
\end{equation}
The normalized valence state on $\mathcal{M}$ in the specified gauge is
\begin{equation}
    \ket{u_v(\theta,\varphi)}=\begin{pmatrix}
        \upsilon\, \cos\theta \\
        -\sin\theta e^{i\Delta\ell\varphi}
    \end{pmatrix}.
\end{equation}
The flux through the surface $\mathcal{M}$ [which we know to be equal to the flux through the hemisphere due to the vanishing divergence $\boldsymbol{\nabla}_\bk\cdot \boldsymbol{\mathcal{F}} =0$ for the momenta (with gapped spectrum) between $\mathcal{M}$ and the hemisphere], is computed as
\e{ \Psi\left(\substack{\textrm{rotation-sym.}\\\textrm{Berry dipole}}\right)\eq 
\int_{\{\bk\in\mathbb{R}^3|k_z>0\}} 
\boldsymbol{\mathcal{F}}\cdot \nabk \varrho \,\delta(\varrho-1) d^3\bk }
with the shorthand $\varrho=||\zeta(\bk)||$, and the delta function constraining the integral to the desired surface. 
Transforming the integrand to the right-handed coordinate system $(\varrho,\theta,\phi)$, then integrating over $\varrho$, one obtains
\e{ \Psi\left(\substack{\textrm{rotation-sym.}\\\textrm{Berry dipole}}\right)\eq 
\int_{\mathcal{M}}  \boldsymbol{\mathcal{F}}\cdot \nabk \varrho \f{d\theta d\phi}{\nabk\varrho \cdot \nabk \theta \times \nabk \phi}.\la{Phi}}
The Berry-curvature vector is given by
\e{ \ket{\nabk u} \eq \nabk \varrho \ket{\partial_{\varrho} u}+\nabk \theta \ket{\partial_{\theta} u}+\nabk \phi \ket{\partial_{\phi}u},\\ 
\boldsymbol{\mathcal{F}}(\bk)\eq i\braopket{\nabk u}{\times}{\nabk u} =-2\text{Im}\braket{\partial_{\theta} u}{\partial_{\phi}u}\nabk \theta \times \nabk \phi+\ldots \nonumber} 
where ``$\ldots$'' denotes terms involving a cross product with $\nabk \varrho$. Substituting this expression into \eqref{Phi}, we find that the Jacobian determinant cancels out,
\begin{align}
    \Psi\left(\substack{\textrm{rotation-sym.}\\\textrm{Berry dipole}}\right)
    \eq \int\limits_0^{2\pi}d\varphi\int\limits_0^{\pi/2}d\theta \left[-2\Im\bra{\partial_\theta u_v}\ket{\partial_\varphi u_v}\right] \lin
    \eq 2\pi \int\limits_0^{\pi/2}d\theta \left[-\Delta\ell\sin 2\theta\right] 
    =-2\pi\Delta\ell,
\end{align}
The final results is precisely Eq.~(\ref{eqn:Berry-dipole-flux}) of the main text.

\subsection{Expressions for Berry connection, Berry curvature, \texorpdfstring{\\}{} 
and Chern-Simons 3-form}\label{app:Berry-curv-formulas}

Without a detailed derivation, we proceed to list several analytical results for rotation-symmetric Berry dipole with spin $\Delta \ell\in\mathbb{Z}$. 
We assume the gauge in which the normalized eigenstate corresponding to the valence band is given by
\begin{align}
    u_v(\bk)
    &=i\sigma_y\zeta^*(\bk)/\|\zeta(\bk)\| \label{eq:app:gauge}\\
    &= \frac{1}{\left[(k_x^2+k_y^2)^\abs{\Delta\ell}+k_z^2 + \phi^2\right]^{\frac{1}{2}}}\left(\begin{array}{c}
    \upsilon k_z - i\phi \\
    \!\!\!-\left[k_x + \sign(\Delta\ell) \, i k_y\right]^{\abs{\Delta\ell}}\!\!\!
    \end{array}\right) \nonumber
\end{align}
Substituting this spinor function into the definition of the Berry connection $\boldsymbol{\mathcal{A}}(\bk)=\bra{u_v(\bk)}\ket{i\boldsymbol{\nabla}_\bk u_v(\bk)}$ we get the following components:
\begin{align}
\mathcal{A}_x &= +\frac{({\Delta \ell})(k_x^2 + k_y^2)^{\abs{\Delta \ell}-1}}{(k_x^2 + k_y^2)^\abs{\Delta\ell}+k_z^2 + \phi^2} k_y, \label{eqn:analytical-A2}\\
\mathcal{A}_y &= -\frac{({\Delta \ell)}(k_x^2 + k_y^2)^{\abs{\Delta \ell}-1}}{(k_x^2 + k_y^2)^\abs{\Delta\ell}+k_z^2 + \phi^2} k_x, \\
\mathcal{A}_z &= -\frac{\upsilon \phi}{(k_x^2 + k_y^2)^\abs{\Delta\ell}+k_z^2 + \phi^2},
\end{align}
where we used $\sign{(\Delta \ell)}\abs{\Delta \ell} = \Delta \ell$. 
For components of the gauge-invariant Berry curvature $\boldsymbol{\mathcal{F}}(\bk) = \boldsymbol{\nabla}_{\bk}\times \boldsymbol{\mathcal{A}}(\bk)$ we find
\begin{align}
\mathcal{F}_x &= -\frac{2 (\Delta \ell)\left[ k_x k_z-\upsilon \,\sign{(\Delta \ell)} k_y \phi\right](k_x^2+k_y^2)^{\abs{\Delta \ell}-1}}{\left[(k_x^2 + k_y^2)^\abs{\Delta\ell}+k_z^2 + \phi^2\right]^2},  \\
\mathcal{F}_y &= -\frac{2({\Delta\ell})\left[k_y k_z+\upsilon \, \sign{(\Delta\ell)} k_x \phi\right](k_x^2+k_y^2)^{\abs{\Delta \ell}-1}}{\left[(k_x^2 + k_y^2)^\abs{\Delta\ell}+k_z^2 + \phi^2\right]^2}, \\
\mathcal{F}_z &= 
-\frac{2(\Delta \ell)\abs{\Delta \ell}(k_z^2+\phi^2)(k_x^2+k_y^2)^{\abs{\Delta\ell}-1}}{\left[(k_x^2 + k_y^2)^\abs{\Delta\ell}+k_z^2 + \phi^2\right]^2}.
\end{align}
Finally, for the Chern-Simons form we obtain
\begin{equation}
\boldsymbol{\mathcal{F}}\cdot\boldsymbol{\mathcal{A}} = \frac{2\upsilon \, (\Delta \ell) \, \abs{\Delta \ell} \, \phi \, (k_x^2 + k_y^2)^{\abs{\Delta \ell}-1}}{\left[(k_x^2 + k_y^2)^\abs{\Delta\ell}+k_z^2 + \phi^2\right]^2}.
\end{equation}
For the purpose of the next subsection, it is convenient to switch to the standard cylindrical coordinates $\bk\sim(k_\rho,k_\varphi,k_z)$, leading to
\begin{equation}
\boldsymbol{\mathcal{F}}\cdot\boldsymbol{\mathcal{A}}\, d^3\bk = \frac{2\upsilon \,\Delta \ell\,\abs{\Delta \ell}\, k_\rho^{2\abs{\Delta \ell}-1}}{\left(k_\rho^{2\abs{\Delta\ell}}+k_z^2 + \phi^2\right)^2}\phi \,dk_\rho dk_\varphi dk_z \label{eqn:analytical-F.A-cylindrical}
\end{equation}
where $d^3\bk = dk_x dk_y dk_z = k_\rho \, dk_\rho dk_\varphi dk_z$.

\subsection{Change of Hopf invariant across Berry-dipole transition} \label{app:change-hopf}

We now derive Eq.~(\ref{eq:Berry-dipole-hopf-change}) of the main text. 
We achieve this in two ways: first, we follow the analytic results of the previous subsection to carry out the integration explicitly; subsequently we present an alternative
proof that does not rely on the knowledge of the complicated expressions in Eqs.~(\ref{eqn:analytical-A2}--\ref{eqn:analytical-F.A-cylindrical}). The latter derivation requires certain familiarity with the language of differential forms, in particular with the $2$-form formulation of the Berry curvature. 

We depart from Eq.~\eqref{eqn:analytical-F.A-cylindrical}, expressed in the cylindrical coordinates, and plug it into Eq.~\eqref{eq:hopfinvar-cont} which defines the Hopf number of a continuum model. 
Noting that the integration over $d k_\varphi$ yields a factor $2\pi$, we proceed to integrate over $k_\rho$ to find
\begin{align}
\chi^\textrm{cont.} &= -\frac{1}{4 \pi^2} \int _{\mathbb{R}^3} d^3 \bk \boldsymbol{\mathcal{F}}\cdot\boldsymbol{\mathcal{A}} \nonumber \\
&= - \frac{\upsilon \phi \Delta \ell}{2\pi}\int_{-\infty}^{+\infty} d k_z \int_0^{+\infty} d k_\rho \frac{2\abs{\Delta \ell} \, k_\rho^{2\abs{\Delta \ell}-1}}{\left(k_\rho^{2\abs{\Delta\ell}}+k_z^2 + \phi^2\right)^2} \nonumber \\
&= \frac{\upsilon \phi \Delta \ell}{2\pi}\int_{-\infty}^{+\infty} d k_z  \left[\frac{1}{k_\rho^{2\abs{\Delta\ell}}+k_z^2 + \phi^2}\right|_{k_\rho=0}^{+\infty} \nonumber \\ &
=-\frac{\upsilon \Delta \ell \phi }{2\pi}\int_{-\infty}^{+\infty} d k_z \frac{1}{k_z^2+\phi^2}.
\end{align}
The antiderivative of the integrated function is $\frac{1}{\abs{\phi}}\arctan\left(\frac{k_z}{\abs{\phi}}\right)$, leading to
\begin{equation}
\chi^\textrm{cont.} = -\frac{1}{2}\upsilon \Delta \ell \,\sign(\phi).\label{eqn:Hopf-integrated-value}
\end{equation}
We find that the change in the continuum Hopf number as $\phi$ changes sign from negative to positive is $\delta \chi^\textrm{cont.} = - \upsilon \Delta \ell$, as advertised in Eq.~(\ref{eq:Berry-dipole-hopf-change}) of the main text.

Alternatively, we can derive the desired result as follows. Note that the change $\delta\chi^\textrm{cont.}$ is defined as a difference in the Chern-Simons form integrals [Eq.~\eqref{eq:hopfinvar-cont}] at two fixed values of the parameter $\phi>0$ and $\phi<0$. 
This change can be equivalently computed as an integral of the Chern-Simons form over the 3-sphere enclosing the Berry-dipole degeneracy in the 
four-dimensional $(\bk,\phi)$-space, i.e.,\footnote{
Our convention to denote objects related to the Berry curvature is as follows: the components of the Berry curvature $2$-form are $\mathcal{F}_{\alpha\beta}=\partial_\alpha \mathcal{A}_\beta - \partial_\beta \mathcal{A}_\alpha$ (with two indices), from which one obtains components of the Berry pseudovector in 3D as $\mathcal{F}_\alpha = \frac{1}{2}\epsilon_{\alpha\beta\gamma}\mathcal{F}_{\beta\gamma}$ (with a single index); here, $\epsilon$ is the fully antisymmetirc Levi-Civita symbol, and summation over repeated indices is implied. We further use the plain calligraphic symbol $\mathcal{F}$ to indicate the complete Berry curvature $2$-form, and the bold version $\boldsymbol{\mathcal{F}}$ to denote the complete Berry curvature pseudovector. Components of the exterior product of a $p$-form $a$ and $q$-form $b$ are given by~\cite{Fecko:2006}
\begin{equation}
(a\wedge b)_{\alpha\ldots \beta \gamma\ldots \delta}=\frac{(p+q)!}{p!\,q!}a_{[\alpha\ldots\beta}b_{\gamma\ldots \delta]}    
\end{equation}
Square brackets around indicates indicate the fully antisymmetric component; in particular, we have
\begin{equation}
\mathcal{F}_{[\alpha\beta}\mathcal{A}_{\gamma]}{=} \frac{1}{3!}\left(\mathcal{F}_{\alpha\beta}\mathcal{A}_\gamma {-} \mathcal{F}_{\alpha\gamma}\mathcal{A}_\beta{+} \mathcal{F}_{\beta\gamma}\mathcal{A}_\alpha{-} \mathcal{F}_{\beta\alpha}\mathcal{A}_\gamma{+} \mathcal{F}_{\gamma\alpha}\mathcal{A}_\beta{-} \mathcal{F}_{\gamma\beta}\mathcal{A}_\alpha \! \right)\!\!\label{eqn:CS3-form-antisymm}
\end{equation}
in the case of Eq.~(\ref{eqn:Hopf-on-S3}). 
Note that we do not distinguish between covariant and contravariant indices, and place both in the subscript.\label{foot:Berry-forms}} 
\begin{align}
\delta\chi ^\textrm{cont.} &= -\frac{1}{4\pi^2} \oint_{S^3} \mathcal{F} \wedge \mathcal{A} \nonumber \\ 
&= -\frac{1}{4\pi^2}\oint_{S^{\!3}}
\,\frac{3!}{2!\;1!}{\mathcal{F}_{[\alpha\beta}}\,{\mathcal{A}_{\gamma]}} \; 
d x_\alpha \, d x_\beta \, d x_\gamma,\label{eqn:Hopf-on-S3}
\end{align}
where $\{x_\alpha\}_{\alpha=1}^3$ are right-handed coordinates on the $3$-sphere.
Such deformation of the integration domain is possible without changing the invariant by virtue of the Stokes' theorem, because the differential of the Chern-Simons 3-form (this differential is proportional to the \emph{second Chern character}~\cite{nakahara1990})
\begin{equation}
\textrm{ch}_2(\mathcal{F})\propto d(d\mathcal{A} \wedge \mathcal{A}) = \mathcal{F} \wedge \mathcal{F}
\end{equation}
is \emph{zero for two-band Hamiltonians} (the proof of this statement is postponed to Appendix~\ref{sec:secondchern}), and therefore the integral of $\mathcal{F}\wedge\mathcal{F}$ over a 4-dimensional space without singularities also vanishes. 

For the same reason, to simplify the calculation, we can slightly modify the surface of integration from a 3-sphere to a diffeomorphic 3-manifold on which the spinor function $\zeta(\bk,\phi)$ is normalized to unity: $\norm{\zeta(\bk,\phi)}=1$. 
We use 3-spherical coordinates to parametrize this manifold as $\phi=\cos\psi$, $k_z=\sin\psi\cos\theta$, $k_y=(\sin\psi\sin\theta)^{1/|\Delta\ell|}\sin\varphi$, and $k_x=(\sin\psi\sin\theta)^{1/|\Delta\ell|}\cos\varphi$, with the 3-spherical angles running through intervals $\psi\in[0,\pi]$, $\theta\in[0,\pi]$, $\varphi\in[0,2\pi]$. 
The spinor function on the defined 3-manifold 
\begin{equation}
    \zeta(\psi,\theta,\varphi)=\left(\begin{array}{c}\sin\psi\sin\theta e^{-i\Delta\ell\varphi} \\ \upsilon\sin\psi\cos\theta+i\cos\psi\end{array}\right),
    \label{eq:app:spinor-parametrization}
\end{equation}
is indeed normalized, $\norm{\zeta(\psi,\theta,\varphi)}=1$. 
We adopt the gauge choice of \q{eq:app:gauge},
in which the conduction eigenstate is $u_c(\psi,\theta,\varphi)=\zeta(\psi,\theta,\varphi)$. 
We opt to explicitly analyze the change in the Hopf invariant only for the conduction band, and then extend the result to the valence band, as both bands of a two-band Hamiltonian possess the same Hopf invariant, owing to  a symmetry $\mathcal{J} \boldsymbol{h}(\bk) \mathcal{J}^{-1} = -\boldsymbol{h}(\bk)$ with $\mathcal{J}= i \sigma_y \mathcal{K}$ of any Hamiltonian that is a sum of Pauli matrices.

The Chern-Simons $3$-form on this $3$-manifold is found
\begin{align}
\mathcal{F}\wedge\mathcal{A} &= 3 \, \mathcal{F}_{[\alpha\beta}\mathcal{A}_{\gamma]} dx_\alpha dx_\beta dx_\gamma  \nonumber \\
&= 3 \left(i \partial_{[\alpha}\zeta^\dagger \partial_\beta \zeta\right) \left(i \zeta^\dagger \partial_{\gamma]}\zeta^\dagger \right) \, dx_\alpha \, dx_\beta \, dx_\gamma,\label{eqn:CS3-from-spinor}
\end{align}
where $\{x_\alpha\}_{\alpha=1}^3 = (\psi,\theta,\varphi)$ are coordinates on the $3$-manifold. We find
\begin{align}
i \zeta^\dagger \partial_\psi \zeta &= \upsilon \cos\theta \nonumber, \\
i \zeta^\dagger \partial_\theta \zeta &=- \upsilon \sin \psi \cos\psi \sin\theta, \\ 
i \zeta^\dagger \partial_\varphi \zeta &= \Delta\ell \sin^2\psi \sin^2\theta \nonumber 
\end{align}
for the components of the Berry connection, and
\begin{align}
i\partial_\psi\zeta^\dagger \partial_\theta \zeta = - i\partial_\theta \zeta^\dagger \partial_\psi \zeta &= \upsilon \sin^2\psi\sin\theta \nonumber, \\
i\partial_\psi \zeta^\dagger \partial_\varphi \zeta = -i\partial_\varphi \zeta^\dagger \partial_\psi \zeta &= \Delta \ell \sin\psi\cos\psi\sin^2\theta,  \\
i\partial_\theta \zeta^\dagger \partial_\varphi \zeta = -i\partial_\varphi \zeta^\dagger \partial_\theta \zeta &=\Delta \ell \sin^2\psi\sin\theta\cos\theta \nonumber
\end{align}
for the components of the Berry curvature. After we insert these components to Eq.~(\ref{eqn:CS3-from-spinor}) and perform the antisymmetrization in Eq.~(\ref{eqn:CS3-form-antisymm}), we finally obtain the Chern-Simons $3$-form
\begin{equation}
\mathcal{F}\wedge\mathcal{A} = 2 \upsilon \Delta\ell \sin^2\psi\sin\theta \, d\psi \, d\theta \,d\varphi
\end{equation}
on the considered $3$-manifold. The integration over the $3$-manifold is straightforward, resulting in the change of the continuum Hopf number

\begin{align}
    \delta\chi^\textrm{cont.} 
    &= -\frac{1}{4\pi^2}\int\limits_0^\pi d\psi\int\limits_0^\pi d\theta \int\limits_0^{2\pi} d\varphi\: 2\upsilon\Delta\ell\sin^2\psi\sin\theta \\
    &= -\upsilon\Delta\ell,
    \label{eq:app:delta-chi-derivation}
\end{align}
which is again the result given in Eq.~(\ref{eq:Berry-dipole-hopf-change}).

\subsection{Vanishing of the 
second Chern character for two-band Hamiltonians\label{sec:secondchern}}

In this part of the present Appendix, we move beyond Berry-dipole Hamiltonians, and review a particular property that applies to \emph{arbitrary} smooth two-band Hamiltonian $h(\bk,\phi)$.
Namely, while the second Chern class necessarily vanishes for each fibre of a line bundle~\cite{Hatcher:2003} (i.e., for any single non-degenerate energy band), this is not generally true of the second Chern \emph{character}.  
Nevertheless, we show here that the second Chern character vanishes for each fibre of a line bundle corresponding to the valence band of a \emph{two-band} Hamiltonian $h(\bk,\phi) = \boldsymbol{h}(\bk,\phi)\cdot\boldsymbol{\sigma}$, with $\boldsymbol{h}\neq\boldsymbol{0}$. 
Note that the vanishing of the second Chern character is a \emph{local} property of the $(\bk,\phi)$-space Hamiltonian, and is (for the case of crystalline models) valid irrespective of whether the Hamiltonian
is periodic over the Brillouin zone. We have applied this result in a proof of the quantization of the Hopf invariant in Appendix~\ref{app:Hopf-periodic-nonperiodic} and in the analytic considerations presented in Appendix~\ref{app:change-hopf}. 

To reveal this local property, we normalize $\boldsymbol{h}/\norm{\boldsymbol{h}}\equiv \boldsymbol{\eta}$. 
It is well-known that Berry curvature plays the role of the skyrmion density, $\mathcal{F}_{\alpha\beta} \propto \epsilon_{ijk} \eta^i \partial_\alpha \eta^j \partial_\beta \eta^k$~\cite{stone_mathematicsforphysics}.
Then the second Chern character $\propto \mathcal{F}\!\wedge\! \mathcal{F}$~\cite{nakahara1990} has components 
\begin{align}
&(\mathcal{F} \wedge \mathcal{F})_{\alpha\beta\gamma\delta} \propto \\ 
&\qquad\sum_{\varsigma\in\mathfrak{S}_4}\sign(\varsigma) \epsilon_{ijk}\epsilon_{lmn} \eta^i (\partial_{\alpha'} \eta^j) (\partial_{\beta'} \eta^k) \eta^l (\partial_{\gamma'} \eta^m) (\partial_{\delta'} \eta^n)\nonumber
\end{align}
where $\mathfrak{S}_4$ is the permutation group of four elements, $\varsigma$ is any of the permutations of the four Greek indices, and we introduced primed notation $\alpha'=\varsigma(\alpha)$ (and similar for $\beta,\gamma,\delta$) to achieve brevity.

The product of the two Levi-Civita symbols can be written~\cite{Fecko:2006} as a linear combination of products of three Kronecker symbols ``$\pm\delta^{ad}\delta^{be}\delta^{cf}$'' where $\{a,b,c\}$ (resp.~$\{d,e,f\}$) is a permutation of indices $\{i,j,k\}$ (resp.~$\{l,m,n\}$) used in the expression above. 
It is easily seen that all the resulting terms will contain at least one product of the type ``$(\partial_\mu \eta^g)(\partial_\nu \eta^g)$'' with $\mu,\nu$ being distinct elements of $\{\alpha,\beta,\gamma,\delta\}$. 
Since such an expression for product of derivatives is symmetric in $\mu\leftrightarrow \nu$, while the $4$-form $\mathcal{F}\!\wedge\! \mathcal{F}$ is antisymmetric in those same indices due to the antisymmetrization over the permutations in $\mathfrak{S}_4$, the corresponding term is cancelled.
As this argument can be repeated for \emph{all} the terms obtained by replacing the Levi-Civita symbols by Kronecker symbols, one concludes that $\mathcal{F} \!\wedge\! \mathcal{F} = 0$ for the valence band of a strictly two-band model.

\section{Berry dipole as a representative phase transition}\label{app:phase-transition-general}

Several arguments presented in the main text, notably in the early Secs.~\ref{sec:berry-dipole} and~\ref{sec:lattice-RTP},
assume that the topological phase transition in which the RTP or the Hopf invariant change value is facilitated by a band touching which occurs at a single isolated value of the tuning parameter. 
Such a `gapless' point on the phase diagram, whose Hamiltonian is given by Eqs.~(\ref{eq:kp-sym-spinor}) and~(\ref{eq:zszs-hamiltonian}), has been referred to as the Berry dipole. 

Although we showed in Sec.~\ref{sec:semi-metallic-transition-region} that transitions which alter topological invariants of crystalline Hopf insulators upon tuning one parameter ($\phi$) are generically \emph{not} point-like, we claim that the Berry dipole still constitutes a \emph{representative} phase transition in the sense that the generic phase transition altering the value of RTP or Hopf invariant can be continuously deformed into a Berry dipole. 
The goal of this Appendix is to formalize this `deformability' into Theorem~\ref{theorem:homotopy}. The discussion is organized into preliminaries (\app{app:prelimthmdeform}), statement of the theorem (\app{app:thmdeform}), proof of a lemma [Eq.~(\ref{eqn:hopfrtpctm})] relating the Hopf invariant and the Zak phase invariant (\app{sec:hopfzak}), and finally the proof of the theorem (\app{app:proofthmdeform}).

\subsection{Preliminaries}\la{app:prelimthmdeform}

For a $\textrm{P}n$-symmetric tight-binding Hamiltonian $h(\bk)$, we consider a topology-changing band touching occurring at (or near) an $m$-fold rotation-symmetric line $\gamma_{\Pi}$, with $m$ a nontrivial divisor of $n$. 
Let the band touching (and subsequent untouching) be tuned by a real Hamiltonian parameter $\phi$, such that $h(\bk,\phi)$ is $\textrm{P}n$-symmetric at each $\phi$. 
For notational convenience, we (1)~redefine the origin of $(k_x,k_y)$ to lie on $\gamma_{\Pi}$, (2)~depending on the context we sometimes utilize the complex coordinates $k_\pm=k_x \pm ik_y$ in place of $(k_x,k_y)$, and (3)~we combine the three-dimensional momentum ($\bk$) and the tuning parameter ($\phi$) into a single \emph{generalized four-momentum space}, $\boldsymbol{\mathbbm{k}}=(\bk;\phi)=(k_x,k_y,k_z;\phi)$.

Then the $\mathrm{P}n$-symmetric Hamiltonian at the $m$-fold rotation-symmetric line $\gamma_\Pi$ satisfies the constraint [cf.~Eq.~(\ref{eq:app:kp-sym-condition})]:
\begin{equation}
\forall  \boldsymbol{\mathbbm{k}}: R_{C_m}h_0(k_\pm,k_z,\phi)R_{C_m}^{-1}=h_0(e^{\pm 2\pi i/m} k_\pm,k_z,\phi),\label{eqn:assumed-rot-sym}
\end{equation}
with the itinerant rotation matrix
\e{ R_{C_m} = e^{-i\frac{\pi}{m}\Delta\widetilde{\mathcal{L}}\, \sigma_z}\label{eq:app:rot-matr}}
parametrized by  $\Delta\widetilde{\mathcal{L}}$, i.e., the mod-$m$-valued difference in itinerant angular momenta (of the two basis band representations) at $\Pi$. 
We further assume that the topology-changing band touching is bounded by a 3-spherical $\boldsymbol{\mathbbm{k}}$-region with radius $\kappa$:
\begin{equation}
S^{\! 3}_\kappa = \{ \boldsymbol{\mathbbm{k}}\,|\,k_x^2+k_y^2+k_z^2+\phi^2 = \kappa^2\}. \label{eqn:3-sphere-kappa}
\end{equation}
We explain below why and how the band touching can be continuously deformed inside $S^{\! 3}_\kappa$ so as to achieve the Berry-dipolar form [Eqs.~(\ref{eq:kp-sym-spinor}) and~(\ref{eq:zszs-hamiltonian})] throughout the $\boldsymbol{\mathbbm{k}}$-region bounded by a smaller three-sphere: $S^{\! 3}_{2\kappa/3}$, and with the sole gapless point enclosed in $S^{\! 3}_\kappa$ being the Berry dipole at the center of coordinates.

It is convenient to reformulate the integer-valued change of polarization at $\Pi$ [cf.~Eq.~(\ref{eq:pol-def})] as
\begin{align}
\delta \mathscr{P} &= {\frac{1}{2\pi}} \int_{-\pi}^{+\pi}\left[\mathcal{A}_z(0,0,k_z;\phi=\kappa) - \mathcal{A}_z(0,0,k_z;\phi=-\kappa) \right] dk_z \nonumber \\
&= {\frac{1}{2\pi}}\oint_{\boldsymbol{\mathbbm{k}}\,\in\,\gamma} \mathcal{A} = {\frac{1}{2\pi}} \int_{\boldsymbol{\mathbbm{k}}\,\in\, \mathcal{D}} \mathcal{F}, 
\label{eqn:rephrased-Polarizaion-change}
\end{align}
where $\gamma$ is a closed path that is the concatenation of (1) a nontrivial 1-cycle of the Brillouin torus extending in the positive $k_z$ direction at $k_x=k_y=0$ and $\phi=\kappa>0$, (2) nontrivial 1-cycle of the Brillouin torus extending in the opposite $k_z$-direction at $k_x=k_y=0$ and $\phi=-\kappa<0$, and (3) two finite segments connecting the two previous lines at $k_z = \pm \pi$, within the BZ boundary. 
In the penultimate expression of \q{eqn:rephrased-Polarizaion-change}, Stokes' theorem allows to equate the line integral of the Berry connection $1$-form $\mathcal{A}$ to an area integral of the Berry curvature $2$-form $\mathcal{F}$, with the integral domain $\mathcal{D}$ being an arbitrary two-dimensional open sheet (or `Gaussian surface') bounded by $\gamma\equiv \partial \mathcal{D}$, with the two-band Hamiltonian being gapped everywhere on $\mathcal{D}$. 
Although not explicitly indicated by Eq.~(\ref{eqn:rephrased-Polarizaion-change}), it is understood that components of $\mathcal{A}$ and $\mathcal{F}$ can be computed if one defines coordinates on $\gamma$ and~$\mathcal{D}$. 

Note that the quantization of the last expression in Eq.~(\ref{eqn:rephrased-Polarizaion-change}) [as well as in Eq.~(\ref{eqn:rephrased-Zak-phase-change}) in footnote~\ref{foot:Zak-on-S3}] to integers follows easily from the fact that the (spectrally flattened) Hamiltonian at all points of $\gamma$ is constant, $h_\textrm{flat} = -\sigma_z$, 
due to the decoupling of the two basis orbitals along the rotation-invariant manifold $k_x=k_y=0$ (note that this defines a \emph{rotation-invariant plane} inside the four-momentum space).  
As a consequence, the two-dimensional sheet $\mathcal{D}$ must wrap around the Bloch sphere an integer number of times.

For the purpose of the following discussion, we now consider a continuous deformation:  $\mathcal{D}\ri \mathcal{D}'_\kappa$, such that the boundary of the sheet remains 
pinned to the plane $k_x = k_y= 0$ throughout the deformation (thus ensuring the quantized value of the integral remains unchanged); specifically,
\begin{equation}
    \partial     \mathcal{D}'_\kappa= \{\boldsymbol{\mathbbm{k}}\,|\,k_x=0,k_y=0,k_z^2+\phi^2 = \kappa^2\},\la{partialdprime}
\end{equation}
where $\kappa>0$ is the same radius as in the definition of the $3$-sphere in Eq.~(\ref{eqn:3-sphere-kappa}). We can therefore express the change of polarization~as\footnote{We remark that an analogous reformulation also applies to the continuum Zak phase of $\bk\cdot\bp$ models with the continuous rotation symmetry. Departing from Eq.~(\ref{eq:pol-continuum}), we find it necessary to modify the integration bounds in Eq.~(\ref{eqn:rephrased-Polarizaion-change}) as $\frac{1}{2\pi}\int_{-\pi}^{+\pi}\mapsto \int_{-\infty}^{+\infty}$ 
i.e.,
\begin{equation}
\delta \mathscr{Z}^\textrm{cont.} 
= \int_{-\infty}^{+\infty} \left[\mathcal{A}_z(0,0,k_z;\phi=\kappa) - \mathcal{A}_z(0,0,k_z;\phi=-\kappa) \right] dk_z. 
\label{eqn:rephrased-Zak-phase-change}
\end{equation}
By repeating the same steps as for the lattice model, we derive Eq.~(\ref{eqn:rephrased-Polarizaion-change-2}) with the polarization on the left-hand side replaced by $\delta\mathscr{Z}^\textrm{cont.}/2\pi$.\label{foot:Zak-on-S3}
} 
\begin{equation}
\delta \mathscr{P}
= {\frac{1}{2\pi}}\int_{\boldsymbol{\mathbbm{k}}\,\in\,\mathcal{D'}}\mathcal{F}. \label{eqn:rephrased-Polarizaion-change-2}
\end{equation}
The advantage gained by the deformation is that the topological invariant is expressed as an integration over a compact submanifold of $\boldsymbol{\mathbbm{k}}$-space; namely on a two-dimensional slice 
\begin{equation}
\mathcal{D}'_\kappa = \{ \boldsymbol{\mathbbm{k}}\,|\,k_x\geq 0,k_y=0, k_x^2+k_z^2+\phi^2 = \kappa^2\} \label{eqn:2-slice-kappa}
\end{equation}
of $S^{\! 3}_\kappa$ defined in \q{eqn:3-sphere-kappa}.

Furthermore, note that the change in the continuum Hopf number, $\delta\chi^{\textrm{cont.}}$, has been similarly reformulated as an integral of the Chern-Simons form on $S^{\! 3}_\kappa$ in Eq.~(\ref{eqn:Hopf-on-S3}). 
By employing the Hopf-Chern relation for the Hopf invariant (introduced above in \s{sec:linking-RTP-Hopf} and detailed below in \app{app:hopf-rtp-whitehead}), we derive that the two invariants on $S^3_\kappa$ are constrained via the
\begin{equation}
{\textrm{\emph{Hopf-Zak lemma}}:\quad} \delta \chi =_m \delta\mathscr{P}\,\Delta\widetilde{\mathcal{L}}
\label{eqn:hopfrtpctm}
\end{equation}
where $\delta\widetilde{\mathcal{L}}$  is the mod-$m$-valued difference of the itinerant angular momenta of the two basis band representations at $\Pi$.\footnote{In the continuum model, cf.~footnote~\ref{foot:Zak-on-S3}, an analogous constraint holds with $\delta\widetilde{\mathcal{L}}$ replaced by the Berry-dipole spin $\Delta \ell$ and with $=_m$ replaced by a strict equality.} {To maximize coherence in presentation, we prove the Hopf-Zak lemma separately in Appendix~\ref{sec:hopfzak}.}

\subsection{Theorem on deformability to a Berry dipole }\la{app:thmdeform} 

Our proof of the theorem relies on the following assumption that we believe holds but are presently unable to justify by algebraic-topological methods:
\medskip  

\noindent \emph{\emph{Assumption:} The topological classification of a gapped Hamiltonian on $3$-sphere $S^{\! 3}_\kappa$ [defined in Eq.~(\ref{eqn:3-sphere-kappa})] with rotation symmetry $C_m: k_x + i k_y \mapsto e^{2\pi i/m}(k_x + i k_y)$ is exhausted by the Hopf invariant on $S^{\! 3}_\kappa$ [defined in Eq.~(\ref{eqn:Hopf-on-S3})] and by the polarization (Zak phase) invariant on the two-dimensional slice $\mathcal{D'}$ of the $3$-sphere [defined in Eqs.~(\ref{eqn:rephrased-Polarizaion-change-2},\ref{eqn:2-slice-kappa})], with the two invariants constrained by Eq.~(\ref{eqn:hopfrtpctm}}).\medskip

\noindent The nontrivial assumption lies in the \emph{completeness} of the classification; in contrast, Eq.~(\ref{eqn:hopfrtpctm}) (proven in Appendix~\ref{sec:hopfzak}) is incontrovertible.

\begin{theorem} \label{theorem:homotopy}
We assume a continuous two-band Hamiltonian $h_0(\boldsymbol{\mathbbm{k}})$ that:
\begin{itemize} 
\item[(1)] respects $C_{m}$ rotation symmetry as in \q{eqn:assumed-rot-sym},
\item[(2)] preserves the spectral gap on 3-sphere ${S}^{\! 3}_\kappa$, and 
\item[(3)] carries non-vanishing Hopf invariant $\delta \chi$ on ${S}^{\! 3}_\kappa$, and a unit-change polarization invariant ($\delta\mathscr{P} \!=\! \pm 1$) on $\mathcal{D}'_\kappa$. 
\end{itemize}
\noindent Then there exists a continuous homotopy $\mathfrak{H}(\boldsymbol{\mathbbm{k}};t)$ that:
\begin{itemize}
\item[(\emph{i})] fulfills properties (1--3)~stated above for all $t\in[0,1]$, 
\item[(\emph{ii})] coincides with the Hamiltonian $h_0(\boldsymbol{\mathbbm{k}})$ at $t=0$ for all $\boldsymbol{\mathbbm{k}}$, as well as at distances $\norm{\boldsymbol{\mathbbm{k}}} > \kappa$ for all $t$, and
\item[(\emph{iii})] at $t=1$ it is equal to the spinor-form Hamiltonian given by Eqs.~(\ref{eq:kp-sym-spinor}) and~(\ref{eq:zszs-hamiltonian}) with
\begin{equation}
    \qquad\;\textrm{spin} \;\Delta\ell^\mathrm{t} = \delta\chi\,\delta \mathscr{P}\quad \textrm{and}\;\;\textrm{helicity} \;\upsilon^\mathrm{t} = -\delta\mathscr{P}\!\!
    \label{eqn:deformed-dipole-invariants}
\end{equation}
for distances $\norm{\boldsymbol{\mathbbm{k}}} \leq \tfrac{2}{3}\kappa$, and is gapped for $0<\norm{\boldsymbol{\mathbbm{k}}}\leq \kappa$. 
\end{itemize} 
\noindent We hereafter label the final Hamiltonian $\mathfrak{H}(\boldsymbol{\mathbbm{k}};t{=}1)$ as $h_3(\boldsymbol{\mathbbm{k}})$. 
The condition (\emph{iii}) implies that $h_3(\boldsymbol{\mathbbm{k}})$ exhibits a Berry dipole at $\boldsymbol{\mathbbm{k}}=\boldsymbol{0}$.
\medskip  
\end{theorem}

\noindent \emph{Remark:} Note that due to Eq.~(\ref{eq:Berry-dipole-pol-change}) and per the correspondence between $\delta\mathscr{P}$ and $\delta\mathscr{Z}^\textrm{cont.}/(2\pi)$, only phase transitions with $\delta\mathscr{P}=\pm 1$ can potentially be deformed into a Berry dipole. 
If the nodal $\boldsymbol{\mathbbm{k}}$-manifold mediating the phase transition exhibits a larger value of $\delta\mathscr{P}$, it is necessary to first split the nodal $\boldsymbol{\mathbbm{k}}$-manifold into several pieces, each satisfying $\delta\mathscr{P}=\pm 1$, and then to deform each piece to a \emph{distinct} Berry dipole.\medskip

\subsection{Proof of Hopf-Zak lemma}\la{sec:hopfzak}

Before proving the Theorem~\ref{theorem:homotopy} (task deferred to Appendix~\ref{app:proofthmdeform}), our first goal here is to prove the Hopf-Zak lemma stated in \q{eqn:hopfrtpctm}. 
Our proof is based on applying the Hopf-Chern relation introduced in Sec.~\ref{sec:linking-RTP-Hopf} [see also Eq.~(\ref{eq:whitehead}) and the discussion in that appendix].
To recapitulate, for Hamiltonian maps from $S^{\! 3}$ to the Bloch sphere $S^{\! 2}$, the Hopf-Chern relation equates the Hopf invariant over $S^{\! 3}$ with the product of (\emph{i}) the Chern number of a Gaussian surface (within $S^{\! 3}$) bounded by an oriented, path-connected preimage loop $\gamma(b)$ of an arbitrary point $b$ on the Bloch sphere, and (\emph{ii}) the positive-integer-valued multiplicity of the preimage loop, which was previously defined as $\mu[\gamma(b)]$ in \s{sec:linking-RTP-Hopf}. 
If the preimage of $b$ is a disjoint union of several path-connected loops, ${h_\textrm{flat}^{-1}(b) = } \cup_i \gamma_i(b)$, then the Hopf invariant equals the sum of (\emph{i}) $\times$ (\emph{ii}) over all the preimage loops. 

In our application, we choose $S^{\! 3}$ to be the $3$-sphere $S^{\! 3}_{\kappa}$ defined in \q{eqn:3-sphere-kappa} and $b$ to be the `south pole' {($b=-\sigma_z$)} of the Bloch sphere, such that the oriented preimage of $b$ consists of (a) $\gamma(b)=\pm \partial\mathcal{D}'_\kappa$, with $\partial\mathcal{D}'_\kappa$ defined in \q{partialdprime} and the $\pm$ prefactor reflecting that $\gamma(b)$ and $\partial\mathcal{D}'_\kappa$ may have equal or opposite orientations, plus (b) any number (possibly zero) of rotation-symmetric $m$-plets of preimage loops, plus (c) any number  of individually-$C_m$-invariant preimage loops that link with the rotation-invariant manifold $\partial\mathcal{D}'_\kappa$. That $\pm \partial     \mathcal{D}'_\kappa$ belongs to the preimage of the south pole follows from the spectrally-flattened Hamiltonian in the rotation-invariant manifold being~simply $-\sigma_z$. 
\medskip

\noindent Two facts follow from our setting.  
\medskip

\noindent  \textit{Fact one:} the Chern number over an oriented Gaussian surface  bounded by $\gamma(b)$ equals the polarization invariant $\delta \mathscr{P}$ if $\gamma(b)= \partial\mathcal{D}'_\kappa$, and otherwise equals $-\delta \mathscr{P}$ if $\gamma(b)=- \partial\mathcal{D}'_\kappa$. This follows from Eq.~(\ref{eqn:rephrased-Polarizaion-change-2}), with the integration domain $\mathcal{D}'_\kappa$ specified in Eq.~(\ref{eqn:2-slice-kappa}). 
\medskip

\begin{figure}
    \centering
    \includegraphics[width=0.3\textwidth]{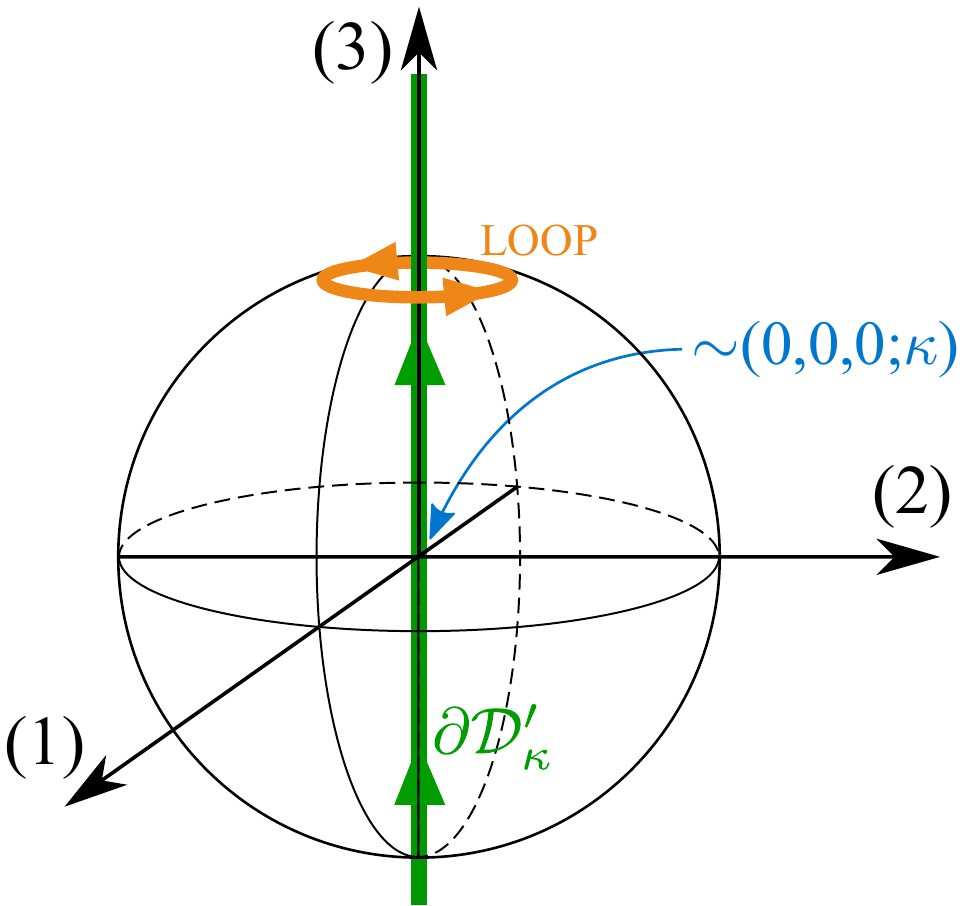}
    \caption{Linking of $\tinyloop$ with $\partial\mathcal{D}'_\kappa$ in the stereographic projection defined by Eq.~(\ref{eqn:steregraphic}). The projection represents $S^{\! 3}_\kappa$ [Eq.~(\ref{eqn:3-sphere-kappa})] as the Euclidean space $\mathbb{R}^3\cup \infty$. The central point (indicated with blue arrow) corresponds to the point $(0,0,0;\kappa)$ on the $3$-sphere, whereas `$\infty$'  corresponds to the point $(0,0,0;-\kappa)$. The axis numbering indicates the components of the projection space in Eq.~(\ref{eqn:steregraphic}). The green oriented line (compactified at $\infty$) corresponds to $\partial\mathcal{D}'_\kappa$ [Eq.~(\ref{eqn:stereo-partial-D})], and the orange oriented loop represents $\tinyloop$ [Eq.~(\ref{eqn:lemma-tiny-loop})]. The indicated $2$-spherical shell corresponds to the cross-section of $S^{\! 3}_\kappa$ with $\phi=0$. Stereographic projection of $\mathcal{D}'_\kappa$ [Eq.~(\ref{eqn:2-slice-kappa}), not illustrated] spans the semi-infinite plane in the positive (1)-direction.
    }
    \label{fig:stereographic}
\end{figure}

\begin{figure*}
    \centering
    \includegraphics[width=0.84\textwidth]{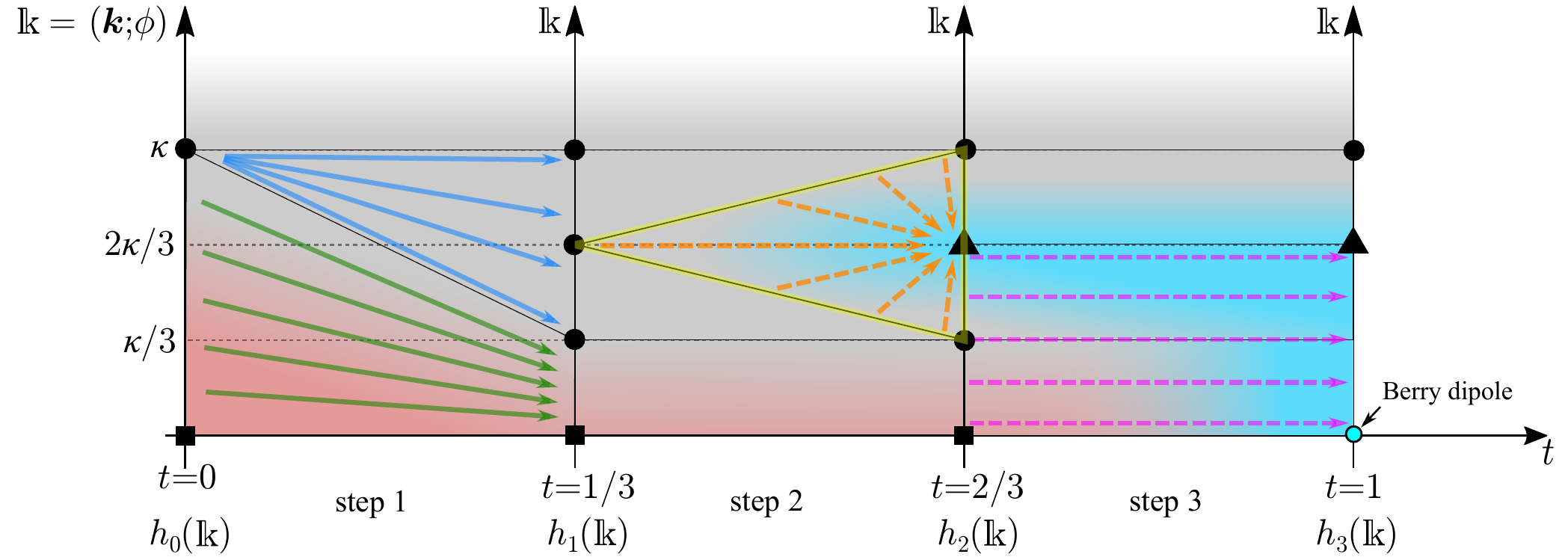}
    \caption{
    Schematic illustration of the continuous homotopy $\mathfrak{H}(\boldsymbol{\mathbbm{k}};t)$ from the original model $h_0(\boldsymbol{\mathbbm{k}})$ at $t=0$, which we leave unconstrained except for respecting the rotation symmetry and spectral gap [cf.~conditions (1--3) in Appendix~\ref{app:phase-transition-general}],  
    into the spinor-form Hamiltonian given by Eqs.~(\ref{eq:kp-sym-spinor}) and~(\ref{eq:zszs-hamiltonian}) where the topological transition occurs via a single gapless point -- the Berry dipole.
    The vertical axis represents the four directions, $(\bk;\phi)\equiv\boldsymbol{\mathbbm{k}}$, and each point in the $(\boldsymbol{\mathbbm{k}};t)$-plane represents a Hamiltonian on a $3$-sphere with radius given by the vertical coordinate of the point. 
    In step~1, we compress the Hamiltonian inside the ball $\norm{\boldsymbol{\mathbbm{k}}} < \kappa$ to a ball with smaller radius $\norm{\boldsymbol{\mathbbm{k}}}<\kappa/3$ [solid green arrows], while making the Hamiltonian constant in the radial direction within the annulus $\kappa/3 < \norm{\boldsymbol{\mathbbm{k}}} < \kappa$ [solid blue arrows]. 
    In step~2, we perform a rotation-symmetry-preserving deformation from the obtained Hamiltonian into the spinor-form Hamiltonian on a $3$-sphere with radius $\tfrac{2}{3}\kappa$. 
    The continuous deformation is ``wedged'' (yellow triangle) into the annulus such the Hamiltonian evolution is identical along every dashed orange arrow. 
    Finally, in step~3, we perform a linear interpolation of the Hamiltonian inside ball $\norm{\boldsymbol{\mathbbm{k}}}<\tfrac{2}{3}\kappa$ [dashed magenta arrows], which results in the desired spinor-form Hamiltonian with a Berry-dipole transition at $\boldsymbol{\mathbbm{k}}=\boldsymbol{0}$ (cyan dot). 
    To help visualize the deformations, we color the background of the diagram according to the following rules: (\emph{i}) The initial Hamiltonian at $\norm{\boldsymbol{\mathbbm{k}}}=\kappa$ (where spectral gap is guaranteed by assumptions) is indicated with gray background. 
    (\emph{ii}) The initial Hamiltonian for $\norm{\boldsymbol{\mathbbm{k}}}<\kappa$ (which may exhibit band nodes) is shown in shades of red.
    (\emph{iii}) The spinor-form Hamiltonian (with single gapless point -- the Berry dipole) is highlighted in cyan. The homotopy function  $\mathfrak{H}(\boldsymbol{\mathbbm{k}};t)$ may exhibit accidental nodes \emph{inside} the region marked with magenta arrows, but these do not alter the considered topological invariants and are thereby inconsequential for the validity of the proof. 
    In regions where no arrow is indicated, the homotopy function is constant along $t$. 
    The black dot ($\mathlarger{\mathlarger{\mathlarger{\bullet}}}$), black square ($\blacksquare$), resp.~black triangle ($\mathlarger{\mathlarger{\blacktriangle}}$) indicate selected points in the $(\boldsymbol{\mathbbm{k}};t)$ diagram where the Hamiltonian is the same on the 3-spherical shell around the origin.
    For details of the individual steps, see the text of Appendix~\ref{app:phase-transition-general}. To simplify the detailed specification of the homotopy, we denote the Hamiltonians $\mathfrak{H}(\boldsymbol{\mathbbm{k}};t)$ at times $t\in\{\tfrac{1}{3},\tfrac{2}{3},1\}$ respectively as $h_{1,2,3}(\boldsymbol{\mathbbm{k}})$.
    }
    \label{fig:homotopy}
\end{figure*}

\noindent \textit{Fact two}: the multiplicity  $\mu[\gamma(b)]=\pm (\Delta\widetilde{\mathcal{L}}+{r}
\,m)>0$ for some integer  
{$r$} if $\gamma(b)=\pm \partial\mathcal{D}'_\kappa$.
Recall here that $\Delta\widetilde{\mathcal{L}}$ is the mod-$m$ parameter entering through the itinerant rotation matrix $R_{C_m}$ in \q{eqn:assumed-rot-sym}. 
To justify this statement, consider an infinitesimal rotation-invariant $\boldsymbol{\mathbbm{k}}$-loop ${\tinyloop}\subset S^{\!3}_\kappa$  linking with $\partial \mathcal{D}'_\kappa$ and {centered at $\boldsymbol{\mathbbm{k}}=(0,0,k_z=\kappa;\phi=0)$, namely}
\begin{eqnarray}
    \tinyloop:  [0,1] &\to& S^{\! 3}_\kappa \nonumber \\
    t &\mapsto& \left(\varrho \cos(2\pi t),\varrho \sin(2\pi t),\!\!\sqrt{\kappa^2{-}\varrho^2};0\right),\quad \label{eqn:lemma-tiny-loop}
\end{eqnarray}
with the loop radius $\varrho$ approaching zero. By describing the loop as `rotation-invariant', we mean that the parametrization is chosen such that $m$-fold rotation $C_m$ preserves the loop as a whole, but advances $t \mapsto t+1/m$.
Our convention is that $C_m$ produces a clockwise rotation in the $(k_x,k_y)$-plane [assuming that $(k_x,k_y,k_z)$ form a right-handed coordinate system]. 

We claim that $\tinyloop$ [Eq.~(\ref{eqn:lemma-tiny-loop})] is linked with $\partial\mathcal{D}'_\kappa$ [Eq.~(\ref{partialdprime})] inside the $3$-manifold $S^{\! 3}_\kappa$, and that [per the discussion in Sec.~\ref{sec:linking-RTP-Hopf}] we can use it to determine the orientation and the multiplicity of $\partial\mathcal{D}'_\kappa\subset h_\textrm{flat}^{-1}(b)$. To visualize this linking, we  represent $S^{\!3}$ as $\mathbb{R}^3 \cup \infty$ in a \emph{stereographic projection}, choosing $(0,0,0;\kappa)$ to be mapped to the origin of $\mathbb{R}^3$ and $(0,0,0;-\kappa)$ to be mapped to infinity: 
\begin{eqnarray}
\mathfrak{P}: S^{\! 3}_\kappa &\to& \mathbb{R}^3 \cup \infty \nonumber \\
(k_x,k_y,k_z;\phi) &\mapsto& \frac{\kappa}{\kappa+\phi}(k_x,k_y,k_z).\; \label{eqn:steregraphic}
\end{eqnarray}
Observe that this transforms $\partial \mathcal{D}'_\kappa$ (parametrized\footnote{Our convention is that $\varpi$ increases along $\partial\mathcal{D}'_\kappa$ in the direction parallel to the orientation of $\partial\mathcal{D}'_\kappa$; the orientation of $\partial\mathcal{D}'_\kappa$ is in turn determined by continuous deformation of $\gamma=\partial \mathcal{D}$, which is oriented 
according to Eq.~(\ref{eqn:rephrased-Polarizaion-change}): note that the $k_z$-component of $\partial\mathcal{D}$ increases (decreases) for $\phi=+\kappa$ ($\phi=-\kappa$), as is consistent with Eq.~(\ref{eqn:stereo-partial-D}).} 
by $\varpi\in[-\pi,\pi]$) as
\begin{eqnarray}
\mathfrak{P}: (0,0,\kappa \sin\varpi;\kappa\cos\varpi) \mapsto \left(0,0,\kappa \tan\frac{\varpi}{2}\right), \label{eqn:stereo-partial-D}
\end{eqnarray}
i.e., into a straight vertical line [compactified at `$\infty$', which corresponds to the point $(0,0,0;-\kappa)\in S^{\! 3}_\kappa$] oriented in the `upward' direction, i.e., along the positive direction of the third axis of the stereographic projection space ($\mathbb{R}^3$). Furthermore, since $\tinyloop$ is located within $\phi=0$, it follows from Eq.~(\ref{eqn:steregraphic}) that its $(\bk;\phi)$-coordinates in $\mathbb{R}^4$ [defined in Eq.~(\ref{eqn:lemma-tiny-loop})] 
are identical to its coordinates in $\mathbb{R}^3\cup \infty$, 
meaning that $\tinyloop$ encircles the vertical coordinate axis. The situation is illustrated in Fig.~\ref{fig:stereographic}.

{Setting the radius $\varrho$ of $\tinyloop$ to be arbitrarily small},
the spectrally-flattened Hamiltonian $h_\textrm{flat}(\boldsymbol{\mathbbm{k}})$ on $\tinyloop$ has form
\begin{equation}
{h_\textrm{flat}}(t)={-\sigma_z} + \delta h(t),
\quad \textrm{where}\quad 
\delta h=\delta h_x \sigma_x+\delta h_y\sigma_y
\end{equation} 
is a correction linear in ${\varrho}$ 
to the leading order. 
The symmetry constraint in \q{eqn:assumed-rot-sym} then implies $R_{C_m}\delta h(t)R_{C_m}^{-1}=\delta h(t+1/m)$ for any $t\in [0,1]$. In combination with Eq.~(\ref{eq:app:rotation-on-sigma-0}), we derive
\e{ \vectwo{\delta h_x}{\delta h_y}\bigg|_{t+1/m}=\matrixtwo{\cos \tfrac{2\pi \Delta\widetilde{\mathcal{L}}}{m}}{-\sin \tfrac{2\pi \Delta\widetilde{\mathcal{L}}}{m}}{+\sin \tfrac{2\pi \Delta\widetilde{\mathcal{L}}}{m}}{\cos \tfrac{2\pi \Delta\widetilde{\mathcal{L}}}{m}} \vectwo{\delta h_x}{\delta h_y}\bigg|_{t}.\la{homo}}
Viewing the normalized  $(\delta h_x,\delta h_y)$ as a unit vector on the equator of the Bloch sphere, \q{homo} implies that the unit vector revolves ($Rev=\Delta\widetilde{\mathcal{L}} + mr$) times in the equatorial plane as $t$ is continuously advanced by unity, with $r$ an integer. 
The sign of $Rev$ determines the relative orientation of $\gamma(b)=\text{sgn}[Rev]\partial \mathcal{D}'_\kappa$, according to our definition of preimage orientation in \s{sec:linking-RTP-Hopf}. 
Equivalently, one may say that the preimage multiplicity $\mu[\gamma(b)]=\pm Rev$ if $\gamma(b)=\pm \partial \mathcal{D}'_\kappa$.
\medskip

By combining \textit{Facts 1} and \textit{2}, the Hopf-Chern relation implies that the preimage loop $\partial     \mathcal{D}'_\kappa$ `contributes' $(\pm \delta\mathscr{P})[\pm (\Delta\widetilde{\mathcal{L}}+m r 
)] =_m \delta\mathscr{P}\Delta\widetilde{\mathcal{L}}$ 
to the Hopf invariant, if we view the Hopf invariant as a sum of contributions from each preimage loop $\gamma_i(b=-\sigma_z)$ in ${h_\textrm{flat}^{-1}(b) = } \cup_i \gamma_i(b)$. 
The remaining contributions stem [cf.~second paragraph of Sec.~\ref{sec:hopfzak}] from: 
(b) the possibility of $m$-plets of preimage loops, which contribute $m\mathbb{Z}$ to the Hopf invariant (since Berry curvature transforms as a vector under rotations, each loop in the $m$-plet contributes the same amount), and from 
(c) the possibility of $C_m$-invariant singlets of preimage loops, each of which necessarily encloses $m\mathbb{Z}$ quanta of Berry flux, according to an argument presented in Sec~V. of the Supplemental Material of \ocite{Nelson:2021}. Thus we arrive at the Hopf-Zak lemma in \q{eqn:hopfrtpctm}.

\subsection{Proof of Theorem~\ref{theorem:homotopy}}\la{app:proofthmdeform}

Let us describe how to construct the homotopy function $\mathfrak{H}(\boldsymbol{\mathbbm{k}};t)$ introduced in Theorem~\ref{theorem:homotopy}. This is achieved in three steps that can be carried out for any initial Hamiltonian $h_0(\boldsymbol{\mathbbm{k}})$.

First, in the interval $t\in[0,\tfrac{1}{3}]$ we consider a continuous deformation that ``inflates'' the 3-spherical shell $S_\kappa^{\! 3}$ into an annulus with outer and inner radii $\kappa$ resp.~$\kappa/3$ [solid blue arrows in Fig.~\ref{fig:homotopy}], while ``compressing'' the interior of the 3-sphere [solid green arrows in Fig.~\ref{fig:homotopy}]. The Hamiltonian values are assumed to follow this continuous deformation. To describe this process explicitly: 
\begin{itemize}
\item[(1)] for $\norm{\boldsymbol{\mathbbm{k}}}>\kappa$ we keep the Hamiltonian unchanged, $\mathfrak{H}(\boldsymbol{\mathbbm{k}};t) = h_0(\boldsymbol{\mathbbm{k}})$,
\item[(2)] for $\norm{\boldsymbol{\mathbbm{k}}}<\kappa$ we translate the Hamiltonian function in the radial direction by two thirds of the distance towards the origin, $\mathfrak{H}\left(\boldsymbol{\mathbbm{k}}(1-2t);t\right) = h_0(\boldsymbol{\mathbbm{k}})$, and 
\item[(3)] for momenta $\kappa>\norm{\boldsymbol{\mathbbm{k}}}>\kappa(1-2t)$ located within the annulus we extend the value from the sphere with radius $\kappa$, $\mathfrak{H}\left(\boldsymbol{\mathbbm{k}};t\right) = h_0(\kappa\, \boldsymbol{\mathbbm{k}} /\norm{\boldsymbol{\mathbbm{k}}})$. 
\end{itemize}
One can easily check that this simple geometric manipulation produces a continuous homotopy function on $t\in[0,\tfrac{1}{3}]$ that is compatible with conditions (\emph{i}) and (\emph{ii}) above. For brevity, we denote $\mathfrak{H}(\boldsymbol{\mathbbm{k}},t{=}\tfrac{1}{3})\equiv h_1(\boldsymbol{\mathbbm{k}})$.

For the purpose of the second step, $t\in[\tfrac{1}{3},\tfrac{2}{3}]$, we first construct a continuous homotopy deformation on a $3$-sphere, $\overline{\mathfrak{H}}(\boldsymbol{\mathbbm{k}}/\norm{\boldsymbol{\mathbbm{k}}};t')$ where $t'\in[0,1]$. (Note that $\overline{\mathfrak{H}}$ depends only on the angular $3$-spherical coordinates and \emph{not} on the magnitude of $\boldsymbol{\mathbbm{k}}$; furthermore, parameter $t'$ is, at this stage, unrelated to $t$). 
The particular homotopy function we need connects the initial Hamiltonian $h_0(\boldsymbol{\mathbbm{k}})$ on $S^{\! 3}_\kappa$ at $t'=0$ to the final spinor-form Hamiltonian $h_3(\boldsymbol{\mathbbm{k}})$ [defined by Eqs.~(\ref{eq:kp-sym-spinor},\ref{eq:zszs-hamiltonian})] on $3$-sphere $S^{\! 3}_{2\kappa/3}$ with radius $2\kappa/3$ at $t'=1$. 
Although we do not provide an explicit expression for such a homotopy function, the existence of the homotopy is guaranteed (per the Assumption formulated in the beginning of Appendix~\ref{app:thmdeform}) if the Hamiltonians exhibit the same value of both the Hopf and the polarization (Zak phase) invariant. To that end, notice that [from Eqs.~(\ref{eq:Berry-dipole-hopf-change}) and~(\ref{eqn:deformed-dipole-invariants})] the change in the Hopf invariant 
\begin{equation}
\chi^\mathrm{t} = -\upsilon^\mathrm{t} \, \Delta\ell^\mathrm{t} = \delta \chi \, (\delta \mathscr{P})^2 = \delta \chi \label{eqn:compare-hopf}
\end{equation}
(cf.~the Remark appearing under the statement of Theorem~\ref{theorem:homotopy} in Appendix~\ref{app:thmdeform} to clarify the last step), and [from Eqs.~(\ref{eq:Berry-dipole-pol-change}) and~(\ref{eqn:deformed-dipole-invariants})] the change of the polarization (Zak phase) invariant is 
\begin{equation}
\delta \mathscr{P}^\mathrm{t} = -\upsilon^\mathrm{t} = \delta\mathscr{P}.\label{eqn:compare-pol}
\end{equation}
Equations~(\ref{eqn:compare-hopf}) and~(\ref{eqn:compare-pol}) confirm that the target Berry-dipole Hamiltonian $h_3(\boldsymbol{\mathbbm{k}})$ at $\norm{\boldsymbol{\mathbbm{k}}}=2\kappa/3$ indeed exhibits the same pair of topological invariants on $S^{\! 3}_{2\kappa/3}$ as the initial Hamiltonian $h_0(\boldsymbol{\mathbbm{k}})$ on $S^{\! 3}_{\kappa}$; thus, the required homotopy function $\overline{\mathfrak{H}}(\boldsymbol{\mathbbm{k}}/\norm{\boldsymbol{\mathbbm{k}}};t')$ exists. For completeness, let us remark that
\begin{equation}
\Delta \ell 
\,\,\stackrel{\textrm{Eq.~(\ref{eqn:deformed-dipole-invariants})}}{=} \,\,
\delta \chi \,\delta\mathscr{P} 
\,\, \underset{\textrm{(Appendix~\ref{app:thmdeform})}}{\overset{\textrm{Rem.~1}}{=}}\,\,
\frac{\delta\chi }{\mathscr{P}}
\,\, \stackrel{\textrm{Eq.~(\ref{eqn:hopfrtpctm})}}{{\phantom{_m}}=_m} \Delta \widetilde{\mathcal{L}}
\end{equation}
is consistent with the $C_m$ rotation symmetry of the crystalline model [compare to Eq.~(\ref{eqn:l-L-mod-m})].

Knowing that the required homotopy $\overline{\mathfrak{H}}(\boldsymbol{\mathbbm{k}}/\norm{\boldsymbol{\mathbbm{k}}};t')$ exists, we ``wedge'' it
into the annulus as illustrated with dashed orange arrows in Fig.~\ref{fig:homotopy}, i.e., such that the homotopy function $\overline{\mathfrak{H}}$ with range $t'\in[0,1]$ is realized along each straight line connecting the outer boundary of the wedge (yellow triangle in Fig.~\ref{fig:homotopy}) to the point $\norm{\boldsymbol{\mathbbm{k}}}=\tfrac{2}{3}\kappa$, $t=\tfrac{2}{3}$. 
Since we preserve the rotation symmetry for all $(\phi,t)$ and keep the energy gap on $S^{\! 3}_\kappa$, it follows that we did not change the topological invariants of the insulating Hamiltonian outside $S^{\! 3}_\kappa$, and the validity of points (1--3) and (i--iii) follows easily. 
We denote the resulting Hamiltonian after step $2$ as $\mathfrak{H}(\mathbbm{k},t=\tfrac{2}{3})\equiv h_2(\mathbbm{k})$. 

Finally, in the third step,  $t\in[\tfrac{2}{3},1]$, we perform inside the four-dimensional ball
\begin{equation}
\mathcal{B}_{2\kappa/3}^{\! 4}=\{\boldsymbol{\mathbbm{k}}\;| \; \norm{\boldsymbol{\mathbbm{k}}}\leq \tfrac{2}{3}\kappa\}
\end{equation}
a linear interpolation of the Hamiltonian $h_2(\boldsymbol{\mathbbm{k}})$ to the spinor-form Hamiltonian defined by Eqs.~(\ref{eq:kp-sym-spinor}) and~(\ref{eq:zszs-hamiltonian}) [dashed magenta arrows in Fig.~\ref{fig:homotopy}].
Explicitly,
\begin{equation}
\mathfrak{H}(\boldsymbol{\mathbbm{k}};t) = (3-3t)h_2(\boldsymbol{\mathbbm{k}}) + (3t-2)h_3(\boldsymbol{\mathbbm{k}}).
\end{equation}
Since $h_2(\boldsymbol{\mathbbm{k}})$ and $h_3(\boldsymbol{\mathbbm{k}})$ readily coincide on the 3-sphere $S^{\! 3}_{2\kappa/3} = \partial \mathcal{B}_{2\kappa/3}^4$, the deformation is clearly continuous.
Furthermore, while the linear interpolation may in principle create accidental gapless points for intermediate values of $t$, these do not interfere with the energy gap on $S^{\! 3}_\kappa$; in particular, all such accidental band nodes are absent at the final time $t=1$. Since the described linear interpolation is manifestly compatible with the rotation symmetry, it follows that the third step is also compatible with all the rules (1--3) and (i--iii), and especially that the Hamiltonian $\mathfrak{H}(\boldsymbol{\mathbbm{k}};t=1)\equiv h_3(\mathbbm{k})$ realizes the Hopf transition via rotation-symmetric Berry dipole.

By combining the three steps in sequence, we obtain a concrete strategy to continuously deform a general Hopf-invariant-altering Hamiltonian to the prescribed Berry-dipolar form while not breaking the rotation symmetry and while preserving the energy gap outside a specified distance from the point $\boldsymbol{\mathbbm{k}}=\boldsymbol{0}$. This completes the proof of Theorem~\ref{theorem:homotopy}.

\section{Relation between the Zak phase and the positional center of the Wannier orbital}\label{app:pol-position}

In this Appendix we show for the (possibly) aperiodic, two-band, fully-gapped Hamiltonian [cf.\ Eq.~(\ref{eq:ham_non-periodic})] that the geometric phase defined by Eq.~\eqref{eq:pol-def} at a rotation-invariant line $\gamma_\Pi$ is equal modulo integer to the position $z_{\alpha'}$ of the basis orbital $\varphi_{\alpha'}$, assuming that at $\gamma_\Pi$-line (\emph{i}) the itinerant angular momenta of the valence and conduction bands are distinct and (\emph{ii}) that the valence band representation coincides with a BBR$[\varphi_{\alpha'}]$ at $\gamma_\Pi$. This is a particular example of the more general relation in \q{berrypol}, proven in Appendix~\ref{app:geom-theo-pol}, restricted to the case when the Fourier transform is performed only in one direction, namely along the $z$-axis, while the reduced momenta $\bk_\perp$ are kept as external parameters. In this case the geometric theory of polarization gives the following relation
\begin{align}
    \int \frac{dk_z}{2\pi}&\bra{u_j(\bk_\perp,k_z)}\ket{i\partial_{k_z}u_j(\bk_\perp,k_z)}\lin 
    &=_1\Big\langle\Big\langle\!\bra{\mathcal{V}_{j,\bk_\perp,0}}\hat{z} 
    \ket{\mathcal{V}_{j,\bk_\perp,0}}\!\Big\rangle\Big\rangle,
    \label{eq:app:berrypol-1D}
\end{align}
where the Bravais vector in $z$ direction assumed to be unit, $\big|{\mathcal{V}_{j,\bk_\perp,0}}\big\rangle\big\rangle\big\rangle$ are the intra-column states, with the corresponding functions defined in Eq.~\eqref{eq:app:intra-column-orbitals} 
and the position operator in the $z$ direction can be decomposed in the intra-column basis as
\begin{equation}
    \hat{z} = \sum_{R_z, \alpha}\ket{R_z,\alpha}\rangle\rangle(R_z+z_\alpha)\langle\langle\bra{R_z,\alpha}.
\end{equation}
The relation in Eq.~\eqref{eq:app:berrypol-1D} can be proven in analogous way as relation in Eq.~\eqref{berrypol}. 

Along the rotation-invariant lines $\gamma_\Pi$ with distinct itinerant angular momenta, the two basis orbitals do not hybridize, and hence, the valence Bloch state is simply a basis state $\ket{\psi_{v,\Pi,k_z}}\!\Big\rangle=\ket{\Pi,k_z,{\alpha'}}\rangle$, where the valence band representation, restricted to reduced momentum $\Pi$ coincides with the basis band representation induced from the basis orbital $\varphi_{\alpha'}$, and also restricted to~$\Pi$. From the valence Bloch state we can compute the intra-column state at reduced momentum $\Pi$ to be $\ket{\mathcal{V}_{v,\Pi,R_z}}\!\Big\rangle\Big\rangle=\ket{R_z, \alpha'}\rangle\rangle$. Plugging this state at $R_z=0$ into Eq.~\eqref{eq:app:berrypol-1D} we can straightforwardly compute that the right-side of the equation is equal to the $z$ position of the basis orbital $\varphi_{\alpha'}$ and we get
\begin{equation}
    \mathscr{P}(\Pi)=\frac{1}{2\pi}\int dk_z \bra{u_v(\Pi,k_z)}\ket{i\partial_{k_z}u_v(\Pi,k_z)} =_1 z_{\alpha'}.
    \label{eq:app:pol-position}
\end{equation}

If on the contrary we choose to work in the periodic convention for the tight-binding Hamiltonian (introduced in Appendix~\ref{app:nonperiodic-convention}), the position operator in $z$ direction will loose the information about the spacial positions of the basis orbitals:
\begin{equation}
    \hat{z} 
    = \sum_{R_z, \alpha}\ket{R_z,\alpha}\rangle\rangle R_z\langle\langle\bra{R_z,\alpha}.
\end{equation}
In this case the geometric phase defined by Eq.~\eqref{eq:pol-def} is not equal to the position of the basis orbital $z_{\alpha'}$, but instead acquires
an integer value.

\section{Chern number from itinerant angular momentum }\label{app:Chern-via-ell}

In this Appendix we show that in the presence of $n$-fold rotational symmetry, given itinerant angular momenta of a set of bands, $\widetilde{\mathcal{L}}_j(\Pi)$, 
their total Chern number on a symmetry-preserving two-dimensional cut is constrained by 
\begin{equation}
\mathscr{C}_{xy}=_n\sum_j\sum_{\Pi}\widetilde{\mathcal{L}}_j(\Pi),
\label{eq:app:Chern-via-ell}
\end{equation} 
where $\Pi$ runs over all rotation-invariant momenta in rBZ and $j$ is an index for all occupied bands.

To see that this is true we utilize the result of Ref.~\cite{Chen_bulktopologicalinvariants} which relates the Chern number of a rotation-symmetric model to the rotation eigenvalues at all high-symmetry momenta: 
\e{
&C_2:
\quad e^{i2\pi{\mathscr{C}_{xy}}/2   } =\prod_{j}\widetilde{\rho}_{2,j}(\Gamma)\;\widetilde{\rho}_{2,j}(\tx)\;\widetilde{\rho}_{2,j}(\mathrm{Y})\;\widetilde{\rho}_{2,j}(\tm), \label{eqn:Chern-fomrula-C2}\\
&C_3:\quad e^{i2\pi{\mathscr{C}_{xy}}/3}=\prod_j\widetilde{\rho}_{3,j}(\Gamma)\;\widetilde{\rho}_{3,j}(\tk)\;\widetilde{\rho}_{3,j}(\tkpr), \label{eqn:Chern-fomrula-C3}\\
&C_4: 
\quad e^{i2\pi{\mathscr{C}_{xy}}/4} =\prod_{j}\widetilde{\rho}_{4,j}(\Gamma)\;\widetilde{\rho}_{2,j}(\tx)\;\widetilde{\rho}_{4,j}(\tm), \\
&C_6:\quad e^{i\pi{\mathscr{C}_{xy}}/3}=\prod_j\widetilde{\rho}_{6,j}(\Gamma)\;\widetilde{\rho}_{3,j}(\tk)\;\widetilde{\rho}_{2,j}(\tm).\label{eqn:Chern-fomrula-C6}
}
The rotation eigenvalue of a band can be expressed in terms of its itinerant angular momentum as
\begin{equation}
 \widetilde{\rho}_{m,j}(\Pi) = \exp(i2\pi\widetilde{\mathcal{L}}_j(\Pi)/m).   \label{eqn:eigenvalue-to-ang-momentum}
\end{equation}
Taking the logarithm on both sides of \q{eqn:Chern-fomrula-C2} [resp.\ \q{eqn:Chern-fomrula-C3}] gives immediately the desired result: \q{eq:app:Chern-via-ell} for $n=2$ [resp.~for $n=3$]. 
The  generalization to $n=4,6$ is less obvious because the right-hand sides of Eqs.~\eqref{eqn:Chern-fomrula-C3}--\eqref{eqn:Chern-fomrula-C6} involve $\widetilde{\rho}_{m,j}$ with $m\neq n$. 
Here, we utilize  that $n/m$ copies of a $C_m$-invariant reduced momentum appear in rBZ of a $C_n$-symmetric model, and therefore such a momentum contributes $n/m$ times to the sum in Eq.~(\ref{eq:app:Chern-via-ell}). 
In contrast, it only contributes once to the products in Eqs.~\eqref{eqn:Chern-fomrula-C2}--\eqref{eqn:Chern-fomrula-C6}. 
By combining the formulas \eqref{eqn:Chern-fomrula-C2}--\eqref{eqn:eigenvalue-to-ang-momentum} and after taking care of the multiplicity of the high-symmetry momenta, comparison of the exponents leads us to the conclusion summarized by Eq.~\eqref{eq:app:Chern-via-ell}. 

Let us remark that when discussing the faceted Chern number in Sec.~\ref{sec:BBC-Hopf-RTP-BBR}, we introduce an additional minus sign: $\mathscr{C}_{f}=-\mathscr{C}_{xy}$. This is because the faceted Chern number is computed with respect to the outward surface normal ($\hat{\bm{n}}$), while $\mathscr{C}_{xy}$ is computed with respect to the positive $z$-direction ($\hat{\bm{e}}_z$), and at the bottom facet the two are antiparallel ($\hat{\bm{n}} = -\hat{\bm{e}}_z$).

\stepcounter{section}

\section{One-to-one correspondence of equivalence class for band representations and itinerant angular momenta}\la{app:onetoone}

We are here to prove a useful fact presented in \s{sec:relaxdipole}, namely that \textit{two rank-one band representations (BRs) of $\mathrm{P}n$ are symmetry-equivalent as band representations if and only if they have exactly the same itinerant angular momenta.} 

Without loss of generality, such an equivalence class of rank-one BRs  is specified by ($\mathcal{L},\br_{\perp})$, with $\br_{\perp}$  the $(x,y)$-projection of a $C_n$-invariant Wyckoff position $\br$; this condition of $C_n$-invariance is necessary~\cite{TBO_JHAA} for the BR to be of unit rank. 
What $\mathcal{L}$ specifies is that for any $m$-fold rotation ($C_m$) in the site-stabilizer group $\mathrm{P}n_{\br}$, the corresponding operator acts on the (exponentially-localized) representative Wannier orbital localized at Wyckoff position $\br$ and translated by a Bravais-lattice vector $\bR$ as 
\e{\hat{C}_m W_{\bR}= e^{i2\pi \mathcal{L}/m} W_{C_m\bR+(C_m\br-\br)}.\la{repwan}
}
The Fourier transform of $W_{\bR}$ defines a Bloch function $\psi_{\bk}$ that is periodic and analytic over the Brillouin torus.
Equation~\eqref{repwan} then implies a corresponding symmetry action on the Bloch function:  
\e{\hat{C}_m \psi_{\bk}= \rho_m(\bk) \psi_{C_m\bk}, \as \rho_m(\bk)=e^{i2\pi \mathcal{L}/m}e^{i(C_m\bk-\bk)\cdot \br}.\la{definerho}} 
[cf.~discussion below Eq.~\eqref{eq:itin-am-from-basis} for the meaning of the two exponents in Eq.~\eqref{definerho}]
For any rotation-invariant reduced momentum $\Pi_j$ with an $m_j$-order little group, 
\e{ C_{m_j}\text{-invariant} \; \Pi_j: \as \rho_{m_j}(\Pi_j,k_z)=e^{i2\pi \widetilde{\mathcal{L}}_j/m_j}\la{cminv}} 
defines the mod-$m_j$ itinerant angular momentum $\widetilde{\mathcal{L}}_j$. All itinerant angular momenta may be collected as $(\Pi_1,\widetilde{\mathcal{L}}_1,\ldots,\Pi_J,\widetilde{\mathcal{L}}_J)$, with the total number ($J$) of rotation-invariant momenta uniquely determined by the $\mathrm{P}n$ space group. 

Let us then consider all possible equivalence classes of unit-rank BRs labelled by ($\mathcal{L},\br_{\perp})$, and all possible itinerant angular momenta (allowed for unit-rank BRs) labelled by  $(\Pi_1,\widetilde{\mathcal{L}}_1,\ldots,\Pi_J,\widetilde{\mathcal{L}}_J)$. 
Proving the italicized statement above reduces to proving that  (\ref{definerho}--\ref{cminv}) gives a bijective map between
\begin{equation}\label{eqn:bijectiveness-BBR}
(\mathcal{L},\br_{\perp})\mapsto (\Pi_1,\widetilde{\mathcal{L}}_1,\ldots,\Pi_J,\widetilde{\mathcal{L}}_J).
\end{equation}
The map~(\ref{eqn:bijectiveness-BBR}) is surjective, because the allowed itinerant angular momenta of a unit-rank BR of $\mathrm{P}n$ are, by construction, given by (\ref{definerho}--\ref{cminv}) for some value of $(\mathcal{L},\br_{\perp})$.

To further show that the map (\ref{eqn:bijectiveness-BBR}) is injective, we would like to prove that  $(\mathcal{L}_a,\br_{a,\perp})$ and $(\mathcal{L}_b,\br_{b,\perp})$ map to distinct itinerant angular momenta if $(\mathcal{L}_a,\br_{a,\perp})\neq (\mathcal{L}_b,\br_{b,\perp})$. Indeed, suppose 
\begin{equation}\label{labels}
\begin{split}
&(\mathcal{L}_a,\br_{a,\perp})\mapsto (\Pi^a_1,\widetilde{\mathcal{L}}^a_1,\ldots,\Pi^a_J,\widetilde{\mathcal{L}}^a_J),\\
&(\mathcal{L}_b,\br_{b,\perp})\mapsto (\Pi^b_1,\widetilde{\mathcal{L}}^b_1,\ldots,\Pi^b_J,\widetilde{\mathcal{L}}^b_J).
\end{split}
\end{equation}
If $\mathcal{L}_a\neq \mathcal{L}_b$, \q{definerho} guarantees that $\widetilde{\mathcal{L}}^a\neq \widetilde{\mathcal{L}}^b$ for the rotation-invariant reduced momentum at the center of the rBZ. 
To complete the proof of injectivity, we just need to show that in the case $\mathcal{L}_a=\mathcal{L}_b$ but $\br_{a,\perp}\neq \br_{b,\perp}$, the two labels for itinerant angular momenta in \q{labels} are distinct. 
This amounts to showing that $e^{i(C_m\bk-\bk)\cdot \br_a}{\neq}e^{i(C_m\bk-\bk)\cdot \br_b}$ for at least one $C_m$-invariant $\bk$, according to \q{definerho}. 
Equivalently, we want to find one $C_m$-invariant reduced momentum $\Pi$ such that 
\e{ (C_m\Pi-\Pi)\cdot (\br_{a,\perp}-\br_{b,\perp}) \neq_{2\pi} 0,\la{want}}
with `$\cdot$' understood here to be a dot product  of two-component vectors. 
\\

\noindent Fact 1:  $(C_m\Pi-\Pi)$ is necessarily a reciprocal vector of the reduced reciprocal lattice (rRL), by our assumption that $\bk=(\Pi,k_z)$ is $C_m$-invariant.\\

\noindent Fact 2: By our assumption that the Wyckoff positions are distinct, $\br_{a,\perp}-\br_{b,\perp}$ is not a vector of the reduced Bravais lattice that is dual to the rRL. \\

\noindent Fact 3: For any $\mathrm{P}n$ space group, one can always find a nontrivial divisor $m$ ($1<m\leq n$) of $n$ together with 
two nonzero, $C_m$-invariant momenta $\{\Pi_1,\Pi_2\}$, namely
\e{ &\mathrm{P}2: \as m=2, \as \Pi_1=\tx, \as \Pi_2=\ty, \lin
&\mathrm{P}3: \as m=3, \as \Pi_1=\tk, \as \Pi_2=\tk', \lin
&\mathrm{P}4: \as m=2, \as \Pi_1=\tx, \as \Pi_2=\ty, \lin
&\mathrm{P}6: \as m=3, \as \Pi_1=\tk, \as \Pi_2=\tk',
}
[with $\{\tx,\ty,\tk,\tk'\}$ illustrated in \fig{fig:lin-indepG}],
such that $C_m\Pi_j-\Pi_j=\bG_j$ gives two linearly-independent, primitive rRL vectors $\{\bG_1,\bG_2\}$, with corresponding primitive vectors $\{\bR_1,\bR_2\}$ that satisfy the duality condition $\bR_i\cdot \bG_j=2\pi \delta_{ij}$ and span the reduced Bravais lattice.\\

\begin{figure}
    \centering
    \includegraphics{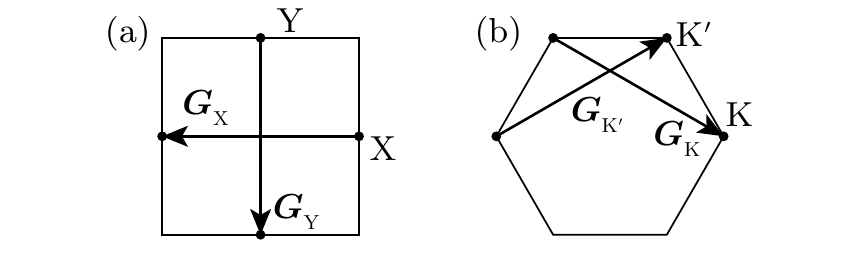}
    \caption{$C_m$-invariant momenta in rBZ for (a) $m=2$ and (b) $m=3$. 
    Two reciprocal lattice vectors $\bG_j=C_m\Pi_j-\Pi_j$, computed for these momenta are linearly-independent.
    }    
    \label{fig:lin-indepG}
\end{figure}

\noindent The above facts imply that one can always find $\Pi$ such that \q{want} holds true. 
Indeed, suppose we pick $\Pi=\Pi_1$, then either \q{want} holds true, or  $(C_m\Pi-\Pi)=\bG_1$ is orthogonal to $(\br_{a,\perp}-\br_{b,\perp})$. The latter case means that $(\br_{a,\perp}-\br_{b,\perp})$ is proportional to $\bR_2$ with a non-integer-valued proportionality constant. 
It would then follow that a different choice for $\Pi=\Pi_2$ would render  \q{want} true. This completes the proof of injectivity, and hence also of bijectivity of the map (\ref{eqn:bijectiveness-BBR}).

\section{A relation between the first Chern class and basis band representations}\label{app:no-band-inv-from-Chern}

In \app{sec:symmetryequiv345}, we prove that \textit{for any two-band, $\mathrm{P}n$-symmetric ($n=3,4,6$), insulating Hamiltonian  with  trivial first Chern class, its valence band (and also its conduction band) is a band representation that is symmetry-equivalent to one of the two basis band representations.}
To remind the reader, a basis band representation BBR$[\varphi_\alpha]$ is induced from one of the two tight-binding-basis orbitals $\varphi_1$ or $\varphi_2$, and to be symmetry-equivalent as band representations means there exists a space-group-equivariant isomorphism between the two corresponding vector bundles. 
In \app{sec:symmetryequiv2}, we explain why the above italicized statement does not hold for $\mathrm{P}2$ symmetry.

\subsection{\texorpdfstring{$\mathrm{P}n$}{Pn} symmetry \texorpdfstring{$(n>2)$}{(n>2)}: Trivial Chern class implies symmetry equivalence to basis band representation}\la{sec:symmetryequiv345}

Because two unit-rank band representations of $\mathrm{P}n$ are symmetry-equivalent if and only if their itinerant angular momenta coincide at all rotation-invariant $\bk$ [as proven in \app{app:onetoone}], the above italicized statement is equivalent to: for any two-band, $\mathrm{P}n$-symmetric ($n=3,4,6$), insulating Hamiltonian with trivial first Chern class, its valence band is characterized by the same set of itinerant angular momenta as one of $\{\mathrm{BBR}[\varphi_\alpha]\}_{\alpha=1,2}$. 
The last statement on itinerant angular momenta may be translated symbolically as\footnote{Although we have phrased (\ref{eq:app:basisBR-statement}) using an `exclusive or', the fact that BBR$[\varphi_1]$ and BBR$[\varphi_2]$ are not symmetry-equivalent implies that the two statements combined in the proposition (\ref{eq:app:basisBR-statement}) can never be true simultaneously. Therefore, one can equivalently treat the proposition as involving the usual `inclusive or', the negation of which is the proposition~(\ref{eq:app:invertion}).\label{foot:ORvsXOR}}
\begin{equation}\label{eq:app:basisBR-statement}
\begin{split}
\mathrm{either}\;\; \forall\,\Pi &:\;\widetilde{\mathcal{L}}_v(\Pi)=\widetilde{\mathcal{L}}_1(\Pi), \\ 
\mathrm{or}\;\; \forall\, \Pi &:\;\widetilde{\mathcal{L}}_v(\Pi)=\widetilde{\mathcal{L}}_2(\Pi).
\end{split}
\end{equation}
To remind the reader of the above symbolic notations: the itinerant angular momenta of the valence (conduction) band are denoted as $\widetilde{\mathcal{L}}_{v(c)}(\Pi)$, while the itinerant angular momenta of BBR$[\varphi_{1(2)}]$ are denoted as $\widetilde{\mathcal{L}}_{1(2)}(\Pi)$. 
If $\Pi$ is $C_m$-invariant (with $m$ the maximal rotational order of the little group), then both $\widetilde{\mathcal{L}}_{v(c)}(\Pi)$ and $\widetilde{\mathcal{L}}_{1(2)}(\Pi)$ are mod-$m$ quantities ($\in\{0,\dots,m-1\}$); at each $\Pi$, $\{\widetilde{\mathcal{L}}_{1}(\Pi),\widetilde{\mathcal{L}}_{2}(\Pi)\}=\{\widetilde{\mathcal{L}}_{v}(\Pi),\widetilde{\mathcal{L}}_{c}(\Pi)\}$. 

Our proof will utilize the mod-$n$ relation (proven in Appendix~\ref{app:Chern-via-ell}) between the Chern number and the itinerant angular momenta [cf.~Eq.~\eqref{eq:app:Chern-via-ell}]. 
We proceed to prove proposition (\ref{eq:app:basisBR-statement}) by contradiction, assuming on the contrary that (cf.~footnote~\ref{foot:ORvsXOR})
\begin{equation}\label{eq:app:invertion}
\begin{split}
    \exists \,\Pi\phantom{'} &:\; \widetilde{\mathcal{L}}_v(\Pi)=\widetilde{\mathcal{L}}_2(\Pi)\neq \widetilde{\mathcal{L}}_1(\Pi)  \\
    \mathrm{and} \as \exists \,\Pi'& :\; \widetilde{\mathcal{L}}_v(\Pi')=\widetilde{\mathcal{L}}_1(\Pi')\neq \widetilde{\mathcal{L}}_2(\Pi').
\end{split}
\end{equation}
In the rBZ of a P$n$ symmetric model with $n\in\{3,4,6\}$, one can always find three and only three rotation-inequivalent rotation-invariant reduced momenta $\{\Pi_j\}_{j=1,2,3}$, such that $\Pi_j$ is $C_{m_j}$-invariant, and all rotation-invariant momenta in the rBZ can be obtained from $\{\Pi_j\}_{j=1,\dots 3}$ by a $C_n$ rotation. 
Momentum $\Pi_j$ belongs to an $(n/m_j)$-multiplet of reduced momenta that is obtained by applying $n$-fold rotation(s) to $\Pi_j$; because the orbital inducing the valence band representation is centered at a $C_n$-invariant Wyckoff position, the itinerant angular momenta at all members of the multiplet are identical. 
Therefore, our assumption \eqref{eq:app:invertion} can be rewritten, modulo trivial relabelling of numeric subscripts, as
\e{
    &\widetilde{\mathcal{L}}_v(\Pi_1)=\widetilde{\mathcal{L}}_2(\Pi_1)\neq \widetilde{\mathcal{L}}_1(\Pi_1),\label{eq:app:one-inversion1}\\
    &\widetilde{\mathcal{L}}_v(\Pi_2)=\widetilde{\mathcal{L}}_1(\Pi_2)\neq \widetilde{\mathcal{L}}_2(\Pi_2),\label{eq:app:one-inversion2}\\
    &\widetilde{\mathcal{L}}_v(\Pi_3)=\widetilde{\mathcal{L}}_1(\Pi_3) \as \textrm{or} \as \widetilde{\mathcal{L}}_v(\Pi_3)=\widetilde{\mathcal{L}}_2(\Pi_3).
    \label{eq:app:one-inversion3}
}

First, assuming the first option in Eq.~\eqref{eq:app:one-inversion3}, we compute the Chern number of the valence band in the rotation-preserving 2D cut by means of Eq.~\eqref{eq:app:Chern-via-ell},
\begin{align}
    \mathscr{C}_{xy}^{v} & =_n\!\sum_{\Pi}\widetilde{\mathcal{L}}_{v}(\Pi)=_n\!\!\sum_{\Pi}\widetilde{\mathcal{L}}_1(\Pi) \!+\! \frac{n}{m_1}[\widetilde{\mathcal{L}}_2(\Pi_1)-\widetilde{\mathcal{L}}_1(\Pi_1)] \\
    & =_n \frac{n}{m_1}\left[\widetilde{\mathcal{L}}_2(\Pi_1)-\widetilde{\mathcal{L}}_1(\Pi_1)\right], \label{eq:non-basis-BR-Chern} 
\end{align}
where the summation goes over all rotation-invariant momenta in the rBZ. 
In the first row we used that momentum $\Pi_1$ is a member of an $n/m_1$-plet, and in the second row we used that the basis band representation has a vanishing Chern number~\cite{nogo_AAJH}. 
Because $\widetilde{\mathcal{L}}_2(\Pi_1)\neq \widetilde{\mathcal{L}}_1(\Pi_1)$ [cf.~Eq.~(\ref{eq:app:one-inversion1})] and $\widetilde{\mathcal{L}}_{1(2)}(\Pi_1)\in\{0,\dots,m_1-1\}$, their difference is nonzero and is never a multiple of $m_1$; therefore, we conclude that the Chern number of the valence band is necessarily non-zero, which contradicts our initial assumption. 
If, instead, we opt for the second option in Eq.~\eqref{eq:app:one-inversion3}, similar steps as in Eq.~\eqref{eq:non-basis-BR-Chern} and the fact that $\widetilde{\mathcal{L}}_2(\Pi_2)\neq \widetilde{\mathcal{L}}_1(\Pi_2)$ [cf.~Eq.~(\ref{eq:app:one-inversion2})] and $\widetilde{\mathcal{L}}_{1(2)}(\Pi_2)\in\{0,\dots,m_2-1\}$ lead to a Chern number
\e{
\mathscr{C}_{xy}^{v}=_n \frac{n}{m_2}\left[\widetilde{\mathcal{L}}_1(\Pi_2)-\widetilde{\mathcal{L}}_2(\Pi_2)\right]\neq_n 0,}
which again contradicts the initial assumption.
This together proves the proposition
\eqref{eq:app:basisBR-statement}.

\subsection{\texorpdfstring{$\mathrm{P}2$}{P2} symmetry: trivial Chern class does not imply symmetry equivalence to a basis band representation}\la{sec:symmetryequiv2}

In the case of $\mathrm{P}2$-symmetric Hamiltonians, the rBZ contains \emph{four} (rather than three, as above) rotation-inequivalent $C_2$-invariant reduced momenta. 
We remark that the presented calculation \eqref{eq:non-basis-BR-Chern}, which results in a non-zero value of the Chern number, also applies to the case of space group P$2$ when the valence itinerant angular momentum is given by Eqs.~\eqref{eq:app:one-inversion1}--\eqref{eq:app:one-inversion3} together with $\widetilde{\mathcal{L}}_v(\Pi_4)=\widetilde{\mathcal{L}}_j(\Pi_4)$ given that $\widetilde{\mathcal{L}}_v(\Pi_3)=\widetilde{\mathcal{L}}_j(\Pi_3)$. 
However a \emph{vanishing} 
Chern number is compatible with the following choice of the valence itinerant angular momenta
\begin{align}
    &\widetilde{\mathcal{L}}_v(\Pi_1)=\widetilde{\mathcal{L}}_2(\Pi_1)\neq \widetilde{\mathcal{L}}_1(\Pi_1),\label{eq:app:two-inversion1}\\
    &\widetilde{\mathcal{L}}_v(\Pi_2)=\widetilde{\mathcal{L}}_2(\Pi_2)\neq \widetilde{\mathcal{L}}_1(\Pi_2),\label{eq:app:two-inversion2}\\
    &\widetilde{\mathcal{L}}_v(\Pi_3)=\widetilde{\mathcal{L}}_1(\Pi_3)\neq \widetilde{\mathcal{L}}_2(\Pi_3),
    \label{eq:app:two-inversion3} \\
    &\widetilde{\mathcal{L}}_v(\Pi_4)=\widetilde{\mathcal{L}}_1(\Pi_4)\neq \widetilde{\mathcal{L}}_2(\Pi_4).
    \label{eq:app:two-inversion4}
\end{align}
The Chern number in this case can be computed as
\begin{align}
    \mathscr{C}_{xy}^v
    =_2\!\sum_{j=1}^4\widetilde{\mathcal{L}}_v(\Pi_j)
    =_2\!\sum_{j=1}^4\widetilde{\mathcal{L}}_1(\Pi_j)+\!\sum_{j=1}^2\Delta\widetilde{\mathcal{L}}(\Pi_j)
    =_2 0,
\end{align}
where we used that by assumption $\Delta\widetilde{\mathcal{L}}(\Pi_j)\neq0$, and therefore $\Delta\widetilde{\mathcal{L}}(\Pi_j)=1$, $j=1,2$.
This scenario (dubbed `case P2-II') 
is therefore studied separately when analyzing 
the relation between the Hopf and the RTP invariants in Sec.~\ref{sec:RTP-Hopf}.

\section{Time reversal of a crystalline Hopf Hamiltonian} \label{app:time-rev-ham}

In the main text we stated that a crystalline Hopf Hamiltonian, with the difference in angular momentum between two basis orbitals equal to $-\Delta\mathcal{L}$, can be obtained from a Hamiltonian with this difference being $\Delta\mathcal{L}$ through the application of time reversal. 
This means that the two Hamiltonians are related by \e{h_{-\Delta\mathcal{L}}(\bk)= h_{\Delta\mathcal{L}}^*(-\bk).\la{timereverse2}}
In this Appendix we prove this relation.

The crystalline-Hopf-insulating Hamiltonian $h_{\Delta\mathcal{L}}(\bk)$ is generally not time-reversal symmetric, because a nontrivial Hopf invariant is incompatible with time-reversal symmetry [cf.~Appendix~\ref{app:incompatibility}].
Nevertheless, to derive \q{timereverse2}, it is useful to consider $h_{\Delta\mathcal{L}}(\bk)$ as a submatrix of a larger-dimensional matrix Hamiltonian $H(\bk)$ which \textit{is} time-reversal-symmetric.

Let us label a Bloch-state basis for $H(\bk)_{\ab}$  by $\{\ket{\bk,\alpha}\rangle\}_{\alpha=1,\ldots, N_{T}}$ [cf.\ analogous \q{eq:app:Bloch-basis}], with $N_T>2$ the total number of basis states; and with $\{\ket{\bk,\alpha}\rangle\}_{\alpha=1,2}$ coinciding with the $\bk$-dependent basis of $h_{\Delta\mathcal{L}}(\bk)$; this implies that at each $\bk$ we have  $[h_{\Delta\mathcal{L}}]_{\ab}=H_{\ab}$ if $\alpha$ and $\beta$ are restricted to $\{1,2\}.$ 
Because the tight-binding Hilbert space given by $\{\ket{\bk,\alpha}\rangle\}_{\alpha=1,\ldots, N_{T}}$ is a band representation, the crystallographic splitting theorem~\cite{crystalsplit_AAJHWCLL} guarantees it is always possible to decompose the Hilbert space into $N_T$ line bundles which each has trivial first Chern class, such that both point-group generators (the $n$-fold rotation $C_n$ and time-reversal $\Theta$) act on the projectors
\e{ P_{\alpha}:= \sum_{\bk} \ket{\bk,\alpha}\rangle\langle\bra{\bk,\alpha}}
as a permutation on the basis labels:
\e{ \hat{C}_n P_{\alpha} \hat{C}_n^{-1} \eq P_{C_n\circ \alpha}  \lin
    \hat{\Theta} P_{\alpha} \hat{\Theta}^{-1} \eq P_{\Theta\circ \alpha}. \la{timereverseP}}
Assuming all $N_T$ tight-binding-basis orbitals are centered on $C_n$-invariant Wyckoff positions, and each orbital gives a one-dimensional representation of the order-$n$ cyclic subgroup of the corresponding  site-stabilizers, then $C_n$ acts as the trivial permutation: $C_n\circ \alpha=\alpha$ for all $\alpha\in \{1\ldots N_T\}$. 
By our assumption that at least one of the on-site angular momenta is nontrivial, time reversal will then act as a nontrivial permutation, \aac{meaning $\Theta\circ \alpha$ is not equal to $\alpha$ for all $\alpha\in \{1\ldots N_T\}.$}
For integer-spin representations (where $\Theta^2$ equals the identity operator), it is possible (for tight-binding-basis orbitals with trivial on-site angular momenta) that $\Theta\circ {\alpha}=\alpha$;  for half-integer-spin representations, it is guaranteed that $\Theta\circ {\alpha}\neq \alpha$. 
In all cases, \q{timereverseP} implies the following constraint on  the Bloch states:
\e{\hat{\Theta}\ket{\bk,\alpha}\rangle =\rho_\alpha(\bk)\ket{\bk,\Theta\circ \alpha}\rangle.}
One can always redefine $\{\ket{\bk,\alpha}\rangle\}_{\alpha=1,2}$ by a multiplicative phase factor such that  $\rho_{\alpha}(\bk)=1$ for $\alpha\in \{1,2\}$. 

We next apply the time-reversal-symmetry constraint to derive a relation between $h_{\Delta\mathcal{L}}$ and a different submatrix of $H$. Letting $\hat{H}$ denote the second-quantized form of the tight-binding Hamiltonian $H$, we apply that  $\hat{H}$ is self-adjoint, time-reversal symmetric ($[\hat{H},\Theta]=0$), and apply also the anti-unitary nature of time reversal ($\braket{\Theta\psi}{\Theta\phi}=\braket{\phi}{\psi}$), to derive \e{ [h_{\Delta\mathcal{L}}]_{\ab}(\bk) \eq \Big\langle\!\braket{\bk,\alpha }{\hat{H} (\bk,\beta)}\!\Big\rangle = \Big\langle\!\braket{\Theta\hat{H}(\bk,\beta) }{\Theta(\bk,\alpha)}\!\Big\rangle \lin 
\eq \big\langle\braopket{\Theta(\bk,\beta) }{\hat{H}}{\Theta(\bk,\alpha)}\big\rangle=\big\langle\braopket{-\bk,\bar{\beta} }{\hat{H}}{-\bk,\bar{\alpha}}\big\rangle\lin
\eq H(-\bk)_{\bar{\beta}\bar{\alpha}}= H(-\bk)^*_{\bar{\alpha}\bar{\beta}}.\la{antiunitary}}
What can be said about this submatrix obtained from restricting $H$ to row/column indices: $\bar{1}$ and $\bar{2}$?
Given that the difference in on-site angular momenta (of the tight-binding-basis orbitals of $h_{\Delta\mathcal{L}}$) is positive: $\mathcal{L}_2-\mathcal{L}_1=\Delta \mathcal{L}>0$, then the corresponding quantity (for the time-reversed basis orbitals) must be opposite in sign:  $\mathcal{L}_{\bar{2}}-\mathcal{L}_{\bar{1}}=-\Delta \mathcal{L}<0$; time reversal being a local operation in real space also implies that the Wyckoff positions of a basis orbital is invariant under time-reversal, hence: $\br_{2,\perp}-\br_{1,\perp}=\br_{\bar{2},\perp}-\br_{\bar{1},\perp}$. To recapitulate, the restriction of $H(\bk)$ to the submatrix with row/column indices $\bar{1}$ and $\bar{2}$ defines a two-band Hamiltonian, whose tight-binding-basis orbitals differ in on-site angular momentum by $-\Delta \mathcal{L}$; hence one may identify $h_{-\Delta \mathcal{L}}(\bk)_{\ab}=H(\bk)_{\bar{\alpha},\bar{\beta}}$, with $\alpha,\beta\in \{1,2\}$. This combines with \q{antiunitary} to give \q{timereverse2}.

\section{Proof of the angular-momentum anomaly of insulators with RTP invariant}\label{app:ang-mom-anomaly}

In this Appendix we prove that for any model Hamiltonian with nonzero RTP invariant $\Delta\mathscr{P}_{\Lambda\Xi}$ between two rotation-invariant reduced momenta, the same model with a surface termination manifests an angular-momentum anomaly. 
We shall concern ourselves with surface states that are localized to the surface but are extended as Bloch waves in the two independent directions parallel to the facet; in particular, we will focus on surface states whose reduced momenta are rotation-invariant. 

We consider the generic case of multi-band models with $N_v$ valence and $N_c$ conduction bulk bands, assuming that the former are enumerated by the index $j\in\{1,\dots N_v\}$ and the latter by the index $j\in\{N_v+1,\dots N_v+N_c\}$. For the RTP invariant $\Delta\mathscr{P}_{\Lambda\Xi}$ to be well-defined, we have assumed that (\emph{i}) bulk states 
along rotation-invariant lines $\gamma_\Lambda$ and $\gamma_\Xi$ 
fulfill the mutually-disjoint condition, meaning that the sets of itinerant angular momenta of all valence bands and of all conduction bands are disjoint: $\{\widetilde{\mathcal{L}}_j(\Pi)\}_{j=1}^{N_v}\cap\{\widetilde{\mathcal{L}}_j(\Pi)\}_{j=N_v+1}^{N_v+N_c}=\varnothing$, for $\Pi\in\{\Lambda,\Xi\}$, and (\emph{ii}) there exist a composite band representation CBR such that the valence band representation, restricted to reduced momenta $\Lambda$ and $\Xi$ equals to the restriction of CBR to these momenta.  
The particular case of a two-band model corresponds to $N_v=N_c=1$; the mutually-disjoint condition is fulfilled when $\widetilde{\mathcal{L}}_v(\Pi)\neq\widetilde{\mathcal{L}}_c(\Pi)$ for $\Pi\in{\Lambda, \Xi}$ and the composite band representation reduces to a basis band representation induced from one of the basis orbitals BBR$[\varphi_\alpha]$.
If the valence subspace comprises of multiple bands, the electric polarization at a reduced momentum $\bk_\perp$ is given by a sum of electric polarization of all valence bands:
\begin{equation}
    \mathscr{P}(\bk_\perp)=\sum_{j=1}^{N_v}\frac{i}{2\pi}\int\limits_{-\pi}^{\pi} dk_z\bra{u_j(\bk)}\ket{\partial_{k_z}u_j(\bk)},
    \label{eq:app:polarization-many-band}
\end{equation}
where $\ket{u_j(\bk)}$, $j=1,\dots N_v$ is a set of differentiable intra-cell functions that span the valence subspace of a tight-binding Hamiltonian at each $\bk$.
The RTP invariant is a difference in electric polarization given by Eq.~\eqref{eq:app:polarization-many-band} between reduced momenta $\Lambda$ and $\Xi$.

The existence of the angular-momentum anomaly implies that for all states (with reduced, rotation-invariant momentum $\Pi$) localized to the bottom facet, the number of states ($\sharp_{f(\mathrm{b})}\widetilde{\mathcal{L}}_v(\Pi)$) whose itinerant angular momentum belongs to the valence set $\{\widetilde{\mathcal{L}}_j(\Pi)\}_{j=1}^{N_v}$ differs between points $\Lambda$ and $\Xi$ by $\Delta\mathscr{P}_{\Lambda\Xi}$:
\begin{equation}
    \sharp_{f(\mathrm{b})}\widetilde{\mathcal{L}}_v(\Xi)-\sharp_{f(\mathrm{b})}\widetilde{\mathcal{L}}_v(\Lambda)=\Delta\mathscr{P}_{\Lambda\Xi}.
    \label{eq:app:ang-mom-anomaly-statement}
\end{equation}
The subscript $(\mathrm{b})$ reminds us that the above relation holds for the bottom facet; for the opposing `top facet' of a finite but long slab, the above relation holds with $\Delta\mathscr{P}_{\Lambda\Xi}\ri -\Delta\mathscr{P}_{\Lambda\Xi}$.

To prove the statement in Eq.~\eqref{eq:app:ang-mom-anomaly-statement}, we need to formalize the notion of states that are localized at the bottom facet. We do this using the Wannier-cut method reviewed in Appendix~\ref{app:wannier-cut}. 
In the next Appendix~\ref{app:AMA-Wannier} we express the two sides of Eq.~\eqref{eq:app:ang-mom-anomaly-statement} in terms of intra-column states, and show that the expressions are equivalent.

\subsection{Wannier cut method to define faceted projector\label{app:wannier-cut}}

The Wannier-cut method~\cite{Trifunovic:2020} allows to define a \textit{faceted projector} onto $\mathcal{N}_f$ polarization bands, closest to a chosen facet. We choose $\mathcal{N}_f$ large enough such that the faceted projector includes all surface-like polarization bands, that were defined in Sec.~\ref{sec:anomalyprecise}.

In the main text we formulated the angular-momentum anomaly for a semi-infinite geometry. However, here, for convenience, we consider a finite in $z$-direction and infinite in $(x,y)$-plane slab, with unit cells in $z$ directions labeled by $R_z=1,\dots \mathcal{N}_0$, and in the following referred to as layers. For large enough $\mathcal{N}_0$ the top and bottom facets are separated by a large enough distance, such that the angular-momentum anomaly derived for the bottom facet of a finite geometry holds as well on the bottom facet of a semi-infinite geometry. First, we define the projector onto $2\mathcal{N}_f<\mathcal{N}_0$ states at the top and bottom facets. For this we utilise the 
intra-column functions defined for an infinite geometry in Eq.~\eqref{eq:app:intra-column-orbitals}. In a finite slab with periodic boundary conditions, one can define bulk-like intra-column functions $\bar{\mathcal{V}}_{j,\bk_\perp,R_z}(\bar{R}_z,\alpha)$ by replacing an integral over momentum $k_z$ by a finite sum in Eq.~\eqref{eq:app:intra-column-orbitals}. For large enough $\mathcal{N}_0$, bulk-like intra-column functions in finite geometry and bulk intra-column functions in infinite geometry converge to each other. Moreover, the intra-column functions of a finite slab with open boundary conditions, computed far enough from the facets, can be described by the bulk-like intra-column functions $\bar{\mathcal{V}}_{j,\bk_\perp,R_z}(\bar{R}_z,\alpha)$ up to exponentially small corrections. Then the projector onto top and bottom facets can be expressed as:
\begin{equation}
    P_{\mathrm{t\&b}}(\bk_\perp)=\mathbbold{1}-\sum_{j=1}^{N_v+N_c}\sum_{R_z=\mathcal{N}_f+1}^{\mathcal{N}_0-\mathcal{N}_f}\ket{\bar{\mathcal{V}}_{j,\bk_\perp,R_z}}\!\Big\rangle\Big\rangle\Big\langle\Big\langle\!\bra{\bar{\mathcal{V}}_{j,\bk_\perp,R_z}}.
    \label{eq:app:wannier-cut}
\end{equation}
The first term on the right-hand side is the projector onto all $\mathcal{N}_0\times(N_v+N_c)$  states of a finite slab with reduced momentum $\bk_{\perp}$, and is given by the identity matrix in an intra-column Hilbert space with basis vectors indexed by $(R_z,\alpha)$. From this term, we subtract projectors onto a set of bulk-like intra-column states (in both the valence and conduction subspaces), that are localized to the bulk of the finite slab. 
Because of the exponential localization of the intra-column states in the $z$~direction, we approximate $\bar{\mathcal{V}}_{j,\bk_\perp,R_z}(\bar{R}_z,\alpha)\approx 0$ for $|\bar{R}_z-R_z|>R_0$, which holds up to exponentially small corrections for a large enough cut-off scale $R_0$. Then for large enough $\mathcal{N}_f>R_0$ and $\mathcal{N}_0-2\mathcal{N}_f>R_0$,
$P_{\mathrm{t\&b}}(\bk_\perp)$ has finite support on two regions, located respectively near the bottom and the top facet. Note $P_{\mathrm{t\&b}}(\bk_\perp)$, being the difference between an identity matrix and a projection matrix, is itself a projector in the intra-column Hilbert space indexed by ($R_z,\alpha)$.

To define the projector onto the bottom facet we additionally `sandwich' Eq.~\eqref{eq:app:wannier-cut} by a projector onto the bottom half of the slab:
\begin{equation}
P_\mathrm{b}(\bk_\perp) = P_\downarrow P_\mathrm{t\&b}(\bk_\perp)P_\downarrow\qquad\textrm{where}\quad P_\downarrow = \begin{pmatrix}
        \mathbbold{1} & \mathbbold{0} \\
        \mathbbold{0} & \mathbbold{0}
    \end{pmatrix}, \label{eq:app:wannier-cut-bottom}
\end{equation}
and $\mathbbold{1}$ and $\mathbbold{0}$ respectively denote the identity and the zero matrices of the size $\mathcal{N}_0/2\times(N_v+N_c)$, where for simplicity we assumed that $\mathcal{N}_0$ is even. Then $P_{\downarrow}$ defines the projector onto first $\mathcal{N}_0/2$ layers of the slab.

\subsection{Proof of angular-momentum anomaly using intra-column functions\label{app:AMA-Wannier}}

In this section we continue proving the angular-momentum anomaly, expressed by Eq.~(\ref{eq:app:ang-mom-anomaly-statement}).
First, using the notations introduced in Appendix~\ref{app:tb-formalism}, we express the RTP invariant in terms of the intra-column functions. For this,
we express the intra-cell function $u_j(\bk)$ in terms of intra-column functions
\begin{equation}
    u_j^\alpha(\bk_\perp,k_z)=\sum_{q=1}^{N_v}[U^\dagger(\bk_\perp,k_z)]_{jq}\sum_{R_z=-\infty}^{\infty}e^{-ik_z(R_z+z_\alpha)}\mathcal{V}_{q\bk_\perp 0}(R_z,\alpha),
    \label{eq:app:u-via-hwf}
\end{equation}
by inverting Eq.~\eqref{eq:app:intra-column-orbitals} and additionally assuming a gauge transformation $U_{jq}(\bk)$ that mixes bands within the valence subspace. Unitarity of the gauge transformation implies that $\sum_{q=1}^{N_v}[U^\dagger(\bk)]_{jq}[U(\bk)]_{qp}=\delta_{jp}$. After plugging Eq.~\eqref{eq:app:u-via-hwf} into the definition of the electric polarization \eqref{eq:app:polarization-many-band} and performing some algebraic transformations, we get the following expression for the polarization:
\begin{align}
    \mathscr{P}(\bk_\perp) &= \sum_{j,q=1}^{N_v}\int\limits_{-\pi}^\pi\frac{dk_z}{2\pi}[U(\bk_\perp,k_z)]_{jq}\partial_{k_z}[U^\dagger(\bk_\perp,k_z)]_{qj} \lin
    & + \sum_{j=1}^{N_v} \sum_{R_z=-\infty}^{\infty} \sum_{\alpha=1}^{N_v+N_c} \left|\mathcal{V}_{j\bk_\perp 0}(R_z,\alpha)\right|^2R_z \lin 
    & + \sum_{j=1}^{N_v}\sum_{\alpha=1}^{N_v+N_c}\int\limits_{-\pi}^{\pi}\frac{dk_z}{2\pi} \left|u_j^\alpha(\bk_\perp,k_z)\right|^2z_\alpha,
\label{eq:app:pol-via-hwf}
\end{align}
where the first term, responsible for the gauge freedom, is quantized to integers. Moreover, the requirement of a global smooth gauge ensures that this integer is independent of $\bk_\perp$, and therefore, this term drops in the difference in polarization between two reduced momenta:
\begin{align}
    \Delta\mathscr{P}_{\Lambda\Xi}=\sum_{j=1}^{N_v}\sum_{\alpha=1}^{N_v+N_c} &\left[\sum_{R_z=-R_0}^{R_0}\left|\mathcal{V}_{j\bk_\perp 0}(R_z,\alpha)\right|^2R_z\right. \lin 
    &\left.+ \int\limits_{-\pi}^{\pi}\frac{dk_z}{2\pi}\left|u_j^\alpha(\bk_\perp,k_z)\right|^2z_\alpha\right]\Bigg|_\Lambda^\Xi,
    \label{eq:app:pol-diff-hwf}
\end{align}
where we used the assumed truncation of the intra-column functions in unit cells $|R_z|>R_0$. In the following we show that the second term in this expression vanish. Note, that it encodes the dependence of electric polarization on the position of the spatial origin and on the choice of the unit cell via the values of the basis orbitals positions $z_\alpha$. We can rewrite the summation over all valence states in terms of a projector onto valence subspace:
\begin{align}
        &\sum_{j=1}^{N_v}\sum_{\alpha=1}^{N_v+N_c}\int\limits_{-\pi}^{\pi}\frac{dk_z}{2\pi}\left|u_j^\alpha(\Pi,k_z)\right|^2z_\alpha \lin 
        =& \int\limits_{-\pi}^{\pi}\frac{dk_z}{2\pi}\Tr_{\alpha}\left[\hat{z}_\mathrm{cell}\sum_{j=1}^{N_v}\ket{u_j(\Pi,k_z)}\bra{u_j(\Pi,k_z)}\right] \lin 
        =& \int\limits_{-\pi}^{\pi}\frac{dk_z}{2\pi}\Tr_{\alpha}\left[\hat{z}_\mathrm{cell}P^{\mathrm{cell}}_{\mathrm{v}}(\Pi,k_z)\right]=\sum_{\alpha=1}^{N_v}z_\alpha,
        \label{eq:app:pol-origin-dependence}
\end{align}
where $\hat{z}_\mathrm{cell}=\sum_\alpha\ket{\alpha}z_\alpha\bra{\alpha}$ and in the last row we denoted the projector onto the valence subspace in the intra-cell Hilbert space as $P^{\mathrm{cell}}_v$. In the last step, we used the fact that at reduced momenta $\Pi\in\{\Lambda,\Xi\}$ with fulfilled mutually-disjoint condition the valence bands form a basis of the valence subspace. Hence, the projector is $k_z$-independent and is equal to
\begin{equation}
    P^{\mathrm{cell}}_v=\begin{pmatrix}
        \mathbbold{1}_{N_v} & \mathbbold{0} \\
        \mathbbold{0} & \mathbbold{0}
    \end{pmatrix},
    \label{eq:app:val-projector}
\end{equation}
where $\mathbbold{1}_{N_v}$ is an identity matrix of the size $N_v$ and $\mathbbold{0}$ denote zero matrices of complementary sizes, such that the whole projector has the size $N_v+N_c$. We see that the origin-dependent term \eqref{eq:app:pol-origin-dependence} does not depend on the concrete value of $\Pi$, and hence it vanishes in the Eq.~\eqref{eq:app:pol-diff-hwf}. Then the RTP invariant expressed in terms of intra-column functions is
\begin{equation}
    \Delta\mathscr{P}_{\Lambda\Xi}=\sum_{j=1}^{N_v}\sum_{\alpha=1}^{N_v+N_c}\sum_{R_z=-R_0}^{R_0}R_z\left[\left|\mathcal{V}_{j\Xi 0}(R_z,\alpha)\right|^2-\left|\mathcal{V}_{j\Lambda 0}(R_z,\alpha)\right|^2\right].
    \label{eq:app:RTP-via-hwf}
\end{equation}

As the second step of the proof, we express the number of states, localized on the bottom facet, with valence-like itinerant angular momentum at a given reduced momentum $\Pi$, in terms of the bulk-like intra-column functions. 
Given the projector onto the bottom facet, this number can be computed as
\begin{equation}
    \sharp_{f(\mathrm{b})} \widetilde{\mathcal{L}}_v(\Pi)=\Tr\left[P^{\mathrm{column}}_v(\Pi)P_\mathrm{b}(\Pi)\right],
\end{equation}
where 
\begin{equation}
    P^{\mathrm{column}}_v(\bk_\perp)=\sum_{j=1}^{N_v}\sum_{R_z=-\infty}^{\infty}\ket{\bar{\mathcal{V}}_{j\bk_\perp R_z}}\!\Big\rangle\Big\rangle\Big\langle\Big\langle\!\bra{\bar{\mathcal{V}}_{j\bk_\perp R_z}}
\end{equation} 
is the projector onto the valence subspace in the intra-column Hilbert space. 
Taking the difference between two rotation-invariant momenta and substituting the projector onto the bottom-facet-localized states as defined using the Wannier cut in Eqs.~\eqref{eq:app:wannier-cut} and \eqref{eq:app:wannier-cut-bottom}, we find
\begin{align}
    &\sharp_{f(\mathrm{b})} \widetilde{\mathcal{L}}_v(\Xi)-\sharp_{f(\mathrm{b})} \widetilde{\mathcal{L}}_v(\Lambda) \lin 
    &=\Tr\bigg[P_\downarrow\sum_{j=1}^{N_v}\sum_{R_z=\mathcal{N}_f+1}^{\mathcal{N}_0-\mathcal{N}_f}\ket{\bar{\mathcal{V}}_{j,\bk_\perp,R_z}}\!\Big\rangle\Big\rangle\Big\langle\Big\langle\!\bra{\bar{\mathcal{V}}_{j,\bk_\perp,R_z}}\bigg]\Bigg|^\Lambda_\Xi 
    \label{eq:app:AMA-wannier-cut}
\end{align}
where in the second line we used that the identity operators computed for $\Lambda$ and $\Xi$ cancel each other and the projector on the valence subspace $P^{\mathrm{column}}_v$ sets the sum over conduction intra-column states to zero. In the last step of the derivation we use the explicit form of the projector $P_\downarrow$ which sets to zero the values of the intra-column function at layers $z>\mathcal{N}_0/2$. We compute the trace in Eq.~\eqref{eq:app:AMA-wannier-cut} and get:
\begin{align}
    &\sharp_{f(\mathrm{b})} \widetilde{\mathcal{L}}_v(\Xi)-\sharp_{f(\mathrm{b})} \widetilde{\mathcal{L}}_v(\Lambda) \lin 
    & = \sum_{j=1}^{N_v}\sum_{\alpha=1}^{N_v+N_c}\sum_{R_z=\mathcal{N}_f+1}^{\mathcal{N}_0-\mathcal{N}_f}\sum_{R_z'=0}^{\mathcal{N}_0/2}\left|\bar{\mathcal{V}}_{j\bk_\perp R_z}(R_z',\alpha)\right|^2 \Bigg|^\Lambda_\Xi  \lin 
    & = \sum_{j=1}^{N_v}\sum_{\alpha=1}^{N_v+N_c}\sum_{R_z=\mathcal{N}_f+1}^{\mathcal{N}_0-\mathcal{N}_f}\sum_{R_z'=R_z-R_0}^{\min(\mathcal{N}_0/2,R_z+R_0)}\left|\bar{\mathcal{V}}_{j\bk_\perp R_z}(R_z',\alpha)\right|^2 \Bigg|^\Lambda_\Xi  \lin
    & = \sum_{j=1}^{N_v}\sum_{\alpha=1}^{N_v+N_c}\sum_{R_z=\mathcal{N}_f+1}^{\mathcal{N}_0-\mathcal{N}_f}\sum_{\bar{R}_z=-R_0}^{\min(\mathcal{N}_0/2-R_z,R_0)}\left|\bar{\mathcal{V}}_{j\bk_\perp 0}(\bar{R}_z,\alpha)\right|^2 \Bigg|^\Lambda_\Xi.
    \label{eq:app:AMA-hwv}
\end{align}
The summation over $R_z$ in Eq.~\eqref{eq:app:AMA-hwv} can be split into three contributions. (\emph{i}) For $R_z>\mathcal{N}_0/2+R_0$ the sum over $\bar{R}_z$ has no elements as the upper limit of the sum is lower than the lower limit. (\emph{ii}) For $R_z\leq \mathcal{N}_0/2-R_0$ the sum over $\bar{R}_z$ runs from $-R_0$ to $R_0$ and together with sum over orbital index $\alpha$ results in the normalization of the intra-column functions which is equal to one. These terms are equal at both rotation-invariant momenta $\Lambda$ and $\Xi$ and therefore are canceled after subtraction. (\emph{iii}) For intermediate $\mathcal{N}_0/2-R_0<R_z\leq \mathcal{N}_0/2+R_0$ the sum is truncated at $\mathcal{N}_0/2-R_z$ and since the intra-column functions depend only on $\bar{R}_z$ we can rewrite summations over $R_z$ and $\bar{R}_z$ as
\begin{align}
    \sum_{R_z=\mathcal{N}_0/2-R_0+1}^{\mathcal{N}_0/2+R_0}&\sum_{\bar{R}_z=-R_0}^{\mathcal{N}_0/2-R_z}\left|\bar{\mathcal{V}}_{j\bk_\perp 0}(\bar{R}_z,\alpha)\right|^2 \lin 
    &= \sum_{\bar{R}_z=-R_0}^{R_0} (R_0-\bar{R}_z)\left|\bar{\mathcal{V}}_{j\bk_\perp 0}(\bar{R}_z,\alpha)\right|^2.
\end{align}
The term proportional to $R_0$ in this sum, when summed over the orbital index $\alpha$, is constant for the same reason as described above, and therefore it vanishes when taking the difference between rotation-invariant momenta $\Lambda$ and $\Xi$. We arrive to the following expression for the angular momentum difference:
\begin{align}
    &\sharp_{f(\mathrm{b})} \widetilde{\mathcal{L}}_v(\Xi)-\sharp_{f(\mathrm{b})} \widetilde{\mathcal{L}}_v(\Lambda) \lin
    & = \sum_{j=1}^{N_v}\sum_{\alpha=1}^{N_v+N_c}\sum_{\bar{R}_z=-R_0}^{R_0}\bar{R}_z\left[\left|\bar{\mathcal{V}}_{j\Xi 0}(\bar{R}_z,\alpha)\right|^2-\left|\bar{\mathcal{V}}_{j\Lambda 0}(\bar{R}_z,\alpha)\right|^2\right].
    \label{eq:app:AMA-final}
\end{align}
{Finally, recall that for large enough slabs the bulk-like intra-column functions $\bar{\mathcal{V}}_{j\Pi 0}(\bar{R}_z,\alpha)$ converge to the bulk intra-column functions of infinite geometry $\mathcal{V}_{j\Pi 0}(\bar{R}_z,\alpha)$. In the large-size limit, the right-side of Eq.~\eqref{eq:app:AMA-final} becomes identical to the expression for the RTP invariant \eqref{eq:app:RTP-via-hwf}.} This concludes the proof of Eq.~\eqref{eq:app:ang-mom-anomaly-statement}

\section{Proof of Zak-phase anomaly\label{app:Zak-phase-anomaly}}

In this Appendix we provide necessary proves in order to compete our discussion of the Zak-phase anomaly presented in Sec.~\ref{sec:BBC_zak} of the main text. As a reminder, this anomaly manifests the faceted Zak phases computed along certain symmetrically-chosen rBZ loops, which have distinct values from the bulk Zak phases (both conduction and valence). We subdivide this Appendix into two parts: in Appendix~\ref{app:two-segm} we prove the Zak-phase anomaly when the loop can be chosen to be two-segment, and in Appendix~\ref{app:four-segm} we generalize this proof to the cases when a four-segment loop is necessary to detect the Zak-phase anomaly.

\subsection{Zak-phase anomaly for \texorpdfstring{$C_m$}{Cm}-related-two-segment loops}\label{app:two-segm}

We assume that the RTP invariant $\Delta\mathscr{P}_{\Lambda\Xi}>0$ for two rotation-invariant reduced momenta $\Lambda$ and $\Xi$ and that there exists an integer $m>1$ such that at both reduced momenta $\Pi\in\{\Lambda,\Xi\}$ the mutually-disjoint condition is fulfilled modulo $m$: $\widetilde{\mathcal{L}}_c(\Pi) \neq_m \widetilde{\mathcal{L}}_v(\Pi)$. Then the loop, on which one can detect the Zak-phase anomaly is $\mathcal{S}=\Xi\Lambda (\Xi+\bG)$ [cf.~Fig.~\ref{fig:Zak-path-notations}(a) for an example]. This loop consists of two segments, $\Xi\Lambda$ and $\Lambda(\Xi+\bG)$, which are related to each other by a $C_m$ rotation, reversal of orientation and possibly shift by a reciprocal lattice vector in rBZ, $\Lambda (\Xi+\bG)=C_m \Lambda \Xi + \bG'$. 

The techniques to evaluate the Zak phase for a $C_m$-related-two-segment loop, given the itinerant angular momenta at the two $C_m$-invariant momenta in said loop, have been developed in \ocite{AA_wilsonloopinversion} and \ocite{TBO_JHAA}. We review these techniques here as a warm-up for the more complicated case of the four-segment loops, which require a generalization of the known techniques. Much of the intuition behind these techniques can be conveyed by calculating the Zak phase of a single faceted band.

\begin{figure*}
    \centering
    \includegraphics{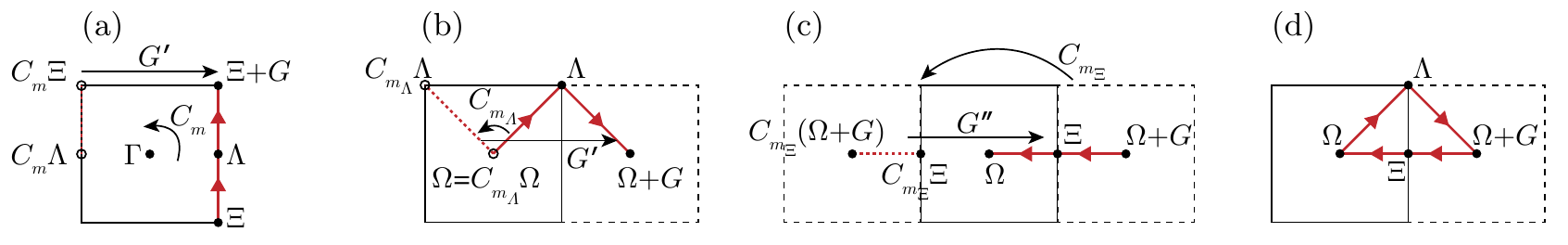}
    \caption{Illustration of notations introduced in Appendix~\ref{app:Zak-phase-anomaly} to compute the Zak phase over a symmetrically chosen loop in rBZ. (a) Notations for a $C_m$-related-two-segment loop. (b--d) Notations for a four-segment loop. In (b) we introduce notations for the first two segments, in (c) for the second two segments. The whole loop is illustrated in (d).}
    \label{fig:Zak-path-notations}
\end{figure*}

\subsubsection{Single isolated band}

To compute the Zak phase of an isolated band we express it as a phase of a corresponding Wilson loop
\begin{equation}
    e^{i\mathscr{Z}_j}=\mathcal{W}^j_{\mathcal{S}}= \hat{\mathrm{T}}\exp\left[i\int\limits_{\mathcal{S}}d\bk_\perp\cdot\boldsymbol{\mathcal{A}}_\perp[u_j(\bk_\perp)]\right],
    \label{eq:app:wilson-loop-rbz}
\end{equation}
\e{
\mathscr{Z}_{\Xi\Lambda (\Xi+\bG),j}
&=_{2\pi} \Im\log \mathcal{W}^j_{(\Xi+\bG)\leftarrow\Lambda\leftarrow\Xi} \lin
&=_{2\pi}\Im\log \mathcal{W}^j_{(\Xi+\bG)\leftarrow\Lambda} \mathcal{W}^j_{\Lambda\leftarrow\Xi} \lin
&=_{2\pi} \Im\log\widetilde{\rho}_{m,j}(\Xi)\left[\mathcal{W}^j_{\Lambda\leftarrow\Xi}\right]^\dagger\,\widetilde{\rho}_{m,j}(\Lambda)^{-1}\mathcal{W}^j_{\Lambda\leftarrow\Xi} \lin
&=_{2\pi} \frac{2\pi}{m}\left[\widetilde{\mathcal{L}}_j(\Xi)-\widetilde{\mathcal{L}}_j(\Lambda)\right]
\label{eq:app:zak-wilson}
}
where in the second row we decomposed a Wilson loop into two Wilson lines. In the third row we related two Wilson lines, computed along symmetry related paths, via  eigenvalues $\widetilde{\rho}_{m,j}(\Pi)$ of the itinerant rotation matrix \eqref{eq:itinerant_rotation}, and also took into account that Wilson lines computed on two paths, which are related to each other by a reciprocal lattice translation, are equal to each other. In the fourth row we applied the definition of itinerant symmetry eigenvalues in terms of itinerant angular momenta $\widetilde{\rho}_{m,j}(\Pi)=\exp(i2\pi\widetilde{\mathcal{L}}_j(\Pi)/m)$ and used that two Wilson lines computed along the same path in opposite directions cancelled out. Note, that a single-band Wilson loop is a complex number and hence all terms in Eq.~\eqref{eq:app:zak-wilson} commute with each other. Substituting the values of the itinerant angular momentum for the faceted, bulk-valence and bulk-conduction bands, we arrive to Eqs.~\eqref{thesame} and \eqref{eq:zak-bulk-surf}.

\subsubsection{Several faceted bands}

\begingroup
\renewcommand{\arraystretch}{1.3}
\setlength{\tabcolsep}{6pt}
\begin{table}
    \centering
    \begin{tabular}{ccc}
    \hline\hline
     & $\Lambda$ & $\Xi$ \\
     \hline
     $\#\widetilde{\mathcal{L}}_v$ & $\mathcal{N}_v^\mathrm{na}$ & $\mathcal{N}_v^\mathrm{na}+\mathcal{N}^{\mathrm{a}}$  \\
     $\#\widetilde{\mathcal{L}}_c$ & $\mathcal{N}_c^\mathrm{na}+\mathcal{N}^{\mathrm{a}}$ & $\mathcal{N}_c^\mathrm{na}$ \\
     \hline\hline
    \end{tabular}
    \caption{Number of bands with fixed itinerant angular momenta at a fixed rotation-invariant point for a set of surface localized bands with $\mathcal{N}_v^\mathrm{na}$ valence-bulk-like $\mathcal{N}_c^\mathrm{na}$ conduction-bulk-like and $\mathcal{N}^{\mathrm{a}}$ anomalous bands.}
    \label{tab:surf_ell}
\end{table}
\endgroup

In the case, when there are several faceted bands we need to apply the theorem, proven in Ref.~\cite{TBO_JHAA}, in order to compute the Zak phases. Let us consider a set of $\mathcal{N}_f$ faceted bands. 
The RTP invariant ensures that there will be $\mathcal{N}^{\mathrm{a}}=|\Delta\mathscr{P}_{\Lambda\Xi}|$ anomalous bands with itinerant angular momenta that are bulk-conduction-like at $\Lambda$ and bulk-valence-like at $\Xi$ [cf.\  Eq.~\eqref{eq:surface_ell}]. 
The amount of non-anomalous faceted bands with bulk-valence-like (bulk-conduction-like) itinerant angular momenta at both $\Lambda$ and $\Xi$, we denote as $\mathcal{N}_v^\mathrm{na}$ ($\mathcal{N}_c^\mathrm{na}$). 
Then the total number of bands with each value of itinerant angular momentum at rotation-invariant $\bk_\perp$-points is summarized in Table \ref{tab:surf_ell}. 
The theorem of Ref.~\cite{TBO_JHAA}, applied to these values of the itinerant angular momenta, guarantees that the given set of faceted bands will have at least $\mathcal{N}^{\mathrm{a}}=\Delta\mathscr{P}_{\Lambda\Xi}$ Zak phases with quantized values
\begin{equation}
    \frac{\mathscr{Z}_f}{2\pi}=_1\frac{\widetilde{\mathcal{L}}_v(\Xi)-\widetilde{\mathcal{L}}_c(\Lambda)}{m}.
\end{equation}

\subsection{Zak-phase anomaly for four-segment loops}\label{app:four-segm}

In this Appendix we discuss in more details the bulk-boundary correspondence of the RTP between $\Lambda$ and $\Xi$ reduced momenta, when 
\begin{align}
    \widetilde{\mathcal{L}}_v(\Lambda)&\neq_{m_\Lambda}\widetilde{\mathcal{L}}_c(\Lambda), \\ 
    \widetilde{\mathcal{L}}_v(\Xi)&\neq_{m_\Xi}\widetilde{\mathcal{L}}_c(\Xi), \\
    m_\Lambda&\neq m_\Xi.
\end{align}

This happens only in two distinct $\mathrm{P}4$-symmetric Hilbert spaces, that were described in the main text. In both cases, one can construct a four-segment loop which is a concatentation of two two-segment rBZ-loops: 
the first loop $\Omega\Lambda(\Omega+\bG)$ consists of two segments related by $C_{m_\Lambda}$ symmetry $\Lambda(\Omega+\bG)=C_{m_\Lambda}(\Lambda\Omega) + \bG'$ [cf.~Fig.~\ref{fig:Zak-path-notations}(b) for an example], 
and second loop $(\Omega+\bG)\Xi\Omega$ consists of two segments related by $C_{m_\Xi}$ symmetry $\Xi\Omega={m_\Xi}(\Xi(\Omega+\bG))+\bG''$ [cf.~Fig.~\ref{fig:Zak-path-notations}(c) for an example]. An example of the entire four-segment loop $\mathcal{S}_4=\Omega\Lambda(\Omega+\bG)\Xi\Omega$ is presented in Fig.~\ref{fig:Zak-path-notations}(d). Here we introduced another rotation-invariant momentum $\Omega$ at which the mutually-disjoint condition is violated $\widetilde{\mathcal{L}}_v(\Omega)=\widetilde{\mathcal{L}}_c(\Omega)=\widetilde{\mathcal{L}}(\Omega)$.

We again consider a set of faceted bands with $\mathcal{N}^{\mathrm{a}}$ anomalous bands and $\mathcal{N}_v^\mathrm{na}+\mathcal{N}_c^\mathrm{na}$ non-anomalous bands. All the values of the itinerant angular momenta are summarized in Table~\ref{tab:itinerant_ell}. Then the Zak phases of these faceted bands are computed as phases of the Wilson loop eigenvalues. The Wilson loop along our chosen loop $\mathcal{S}_4$ is
\e{
&\mathcal{W}_{\Omega\leftarrow\Xi\leftarrow(\Omega+\bG)\leftarrow\Lambda\leftarrow\Omega} \lin
&=\mathcal{W}_{\Omega\leftarrow\Xi}\mathcal{W}_{\Xi\leftarrow(\Omega+\bG)}\mathcal{W}_{(\Omega+\bG)\leftarrow\Lambda}\mathcal{W}_{\Lambda\leftarrow\Omega}  \lin
&= \left[\widetilde{R}_{C_{m_\Xi}}(\Omega)\mathcal{W}^\dagger_{\Xi\leftarrow(\Omega+\bG)}\widetilde{R}_{C_{m_\Xi}}^{-1}(\Xi)\mathcal{W}_{\Xi\leftarrow(\Omega+\bG)}\right. \lin
&\as\cdot\left.\widetilde{R}_{C_{m_\Lambda}}(\Omega)\mathcal{W}^\dagger_{\Lambda\leftarrow\Omega}\widetilde{R}_{C_{m_\Lambda}}^{-1}(\Lambda)\mathcal{W}_{\Lambda\leftarrow\Omega}\right].
\label{eq:app:zak_manybands_closedloop}
}
Here we decomposed the Wilson loop into four Wilson lines, and related every pair of Wilson lines, computed along the symmetry related paths, via the itinerant rotation matrices \eqref{eq:itinerant_rotation}.
Contrary to the derivation of Appendix~\ref{app:two-segm} the Wilson loops and the rotation matrices are square matrices of size $\mathcal{N}_f=\mathcal{N}^\mathrm{a}+\mathcal{N}_v^{\mathrm{na}}+\mathcal{N}_c^{\mathrm{na}}$ and therefore do not commute. To analyze the eigenvalues of this matrix, first, observe that $\widetilde{R}_{C_m}(\Omega)$ is proportional to identity matrix; this is because an eigenvalue of  $\widetilde{R}_{C_m}(\Omega)$ can only take one of two values -- the rotational representation of the bulk-conduction band: $e^{i2\pi \widetilde{\mathcal{L}}_c(\Omega)/m}$ or the rotational representation of the bulk-valence band: $e^{i2\pi \widetilde{\mathcal{L}}_v(\Omega)/m}$; but by assumption, the mutually-disjoint condition is violated at $\Omega$: $\widetilde{\mathcal{L}}_c(\Omega)=\widetilde{\mathcal{L}}_v(\Omega)\equiv \widetilde{\mathcal{L}}(\Omega)$. Thus $\widetilde{R}_{C_m}(\Omega)$ can only modify the Wilson loop  by a constant, multiplicative  factor
\e{
&\mathcal{W}_{\Omega\leftarrow\Xi\leftarrow(\Omega+\bG)\leftarrow\Lambda\leftarrow\Omega} \lin
&= \left[\mathcal{W}^\dagger_{\Xi\leftarrow(\Omega+\bG)}\widetilde{R}_{C_{m_\Xi}}^{-1}(\Xi)\mathcal{W}_{\Xi\leftarrow(\Omega+\bG)}\mathcal{W}^\dagger_{\Omega\Lambda}\widetilde{R}_{C_{m_\Lambda}}^{-1}(\Lambda)\mathcal{W}_{\Lambda\leftarrow\Omega}\right] \lin
&\as*\exp\left( i2\pi\left[\widetilde{\mathcal{L}}(\Omega)/m_\Lambda+\widetilde{\mathcal{L}}(\Omega)/m_\Xi\right]\right)
\label{eq:app:zak-D-contribution}
}
In the following denote the Wilson lines as $Z_1=\mathcal{W}_{\Lambda\leftarrow\Omega}$ and $Z_2=\mathcal{W}_{\Xi\leftarrow(\Omega+\bG)}$ which are unitary matrices. Thus eigenvalues of the matrices $\widetilde{R}^{-1}_{C_{m_\Lambda}}(\Lambda)$ (resp.~$\widetilde{R}^{-1}_{C_{m_\Xi}}(\Xi)$) and  $Z_1^\dagger\widetilde{R}^{-1}_{C_{m_\Lambda}}(\Lambda)Z_1$ (resp.~$Z_2^\dagger\widetilde{R}^{-1}_{C_{m_\Xi}}(\Xi)Z^2$) are the same. 
One can always find a projector $\mathscr{B}_1$ (resp.~$\mathscr{B}_2$) onto the eigenspace of $Z_1^\dagger\widetilde{R}^{-1}_{C_{m_\Lambda}}(\Lambda)Z_1$ (resp.~$Z_2^\dagger\widetilde{R}^{-1}_{C_{m_\Xi}}(\Xi)Z_2$) corresponding to the eigenvalue $\varrho_1=\exp(-i2\pi\widetilde{\mathcal{L}}_c(\Lambda)/{m_\Lambda})$  (resp.~$\varrho_2=\exp(-i2\pi\widetilde{\mathcal{L}}_v(\Xi)/m_\Xi)$). 
Denote projectors onto the orthogonal complements as $\mathscr{D}_j=I_{\mathcal{N}_f} - \mathscr{B}_j$, $j=1,2$. Then the Wilson loop can be rewritten as
\e{
&\mathcal{W}_{\Omega\leftarrow\Xi\leftarrow(\Omega+\bG)\leftarrow\Lambda\leftarrow\Omega} \lin
&=\left[\mathscr{B}_2\mathscr{B}_1\varrho_2\varrho_1 + \mathscr{B}_2\varrho_2\mathscr{D}_1Z_1^\dagger\widetilde{R}^{-1}_{C_{m_\Lambda}}(\Lambda)Z_1\mathscr{D}_1 \right.\lin
&\as  
+ \mathscr{D}_2 Z_2^\dagger\widetilde{R}^{-1}_{C_{m_\Xi}}(\Xi)Z_2\mathscr{D}_2\mathscr{B}_1\varrho_1 \lin
&\as\left. 
+ \mathscr{D}_2Z_2^\dagger\widetilde{R}^{-1}_{C_{m_\Xi}}(\Xi)Z_2\mathscr{D}_2\mathscr{D}_1Z_1^\dagger\widetilde{R}^{-1}_{C_{m_\Lambda}}(\Lambda)Z_1\mathscr{D}_1
\right] \lin
& \as*\exp\left( i2\pi\left[\widetilde{\mathcal{L}}(\Omega)/m_\Lambda+\widetilde{\mathcal{L}}(\Omega)/m_\Xi\right]\right).
\label{eq:app:zak-BD-projectors}
}
Knowing the itinerant angular momenta of the faceted bands at rotation-invariant points $\Lambda$ and $\Xi$ [cf.~Table~\ref{tab:surf_ell}], we find the ranks of the projectors $\rank(\mathscr{B}_1)=\mathcal{N}_c^\mathrm{na}+\mathcal{N}^\mathrm{a}$ and $\rank(\mathscr{B}_2)=\mathcal{N}_v^\mathrm{na}+\mathcal{N}^\mathrm{a}$.
Because the dimensions of all matrices considered are identically equal to ($\mathcal{N}_v^\mathrm{na}+\mathcal{N}_c^\mathrm{na}+\mathcal{N}^\mathrm{a})$, it follows from dimension counting that
the two vector spaces projected by $\mathscr{B}_1$ and $\mathscr{B}_2$ have a nonzero intersection with dimension at least equal to $\mathcal{N}^\mathrm{a}$. In other words, any vector $\ket{v}$ in this at-least-$\mathcal{N}^\mathrm{a}$-dimensional subspace satisfies
\e{ \mathscr{B}_1\ket{v}=\mathscr{B}_2\ket{v}=\ket{v}, \as \mathscr{D}_1\ket{v}=\mathscr{D}_2\ket{v}=0.}
It follows from the decomposition of the Wilson-loop matrix in \q{eq:app:zak-BD-projectors} that \e{\mathcal{W}\ket{v}=\varrho_2\varrho_1\exp\left( i2\pi\left[\widetilde{\mathcal{L}}(\Omega)/m_\Lambda+\widetilde{\mathcal{L}}(\Omega)/m_\Xi\right]\right)\ket{v},}
with an eigenvalue that does not depend on $Z_j$. Thus we conclude that there are at least $\mathcal{N}^\mathrm{a}$-number of Zak phases quantized to
\e{
\mathscr{Z}=2\pi\left[\frac{\widetilde{\mathcal{L}}(\Omega)-\widetilde{\mathcal{L}}_c(\Lambda)}{m_\Lambda} + \frac{\widetilde{\mathcal{L}}(\Omega)-\widetilde{\mathcal{L}}_v(\Xi)}{m_\Xi}\right].
\label{eq:app:zak-final}
}
As discussed in Sec.~\ref{sec:zakanomalymain} this phase is distinct from both bulk-valence and bulk-conduction Zak phases. This completes the proof of the Zak-phase anomaly.

\section{Hopf-RTP relation proof via Hopf-Chern relation}\label{app:hopf-rtp-whitehead}

In this Appendix we present an alternative proof of the relation between the Hopf and the RTP invariants. {We focus only on the case when the valence band representation is symmetry-equivalent to a basis band representation, as detailed in Sec.~\ref{sec:mainresult} of the main text. In this case the invariants are related as }
\begin{equation}
    \chi=_n\sum_{j=2}^{J}\Delta\widetilde{\mathcal{L}}(\Pi_j)\Delta\mathscr{P}_{\Pi_1\Pi_j},
    \label{eq:app:Hopf-RTP}
\end{equation}
{where $\{\Pi_j\}_{j=1\dots J}$ denotes a maximal subset of rotation-invariant reduced momenta which each satisfies the mutually-disjoint condition: $\Delta\widetilde{\mathcal{L}}(\Pi_j)\neq0$.}
This proof is based {on the Hopf-Chern relation that was originally introduced for the Hopf map by Whitehead in Ref.~\cite{Whitehead:1947}.} 
We first review this original relation in Sec.~\ref{sec:whitehead}. We then extend {its} applicability 
to $\bk$-periodic tight-binding Hamiltonians in  Sec.~\ref{sec:whiteheadthree} and to $\bk$-nonperiodic Hamiltonians in Sec.~\ref{sec:whiteheadnon}, then finally we apply {the Hopf-Chern relation} to derive the Hopf-RTP relation \eqref{eq:app:Hopf-RTP} in Sec.~\ref{sec:whitehead_rotation}.

\subsection{Hopf-Chern relation for the Hopf map \label{sec:whitehead}}

Let us review Whitehead's equivalent formulation\cite{Whitehead:1947} of the Hopf invariant of the Hopf map. It says that the Hopf invariant of a map 
\begin{align}
\eta: \begin{cases} S^3\to S^2,\\
        \gamma(b) \mapsto b                      \end{cases}
\end{align}
is given by the Chern number $\mathscr{C}_{\Sigma}$ computed on the oriented surface $\Sigma$ stretched over an inverse image $(\partial\Sigma = \gamma)$ of a given base point $b$ on the 2-sphere
\begin{align}
    \chi = \mathscr{C}_{\Sigma(b)},\as  \partial\Sigma(b) = \gamma(b).
    \label{eq:whitehead}
\end{align}
We call this relation between the invariants the \emph{Hopf-Chern relation}.
As the boundary $\partial [\Sigma (b)]$ is mapped to the fixed base point $b$ on the 2-sphere, it can be treated as a closed surface with  a well-defined Chern number.

\subsection{Hopf-Chern relation for \texorpdfstring{$\bk$}{k}-periodic Hamiltonians \label{sec:whiteheadthree}}

The Hopf-Chern relation can be directly applied to compute the Hopf invariant of a two-band Hamiltonian $\tilde{h}(\bk)$ defined over the Brillouin torus (meaning $\tilde{h}(\bk)=\tilde{h}(\bk+\bG)$ for any reciprocal lattice vector $\bG$), assuming that $\tilde{h}(\bk)$ is in the trivial first Chern class.

To see this, consider a flattened Hamiltonian $\tilde{h}_\mathrm{flat}(\bk)=\tilde{h}(\bk)/E_+(\bk)=\tilde{\boldsymbol{h}}_\mathrm{flat}(\bk)\cdot\bsigma$ which defines a map $\tilde{\boldsymbol{h}}_\mathrm{flat}:T^3\to S^2$ with the 3-torus given by momenta in the first BZ and the 2-sphere given by a normalized vector of the coefficients in a Pauli matrix decomposition of a flattened Hamiltonian.
Given the vanishing Chern numbers $\mathscr{C}_i =0$ on the three BZ subtori, there exists a continuous deformation that deforms $\tilde{\boldsymbol{h}}_\mathrm{flat}(\bm{k})$ on the boundary of the first BZ to a constant. 
From the topological perspective, this allows us to identify the whole boundary as a single point~\cite{Hatcher:2002}, turning the first BZ into a $3$-sphere. {Such a deformation does not modify nor the Hopf invariant, neither the Chern number computed over the surface bounded by a preimage. Therefore, the Hopf-Chern relation, that is valid for a map from a 3-sphere is also valid for a map from a 3-torus, given by the Hopf insulator Hamiltonian.}

To understand how the Hopf-Chern relation \eqref{eq:whitehead} can be applied to the Hopf insulator, consider some base point on a 2-sphere $b\in S^2$ [green point in Fig.~\ref{fig:whitehead}(b)]. The inverse image $\tilde{\boldsymbol{h}}_\mathrm{flat}^{-1}(b)\subset T^3$ [green path in Fig.~\ref{fig:whitehead}(a)] is not necessarily path-connected, but it can be decomposed into path-connected components, $\tilde{\boldsymbol{h}}_\mathrm{flat}^{-1}(b)=\bigcup_i\gamma_i$, where $\gamma_i$ are generically one-dimensional (possibly non-contractible) closed loops~\cite{kennedy_hopfchern}.
Generically, each loop $\gamma_i$ can be uniquely assigned an \emph{orientation} by the following procedure. Consider an arbitrary reference point $\bm{k}_0\in\gamma_i$ and define a plane passing through this point and normal to $\gamma_i$ at $\bm{k}_0$. Choose two orthogonal basis vectors in this plane, labelled $\hat{\bm{e}}_1$ and $\hat{\bm{e}}_2$. Then, in this plane, a circle infinitesimally close to $\gamma_i$ [blue line in Fig.~\ref{fig:whitehead}(a)] is parametrized by $\bm{k}({\theta})=\bm{k}_0+\varepsilon [\hat{\bm{e}}_1\cos(\theta)+\hat{\bm{e}}_2\sin(\theta)]$ with $\theta\in[0,2\pi]$ and $\varepsilon\to 0$. By considering the generic linear-order expansion of $\tilde{\boldsymbol{h}}_\mathrm{flat}(\bm{k})$ around the reference point, it follows that the image of $\bm{k}({\theta})$ under the Hamiltonian map $\tilde{\boldsymbol{h}}_\mathrm{flat}$ is necessarily a closed loop that winds around the base point $b$ [blue line in Fig.~\ref{fig:whitehead}(b)].
The orientation of $\gamma_i$ at $\bm k_0$ is defined as $\bm{\tau} (\bm k_0)= +(-)\hat{\bm{e}}_1\cross\hat{\bm{e}}_2$ for a clockwise (counterclockwise) winding of $\tilde{\boldsymbol{h}}_\mathrm{flat}[\bm{k}(\theta)]$. {Note that this definition does not depend on the choice of the basis vectors $\hat{\bm{e}}_1$ and $\hat{\bm{e}}_2$.}\footnote{Let us assume that we have picked another pair of basis vectors $\hat{\bm{e}}'_1=\hat{\bm{e}}_1$ and $\hat{\bm{e}}'_2=-\hat{\bm{e}}_2$. Then for the same $\theta$ parametrization, momentum $\bk(\theta)$ will rotate in the opposite direction, and therefore the image $\tilde{\boldsymbol{h}}_\mathrm{flat}[\bm{k}(\theta)]$ will also rotate in the opposite direction. At the same time the cross product $\hat{\bm{e}}'_1\times\hat{\bm{e}}'_2=-\hat{\bm{e}}_1\times\hat{\bm{e}}_2$. Overall, the value of the preimage orientation will not change.} In Fig.~\ref{fig:whitehead}(a) the orientations of $\gamma_i$ are shown by green arrows.  
\medskip

Let us make three remarks on the presented definition of the orientation of the inverse images $\gamma_i$:\medskip 

\noindent (\emph{a}) The orientation cannot discontinuously flip along $\gamma_i$, as that would require $\tilde{\boldsymbol{h}}_\mathrm{flat}[\bm{k}({\theta})]$ to pass through $b$ while adjusting the reference point $\bm{k}_0$. However, this cannot happen because the circle $\bm{k}({\theta})$ never crosses $\gamma_i$ (the inverse image of $b$). \medskip

\noindent (\emph{b}) Recall that a solid angle enclosed by a path on the Bloch sphere is equal to twice the Berry phase picked by the corresponding state~\cite{berry_quantalphase}. 
In the present situation, the clockwise (counterclockwise) winding of $\tilde{\boldsymbol{h}}_\mathrm{flat}[\bm k(\theta)]$ implies a negative (positive) Berry phase on path $\bm k(\theta)$. Therefore, the orientation $\bm{\tau}$ indicates the direction in which the component of Berry curvature $\bm{\mathcal{F}}$ parallel with $\gamma_i$ is negative.\medskip 

\noindent (\emph{c}) Finally, while the connected component $\gamma_i$ might form a non-contractible loop in the BZ along one momentum direction, the number of lines that are oriented along this $\bk$-direction minus the number of lines oriented in the opposite direction is equal to the Chern number of the corresponding perpendicular sheet\footnote{To understand this, recall that a Chern number of a 2-band Hamiltonian defined over this 2-torus sheet equals to the number of times that the 2-torus covers the 2-sphere. When a preimage line of some point on a 2-sphere crosses the 2-torus with one orientation, the torus covers this point. When a preimage line crosses the 2-torus with another orientation, the 2-torus uncovers the sphere at this point. Therefore, in order to have no covering of the sphere, all preimage lines must cross the 2-torus with zero net-winding.} (cf.~SM of Ref.~\cite{Nelson:2021}). 
Since we assume Hopf insulators with $C_i = 0$, each non-contractible path $\gamma_i$ that is oriented along one of $\bk$-axes can be paired with another path(s) oriented in the opposite direction, as illustrated in \fig{fig:whitehead}(a). \medskip

Having fixed the orientation of $\gamma_i$, we can define an oriented surface $\Sigma_i$ (the \emph{Seifert surface}) that has the oriented loop $\gamma_i$ as a boundary [light green area in Fig.~\ref{fig:whitehead}(a)]. The surface $\Sigma_i$ always exists for contractible paths~\cite{Baez:1994}, while for non-contractible paths it exists if one considers a collections of several paths with opposite orientation [cf.~remark (\emph{c}) above]. The orientation of the Seifert surface $ \Sigma_i$ is inherited from the orientation $\bm{\tau}(\bm k_0)$ at $\bm k_0 \in \gamma_i$ by applying the right-hand rule near the boundary. Conveniently, for the \emph{union} of the inscribed surfaces, $\bigcup_i \Sigma_i \equiv \Sigma(b)$, the boundary fulfills $\partial [\Sigma(b)] = \tilde{\boldsymbol{h}}_\mathrm{flat}^{-1}(b)$ with matching orientation. The orientation of the surface fixes the overall sign of the Chern number in Eq.~\eqref{eq:whitehead}.

There is one additional subtlety in Hopf-Chern relation when one accounts for the crystallographic point group of a $\bk$-periodic Hamiltonian. For generic Hamiltonians without point-group symmetry, the preimage lines typically avoid crossing each other, as illustrated in \fig{fig:whitehead}(a). In other words, Hamiltonians where preimage lines cross or overlap require fine-tuning of Hamiltonian parameters. On the other hand, with point-group symmetry, it is possible for two or more preimage lines to cross and even overlap/merge completely, without fine-tuning of Hamiltonian parameters. In such `degenerate' cases, multiple Seifert surfaces become path-connected. The Hopf-Chern relation still holds by summing Chern numbers over all Seifert surfaces. To distinguish a completely-degenerate preimage line from a nondegenerate one, we characterize each line by a positive-integer-valued \textit{multiplicity}, which is given by the (unoriented) winding number of the map $\tilde{\boldsymbol{h}}_\mathrm{flat}[\bm{k}(\theta)]$. Explicit examples of nontrivial multiplicity in rotation-symmetric Hamiltonians will be discussed in Sec. \ref{sec:whitehead_rotation}.

\begin{figure}
    \centering
    \includegraphics{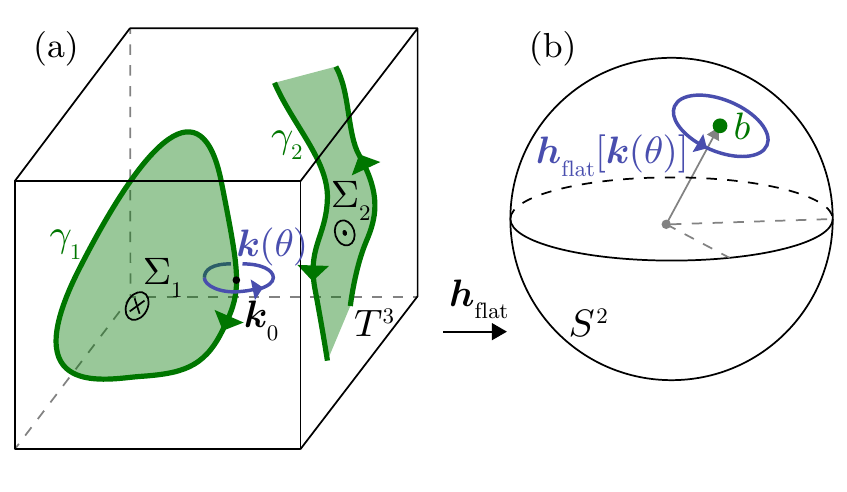}
    \caption{Characterizing a continuous map (the Hamiltonian) $\boldsymbol{h}_\mathrm{flat}$ from the (Brillouin zone) torus $T^3$ to the (Bloch) sphere $S^2$. For a chosen base point $b\in S^2$ [green dot in (b)], the inverse image $\boldsymbol{h}_\mathrm{flat}^{-1}(b)$ is a collection of closed loops $\gamma_i$ [green lines in (a)] with orientation ${\bm \tau}$ defined by the arrows. The orientation is defined by considering an infinitesimal circle $\bm{k}(\theta)$ around a reference point $\bm{k}_0\in \gamma_i$, and by checking whether the image $\boldsymbol{h}_\mathrm{flat}[\bk (\theta)]$ winds (counter)clockwise around $b$ [blue line in (b)]. The paths $\gamma_i$ can be inscribed a (Seifert) surface $\Sigma_i$ [light green sheet in (a)] with an orientation [$\odot$ vs.~$\otimes$ in (a)] inherited from the orientation of the corresponding boundary loop via the right-hand rule.}
    \label{fig:whitehead}
\end{figure}

\subsection{Hopf-Chern relation for \texorpdfstring{$\bk$}{k}-nonperiodic Hamiltonians \label{sec:whiteheadnon}}

So far we showed that the Hopf-Chern relation can be used to compute the Hopf invariant of a $\bk$-periodic Hamiltonian $\tilde{h}$. Here we extend our discussion to $\bk$-nonperiodic Hamiltonians $h(\bk)$ which relate to $\tilde{h}(\bk)$
\begin{align}
   & \tilde{h}(\bk)=V(\bk)h(\bk)V(\bk)^{-1} \lin
   & \tilde{h}(\bk+\bG)=\tilde{h}(\bk).
   \label{eq:periodic_ham}
\end{align}
by conjugation with a diagonal unitary matrix with diagonal elements:
$V_{\alpha\beta}(\bk)=e^{i\bk\br_\alpha}\delta_{\alpha\beta}$ encoding the spatial positioning of basis orbitals within each primitive unit cell. (In cases when the orbitals are spatially separated within each unit cell, the advantage of $\bk$-nonperiodicity has been espoused in Sec.~\ref{app:geom-theo-pol}.)

In Appendix~\ref{app:Hopf-periodic-nonperiodic} we demonstrated that a generalization of the Hopf invariant exists for the $\bk$-nonperiodic Hamiltonian, and can be computed by integrating the Chern-simons three-form of the Berry gauge field [cf.\ \q{eq:hopfinvar}].
Here we show that it can equivalently be computed via a restricted form of the Hopf-Chern relation, where we consider only the preimages of the poles of the Bloch 2-sphere.

To show this, consider a continuous interpolation between the periodic and non-periodic Hamiltonians (that was already introduced in Eq.~\eqref{eq:app:per-nonper-interpol}) $h_f(\bk)=V(t\bk)^{-1}\tilde{h}(\bk)V(t\bk)$ which at $t=0$ gives $h_0(\bk)=\tilde{h}(\bk)$ and at $t=1$ gives $h_1(\bk)=h(\bk)$. Translational symmetry and the energy gap persist throughout this interpolation.

If $h_t(\bk)$ were diagonal at a particular $\bk$ and $t$, it would remain diagonal at that $\bk$ and for any $t\in [0,1]$, since the matrix $V(t\bk)$ is also diagonal. Conversely, if $h_t(\bk)$ had off-diagonal elements at a particular $\bk$ and $t$, its off-diagonal elements can never vanish at that $\bk$ and for any $t$. (The reason is that the off-diagonal component of  the Hamiltonian is of the form $d_x\sigma_x+d_y\sigma_y$ with $(d_x,d_y)$ a real two-vector, and conjugation by  $V=e^{i\theta\sigma_3/2}$ merely rotates the two-vector by the angle $\theta$.) Restriction of the base point in the Hopf-Chern relation to the north and south pole, ensures that the corresponding Hamiltonian is diagonal. Thus the inverse image of either pole remains the same for all $t$. This and the fact that the Berry curvature is a periodic quantity in any convention, $\bm{\mathcal{F}}(\bk+\bG)=\bm{\mathcal{F}}(\bk)$, leads to a well defined and $t$-independent Chern number computed on the corresponding Seifert surface, meaning that $\tilde{C}_{\Sigma(b)}=C_{\Sigma(b)}$ when $b$ is the north or the south pole of the 2-sphere. From this we conclude that the Hopf-Chern relation can be applied to compute the Hopf invariant in both conventions.

\subsection{Proving the Hopf-RTP relation \label{sec:whitehead_rotation}}

\subsubsection{Symmetry defines preimage lines of the south pole}

Now let us apply the Hopf-Chern relation to a Hopf insulator with rotation symmetry.  We assume that the model has space group symmetry P$n$ and the valence (conduction) band is symmetry-equivalent to a basis band representation induced from a  basis orbital $\varphi_1$ ($\varphi_2$) centered at position $\br_1$ ($\br_2$) with  on-site angular momentum $\mathcal{L}_1$ ($\mathcal{L}_2$). In this discussion we assume that the valence band representation is symmetry-equivalent to a basis band representation. On the one hand, lifting this assumption complicates the discussion without bringing further insights. On the other hand, the Hopf-RTP relation in the case of valence band representation non-symmetry-equivalent to a basis band representation was already proven using the boundary properties in Sec.~\ref{sec:hopf-RTP-proof-non-BBR}.

Our choice of the valence (conduction) band representation determines the itinerant angular momentum of the valence (conduction) band at a $C_m$-invariant momentum $\Pi$ to be $\widetilde{\mathcal{L}}_{v(c)}(\Pi)=\mathcal{L}_{1(2)}+\frac{m}{2\pi}\left(\br_{1(2),\perp}\cdot\bG_\Pi\right) \mod m$, where $\bG_\Pi=C_m\Pi-\Pi$. Then the Hamiltonian is diagonal at the rotation-invariant lines $\gamma_\Pi$ at which the mutually-disjoint condition is fulfilled $\Delta\widetilde{\mathcal{L}}(\Pi)=\widetilde{\mathcal{L}}_{c}(\Pi)-\widetilde{\mathcal{L}}_{v}(\Pi)\pmod{m}\neq 0$. For the Hamiltonian with flattened spectrum this means that it is $h_\mathrm{flat}(\Pi,k_z)=\pm\sigma_z$ which corresponds to a mapping to a north or south pole of a 2-sphere $\boldsymbol{h}_\mathrm{flat}(\Pi,k_z)=(0,0,\pm 1)$. The fact that the valence band representation is symmetry-equivalent to a basis band representation guarantees that all $\gamma_\Pi$ lines with fulfilled mutually-disjoint condition are mapped to the same point on the 2-sphere, despite the fact that $\Delta\widetilde{\mathcal{L}}(\Pi)$ need not be identical for all $\Pi$. This one image point is the south pole $h_\mathrm{flat}(\Pi,k_z)=-\sigma_z$ for our choice of valence band representation.
From this we conclude that if we choose the south pole ($\textrm{S}$) to be the base point for the Hopf-Chern relation, all lines $\gamma_\Pi$ with $\Delta\widetilde{\mathcal{L}}(\Pi)\neq0$ belong to the inverse image of $\textrm{S}$. In Fig.~\ref{fig:whitehead-c4-max}(a) we illustrate this for a P4 symmetric model with basis orbitals having on-site angular momenta $\mathcal{L}_1=0$, $\mathcal{L}_2=0$ positioned at $\br_{1,\perp}=(0,0)$, $\br_{2,\perp}=(1/2,1/2)$. The corresponding itinerant angular momenta differences are $\Delta\widetilde{\mathcal{L}}(\Gamma)=0$, $\Delta\widetilde{\mathcal{L}}(\tm)=2$, $\Delta\widetilde{\mathcal{L}}(\tx)=1$, which forces the lines $\gamma_\tx$ and $\gamma_\tm$ to map to the south pole.

\begin{figure*}
    \centering
    \includegraphics{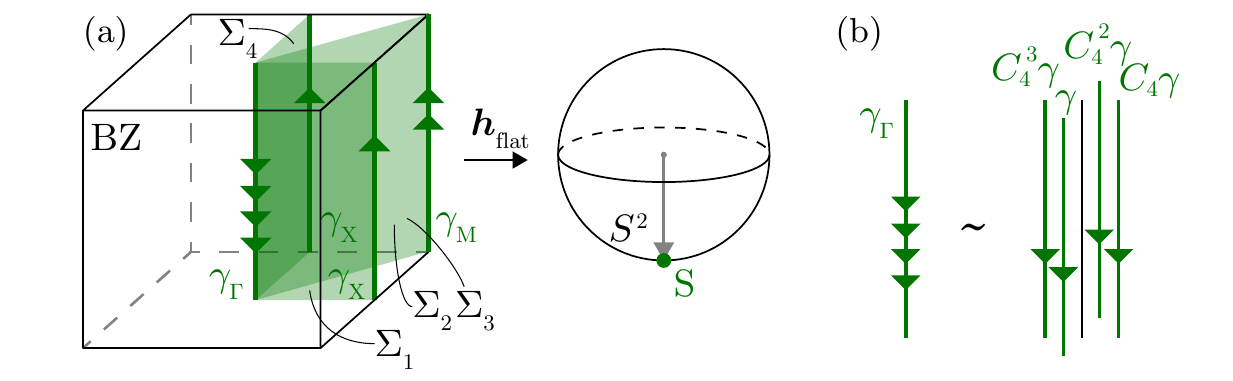}
    \caption{(a) Illustration of a Hopf-Chern relation in the case of a model with P4 space group and itinerant angular momenta differences $\Delta\widetilde{\mathcal{L}}(\Gamma)=0$, $\Delta\widetilde{\mathcal{L}}(\tm)=2$, $\Delta\widetilde{\mathcal{L}}(\tx)=1$. The base point is chosen to be the south pole $\textrm{S}\sim -\sigma_z$. The symmetry forces the lines $\gamma_{\tx}$ and $\gamma_{\tm}$ (denoted with green) to be mapped onto $\textrm{S}$. In addition there can be other non-contractible preimages of $\textrm{S}$ that appear in $4$ $C_4$ related copies. In the presented example these four lines are degenerate along $\gamma_\Gamma$ (shown in green). The green arrows show possible orientation of the lines that is compatible with symmetries and the no-net-winding rule. Multiple arrows show the multiplicity of the line. The oriented preimage lines are paired up to bound 4 Seifert surfaces $\Sigma_j$, $j=1,2,3,4$ (denoted by green shades). The lines with multiplicity $\mu$ form a boundary of $\mu$ distinct Seifert surfaces. (b) The four degenerate preimage lines at $\gamma_\Gamma$ are not fixed to the rotation-invariant line by symmetry and can be split into 4 $C_4$-related copies by a small perturbation of the Hamiltonian.}
    \label{fig:whitehead-c4-max}
\end{figure*}

\subsubsection{Orientation and multiplicity of \texorpdfstring{$\gamma_\Pi$}{gamma Pi} preimages}\label{app:orientation-Cm}

Before discussing other possible preimage lines of the south pole, let us determine the orientation of the $\gamma_\Pi$ lines. According to the comment (\textit{b}) of Sec.~\ref{sec:whitehead}, this can be done by computing the tangential-to-$\gamma_\Pi$ component of the Berry-curvature vector. As all $\gamma_\Pi$ lines are vertical, their orientation is $\bm{\tau}(\Pi)=-\sign(\mathcal{F}_z(\Pi,k_z))\hat{\bm{e}}_z$ where $\hat{\bm{e}}_z$ denotes the unit vector in the $z$ direction of momentum space and the Berry curvature can be computed at any $k_z$ momentum component. 
In the presence of rotation symmetry, as opposed to the generic case, the rotation-invariant preimage lines are characterized not only by orientation but also by a multiplicity introduced in Sec.\ \ref{sec:whiteheadthree}. 
The latter is defined as a number of times an image of a small loop around the rotation-invariant line winds around the south pole on the 2-sphere. 
Consider a $\bk$-loop around $\gamma_\Pi$ line parametrized as $\bk(\theta)=(\Pi,k_z)+\varepsilon [\hat{\bm{e}}_1\cos(\theta)+\hat{\bm{e}}_2\sin(\theta)]$, $\theta\in[0,2\pi]$; if the Hamiltonian in this parametrization obeys the condition $h(\theta+2\pi/\mu)=h(\theta)$ for the largest possible integer $\mu$, we conclude that the image of a $\bk$-loop winds around the south pole $\mu$ times and the multiplicity of the inverse image is $\mu(\Pi)$. 
The $\gamma_\Pi$ line can be viewed as a degenerate set of $\mu(\Pi)$ lines, with each line having unit multiplicity and carrying the same orientation $\bm{\tau}(\Pi)$. The degeneracy is protected by the rotation symmetry.

Let us relate, as best as we can, the orientations and multiplicities of the preimage lines $\gamma_\Pi$ to the itinerant angular-momentum difference at $\Pi$. For this purpose, we consider a Taylor expansion of the Hamiltonian around $\bk_0\in\gamma_\Pi$ up to the lowest order in $\bm{\kappa}=\bk-\bk_0$ in the presence of rotational symmetry. As shown in Appendix~\ref{app:kp-sym-nonuniax}, the expanded Hamiltonian must fulfill
\begin{equation}
    \widetilde{R}_{C_m}(\Pi)h^{\bk\cdot\boldsymbol{p}}(\bkappa)\widetilde{R}_{C_m}^{-1}(\Pi)=h^{\bk\cdot\boldsymbol{p}}(C_m\bkappa)
\end{equation}
with the itinerant rotation matrix $\widetilde{R}_{C_m}(\Pi)$ given in Eq.~\eqref{eqn:widetilde-R-C-m}. For our choice of valence (conduction) band representation we can replace $\widetilde{\mathcal{L}}_2(\Pi)-\widetilde{\mathcal{L}}_1(\Pi)$ by $\Delta\widetilde{\mathcal{L}}(\Pi)=\widetilde{\mathcal{L}}_c(\Pi)-\widetilde{\mathcal{L}}_v(\Pi) \pmod{m}$ in the expression for $\widetilde{R}_{C_m}(\Pi)$. For all possible rotation symmetries $C_m$ and differences in itinerant angular momenta $\Delta\widetilde{\mathcal{L}}(\Pi)$ the symmetry-allowed form of a flatted Hamiltonian $h_\mathrm{flat}(\boldsymbol{\kappa})$ is presented in the third column of Table~\ref{tab:sym_kp_ham} with introduced notations $\kappa_\pm=\kappa_x\pm i\kappa_y$ and $\sigma_\pm=\sigma_x\pm i\sigma_y$. The corresponding vector of Hamiltonian coefficients is normalized as the corrections to the vector $\boldsymbol{h}_\mathrm{flat}(\boldsymbol{\kappa}=0)=(0,0,-1)$ are linear in $\boldsymbol{\kappa}$ and therefore can be neglected in the norm $|\boldsymbol{h}_\mathrm{flat}(\boldsymbol{\kappa})|=1 +O(\boldsymbol{\kappa})$. The $z$-component of the Berry curvature in a 2-band Hamiltonian can be computed as
\begin{equation}
\mathcal{F}_z=\frac{1}{2}\epsilon_{abc}h_{\mathrm{flat},a}\,\partial_x h_{\mathrm{flat},b}\,\partial_y h_{\mathrm{flat},c},
\label{eq:Fz}
\end{equation}
which can be simplified to 
\begin{equation}
\mathcal{F}_z=-2h_{\mathrm{flat},z}(\partial_+ h_{\mathrm{flat},+}\,\partial_- h_{\mathrm{flat},-} - \partial_+ h_{\mathrm{flat},-}\,\partial_- h_{\mathrm{flat},+})
\label{eq:Fz_sym_constr}
\end{equation}
as the symmetry constraints the $h_{\mathrm{flat},z}$ component to not depend on $\kappa_x$, $\kappa_y$; the introduced notations are $\partial_\pm=(\partial_x\mp i\partial_y)/2$ and $h_{\mathrm{flat},\pm}=(h_{\mathrm{flat},x}\mp ih_{\mathrm{flat},y})/2$. Analyzing the results of Table~\ref{tab:sym_kp_ham} we see that whenever $\Delta\widetilde{\mathcal{L}}\neq m/2$ a coefficient in front of $\sigma_+$ depends only on one of $\kappa_\pm$ and the sign of $\mathcal{F}_z$ is determined by symmetry representation alone, it is presented in the second to last column in Table~\ref{tab:sym_kp_ham}. At the same time the power of a coefficient in front of $\sigma_+$ gives the multiplicity of the preimage line. Assume $h_{\mathrm{flat},0}=-\sigma_z+(a\kappa_+^\mu+b\kappa_-^\mu)\sigma_++h.c.$ The loop around $\bk_0$ can be parametrized with $\kappa_\pm=\varepsilon e^{\pm i\theta}$, $\theta\in[0,2\pi]$. We see that in this parametrization $h_{\mathrm{flat},0}(\theta+2\pi/\mu)=h_{\mathrm{flat},0}(\theta)$ and hence the multiplicity of $\gamma_\Pi$ line is $\mu$, which is listed in the last column of Table~\ref{tab:sym_kp_ham}. For the example case (outlined in the previous subsection) one choice of preimage line orientations is shown in Fig.~\ref{fig:whitehead-c4-max}(b) with green arrows, while the number of arrows denote the multiplicities of the lines.

\begingroup
\renewcommand{\arraystretch}{1.5}
\setlength{\tabcolsep}{8pt}
\begin{table*}
\centering
\begin{tabular}{llllll}
\hline\hline
$m$ & $\Delta\widetilde{\mathcal{L}}$ & $h_\mathrm{flat}(\bm{\kappa})+\sigma_z$ & $\mathcal{F}_z/2$ & $\tau_z$ & $\mu$ \\
\hline
$2$ & $1$ & $(a\kappa_++b\kappa_-)\sigma_+ + h.c.$ & $|a|^2-|b|^2$ & undetermined & 1 \\
$3$ & $\pm 1 \mod{3}$ & $a\kappa_\mp\sigma_+ + h.c.$ & $\mp |a|^2$ & $\pm1$ & 1 \\
$4$ & $\pm 1\mod{4}$ & $a\kappa_\mp\sigma_+ + h.c.$ & $\mp |a|^2$ & $\pm1$ & 1 \\
 & $2$ & $(a\kappa_+^2+b\kappa_-^2)\sigma_+ + h.c.$ & $4\left[|a\kappa_+|^2-|b\kappa_+|^2\right]$ & undetermined  & 2\\
$6$ & $\pm 1\mod{6}$ & $a\kappa_\mp\sigma_+ + h.c.$ & $\mp |a|^2$ & $\pm1$ & 1 \\
 & $\pm 2\mod{6}$ & $a\kappa_\mp^2\sigma_+ + h.c.$ & $\mp 4|a\kappa_+|^2$ & $\pm1$ & 2 \\
 & $3$ & $(a\kappa_+^3+b\kappa_-^3)\sigma_+ + h.c.$ & $9\left[|a\kappa_+^2|^2-|b\kappa_+^2|^2\right] $& undetermined & 3 \\
\hline\hline
\end{tabular}
\caption{ For all differences between itinerant angular momenta of conduction and valence bands at a given reduced momentum $\Pi$ $\Delta\widetilde{\mathcal{L}}(\Pi)=\widetilde{\mathcal{L}}_c(\Pi)-\widetilde{\mathcal{L}}_v(\Pi)\pmod m$, we present a form of flattened Hamiltonian in the vicinity of $\bk_0$ point allowed by the symmetry. For this Hamiltonian we compute the $\mathcal{F}_z$ component of the Berry curvature, and deduce the resulting orientation $\bm{\tau}=\tau_z\hat{\bm{e}}_z$ and multiplicity $\mu$ of a corresponding inverse image $\gamma_\Pi$. In some cases the orientation can not be determined solely from the itinerant angular momentum, however, as we will see, this does not impair the proof of the Hopf-RTP relation.}
\label{tab:sym_kp_ham}
\end{table*}
\endgroup

\subsubsection{Other non-contractible preimage lines \texorpdfstring{$\gamma_i$}{gamma i}}

Importantly, the BZ can contain other than $\gamma_\Pi$ lines that are preimages of the south pole. At this stage of the proof we focus only on the preimages that form non-contractible loops along the $k_z$ direction in the BZ and discuss how other preimages affect the result later in this section. 
In addition to $\gamma_\Pi$ lines with fulfilled mutually-disjoint condition, the preimage of the south pole can contain two types of non-contractible lines: (\emph{i}) lines that appear at generic positions in the BZ and (\emph{ii}) lines that appear at rotation-invariant momenta $\gamma_{\Pi'}$ with violated mutually-disjoint condition, $\Delta\widetilde{\mathcal{L}}(\Pi')=0$.

For both types of preimage lines we show, that the $C_n$ rotation symmetry forces them to appear in $n$ rotation-related copies. 
(\emph{i}) For the preimage lines at generic positions, this follows from the symmetry condition on the Hamiltonian $R_{C_n}h(\bk)R_{C_n}^{-1}=h(C_n\bk)$, which equates the Hamiltonians at $\bk$ and at $C_n\bk$ when the Hamiltonian is diagonal.
(\emph{ii}) If a preimage line happened to coincide with a $C_m$-invariant $\gamma_{\Pi'}$ line with violated mutually-disjoint condition, first there will be a $n/m$-multiplet of preimage lines at all $C_n$-related $\gamma_{\Pi'}$ lines. Second, without the mutually-disjoint condition, the south pole preimage is not forced to lie at a $C_m$-invariant line by symmetry. 
Therefore, a small perturbation of the Hamiltonian can shift this preimage line away from the rotation-invariant line. But as we already argued in the case (\emph{i}), the preimage lines at generic positions must appear in $n$ $C_n$-related copies. 
This means that the initial preimage line, at each $C_m$-invariant line of a $n/m$-multiplet, forms $m$ degenerate copies [cf.~Fig.~\ref{fig:whitehead-c4-max}(b)].

We denote these additional preimage lines as $\gamma_i$ and all $n$ copies as $C_n^j\gamma_i$, $j=0,\dots n-1$. We remark, that the preimages of the north pole ($h_\mathrm{flat}=\sigma_z$) also appear as $n$ rotation-related copies for the same reason as for the south pole; this property will be used at the end of this section. 

\subsubsection{Orientation of non-contractible preimage lines}

Having specified all non-contractible preimages of $\textrm{S}$, we proceed by defining their orientations. First, let us introduce a scalar orientation along $k_z$-axis as 
\begin{equation}
    \tau_z(\bk_0)=\sgn(\boldsymbol{\tau}(\bk_0)\cdot\hat{\bm{e}}_z)\in\{-1, 1\},
\end{equation}
where $\bk_0$ is a point on a considered line, and for simplicity we assume that the line crosses each plane at fixed $k_z$ only once. Then $\tau_z$ does not depend on the specific point $\bk_0$. The definition can be straightforwardly generalized if this assumption is relaxed. We proceed to define the $k_z$-orientations of non-contractible lines. 

In addition to the symmetry constraints defined in Sec.~\ref{app:orientation-Cm}, the orientations of the non-contractible lines are restricted by two more rules:\\

\noindent{(\emph{i})} Any two inverse images ($\gamma_1$ and $\gamma_2=C_n\gamma_1$) that are related by rotation have the same $k_z$-orientation $\tau_z(\gamma_1)=\tau_z(\gamma_2)$, because the Berry curvature transforms as a vector under rotation.\\

\noindent{(\emph{ii})} As discussed in the comment (\textit{c}) in Sec.~\ref{sec:whitehead} the number of upwards oriented lines must be equal to the number of downwards oriented lines, meaning that 
\begin{equation}
    \sum_{\gamma \textrm{, noncontractible along $\hat{\bm{e}}_z$}}\mu(\gamma)\tau_z(\gamma)=0.
    \label{eq:app:balanced-pairing}
\end{equation}
Crucially, each non-contractible preimage contribute to the sum number of times, equal to its multiplicity. We refer to this constraint as the \emph{balanced orientation rule}. 
 \\

For a set of all non-contractible preimage lines, $\gamma_\Pi$ and $C_n^j\gamma_i$, the symmetry conditions and the balanced orientation rule restrict possible orientations to a limited number of options. For several concrete examples of the preimage orientations, see Table \ref{tab:orient-example} and Fig.~\ref{fig:whitehead-c4-max}(a). Although in some cases, the outlined conditions are not enough to uniquely determine the orientations from the symmetry representations alone, we will see that the obtained information about orientations is enough to prove the Hopf-RTP relation. 

\begingroup
\renewcommand{\arraystretch}{2}
\setlength{\tabcolsep}{8pt}
\begin{table*}
    \centering
    \caption{Possible preimage orientations for some crystalline Hopf models given by the space group and the symmetry $\Delta\mathcal{L}=\mathcal{L}_2-\mathcal{L}_1$ and spatial positions $\Delta\br_\perp=\br_{2,\perp}-\br_{1,\perp}$ of the basis orbitals. Arrows $\uparrow$ and $\downarrow$ denote positive and negative orientations along $k_z$ axis and the number of arrows [$\upuparrows$] denote multiplicity. Value in brackets [$(2)$] specifies how many rotation-related lines of this type constitutes the south pole preimage.}
    \label{tab:orient-example}
    \begin{tabular}{llllll}
    \hline\hline
    SG & $\Delta\mathcal{L}$ & $\Delta\br_\perp$ & $\Delta\widetilde{\mathcal{L}}(\Pi)$ & orientations from symmetry & possible total orientations \\
    \hline
    P4 & 1 & $(0,0)$ & $\Delta\widetilde{\mathcal{L}}(\Gamma)=1$, $\Delta\widetilde{\mathcal{L}}(\tx)=1$, $\Delta\widetilde{\mathcal{L}}(\tm)=1$ & $\Gamma\uparrow$, $\tm\uparrow$ & 
    \renewcommand{\arraystretch}{1}
    \begin{tabular}[t]{l}
          $\Gamma\uparrow$, $\tm\uparrow$, $(2)\tx\downarrow$ \\
          or \\
          $\Gamma\uparrow$, $\tm\uparrow$, $(2)\tx\uparrow$, $(4)\gamma_1\downarrow$
    \end{tabular} \renewcommand{\arraystretch}{2} \\
    P4 & 0 & $(1/2,1/2)$ & $\Delta\widetilde{\mathcal{L}}(\Gamma)=0$, $\Delta\widetilde{\mathcal{L}}(\tx)=1$, $\Delta\widetilde{\mathcal{L}}(\tm)=2$ &  &
    \renewcommand{\arraystretch}{1}
    \begin{tabular}[t]{l}  $(2)\tx\uparrow$, $\tm\downdownarrows$ \\ or \\ $(2)\tx\uparrow$, $\tm\upuparrows$, $(4)\gamma_1\downarrow$ \end{tabular}
    \renewcommand{\arraystretch}{2}\\
     P4 & 1 & $(1/2,1/2)$ & $\Delta\widetilde{\mathcal{L}}(\Gamma)=1$, $\Delta\widetilde{\mathcal{L}}(\tx)=0$, $\Delta\widetilde{\mathcal{L}}(\tm)=3$ & $\Gamma\uparrow$, $\tm\downarrow$ &
      \renewcommand{\arraystretch}{1}
      \begin{tabular}[t]{l}$\Gamma\uparrow$, $\tm\downarrow$ \end{tabular}
       \renewcommand{\arraystretch}{2}\\
     P3 & 1 & $(0,0)$ & $\Delta\widetilde{\mathcal{L}}(\Gamma)=1$, $\Delta\widetilde{\mathcal{L}}(\tx)=1$, $\Delta\widetilde{\mathcal{L}}(\tm)=1$ & $\Gamma\uparrow$, $\tk\uparrow$, $\tkpr\uparrow$ & 
      \renewcommand{\arraystretch}{1}
      \begin{tabular}[t]{l} $\Gamma\uparrow$, $\tk\uparrow$, $\tkpr\uparrow$, $(3)\gamma_1\downarrow$ \end{tabular} 
       \renewcommand{\arraystretch}{2}\\
       \hline\hline
    \end{tabular}
\end{table*}
\endgroup

\subsubsection{Seifert surface and Hopf-Chern relation for \texorpdfstring{$h_\mathrm{flat}= -\sigma_z$}{hflat}
}

After determining the orientations and multiplicities of the non-contractible preimage lines, we proceed by defining the Seifert surface stretched over these lines. For this we arbitrarily pair up the preimage lines such that each pair contain upwards $\gamma_{\uparrow i}$ and downwards $\gamma_{\downarrow i}$ oriented lines. This is possible due to validity of the balanced orientation rule. Note, that a preimage line with multiplicity $\mu$ forms a boundary of $\mu$ distinct Seifert surfaces.  Denoting each distinct Seifert surface bounded by $\gamma_{\downarrow i}$ and $\gamma_{\uparrow i}$ as $\Sigma_i$, we express the Hopf invariant as a sum of Chern numbers $C_{\Sigma_i}$ over all Seifert surfaces:
\begin{equation}
    \chi=\sum_{i}C_{\Sigma_i},
    \label{eq:app:whitehead-non-contractible}
\end{equation}
(A nontrivial Chern number over a Seifert surface is consistent with the triviality of the first Chern class, meaning the Chern number on any 2D sub-torus of the Brillouin zone vanishes.) 
The Chern number of each Seifert surface can be computed as the difference in the Berry-Zak phases computed along the paths
\begin{equation}
    C_{\Sigma_i}=(\mathscr{Z}_{\gamma_{\uparrow i}}-\mathscr{Z}_{\gamma_{\downarrow i}})/2\pi.
    \label{eq:app:chern-via-Zak}
\end{equation}
Though $\gamma_{\uparrow i}$ and $\gamma_{\downarrow i}$ have opposite orientations, the corresponding two Zak phases are computed over paths of the same orientation, hence the minus sign in the above equation. By itself, $\mathscr{Z}_{\gamma_\uparrow}$ is a phase defined modulo $2\pi$; however, given any global, differential gauge whose existence is guaranteed by the trivial first Chern class, the difference $(\mathscr{Z}_{\gamma_\uparrow}-\mathscr{Z}_{\gamma_\downarrow})$ is uniquely defined. 
Plugging Eq.~\eqref{eq:app:chern-via-Zak} into Eq.~\eqref{eq:app:whitehead-non-contractible} and taking into account the orientation along $k_z$-axis $\tau_z$ of each line $\gamma$ we get
\begin{align}
    \chi\eq \sum_i \sum_{\gamma\in \{\gamma_{\uparrow i},\gamma_{\downarrow i}\}}\tau_z(\gamma)\mathscr{Z}_{\gamma}/2\pi \la{doublesum}
\end{align}  
From this sum, we isolate the contribution from  rotation-invariant lines $\gamma_{\Pi}$ with fulfilled mutually-disjoint condition $\Delta\widetilde{\mathcal{L}}(\Pi)\neq 0$.  For each $\gamma_{\Pi}$, we replace $\mathscr{Z}/2\pi\rightarrow  \mathscr{P}$ to remind ourselves that  the normalized Zak phase can be interpreted as the dimensionless polarization for a one-$k$-parameter Hamiltonian \eqref{eq:pol-def}:
\begin{align}    
    \chi \eq \smashoperator{\sum_{\gamma_\Pi\in\boldsymbol{h}_\mathrm{flat}^{-1}(\textrm{S});\Delta\widetilde{\mathcal{L}}(\Pi)\neq 0}}\mu(\Pi)\tau_z(\Pi)\mathscr{P}(\Pi) +
    \smashoperator[r]{\sum_{{\textrm{all other }\gamma\in\boldsymbol{h}_\mathrm{flat}^{-1}(\textrm{S})}}} 
    \tau_z(\gamma)\mathscr{Z}_{\gamma}/2\pi ,\la{secondsum}
\end{align}
The meaning of $\sum_{\gamma_\Pi}$ is to sum over all components of $\boldsymbol{h}_\mathrm{flat}^{-1}(\textrm{S})$ with the restriction: $\Delta\widetilde{\mathcal{L}}(\Pi)\neq 0$; note each term appearing once in $\sum_{\gamma_\Pi\in\boldsymbol{h}_\mathrm{flat}^{-1}(\textrm{S})}$ appears $\mu(\Pi)$  times in the double summation of \q{doublesum}. 
The second term in \q{secondsum} sums over all {other preimage lines},
not already included in the first sum.
Because $\gamma$ 
is not necessarily vertical, it does not necessarily have the interpretation of a  polarization. However, since along these lines the Hamiltonian is diagonal, the valence band does not hybridize with the conduction band and hence, $\mathscr{Z}_{\gamma
}/2\pi=_1 z_1$. {Using Eq.~\eqref{eq:app:balanced-pairing}, we can subtract zero which is the sum over all $k_z$-orientations (weighted with multiplicities) of the preimage lines, multiplied by the value of $z_1$} and get
\begin{align}
    \chi&=\sum_{\gamma_\Pi\in\boldsymbol{h}_\mathrm{flat}^{-1}(\textrm{S});\Delta\widetilde{\mathcal{L}}(\Pi)\neq 0}\mu(\Pi)\tau_z(\Pi)(\mathscr{P}(\Pi)-z_1) \lin
    & \as+ \sum_{\textrm{all other }\gamma\in \boldsymbol{h}_\mathrm{flat}^{-1}(\textrm{S})}\tau_z(\gamma)(\mathscr{Z}_{\gamma}/2\pi-z_1).
    \label{eq:app:Hopf-RTP-proof-1}
\end{align}
The sum in the second row vanishes modulo $n$, because the preimage lines $\gamma$, that do not coincide with rotation-invariant lines with mutually-disjoint condition, necessarily come in $n$ copies with identical $k_z$-orientations. $n$ copies of an integer-valued expression $(\mathscr{Z}_{\gamma}/2\pi-z_1)\in\mathbb{Z}$ sum up to a multiple of $n$. Let us further reorganize the first term in the sum by substituting the product of orientation and multiplicity of the rotation-invariant preimage lines $\mu(\Pi)\tau_z(\Pi)$ by the itinerant angular momentum difference $\Delta\widetilde{\mathcal{L}}(\Pi)$.
\begin{equation}
\chi=_n \sum_{\gamma_\Pi\in\boldsymbol{h}_\mathrm{flat}^{-1}(\textrm{S});\Delta\widetilde{\mathcal{L}}(\Pi)\neq 0}\Delta\widetilde{\mathcal{L}}(\Pi)(\mathscr{P}(\Pi)-z_1).
\label{eq:app:Hopf-RTP-proof-2}
\end{equation}
This transformation is allowed for the following reason. According to Table \ref{tab:sym_kp_ham}, $\Delta\widetilde{\mathcal{L}}=\tau_z\,\mu\mod m$. We also know that the BZ contains $n/m$ $C_n$-related copies of a $C_m$-invariant line, which are all part of the preimage of $\textrm{S}$. This means that in the sum over all preimages, $n/m\times \tau_z\,\mu=_nn/m\times\Delta\widetilde{\mathcal{L}}$, with $\Delta\widetilde{\mathcal{L}}\in \{0,1,\ldots, m-1\}$. In the obtained relation \eqref{eq:app:Hopf-RTP-proof-2}, the sum over preimage lines can be replaced by the sum over a maximal set of reduced momenta $\{\Pi_j\}_{j=1\dots J}$ with fulfilled mutually-disjoint and iso-orbital conditions. In the last step of the derivation we replace $\sum_{j=1}^J\Delta\widetilde{\mathcal{L}}(\Pi_j)z_1=_n\sum_{j=1}^J\Delta\widetilde{\mathcal{L}}(\Pi_j)\mathscr{P}(\Pi_1)$, as for any arbitrarily chosen $\Pi_1\in \{\Pi_j\}_{j=1\dots J}$, $\mathscr{P}(\Pi_1)=_1 z_1$ and $\sum_{j=1}^J\Delta\widetilde{\mathcal{L}}(\Pi_j)\in n\mathbb{Z}$ because of a trivial Chern class of the model [cf.~Appendix~\ref{app:Chern-via-ell}]. With this we obtain the desired relation \eqref{eq:app:Hopf-RTP}.

\subsubsection{Contractible south pole preimages do not modify relation modulo \texorpdfstring{$n$}{n}\label{sec:other_preimages}}
To finish the proof of the Hopf-RTP relation we need to take into account other possible preimages of the south pole, that form contractible loops. There are two types of such preimages: \\
\noindent{(\emph{i})} A closed loop at a generic position in the BZ that is repeated $n$ times due to $C_n$ symmetry [cf.~Fig.~\ref{fig:whitehead-mod4}(a) for a P4 illustration]. \\
\noindent{(\emph{ii})} A closed loop that encircles a $C_m$-invariant line in the BZ. This loop must appear in $n/m$ copies and has a $C_m$-symmetric shape [cf.~Fig.~\ref{fig:whitehead-mod4}(b) for a P4 illustration].

\begin{figure}
    \centering
    \includegraphics{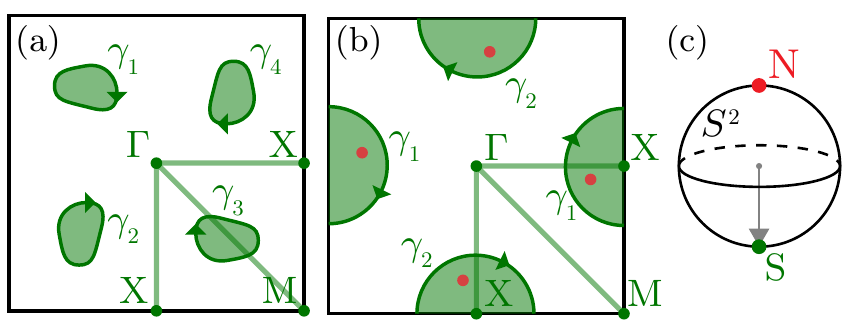}
    \caption{Illustration of contractible south pole preimages in a 2D horizontal cut of the BZ. (a) Closed loops at generic positions that form 4 $C_4$-related copies. Each of them carries an integer Chern number, therefore their contribution to the Hopf invariant is $4\mathbb{Z}$. (b) Closed loops that encircle a $C_2$-invariant line $\gamma_\tx$ form 2 $C_4$-related copies $\gamma_1$ and $\gamma_2$. The linking number of these preimage lines with the preimages of the north pole is forced to be a multiple of 4 by rotation symmetry. (c) A 2-sphere with the north and south pole, whose preimage lines are shown in panels (a) and (b).}
    \label{fig:whitehead-mod4}
\end{figure}

Both types of preimages modify the Hopf-RTP relation only by multiples of $n$:\\
\noindent{(\emph{i})} Assume that a closed loop carries a Chern number $C'$. Since the loops appear as $n$ copies in the BZ, each copy carries the same Chern number $C'$ and together they contribute to the Hopf invariant $nC'$, which does not modify the relation modulo $n$.\\
\noindent{(\emph{ii})} To determine the contribution to the Hopf invariant from each of the $n/m$ copies of loops we consider preimages of the north pole of the 2-sphere [cf.~Fig~\ref{fig:whitehead-mod4}(c)]. These preimages also come in $n$ $C_n$-related copies and if one of them happens to be linked with a component of the preimage of the south pole, the $m$ $C_m$-related copies of it will also be linked with the same component [cf.~Fig.~\ref{fig:whitehead-mod4}(b)]. Since the Hopf invariant can alternatively be calculated as the number of links between preimages of two points on the 2-sphere, we can conclude that each of $n/m$ loops contributes $m\mathbb{Z}$ value to the Hopf invariant. Therefore, all loops in total contribute $n\mathbb{Z}$ and do not modify the relation modulo $n$.

We conclude that we considered contributions of all possible preimages of the south pole to the Hopf invariant and it is given by Eq.~\eqref{eq:app:Hopf-RTP} up to multiples of $n$.

\section{Proof of conditionally-robust surface states}\label{app:Toeplitz-bounds}

In this Appendix we derive the constraint on energy eigenvalues, which we applied in Sec.~\ref{sec:anganomalymore2} to infer 
the metallic signatures at open boundaries with sharply terminated hoppings of RTP insulators with \emph{multiple} energy bands. 
The first proof we present in Sec.~\ref{app:block-Toeplitz} follows directly from spectral properties of finite-dimensional, block Toepliz matrices derived in Ref.~\cite{Miranda:2000}, and here translated to the language of the Bloch band theory. 
An alternative proof valid in the limit of an infinite-sized lattice is provided in  Sec.~\ref{app:transfermatrix} within the transfer-matrix formalism, with certain restrictions on the number of bands and the range of hopping.

\subsection{Spectral bounds for block Toeplitz matrices}\label{app:block-Toeplitz}

\begin{theorem} 
Consider a one-dimensional $N_\mathrm{b}$-band model with a momentum-space Bloch Hamiltonian $h(k)$. 
For the purpose of this Appendix we adopt the periodic convention, $h(k) = h(k+2\pi)$. 
We use $E_a(k)$  to denote the eigenvalues of the Hamiltonian, where $a\in\{1,\ldots,N_\mathrm{b}\}$ is the band index and $k\in[-\pi,+\pi]$ is the momentum in the first Brillouin zone. 
We define the spectral bounds of the bulk Hamiltonian as
\begin{align}
E_\mathrm{min}&=\sup\left\{E\in\mathbb{R}\;\big|\;\forall a,k: E_\mathrm{min}\leq E_a(k)\right\}, \label{eqn:E-min-sup}\\
E_\mathrm{max}&=\inf\left\{E\in\mathbb{R}\;\big|\;\forall a,k: E_\mathrm{min}\geq E_a(k)\right\}. \label{eqn:E-max-inf}
\end{align}
Let further $\lambda$ be an eigenvalue of a finite chain with open boundaries, with $\mathcal{N}$ unit cells, and with sharply terminated hopping. 
We claim that $\lambda$ is constrained by the same bounds,
\begin{equation}
E_\mathrm{min} \leq \lambda \leq E_\mathrm{max}.\label{eqn:sharp-spectrum-bound}
\end{equation}
In other words, in the presence of the sharp boundary condition there cannot be an edge state lying outside the energy range set by the bulk states.\medskip

\noindent  \emph{Remark:} 
In a physical setting, we usually assume $h(k)$ to be a smooth function that obeys the global spectral bound at every $k$.
However, for the purpose of the proof, we only need to assume that $h(k)$ is $L^1$-integrable (possibly discontinuous) and that it obeys the spectral bound in Eqs.~(\ref{eqn:E-min-sup}) and~(\ref{eqn:E-max-inf}) at \emph{almost} every momentum $k$.\medskip
\end{theorem}

\begin{proof}
For each $r\in\mathbb{Z}$ we define the $N_\mathrm{b}\times N_\mathrm{b}$ matrix 
\begin{equation}
\mathcal{T}_r = \frac{1}{2\pi} \int_{-\pi}^{+\pi} dk \, e^{-irk} h(k), \label{eqn:block-Toeplitz-elements}
\end{equation}
which in the physical setting encodes the hopping amplitude between each pair of orbitals separated by $r$ unit cells. 
More generally, the matrices $\{\mathcal{T}_r\}_{r\in\mathbb{Z}}$ are all well-defined if $h(k)$ is $L^1$-integrable. 
The Hamiltonian of the finite chain
with $\mathcal{N}$ unit cells and with the sharp boundary condition is then given by the block Toeplitz matrix
\begin{equation}
\mathcal{H}_\mathcal{N} = \left(\begin{array}{cccc}
\mathcal{T}_0 & \mathcal{T}_1 & \cdots & \mathcal{T}_{\mathcal{N}-1} \\
\mathcal{T}_{-1} & \mathcal{T}_0 & \ddots & \vdots \\
\vdots & \ddots & \ddots & \mathcal{T}_1 \\
\mathcal{T}_{-1+\mathcal{N}} & \cdots & \mathcal{T}_{-1} & \mathcal{T}_0
\end{array}\right),\label{eqn:block-Toeplitz-Hamiltonian}
\end{equation}
i.e., with the block structure $(\mathcal{H}_\mathcal{N})_{j,j'} = \mathcal{T}_{j'-j}$. 

Let further $\psi^{(\mathcal{N})} = \left(\psi_1^{(\mathcal{N})}, \psi_2^{(\mathcal{N})}, \ldots, \psi_\mathcal{N}^{(\mathcal{N})} \right)^\top$ be an arbitrary eigenstate of the Hamiltonian in Eq.~(\ref{eqn:block-Toeplitz-Hamiltonian}) with eigenvalue $\lambda$, i.e., $ \mathcal{H}_\mathcal{N} \psi^{(\mathcal{N})} = \lambda \psi^{(\mathcal{N})}$, where each of $\left\{\psi_j^{(\mathcal{N})}\right\}_{j=1}^\mathcal{N}$ is an $N_\mathrm{b}$-component vector of amplitudes of the eigenvector within the $j^\textrm{th}$ unit cell. 
We know that eigenstates of finite Hermitian matrices, such as the Hamiltonian in Eq.~(\ref{eqn:block-Toeplitz-Hamiltonian}), have finite $L^2$ norm that can be rescaled to unity, i.e.,
\begin{equation}
\norm{\psi^{(\mathcal{N})}}^2= \sum_{j=1}^\mathcal{N} \norm{\psi_j^{(\mathcal{N})}}^2 = \sum_{j=1}^\mathcal{N}\sum_{a=1}^{N_\mathrm{b}} \abs{\psi^{(\mathcal{N})}_{j,a}}^2 = 1,\label{eqn:L2-normalizable-eigenstate}
\end{equation}
where $\left\{\psi_{j,a}^{(\mathcal{N})}\right\}_{a=1}^{N_\mathrm{b}}$ are the components of the wave function within $j^\textrm{th}$ unit cell.

We introduce the function
\begin{equation}
v^{(\mathcal{N})}(k) = \sum_{j=1}^\mathcal{N} e^{-ijk} \psi_j^{(\mathcal{N})}. \label{eqn:block-Fourier-in-proof}
\end{equation}
The function $v^{(\mathcal{N})}(k)$ at momenta $k=\frac{2\pi\alpha}{\mathcal{N}},\alpha\in\mathbb{Z}$ can be interpreted as the Fourier transform of the wave function $\psi^{(\mathcal{N})}$.
Nevertheless, we here mean to define $v^{(\mathcal{N})}(k)$ for \emph{all} $k\in[-\pi,+\pi]$.
Since the sum in Eq.~(\ref{eqn:block-Fourier-in-proof}) is finite and the summands are bounded by the norm in Eq.~(\ref{eqn:L2-normalizable-eigenstate}), one finds that $v^{(\mathcal{N})}(k)$ is also a bounded function, thus one does not have to worry about potential divergences. 
In fact, one can explicitly derive that 
\begin{align}
\norm{v^{(\mathcal{N})}}^2 &\equiv  \frac{1}{2\pi}\int_{-\pi}^{+\pi} dk \, \left[v^{(\mathcal{N})}(k)\right]^\dagger \cdot v^{(\mathcal{N})}(k)  \nonumber \\
&= \frac{1}{2\pi} \sum_{j,j'=1}^\mathcal{N} \left[\psi^{(\mathcal{N})}_j\right]^\dagger \cdot \psi^{(\mathcal{N})}_{j'} \int_{-\pi}^{+\pi} dk \, e^{i(j-j')k}   \nonumber \\
&= \sum_{j=1}^\mathcal{N} \left[\psi_j^{(\mathcal{N})}\right]^\dagger \cdot \psi_{j}^{(\mathcal{N})} \equiv \norm{\psi^{(\mathcal{N})}}^2,\label{eqn:norm-v-N}
\end{align}
which is similar to the Parseval's identity (a remark on this is added below the proof). We limit our attention to normalized states, $\norm{\psi}^2 = \norm{v}^2 = 1$. 

We proceed to study the expression
\begin{align}
\lambda &= \left[\psi^{(\mathcal{N})}\right]^\dagger \mathcal{H}_\mathcal{N} \psi^{(\mathcal{N})} \stackrel{(\mathrm{\ref{eqn:block-Toeplitz-Hamiltonian}})}{=} \sum_{j,j'=1}^\mathcal{N} \left[\psi_j^{(\mathcal{N})}\right]^\dagger \cdot  (\mathcal{H}_\mathcal{N})_{j,j'} \cdot \psi_{j'}^{(\mathcal{N})} \nonumber \\
&\stackrel{(\mathrm{\ref{eqn:block-Toeplitz-elements}})}{=} \frac{1}{2\pi}\int_{-\pi}^{+\pi} dk \, \sum_{j,j'=1}^\mathcal{N} e^{i(j-j')k} \left[\psi_j^{(\mathcal{N})}\right]^\dagger \cdot h(k) \cdot \psi_{j'}^{(\mathcal{N})} \nonumber \\
&\stackrel{(\mathrm{\ref{eqn:block-Fourier-in-proof}})}{=} \frac{1}{2\pi} \int_{-\pi}^{+\pi} dk \, \left[v^{(\mathcal{N})}(k)\right]^\dagger\cdot h(k) \cdot v^{(\mathcal{N})}(k).\label{eqn:bound-on-norm}
\end{align}
Note that the last expression in Eq.~(\ref{eqn:bound-on-norm}) is a weighted (over~$k$) expectation value of operator $h(k)$ for the vector $v^{(\mathcal{N})}(k)$, which we expect it to be bounded by the values in Eqs.~(\ref{eqn:E-min-sup}) and~(\ref{eqn:E-max-inf}). 
To explicitly verify this, we decompose $v^{(\mathcal{N})}(k)$ into the complete basis of orthonormal eigenstates $\left\{u_a(k)\right\}_{a=1}^{N_\mathrm{b}}$ of $h(k)$ at each $k$, i.e.,
\begin{equation}
v^{(\mathcal{N})}(k) =  \sum_{a=1}^{N_\mathrm{b}} c_a^{(\mathcal{N})}(k) u_a(k).
\end{equation}
From Eq.~(\ref{eqn:norm-v-N}) we obtain a constraint on the coefficient functions $c_a(k)$, namely
\begin{align}
1 &= \norm{v^{(\mathcal{N})}}^2 = \frac{1}{2\pi}\int_{-\pi}^{+\pi} d k \sum_{a,a'=1}^{N_\mathrm{b}} [c_a^{(\mathcal{N})}(k)]^* c_{a'}^{(\mathcal{N})}(k) u_a^\dagger(k)\cdot u_{a'}(k)\nonumber \\
&= \frac{1}{2\pi} \int_{-\pi}^{+\pi} \sum_{a=1}^{N_\mathrm{b}}\abs{c_a^{(\mathcal{N})}(k)}^2.\label{eqn:v-state-decomp}
\end{align}
We therefore manipulate the last expression in Eq.~(\ref{eqn:bound-on-norm}) as
\begin{align}
(\mathrm{\ref{eqn:bound-on-norm}}) &= \frac{1}{2\pi} \int_{-\pi}^{+\pi}d k \sum_{a,a'=1}^{N_\mathrm{b}}  [c_a^{(\mathcal{N})}(k)]^* c_{a'}^{(\mathcal{N})}(k) u_a^\dagger(k)\cdot h(k) \cdot u_{a'}(k)\nonumber \\
&= \frac{1}{2\pi} \int_{-\pi}^{+\pi}d k \sum_{a=1}^{N_\mathrm{b}} \abs{c_a^{(\mathcal{N})}(k)}^2 E_a(k)\nonumber \\
&\stackrel{\mathrm{(\ref{eqn:E-min-sup})}}{\geq} \frac{1}{2\pi} \int_{-\pi}^{+\pi}d k \sum_{a=1}^{N_\mathrm{b}} \abs{c_a^{(\mathcal{N})}(k)}^2 E_\textrm{min} \stackrel{\mathrm{(\ref{eqn:v-state-decomp})}}{=} E_\textrm{min}.\label{eqn:lambda-bound}
\end{align}
Thus, Eqs.~(\ref{eqn:bound-on-norm}) and~(\ref{eqn:lambda-bound}) together constitute the first inequality in Eq.~(\ref{eqn:sharp-spectrum-bound}). 
The complementary part of Eq.~(\ref{eqn:sharp-spectrum-bound}) follows from utilizing in the last line of Eq.~(\ref{eqn:lambda-bound}) the inequality in Eq.~(\ref{eqn:E-max-inf}). 
This completes the proof of the spectral bound for any finite (but arbitrarily large) system size $\mathcal{N}$. 
\end{proof}
\smallskip

Before proceeding with alternative insights into the spectral bounds for block Toeplitz matrices, let us briefly clarify why the above proof does not generalize to a semi-infinite chain of unit cells.
The crucial complication is that for infinite-dimensional Hermitian operators there are eigenstates that do not have a well-defined $L^2$ norm. 
Examples include plane waves of an infinite chain of unit cells with no boundaries, and also the bulk states of a semi-infinite chain
which arise in the presently discussed context. 
As a consequence, the limit $\lim_{\mathcal{N}\to\infty} v^{(\mathcal{N})}(k)$ defined in Eq.~(\ref{eqn:block-Fourier-in-proof}) may be non-analytic, and the manipulations in Eqs.~(\ref{eqn:norm-v-N}--\ref{eqn:lambda-bound}) become ill-defined. 
However, if one assumes that $\psi_j$ in \q{eqn:block-Fourier-in-proof}   exponentially decays for large $|j|$, then one can take the limit ${\mathcal{N}\to\infty}$ while keeping $v(k)$ analytic, and the above proof does indeed carry forward to a semi-infinite chain. 
Exponential decay of wave functions in the `forbidden' energy range is best illustrated by the transfer-matrix formalism~\cite{Lee:1981}, which is elaborated next.

\subsection{Spectral bounds from transfer matrix formalism}\la{app:transfermatrix}

The following proof applies to a semi-infinite chain of unit cells with one edge, and manifests a new principle that is not present in the previous proof. 
Namely, suppose an energy eigenstate satisfies (\emph{i}) the sharp boundary condition, and (\emph{ii}) has energy outside the energy bounds of an infinite chain without edges [cf.~Eqs.~(\ref{eqn:E-min-sup}) and~(\ref{eqn:E-max-inf})]; then the eigenstate amplitude grows exponentially away from the chain edge, and is therefore not $L^2$-normalizable. 
We will demonstrate this principle, and also prove the spectral bound under sharp boundary conditions, with the simplifying assumption of a one-band Hamiltonian with nearest-neighbor hopping matrix elements.
  
Let us denote the wave function on a chain by $\psi_j$, with $j=1,2,\ldots$ taking any positive-integer values denoting the coordinate of each `site' on a chain.  
With only nearest-neighbor hoppings given by complex matrix element $b$, and an on-site potential with real matrix element $a$, the eigenproblem to solve is
\e{ \bar{b}\psi_{j-1}+b\psi_{j+1}= (E-a)\psi_j\la{seq}}
for $j \geq 2$, with the sharp boundary condition imposed by 
\e{ b\psi_2=(E-a)\psi_1.\la{sharp}}
Without loss of generality, one can (1) redefine the origin on the energy axis to remove $a$ from Eqs.~(\ref{seq}) and~(\ref{sharp}), then (2) perform a $j$-dependent phase transformation of $\psi_j$ such that $b$ becomes real, and finally (3) rescale the Hamiltonian by a multiplicative constant such that $b=1$. 
Then Eq.~(\ref{seq}) can be transformed to
\e{ \vectwo{\psi_{j+1}}{\psi_{j+2}}= 
\mathscr{T}\vectwo{\psi_{j-1}}{\psi_{j}}, \as \mathscr{T}(E) = \matrixtwo{-1}{E}{-E}{E^2-1},}
with $\mathscr{T}$ the \textit{transfer matrix} relating the wave functions across any four adjacent sites (not including the first two). 
The transfer matrix has eigenvalues
\e{ \lambda_{\pm}(E)=-1+\f{E^2}{2} \pm  \sqrt{ \f{E^4}{4}-E^2}, \as \lambda_+\lambda_-=1.}
The last equality follows from $\mathscr{T}$ having unit determinant, which is a general property of transfer matrices~\cite{Lee:1981}.

For $E^2<4$, $\lambda_{\pm}(E)=e^{\pm i 2k(E)}$ with $k(E)$ a real number determined by $E=2\cos(k)$. 
Generally, unimodular eigenvalues of the transfer matrix translate to Hamiltonian eigenstates whose wave functions (on two sites) are invariant (in amplitude) upon repeated application of the transfer matrix, meaning that these are extended, Bloch-wave states. 
For our model, $\lambda_{\pm}(E)$ is simply the phase factor acquired by translating a plane wave with wave number $k(E)$ by two sites.

For $E^2=4,$ the transfer matrix has an eigenvalue degeneracy: $\lambda_{\pm}=1$. 
Introducing a small parameter $\delta$ in  $E=\eta 2 +\delta/4$ with $\eta=\pm 1$, one finds $\lambda_{\pm}= 1\pm \sqrt{\eta \delta } +O(\delta).$ 
Thus tuning $\delta$ across zero (a branch point) changes $\lambda_+$ from a unimodular phase factor to a real number that is larger than one, and $\lambda_-$ into a real number less than one. 
In fact, since $(E^4/4-E^2)$ is a monotone-increasing function of $|E|$ for $E^2>4$, it follows that $\lambda_+$ is also  monotone-increasing in $|E|$ for $E^2>4$. Combining this fact with our constraint $\lambda_+\lambda_-=1$, one finds $\lambda_-$ is monotone-decreasing in $|E|$ but remains positive for any finite $E$. 
Thus it may be concluded for $E^2>4$ that $\lambda_+>1$ and $0\leq \lambda_-< 1$. 
For what follows, it is useful to record the left eigenvectors of the transfer matrix corresponding to $\lambda_\pm$:
\e{ \phi^L_{\pm}= \bigg(-\f{E^2}{2}\pm \sqrt{\frac{E^4}{4}-E^2},E \bigg),}
with the right eigenvectors $\phi^R_{\pm}$ satisfying the bi-orthonormality conditions: $\phi^L_{\pm}\cdot \phi^R_{\pm}=1$ and $\phi^L_{\mp}\cdot \phi^R_{\pm}=0$. 
(We use ``$\cdot$'' to denote the inner product of a two-element row vector with a two-element column vector.)

Let us study whether a wave function satisfying the sharp boundary condition can be normalized. 
The wave function on the first two sites is constrained by the sharp boundary condition of \q{sharp} as $\psi_2=E\psi_1$. 
Let us express the wave function on the first two sites as a linear combination of right eigenvectors of the transfer matrix:
\e{ \vectwo{\psi_1}{\psi_2} = c_+ \phi^R_++c_- \phi^R_-.}
The complex coefficient $c_+$ is determined by utilizing the bi-orthonormality condition:
\e{ c_+=\phi^L_+ \cdot \vectwo{\psi_1}{\psi_2} = \bigg(\f{E^2}{2}+ \f1{2}\sqrt{E^4-4E^2} \bigg)\,\psi_1 \neq 0.\la{overlapcoef}}
Since the bracketed term is strictly positive for $E^2>4$, one deduces $c_+ \neq 0$ ($c_-$ is obtained analogously, but its value is less important for our purposes.) 
It follows that applying the transfer matrix $r$ times to the wave function on the first two sites: 
\e{ \vectwo{\psi_{1+2r}}{\psi_{2+2r}}=\mathscr{T}^r\vectwo{\psi_1}{\psi_2}= c_+(\lambda_+)^r \phi_+^R+c_-(\lambda_-)^r \phi_-^R}
gives a wave function that is exponentially growing as $e^{r\text{ln}\lambda_+}$ as one moves away for the chain edge $(j=1)$. 
The desired conclusion, then, is that there exists no normalizable Hamiltonian eigenstate (under sharp boundary conditions) outside the energy interval $[-2,2]$, this interval being the energy band of extended Hamiltonian eigenstates for an infinite chain without edges. 

To briefly rationalize why this proof works, observe that the transfer-matrix eigenvalue $|\lambda_+|$ and coefficient $|c_+|$ [cf.~Eq.~\ref{overlapcoef}] being monotone-increasing functions of $|E|$ is special: it requires that there is a single dominant energy scale in the problem. We mean that the energy $E$ (where one considers a sharp boundary state) is dominant over all other relevant energy scales in the Hamiltonian; in this model, the hopping parameter $b=1$ gives the only other  (unit) energy scale to compare with $E$. (The on-site energy $a$ is irrelevant.) The existence of a single energy scale for sufficiently large $|E|$ is of course also true for Hamiltonians with more bands and longer-ranged hoppings, and for this reason we expect our transfer-matrix proof of spectral bounds will generalize. We leave this to future investigations.

Finally, it is worth remarking that with two or more energetically separated bulk bands, there are clearly multiple, competing energy scales when $E$ lies within a bulk band gap, and our argument breaks down, admitting the possibility of surface states under sharp boundary conditions.

\section{Incompatibility criterion for the Hopf invariant}\label{app:incompatibility}

We show that a nontrivial Hopf invariant is \textit{incompatible} with certain crystallographic spacetime symmetries, as summarized by the following criterion:\\

\noindent \textbf{Hopf incompatibility criterion:} Let $u(\bk)$ be a nondegenerate energy band of a tight-binding Hamiltonian with the symmetry of a symmorphic magnetic space group $\mathcal{G}$ in three spatial dimensions. Assuming $u(\bk)$ has trivial first Chern class on all planes, the Hopf invariant [cf.~\q{eq:hopfinvar}] of band $u(\bk)$ vanishes if $\mathcal{G}$ contains either of these elements:  (\emph{i}) an element that inverts the arrow of time but preserves the orientation of space, or (\emph{ii}) an element that preserves the arrow of time but inverts the orientation of space.\\ 

In the following we first provide an intuitive understanding of this criterion in Sec.~\ref{app:incompat-intuit}. After that we proceed to formally prove the criterion in Sec.~\ref{app:incomp-proof}.

\subsection{Intuitive understanding of the criterion}\label{app:incompat-intuit}

The incompatibility criterion for a nontrivial Hopf invariant can intuitively be understood in the following way. Since the Berry connection $\boldsymbol{\mathcal{A}}(\bk)$ forms part of the matrix elements of the position operator in the basis of Bloch functions \cite{Blount}, $\boldsymbol{\mathcal{A}}$ transforms under spatial symmetries like a vector, and is moreover even under time reversal. (There is a separate gauge-dependent component of $\boldsymbol{\mathcal{A}}$ that does not transform as a vector, but this component vanishes when integrated over the Brillouin zone, owing to the triviality of the first Chern class.) Likewise, the curvature $\bcalf(\bk)=\nabla_{\bk}\times \boldsymbol{\mathcal{A}}$ transforms like a pseudovector that is odd under time reversal~\cite{100page}. In combination, $ \bcalf \!\cdot\! \boldsymbol{\mathcal{A}}$ transforms  as a pseudoscalar that is odd under time reversal, and therefore its integral over the BZ vanishes. The Hopf incompatibility criterion is a strightforward consequence of these transformation properties.\\

In the following we provide some clarifying remarks:\medskip

\noindent (\emph{a})
We restrict the above criterion to \emph{symmorphic} magnetic space groups $\mathcal{G}$, because to our knowledge there exists no unit-rank band in a nonsymmorphic space group. (For a band whose projector is periodic and analytic in $\boldsymbol{k}$ throughout the BZ, the rank of a band    is the number of linearly independent Bloch functions at every $\boldsymbol{k}$, e.g., the Hopf-insulating conduction and valence bands each has unit rank.)  For nonsymmorphic space groups with a nonsymmorphic symmetry element (screw or glide), the non-existence of unit-rank bands has been proven.\footnote{For band representations, this has been proven by Michel and Zak \cite{connectivityMichelZak}; more generally, this has been proven in the Supplementary Material of Ref.~\cite{TBO_JHAA}.} 
For nonsymmorphic space groups without a nonsymmorphic element, the non-existence of unit-rank bands has not been generally proven, though this seems empirically to be true for a wide variety of space groups~\cite{watanabe2018, nonsymmsid, pofillingenforced}. \medskip

\noindent (\emph{b}) Possible magnetic space groups $\mathcal{G}$ which are not ruled out by the Hopf incompatibility criterion are restricted to Type I (colorless groups) derived from chiral non-magnetic space groups, and to Type III (black-and-white groups with ordinary Bravais lattice) with symmetry elements that include space inversion and time reversal only in combination. \medskip

\noindent (\emph{c}) The criterion is valid for both integer-spin and half-integer-spin energy bands. This suggests that a Hopf insulator can be sought in electronic systems with and without spin-orbit coupling as well as in bosonic metamaterials. \medskip

\noindent (\emph{d}) Finally, let us note that the presented proof of the incompatibility criterion explicitly assumes vanishing Cher numbers, and does not generalize to the Kennedy invariant of Hopf-Chern insulators. 

\subsection{Formal proof of the Hopf incompatibility criterion}\label{app:incomp-proof}

In order to proof the incompatibility criterion formulated at the beginning of this Appendix we assume that the tight-binding Hamiltonian has exponentially-decaying matrix elements in real space, and therefore the Fourier transform of the Hamiltonian is an analytic function of $\bk$ throughout the Brillouin zone (BZ).

To prove the criterion, we need to address the transformation of $\boldsymbol{\mathcal{A}}(\bk)$ under a symmetry $g\in \mathcal{G}$:
\e{
\boldsymbol{\mathcal{A}}(s_g\check g \bsk)= \check g \boldsymbol{\mathcal{A}}(\bsk)  + s_g\check g \nabla_{\bsk}\beta_g.
\label{eq:conntransform}
}
where $s_g=\pm 1$ indicates the reversal of time, and $\check g \in O(3)$ is the matrix representation of $g\in \mathcal{G}$ acting on spatial coordinates. The second term on the right-hand side of Eq.~(\ref{eq:conntransform}) is derived from the symmetry  constraint between tight-binding wave functions at symmetry-related momenta $\bsk \leftrightarrow s_g\check g \bsk$,
\e{ 
R_g u(\bsk) = e^{i\beta_g(\bsk)} u(s_g\check g \bsk),
\label{eq:sewingmatrix}
}
where $R_g$ is a matrix representation of $g$ that is antiunitary if $g$ inverts time, and unitary otherwise. Crucially, the $U(1)$ phase
$\beta_g(\bsk)$ in \q{eq:sewingmatrix}  encodes the change of the Wannier-center position due to $g$ \cite{nogo_AAJH} and cannot generally be made $\bsk$-independent \cite{hughes_inversionsymmetricTI,Chen_bulktopologicalinvariants}. On the other hand, the transformation of the curvature does not involve $\beta_g$ and is exactly as claimed in the intuitive argument from Sec.~\ref{app:incompat-intuit},
\e{ 
\bcalf(s_g \check g \bsk) = s_g \,\text{det}[\check g] \,\check g \bcalf(\bsk).
\label{eq:curvtransform}
}
Combining \q{eq:conntransform} with \q{eq:curvtransform} we derive that the Hopf invariant is proportional to
\e{ 
&\int d\bsk \bcalf \cdot \boldsymbol{\mathcal{A}}\big|_{\bsk} = \int d\bsk \bcalf \cdot \boldsymbol{\mathcal{A}}\big|_{s_g\check g\bsk} \lin
\eq s_g \,\text{det}[\check g] \,\int d\bsk \bcalf \cdot \boldsymbol{\mathcal{A}}\big|_{\bsk}  + \text{det}[\check g] \,\int d\bsk \bcalf \cdot \nabla_{\bsk}\beta_g,
\label{eq:hopfseparate}
}
where the first step corresponds to a change of coordinates, and in the second step we used the invariance of the dot product under rotations.

Next we would like to prove that the last term in this expression vanishes. Having trivial first Chern class means that $u(\bsk)$ can be made an analytic and periodic function of $\bsk$ throughout the BZ \cite{Panati_trivialityblochbundle}. Applying these properties to \q{eq:sewingmatrix}, we deduce that the phase factor $e^{i\beta_g(\bsk)}$  can also be made analytic and periodic. The phase $\beta_g$ is then also analytic, but can in principle wind by integer multiples of $2\pi$ if $\bsk$ is advanced by a reciprocal vector $\bG$. Without loss of generality, we may separate $\beta_g(\bsk) = \tilde{\beta}_g(\bsk)+\bsk \cdot \bR_g$ into a periodic component [$\tilde{\beta}_g(\bsk)=\tilde{\beta}_g(\bsk+\bG)$]   and a winding component, with  $\bR_g$ a $\bsk$-independent Bravais-lattice vector. 

The term in \q{eq:hopfseparate} that is proportional to $\tilde{\beta}_g$ may be expressed via Leibniz rule as 
\e{ 
\int d\bsk \bcalf \cdot \nabla_{\bsk}\tilde{\beta}_g=\int d\bsk \nabla_{\bsk}\cdot (\tilde{\beta}_g \bcalf)- \int d\bsk \beta_g \nabla_{\bsk}\cdot \bcalf; 
}
the first term on the right-hand side is seen to vanish by applying Stoke's theorem and the periodicity of both $\tilde{\beta}_g$ and $\bcalf$; the second term on the right-hand side vanishes because the curvature is generally divergence-free for a nondegenerate energy band of a tight-binding Hamiltonian (assumed to be analytic in $\bsk$). Similarly, the winding component of $\beta_g$ leads to a term in \q{eq:hopfseparate} that is proportional to $\int d\bsk \bcalf \cdot \bR_g$, with $\bR_g$ that is $\bsk$-independent.  This quantity vanishes because the first Chern number has been assumed to vanish on any 2D cut of the BZ. Eq.~\ref{eq:hopfseparate} therefore reduces to
\e{ 
\chi \eq s_g \,\text{det}[\check g] \chi, 
\label{eq:hopfseparate0}
}
implying that $\chi$ transforms as a pseudoscalar that is odd under time reversal, and therefore vanishes if $s_g \,\text{det}[\check g]=-1$ for some $g\in \mathcal{G}$. 

\newpage

\bibliography{bib_Apr2018}
\end{document}